\documentclass[fleqn]{amsart}

\let\olddiv\div

\usepackage{amsmath,amsfonts,amssymb,mathrsfs,physics,bm,mathtools,esint,tensor}

\usepackage{subcaption}

\usepackage[margin = 0.8in]{geometry}

\usepackage{parskip}

\usepackage{hyperref}

\allowdisplaybreaks

\usepackage{enumitem}
\setlist[itemize]{leftmargin = 0.3in}

\makeatletter
\renewcommand{\eqref}[1]{\textup{\eqreftagform@{\ref{#1}}}}
\let\eqreftagform@\tagform@
\def\tagform@#1{%
  \maketag@@@{\makebox[1sp][l]{\hspace{-2em}(\ignorespaces#1\unskip\@@italiccorr)}}%
}
\makeatother

\newtheorem{theorem}{Theorem}[section]
\newtheorem{corollary}[theorem]{Corollary} 
\newtheorem{proposition}[theorem]{Proposition} 
\newtheorem{lemma}[theorem]{Lemma} 

\theoremstyle{definition}
\newtheorem{definition}{Definition}[section] 
\theoremstyle{remark}
\newtheorem{remark}{Remark}[section] 

\numberwithin{equation}{section}

\newcommand{\dimMM}{\texttt{M}}
\newcommand{\dimLL}{\texttt{L}}
\newcommand{\dimTT}{\texttt{T}}
\newcommand{\dimone}{\texttt{[1]}}
\newcommand{\dimM}{\texttt{[Mass]}}
\newcommand{\dimL}{\texttt{[Length]}}
\newcommand{\dimT}{\texttt{[Time]}}
\newcommand{\dimQ}{\texttt{[Charge]}}
\newcommand{\dimE}{\texttt{[Energy]}}
\newcommand{\dimP}{\texttt{[Momentum]}}
\newcommand{\dimV}{\texttt{[Velocity]}}

\newcommand{\bbN}{\mathbb{N}}

\newcommand{\bbR}{\mathbb{R}}
\newcommand{\bbC}{\mathbb{C}}

\newcommand{\frake}{\mathfrak{e}}
\newcommand{\frakg}{\bm{\mathfrak{g}}}
\newcommand{\frakS}{\bm{\mathfrak{S}}}
\newcommand{\frakT}{\bm{\mathfrak{T}}}

\newcommand{\frakF}[1]{\mathcal{F}(#1)}
\newcommand{\frakX}[1]{\mathfrak{X}(#1)} 

\newcommand{\extprod}[2]{\bigwedge\nolimits^{#1}(#2)}

\newcommand{\Star}{{\star}}

\newcommand{\rng}[1]{\mathring{#1}}

\newcommand{\eps}{\varepsilon}

\newcommand{\rmd}{\mathrm{d}}

\newcommand{\sfg}{\mathsf{g}}
\newcommand{\sfh}{\mathsf{h}}

\DeclarePairedDelimiter\av{\lvert}{\rvert}
\DeclarePairedDelimiter\nn{\lVert}{\rVert}
\makeatletter
\let\oldav\av
\def\av{\@ifstar{\oldav}{\oldav*}}
\let\oldnn\nn
\def\norm{\@ifstar{\oldnn}{\oldnn*}}
\makeatother

\begin{document}

{\fontfamily{lmr}\selectfont

\title[GR solutions with finite energy for a point charge under BLTP electromagnetism]{General relativistic solutions with finite self-energy for a point charge under Bopp-Podolsky electromagnetism}

\author{Érik Amorim}

\address{Department Mathematik, Universität zu Köln, Cologne, Germany, eamorim@uni-koeln.de}



\date{August 2025}

\begin{abstract}
    We establish the existence of a family of static, spherically symmetric spacetimes that are solutions of the Einstein Field Equations of General Relativity coupled to the electric field of a static point charge obeying the equations of electromagnetism of Maxwell-Bopp-Landé-Thomas-Podolsky. The point charge is modeled as a naked singularity with non-positive bare mass. The singularity at the location of the charge is milder than that of the Reissner-Weyl-Nordström solution to the conventional Einstein-Maxwell equations, and, contrary to what happens for the latter, the electric-field energy of these solutions is finite.
\end{abstract}

\maketitle

\tableofcontents

\section{Introduction}
\label{sec_intro}

The problem of the \textbf{self-force} in electrodynamics consists of finding an expression for the force that the electromagnetic field generated by a charged particle exerts on the particle itself. The possibility to write well-posed classical systems of equations for the joint evolution of electromagnetic fields and their sources, without resorting to \textit{ad hoc} field averaging or bare-mass renormalization at point charges, requires working with EM (electromagnetism) laws whose associated field-density energy and momenta are integrable at the location of the point charges, which is not the case for the conventional Maxwell theory. A few generalized theories of electromagnetism have been proposed over the years specifically to address this problem, among them one originating in the 1940s that is often called \textbf{Bopp-Podolsky theory}, although it should rightfully be called \textbf{Bopp-Land\'{e}-Thomas-Podolsky theory} in honor of its original proponents Bopp~\cite{bopp}, Land\'{e} and Thomas~\cite{lande}, and Podolsky~\cite{podolsky}. We will refer to it as the \textbf{BLTP theory}.

Working in the Minkowski spacetime of SR (Special Relativity), and taking advantage of the linearity of the equations of BLTP theory, Kiessling and Tahvildar-Zadeh~\cite{kiessling_tahvildar-zadeh_SR} have shown how to formulate a well-posed system for the joint evolution of point particles and their BLTP electromagnetic fields. Naturally, one might now ask whether this result could be extended to the context of GR (General Relativity), with the point particles being modeled as naked singularities of a manifold that obeys the Einstein Field Equations with a source term that necessarily includes a contribution by the electromagnetic fields' energy and momenta densities. The present work is a natural first step towards this greater goal: The study of a universe containing a \textit{single}, \textit{static} point charge in this context.

More precisely: We rigorously establish the existence of a one-parameter family of finite-energy solutions to the Einstein equations for the spacetime of a static point charge whose electric field obeys the equations of BLTP theory and has a finite associated energy. The parameter in question is the so-called \textbf{bare mass} of the particle, namely the value of the spacetime's mass function at area-radial coordinate $r=0$, which (if finite) must be a non-positive number in order for the singularity at the location of the particle to be naked. These \textbf{E-M-BLTP (Einstein-Maxwell-BLTP) spacetimes} are to BLTP theory what the well-known RWN (Reissner-Weyl-Nordström) solution is to the conventional Maxwell theory; the RWN spacetimes, however, are to be considered \textit{too singular} at the location of the charge: Similarly to what happens in SR, the blow-up of the electric potential $\varphi(r)$ at $r=0$ causes the particle's electric-field energy to be infinite. Thus the E-M-BLTP spacetimes should be viewed as a remedy to the problem of Maxwell theory that a point charge modeled in GR as a naked singularity possesses infinite self-energy.

\begin{remark}
    There is another well-known modification of the Maxwell-Maxwell equations which also deals with the problem of infinite field energy of a point particle in flat space. Originally proposed by Born~\cite{born}, it is part of what nowadays is commonly called \textbf{Born-Infeld electrodynamics}. It was then first observed by Hoffmann~\cite{hoffmann} that, under this formulation of EM, the singularity of the static, spherically symmetric spacetime of a resting point charge is milder than that of the RWN spacetime, in the sense that the blowup of certain curvature scalars at $r=0$ is less severe in it. In~\cite{tahvildar-zadeh}, Tahvildar-Zadeh studies a class of static, spherically symmetric spacetimes that generalize that of Hoffmann with some of the same goals as ours here, but through fundamentally distinct methods: The metric coefficients $g_{\mu\nu}$ of the Hoffmann spacetimes can be explicitly solved by quadrature, partly due to the helpful fact that they satisfy the relation $g_{tt}g_{rr} = -1$, which will not be true for us.

    In~\cite{burtscher}, a weak form of the twice-contracted second Bianchi identity for static, spherically symmetric spacetimes with a ``singularity worldline'' is studied. Sufficient conditions for it to hold are found when $g_{tt}g_{rr} = -1$ and the bare mass is strictly negative. Considering how the classical Bianchi identities imply the local conservation laws $\nabla_\mu T^{\mu\nu} = 0$ (equations of motion of the matter/fields), the importance of this work is that it provides a rigorous formulation of an equation of motion for a spacetime singularity, such as a point charge in the way that it is modeled in our spacetime. However, we will find in the present work that the E-M-BLTP spacetimes satisfy $g_{tt}g_{rr} < -1$ at all points, which means that the work in~\cite{burtscher} is not directly applicable to study the motion of point-charge singularities in GR under the BLTP laws of electromagnetism, and thus more work in this topic will be needed when considering spacetimes containing two or more point charges that obey the BLTP laws.
\end{remark}

\begin{remark}
    No experimental results in physics have ever indicated that BLTP theory should be the correct classical theory of electromagnetism. The conventional Maxwell theory has stood the test of time, and it is generally assumed that quantum theories need to be taken into account for the resolution of problems such as a charged particle's apparently infinite self-energy. With that said, it bears noting that our present study is not motivated by direct application to the physical world, but rather as a theoretical toy that enables mathematically rigorous studies such as the formulation of well-posed equations of motion for charged naked singularities.
\end{remark}

\subsection{Layout of the paper}
\label{subsec_layout}

Section~\ref{sec_EM_SR} formulates both the conventional Maxwell theory and the BLTP theory of electromagnetism in the flat space of Special Relativity. Section~\ref{sec_EM_GR} explains how the equations of the previous section can be adapted to General Relativity. For completeness, the general-relativistic stress-energy tensor of BLTP theory is derived in details, and the system of Einstein equations for our case of interest of a static point charge is formulated with care, although a large portion of the calculations is relegated to the appendix. In Section~\ref{sec_main}, our main result is stated in Theorem~\ref{main_theorem}, with the rest of the section constituting its proof. It is a self-sufficient section containing all of the math in our work, and the reader can safely skip the physics from the previous sections if so inclined.

\subsection{Notations and conventions}
\label{subsec_notations}

Here we collect some basic definitions that will be needed throughout the paper, some of which are especially important to declare now because several different conventions exist for them in the literature. A good reference for all of the tensor formulas that we will introduce shortly is~\cite{straumann} --- but the reader is warned that our definition~\eqref{delta_def} of the \textit{codifferential} operator differs from the one in that work by a minus sign.

We will be working under the system of \textbf{Gaussian units}. The three fundamental physical dimensions that we need are $\dimMM = \dimM$, $\dimLL = \dimL$ and $\dimTT = \dimT$. To state the dimension of a quantity $x$ in terms of them, we use the self-explanatory notation
\begin{equation*}
    \mathrm{dim}(x) = \dimMM^a \, \dimLL^b \, \dimTT^c,
\end{equation*}
while the name of the dimension is written in square brackets. Of particular note is the fact that electric charge is a derived dimension in this system: Its basic unit, the \textbf{statcoulomb}, is defined as $1 \, \mathrm{statC} = 1 \, \mathrm{g}^{\frac12} \cdot \mathrm{cm}^{\frac32} \cdot \mathrm{s}^{-1}$, and therefore
\begin{equation*}
    \dimQ = \dimMM^{\frac12} \dimLL^{\frac32} \dimTT^{-1}.
\end{equation*}
Dimensionless quantities have dimension $\dimone$ by definition.

Einstein's summation convention is used: For any term featuring two copies of the index $\mu$ or any other Greek letter, with one of them being a lower index and the other an upper index, the sum $\sum_{\mu = 0}^3$ is implicit. Given variables $x^0,\ldots,x^3$, the partial derivative of a function $f = f(x^0,\ldots,x^3)$ with respect to $x^\mu$ is denoted by either $\partial_\mu f$ or $\partial f/(\partial x^\mu)$, and the index $\mu$ is considered to be a lower index in either of these symbols. We do not apply the summation convention to Latin indices --- sums over Latin indices are always explicitly written.

We will be working on an oriented 4-dimensional spacetime manifold $\mathcal{M}$ equipped with a Lorentzian metric $g$ of signature $({-}\,{+}\,{+}\,{+})$ and with its corresponding Levi-Civita connection $\nabla$. We denote the set of differentiable vector fields on $\mathcal{M}$ by $\frakX{\mathcal{M}}$, and that of differentiable $p$-forms (that is, totally antisymmetric tensor fields of rank $(0,p)$) by $\extprod{p}{\mathcal{M}}$. Consequently, $\extprod{0}{\mathcal{M}} =: \frakF{\mathcal{M}}$ and $\extprod{1}{\mathcal{M}} =: \mathfrak{X}^\ast(\mathcal{M})$ are the sets of differentiable scalar fields and covector fields on $\mathcal{M}$, respectively. More generally, the set of differentiable tensor fields of rank $(p,q)$ on $\mathcal{M}$ is denoted by $\mathcal{T}^p_q(\mathcal{M})$; in particular $\frakX{\mathcal{M}} = \mathcal{T}^1_0(\mathcal{M})$ and $\mathfrak{X}^\ast(\mathcal{M}) = \mathcal{T}^0_1(\mathcal{M})$. The canonical basis of vector fields corresponding to a coordinate system $(x^\mu)$ on $\mathcal{M}$ is denoted by $\{\partial_\mu: \ \mu=0,1,2,3\}$, while its dual basis of 1-forms is $\{\rmd x^\mu: \ \mu=0,1,2,3\}$. The components of the covariant derivative of a tensor field in the appropriate tensor basis formed from the canonical coordinate bases are denoted using the $\nabla_\ast$ symbol; for example for a tensor $U$ of rank $(2,1)$:
\begin{equation*}
    U = \tensor{U}{^{\mu\nu}_\lambda} \ \partial_\mu\otimes\partial_\nu\otimes\rmd x^\lambda \quad\Longrightarrow\quad \nabla U = \nabla_\xi\tensor{U}{^{\mu\nu}_\lambda} \ \rmd x^\xi\otimes\partial_\mu\otimes\partial_\nu\otimes\rmd x^\lambda.
\end{equation*}
The components of the metric $g\in\mathcal{T}^0_2(\mathcal{M})$ and of the inverse metric $g^{-1}\in\mathcal{T}^2_0(\mathcal{M})$ in a given coordinate basis are denoted respectively by $g_{\mu\nu}$ and $g^{\mu\nu}$, and they are used to lower and raise indices of tensor fields as usual. A handy property that will be repeatedly used in our tensor calculations is the fact that index contractions do not change if the upper and lower roles of the contracted indices are reversed; for example $\tensor{R}{^\alpha_\mu_\alpha_\nu} = \tensor{R}{_\alpha_\mu^\alpha_\nu}$. We also define the coordinate-dependent differential operator $\partial^\mu = g^{\mu\alpha}\partial_\alpha$. The symbol $\av{g}$ is used for $-\det\big( g_{\mu\nu}\big)$, a scalar which is also coordinate-dependent. There are no preferred coordinate systems on $\mathcal{M}$: Unless otherwise specified, if a formula makes explicit use of tensor components, it is because the formula is covariant, that is, it retains the same expression regardless of the coordinate system being used.

Einstein's equations with a stress-energy tensor $T\in\mathcal{T}^0_2(\mathcal{M})$ acting as source term are
\begin{equation} \label{EFE}
    R_{\mu\nu} - \frac{R}{2}g_{\mu\nu} = \frac{8\pi G}{c^4} T_{\mu\nu}
\end{equation}
for the Ricci scalar, Ricci tensor, Riemann tensor and Christoffel symbols
\begin{equation} \label{Ricci_Christoffel}
    R = \tensor{R}{^\mu_\mu}, \ R_{\mu\nu} = \tensor{R}{^\alpha_{\mu\alpha\nu}}, \ \tensor{R}{^\mu_{\alpha\beta\gamma}} = \partial_\beta\Gamma^\mu_{\alpha\gamma} - \partial_\gamma \Gamma^\mu_{\alpha\beta} + \Gamma^\mu_{\sigma\beta} \Gamma^\sigma_{\gamma\alpha} - \Gamma^\mu_{\sigma\gamma} \Gamma^\sigma_{\beta\alpha}, \ \Gamma^\alpha_{\mu\nu} = \frac{g^{\alpha\lambda}}{2}\big( \partial_\mu g_{\nu\lambda} + \partial_\nu g_{\mu\lambda} - \partial_\lambda g_{\mu\nu} \big).
\end{equation}

Our convention for the normalization of the wedge product between basis 1-forms is that
\begin{equation*}
    \rmd x^{\mu_1}\wedge\cdots\wedge\rmd x^{\mu_p} = \sum_{\sigma\in S_p} \mathrm{sgn}(\sigma) \ \rmd x^{\mu_{\sigma(1)}}\otimes\cdots\otimes\rmd x^{\mu_{\sigma(p)}}.
\end{equation*}
In particular, a general 2-form $\omega$ is written in the basis $\{ \rmd x^\mu\wedge\rmd x^\nu: \ \mu<\nu \}$ using the same components that appear in its expression $\omega_{\mu\nu} \ \rmd x^\mu\otimes\rmd x^\nu$ as a tensor:
\begin{equation*}
    \omega = \sum_{j=1}^3 \omega_{0j} \ \rmd x^0\wedge\rmd x^j + \omega_{12} \ \rmd x^1\wedge\rmd x^2 + \omega_{13} \ \rmd x^1\wedge\rmd x^3 + \omega_{23} \ \rmd x^2\wedge\rmd x^3.
\end{equation*}
But we will also display the components of such a differential form as an antisymmetric matrix:
\begin{equation*}
    (\omega_{\mu\nu}) = \left(
        \begin{array}{rrrr}
            0 & \omega_{01} & \omega_{02} & \omega_{03} \\
            -\omega_{01} & 0 & \omega_{12} & \omega_{13} \\
            -\omega_{02} & -\omega_{12} & 0 & \omega_{23} \\
            -\omega_{03} & -\omega_{13} & -\omega_{23} & 0
        \end{array}
    \right).
\end{equation*}
A tensor product of the form $\rmd f\otimes\rmd f$ is denoted by $\rmd f^2$. This confusing notation is standard when specifying a metric tensor in a coordinate system.

The scalar product on the space of 1-forms, that is, the map defined by the inverse metric
\begin{equation*}
    g^{-1}: \extprod{1}{\mathcal{M}}\times\extprod{1}{\mathcal{M}} \longrightarrow \frakF{\mathcal{M}}, \quad g^{-1}( \psi,\omega ) = g^{\mu\nu} \psi_\mu\omega_\nu,
\end{equation*}
generalizes to scalar products $\langle \cdot,\cdot\rangle_g: \extprod{p}{\mathcal{M}}\times\extprod{p}{\mathcal{M}} \longrightarrow \frakF{\mathcal{M}}$ for all $p$: If $p=0$ this is simply defined as $\langle f,h\rangle_g := fh$, while if $p>0$ it is uniquely defined by its action on blades:
\begin{equation} \label{blades_inner_product}
    \big\langle \psi^{(1)}\wedge\cdots\wedge\psi^{(p)} , \omega^{(1)}\wedge\cdots\wedge\omega^{(p)} \big\rangle_g := \det\big( g^{\mu\nu} \psi^{(j)}_\mu \omega^{(k)}_\nu \big)_{j,k=1,\ldots,p}, \quad \psi^{(j)},\omega^{(k)}\in\extprod{1}{\mathcal{M}}.
\end{equation}
This yields the formula
\begin{equation} \label{definition_product_pforms}
    \big\langle \psi,\omega \big\rangle_g = \frac{1}{p!} \psi_{\mu_1\cdots\mu_p}\omega^{\mu_1\cdots\mu_p}, \quad \psi,\omega\in\extprod{p}{\mathcal{M}}.
\end{equation}
Then, for each $p=0,\ldots,4$, the Hodge star $\Star: \extprod{p}{\mathcal{M}} \longrightarrow \extprod{4-p}{\mathcal{M}}$ is defined as the unique $\frakF{\mathcal{M}}$-linear map satisfying
\begin{equation} \label{definition_hodge}
    \psi \wedge (\Star\omega) = \big\langle \psi, \omega \big\rangle_g \ \mathrm{vol}_g, \quad  \psi,\omega\in\extprod{p}{\mathcal{M}},
\end{equation}
where
\begin{equation*}
    \mathrm{vol}_g = \sqrt{\av{g}} \ \rmd x^0\wedge\cdots\wedge\rmd x^3 \quad\text{for positively oriented coordinate systems } (x^\mu)
\end{equation*}
is the canonical volume form defined by the metric $g$. We note some useful properties of $\Star$:
\begin{equation} \label{properties_hodge}
    (\Star\psi)\wedge(\Star\omega) = -\psi\wedge\omega, \quad \Star\Star\psi = \left\{\begin{array}{rl}
        \psi &\text{ if } \psi\in\extprod{p}{M} \text{ for odd } p, \\
        -\psi &\text{ if } \psi\in\extprod{p}{M} \text{ for even } p,
    \end{array}\right. \quad \Star 1 = \mathrm{vol}_g.
\end{equation}
If $g = \eta$ is the Minkowski metric and a system of inertial coordinates $\{x^\mu\}$ is used (see~\eqref{mink_mettric} ahead), one checks that this definition of the Hodge-star operator yields
\begin{equation*}
    \hspace{-1em}\begin{array}{rclcrcl}
        \Star (1) &=& \phantom{-} \rmd x^0\wedge\rmd x^1\wedge\rmd x^2\wedge\rmd x^3, &\,& \Star (\rmd x^0\wedge\rmd x^1\wedge\rmd x^2\wedge\rmd x^3) &=& -1 \\
        \Star (\rmd x^0) &=& - \rmd x^1\wedge\rmd x^2\wedge\rmd x^3, &\,& \Star (\rmd x^1\wedge\rmd x^2\wedge\rmd x^3) &=& - \rmd x^0 \\
        \Star (\rmd x^1) &=& - \rmd x^0\wedge\rmd x^2\wedge\rmd x^3, &\,& \Star (\rmd x^0\wedge\rmd x^2\wedge\rmd x^3) &=& - \rmd x^1 \\
        \Star (\rmd x^2) &=& \phantom{-} \rmd x^0\wedge\rmd x^1\wedge\rmd x^3, &\,& \Star (\rmd x^0\wedge\rmd x^1\wedge\rmd x^3) &=& \phantom{-} \rmd x^2 \\
        \Star (\rmd x^3) &=& - \rmd x^0\wedge\rmd x^1\wedge\rmd x^2, &\,& \Star (\rmd x^0\wedge\rmd x^1\wedge\rmd x^2) &=& - \rmd x^3 \\
        \Star (\rmd x^0\wedge\rmd x^1) &=& - \rmd x^2\wedge\rmd x^3, &\,& \Star (\rmd x^2\wedge\rmd x^3) &=& \phantom{-} \rmd x^0\wedge\rmd x^1 \\
        \Star (\rmd x^0\wedge\rmd x^2) &=& \phantom{-} \rmd x^1\wedge\rmd x^3, &\,& \Star (\rmd x^1\wedge\rmd x^3) &=& - \rmd x^0\wedge\rmd x^2 \\
        \Star (\rmd x^0\wedge\rmd x^3) &=& - \rmd x^1\wedge\rmd x^2, &\,& \Star (\rmd x^1\wedge\rmd x^2) &=& \phantom{-} \rmd x^0\wedge\rmd x^3 \\
    \end{array}
\end{equation*}

The codifferential $\delta: \extprod{p}{\mathcal{M}}\longrightarrow\extprod{p-1}{\mathcal{M}}$ is the operator
\begin{equation} \label{delta_def}
    \delta = \Star \circ \rmd \circ \Star = \Star\rmd\Star
\end{equation}
(note that the composition sign $\circ$ is usually omitted when composing operators that act on differential forms), where $\rmd:\extprod{p}{\mathcal{M}}\longrightarrow\extprod{p+1}{\mathcal{M}}$ is the exterior derivative. We remark that a more general formula than~\eqref{delta_def} is normally taken as the definition of $\delta$, but it reduces to our formula as above in the case of a 4-dimensional manifold with a $({-}\,{+}\,{+}\,{+})$ metric. The following expressions for the operators $\rmd$ and $\delta$ are valid in any coordinate system and will be extensively used: For $\psi\in\extprod{p}{\mathcal{M}}$,
\begin{equation} \label{d_coordinates}
    (\rmd\psi)_{\mu_1\cdots\mu_{p+1}} = \sum_{j=1}^{p+1} (-1)^{j+1} \partial_{\mu_j}\psi_{\mu_1\cdots\mu_{j-1}\mu_{j+1}\cdots\mu_p},
\end{equation}
\begin{equation} \label{delta_coordinates}
    (\delta \psi)^{\mu_1\cdots\mu_{p-1}} = -\frac{1}{\sqrt{\av{g}}} \ \partial_\lambda \big( \sqrt{\av{g}} \psi^{\lambda\mu_1\cdots\mu_{p-1}} \big).
\end{equation}
Note that the indices of $\psi$ appearing in~\eqref{delta_coordinates} are raised, highlighting the fact that $\delta$ depends on the metric, as is clear from its definition~\eqref{delta_def} in terms of the Hodge star. Formulas~\eqref{d_coordinates} and~\eqref{delta_coordinates} are remarkable because, although they feature \textit{partial} derivatives as opposed to \textit{covariant} ones, they still manage to be valid in any coordinate system --- in fact they are equivalent to
\begin{equation} \label{ddelta_cov}
    (\rmd\psi)_{\mu_1\cdots\mu_{p+1}} = \sum_{j=1}^{p+1} (-1)^{j+1} \nabla_{\mu_j}\psi_{\mu_1\cdots\mu_{j-1}\mu_{j+1}\cdots\mu_p}, \quad (\delta \psi)^{\mu_1\cdots\mu_{p-1}} = -\nabla_\lambda \psi^{\lambda\mu_1\cdots\mu_{p-1}}.
\end{equation}
But, contrary to what happens for~\eqref{ddelta_cov}, the indices of formulas~\eqref{d_coordinates} and~\eqref{delta_coordinates} cannot be freely raised/lowered on both sides of the equation as in manifestly tensorial equations (unless $g = \eta$ is the Minkowski metric), so care must be taken when applying them on a general manifold.

Finally, for a function $f:(a,\infty)\longrightarrow\bbR$, where $a\in\bbR$, we often denote the limits $\lim_{r\to a^+} f(r)$ and $\lim_{r\to\infty} f(r)$ (should they exist) by the symbols $f(a)$ and $f(\infty)$, respectively. We also make the following definition, to be used throughout Section~\ref{sec_main}:
\begin{definition} \label{def_equivalent}
    Let $a\in\bbR$. For two strictly positive functions $f,g:(a,\infty)\longrightarrow\bbR$, we say that $g$ is \textbf{bracketed} by $f$ and write $f\simeq g$ (it is an equivalence relation) when there exist constants $C_1,C_2>0$ such that
    \begin{equation*}
        C_1 f(r) \leq g(r) \leq C_2 f(r) \quad\text{for all } r>a.
    \end{equation*}
\end{definition}
If $f$ and $g$ in the above definition are continuous, a sufficient condition for $f\simeq g$ is that the quotient $f/g$ attain finite, non-zero limits as $r\to a^+$ and $r\to\infty$. Indeed, the existence of these limits implies that the continuous, strictly positive functions $f/g$ and $g/f$ remain bounded above in two intervals of the form $(a,a+\eps)$ and $(N,\infty)$, while continuity also yields finite bounds for them in $[a+\eps,N]$.

\section{Electromagnetism in Special Relativity}
\label{sec_EM_SR}

In this section we define on Minkowski spacetime the theories of electromagnetism that we wish to study, and we explain why BLTP theory gives rise to better-behaved fields and potentials than the conventional Maxwell theory, at least in the case of a single, static point charge.

\subsection{The Maxwell and BLTP theories of electromagnetism}
\label{subsec_SR_intro}

Minkowski spacetime is by definition the manifold $\bbR^4$ endowed with a Lorentzian metric $\eta$ which, in a suitable coordinate system $(x^0,x^1,x^2,x^3)$ (called \textbf{inertial}), takes the form
\begin{equation} \label{mink_mettric}
    \eta = - \rmd x^0 \otimes\rmd x^0 + \sum_{j=1}^3 \rmd x^j \otimes\rmd x^j.
\end{equation}
All four coordinates are assumed to have a dimension of $\dimL$. The last three, $\bm{x} = (x^1,x^2,x^3)$, are interpreted as giving a spatial location, while $x^0/c$ (where $c$ is the speed of light) has the physical meaning of time as perceived by an inertial observer. Given a differentiable function $f:\bbR^4\longrightarrow\bbR$ of the variables $(x^0,\ldots,x^3)$, we often abuse notation and write $(1/c)\partial_t f$ in place of $\partial_0 f$ for its partial derivative with respect to the variable $x^0$.

Now let a \textbf{charge-density} scalar field $\rho: \bbR^4\longrightarrow\bbR$, with dimension $\dimQ \dimL^{-3} = \dimMM^{\frac12}\dimLL^{-\frac32}\dimTT^{-1}$, and a \textbf{current-density} vector field $\bm{j}:\bbR^4\longrightarrow\bbR^3$, whose components $j_1,j_2,j_3:\bbR^4\longrightarrow\bbR$ are measured with dimension $\dimQ \dimV \dimL^{-3} = \dimMM^{\frac12}\dimLL^{-\frac12}\dimTT^{-2}$, be given satisfying the \textbf{continuity equation}
\begin{equation} \label{continuity}
    \frac{\partial\rho}{\partial t} + \nabla\cdot\bm{j} = 0,
\end{equation}
where the divergence operator $\nabla\cdot$ uses derivatives with respect to the three spatial coordinates --- this will also be the case for gradients $\nabla$, curls $\nabla\times$ and Laplacians $\Delta$ in what follows. We make no assumption on the regularity of $\rho$, $\bm{j}$ or any of the other fields described in this section; our goal is simply to formally state equations that are applicable to sufficiently regular unknowns.

\textbf{Maxwell's equations} in their most general form, often called the \textbf{macroscopic} form, stipulate that the distribution of charges and currents modeled by $\rho$ and $\bm{j}$ generate EM fields $\bm{E},\bm{D},\bm{B},\bm{H}: \bbR^4\longrightarrow\bbR^3$, all of whose components having the dimension $\dimMM^{\frac12}\dimLL^{-\frac12}\dimTT^{-1}$, that satisfy the systems
\begin{equation} \label{Maxwell}
    \left\{\begin{array}{l}
    \displaystyle\frac{1}{c}\displaystyle\frac{\partial\bm{B}}{\partial t} + \nabla\times\bm{E} = \bm{0}, \\ \\
    \nabla\cdot\bm{B} = 0,
\end{array}\right. \qquad \left\{\begin{array}{l}
    \displaystyle\frac{1}{c}\displaystyle\frac{\partial\bm{D}}{\partial t} - \nabla\times\bm{H} = -\displaystyle\frac{4\pi}{c}\bm{j}, \\ \\
    \nabla\cdot\bm{D} = 4\pi\rho.
\end{array}\right.
\end{equation}
The fields $\bm{E},\bm{D},\bm{B},\bm{H}$ are respectively known as the \textbf{electric}, \textbf{electric displacement}, \textbf{magnetic} and \textbf{magnetizing} fields. The continuity equation~\eqref{continuity} is a compatibility condition for this system, as it follows immediately from the two non-homogeneous equations. The homogeneous ones are equivalent, at least locally, to the existence of a \textbf{scalar potential} $\varphi$ and a \textbf{vector potential} $\bm{A}$ --- which however are not uniquely defined --- such that
\begin{equation} \label{field_potentials}
    \bm{B} = \nabla\times\bm{A}, \quad \bm{E} = -\nabla\varphi - \frac{1}{c}\frac{\partial\bm{A}}{\partial t}.
\end{equation}
The Maxwell equations are often considered to be a system for the unknowns $\varphi$ and $\bm{A}$. We also note their dimensions:
\begin{equation*}
    \dim(\varphi) = \dim(A_j) = \dimMM^{\frac12}\dimLL^{\frac12}\dimTT^{-1}, \quad \text{where } \bm{A} = (A_1,A_2,A_3).
\end{equation*}
However, these equations alone are clearly not enough to determine the fields from their sources. A so-called \textbf{vacuum law}, also commonly referred to by the name \textbf{constitutive relations}, must be stipulated: one relation between the electric fields $\bm{E}$ and $\bm{D}$, and another one between the magnetic fields $\bm{B}$ and $\bm{H}$. We call the relations
\begin{equation} \label{Maxwell_vacuum}
    \bm{D} = \bm{E}, \quad \bm{H} = \bm{B}
\end{equation}
\textbf{Maxwell's vacuum law} and label the set of equations~\eqref{Maxwell} and~\eqref{Maxwell_vacuum} the \textbf{Maxwell-Maxwell equations}, which form a system of linear, second-order PDEs (that is, in the unknowns $\varphi$ and $\bm{A}$) that defines the conventional Maxwell theory. In turn, BLTP theory is defined by imposing the \textbf{BLTP vacuum law}
\begin{equation} \label{BLTP_vacuum}
    \bm{D} = \bm{E} - \varkappa^{-2}\square\bm{E}, \quad \bm{H} = \bm{B} - \varkappa^{-2}\square\bm{B}.
\end{equation}
Here $\varkappa > 0$ is a postulated constant with dimension $\dimL^{-1}$, while
\begin{equation*}
    \square = \partial_\mu \partial^\mu = \eta^{\mu\nu}\partial_\mu\partial_\nu = -\frac{1}{c^2}\frac{\partial^2}{\partial t^2} + \Delta
\end{equation*}
is D'Alembert's operator, acting componentwise on vector fields. Then~\eqref{Maxwell} and~\eqref{BLTP_vacuum} are called the \textbf{Maxwell-BLTP equations}, a system of fourth-order linear PDEs for $\varphi$ and $\bm{A}$. See Remark~\ref{rem_BLTP_generalization} ahead for an explanation on why BLTP theory can be viewed as a \textit{natural} modification of the conventional Maxwell theory.

\begin{remark}
    The constant featuring in the equations of BLTP theory is not $\varkappa$ itself, but rather its inverse $\varkappa^{-1}$ (or powers thereof). If $\varkappa^{-1}$ is set to 0, we see from~\eqref{BLTP_vacuum} that the equations reduce to the Maxwell-Maxwell system. Hence, in everything that follows, we will be able to treat both theories under the same set of equations, with the understanding that the Maxwell-Maxwell equations are obtained by setting $\varkappa = \infty$, that is, by removing from the equations all terms accompanied by a positive power of $\varkappa^{-1}$.
\end{remark}

\subsection{Field energy and momenta} \label{subsec_SR_energy}

The fields $\bm{E}$ and $\bm{B}$ determine the motion of test particles by acting on them via the Lorentz force: Let a test particle of rest mass $m$ and charge $q$ be observed in the coordinate system $(x^\mu)$ to be moving with velocity vector $\bm{v} = (v_1,v_2,v_3)$, kinetic energy $\mathcal{E} = mc^2\gamma$ and linear momentum $\bm{p} = mc\gamma\bm{v}$, where $\gamma = \big(1-(v/c)\big)^{-\frac12}$ for $v = \sqrt{(v_1)^2+(v_2)^2+(v_3)^2}$. Then the particle's law of motion under the influence of the EM fields is the statement that the rate of change of its momentum is
\begin{equation} \label{dpdt}
    \frac{\rmd\bm{p}}{\rmd t} = \bm{F} := q\bm{E} + \frac{1}{c} (q\bm{v})\times\bm{B},
\end{equation}
where $\times$ denotes the vector cross product. In particular, the power imparted by the EM fields to the particle is
\begin{equation} \label{dEdt}
    \frac{\rmd\mathcal{E}}{\rmd t} = \bm{F}\cdot\bm{v} = \bm{E}\cdot (q\bm{v}).
\end{equation}
With the help of~\eqref{dpdt} and~\eqref{dEdt}, it is now possible to obtain expressions (in any general EM theory) for the following observer-dependent quantities that measure, at each event $(x^0,\ldots,x^3)$, how much energy and momentum are contained in the EM fields and where they are flowing into:
\begin{itemize}
    \item The EM fields' \textbf{energy density}: a scalar field $\mathfrak{e}$ measuring the volumetric density of energy contained in them. Its dimension is $\dimE \dimL^{-3} = \dimMM \dimLL^{-1} \dimTT^{-2}$.
    \item The EM fields' \textbf{momentum density}: a vector field $\bm{\mathfrak{g}}$ measuring the volumetric density of momentum contained in them. In particular, for each $j=1,2,3$, its component $g_j$ in $(x^1,x^2,x^3)$ coordinates is a scalar measuring the density of field momentum in the direction of $x^j$, with dimension $\dimP \dimL^{-3} = \dimMM \dimLL^{-2} \dimTT^{-1}$.
    \item The EM fields' \textbf{energy flux}: a vector field $\bm{\mathfrak{S}}$, pointing in the direction where energy flows, whose magnitude measures the amount of energy per unit time and unit area that crosses a small hypersurface orthogonal to its flow direction. It is known as the \textbf{Poynting vector}. In particular, for each $k=1,2,3$, its component $\mathfrak{S}_k$ in $(x^1,x^2,x^3)$ coordinates is a scalar measuring the rate at which field energy crosses hypersurfaces normal to the direction of $x^k$, towards the positive direction, with dimension $\dimE \dimL^{-2}\dimT^{-1} = \dimMM\dimTT^{-3}$.
    \item The EM fields' \textbf{momentum flux}: a $(3\times 3)$-matrix field $\bm{\mathfrak{T}}$ whose component ${\mathfrak{T}}_{jk}$ in $(x^1,x^2,x^3)$ coordinates is a scalar measuring the rate at which the $j$-th component of field momentum crosses hypersurfaces normal to the direction of $x^k$, towards the positive direction. The dimension of each component of $\bm{\mathfrak{T}}$ is $\dimP \dimL^{-2} \dimT^{-1} = \dimMM \dimLL^{-1} \dimTT^{-2}$.
\end{itemize}
For the sake of completeness of our work, we wish to derive from scratch an expression for $\frake$ in BLTP theory. So let us briefly explain how $\frake$, $\frakg$, $\frakS$ and $\frakT$ are derived for any general EM theory. Suppose that the fields act back on their source charge distribution modeled by $\rho$ and $\bm{j}$ via the Lorentz force. If $V\subset\bbR^3$ is an arbitrary bounded open set with a smooth boundary $\partial V$, then the rate of change of field energy contained inside it, $\iiint_V \partial_t\frake \ \rmd^3\bm{x}$, is brought about by two processes: field energy entering the boundary of $V$ at a rate of $-\oiint_{\partial V} \frakS\cdot\rmd S = -\iiint_V \nabla\cdot\frakS \ \rmd^3\bm{x}$, and field energy being lost into work done by the Lorentz force on the sources inside $V$ at a rate of $-\iiint_V \bm{E}\cdot\bm{j} \ \rmd^3\bm{x}$. For this last expression, note that the \textit{density} of work is found from~\eqref{dEdt} as $\bm{E}\cdot\bm{j}$, since the volumetric density of $q\bm{v}$ is by definition the current density $\bm{j}$. Now energy conservation in $V$ dictates a relation between these integrals which, given the arbitrariness of $V$, yields the conservation law
\begin{equation} \label{field_energy_conservation}
    \frac{\partial \frake}{\partial t} + \nabla\cdot\frakS = - \bm{E}\cdot\bm{j}.
\end{equation}
Proceeding similarly for the components of field momentum --- note that~\eqref{dpdt} with $q$ and $q\bm{v}$ replaced by $\rho$ and $\bm{j}$, respectively, is the \textit{density} of momentum being transferred from the fields to the sources, --- we find the equations of momentum conservation:
\begin{equation} \label{field_momentum_conservation}
    \frac{\partial \mathfrak{g}_j}{\partial t} + \nabla\cdot\bm{\mathfrak{T}_j} = -\rho\bm{E} - \frac{1}{c}\bm{j}\times\bm{B}, \quad j=1,2,3,
\end{equation}
where $\bm{\mathfrak{T}_j} = (\mathfrak{T}_{j1},\mathfrak{T}_{j2},\mathfrak{T}_{j3})$. It is at this point that the equations of the EM theory of our choice must be used: First the non-homogeneous Maxwell equations~\eqref{Maxwell} are used to replace $\rho$ and $\bm{j}$ on the right side of~\eqref{field_energy_conservation} and~\eqref{field_momentum_conservation} by expressions involving all four EM fields, which are then rewritten in terms of only the fields $\bm{E}$ and $\bm{B}$ by applying the vacuum law~\eqref{BLTP_vacuum}. Finally, using the homogeneous Maxwell equations~\eqref{Maxwell} and some lengthy algebra, these right sides can be cast in the form $\partial_t(\ast) + \nabla\cdot(\bm{\ast})$, where $(\ast)$ and $(\bm{\ast})$ stand for some particular scalar and vector, which are then \textit{defined} to be $\frake$ and $\frakS$ when this procedure is applied to~\eqref{field_energy_conservation}, or $\mathfrak{g}_j$ and $\bm{\mathfrak{T}}_j$ when applied to~\eqref{field_momentum_conservation}. In the $\varkappa = \infty$ case, this calculation can be found in many textbooks on electromagnetism (see~\cite{jackson}, p.~189). We wish to carry it out in the $0<\varkappa<\infty$ case, but not now; we postpone it to the end of this section, where the language of tensor fields makes the task somewhat simpler. For now, we present the result that will be found for $\frake$ (we will not need to consider $\frakg$, $\frakS$ or $\frakT$ in this paper):
\begin{align}
    4\pi \frake &= \frac{\av{\bm{E}}^2+\av{\bm{B}}^2}{2} + \varkappa^{-2} \left( -\bm{E}\cdot\square\bm{E} -\bm{B}\cdot\square\bm{B} - \frac{(\nabla\cdot\bm{E})^2}{2} - \frac{1}{2}\left| \nabla\times\bm{B} - \frac{1}{c}\frac{\partial \bm{E}}{\partial t} \right|^2 \right) \label{e_SR} \\
    &= \bm{B}\cdot\bm{H} + \bm{E}\cdot\bm{D} - \frac{|\bm{E}|^2+|\bm{B}|^2}{2} - \frac{\varkappa^{-2}}{2}\left( (\nabla\cdot\bm{E})^2 + \left| \nabla\times\bm{B} - \frac{1}{c}\frac{\partial \bm{E}}{\partial t} \right|^2 \right). \label{e_SR_aux}
\end{align}
This is the same formula that is stated, for example, in~\cite{bopp} and~\cite{podolsky}.

\begin{remark}
    Given that there are infinitely many possible expressions for $(\ast)$ and $(\bm{\ast})$ preserving the value of $\partial_t(\ast) + \nabla\cdot(\bm{\ast})$, there are also infinitely many possible definitions for the quantities $\frake$, $\frakg$, $\frakS$ and $\frakT$. This fact remains true even if one restricts the search to a \textit{symmetric} $\frakT$ as well as to quantities $\frakg$ and $\frakS$ related by $\frakS = c^2\frakg$, which are physically motivated impositions that hold for matter models in Classical Mechanics. Hence there is no definitive answer to the question of how the field energy and momenta are defined in BLTP theory or even conventional Maxwell theory. One generally settles for expressions that, in some way, feel the simplest or most natural to derive.
\end{remark}

\subsection{SR solutions for a static point charge}
\label{subsec_SR_pointcharge}

Now we would like to restrict the EM equations to the case of a single static point of charge $Q \neq 0$ sitting at position $(x^1,x^2,x^3) = \bm{0}$ at all times $t$. The goal is to write down the fields that solve both the Maxwell-Maxwell and the Maxwell-BLTP systems corresponding to this particle's associated densities, which are
\begin{equation} \label{static_rho_j}
    \rho = Q\delta_{\bm{0}}, \quad \bm{j} = \bm{0},
\end{equation}
where $\delta_{\bm{0}}:\bbR^3\longrightarrow\bbR$ is Dirac's Delta concentrated at $\bm{0}\in\bbR^3$, formally viewed as a generalized function. Anticipating that the fields and their energy will be constant in time and only $\bm{E}$ and $\bm{D}$ will be nonzero, we see that~\eqref{e_SR_aux} yields the following expression for the total \textbf{electric-field energy} at any given time:
\begin{equation} \label{energy_SR}
    \mathcal{E} = \iiint_{\bbR^3} \frake \ \rmd\bm{x} = \frac{1}{8\pi}\iiint_{\bbR^3} \left( 2\bm{E}\cdot\bm{D} - |\bm{E}|^2 - \varkappa^{-2}(\nabla\cdot\bm{E})^2 \right) \ \rmd\bm{x}.
\end{equation}
The well-known solution to the Maxwell-Maxwell system in this context is given at all events $(x^0,\ldots,x^3)$ with $(x^1,x^2,x^3)\neq\bm{0}$ by
\begin{equation*}
    \bm{D}^{\mathrm{(M)}} = \bm{E}^{\mathrm{(M)}} = \frac{Q}{r^2}\bm{e_r}, \quad \bm{H}^{\mathrm{(M)}} = \bm{B}^{\mathrm{(M)}} = \bm{0},
\end{equation*}
where $r = \sqrt{(x^1)^2+(x^2)^2+(x^3)^2}$ and $\bm{e_r} = (x^1,x^2,x^3)/r$, and we have omitted the argument $(x^0,\ldots,x^3)$ from the fields. The label M stands, of course, for \textit{Maxwell}. We speak of ``the'' solution as opposed to ``a'' solution because one usually imposes reasonable boundary-type conditions to ensure its uniqueness --- in the present case, this is the only solution that vanishes at spatial infinity. Next note that the electric field is conservative, that is, it is given in the form $\bm{E}^{\mathrm{(M)}} = -\nabla\varphi^{\mathrm{(M)}}$ for an \textbf{electric potential} $\varphi^{\mathrm{(M)}}:\bbR^3\setminus\{\bm{0}\}\longrightarrow\bbR$, namely the \textbf{Coulomb potential}, which we view as a function $\varphi^{(\mathrm{M})}:(0,\infty)\longrightarrow\bbR$ of $r$:
\begin{equation} \label{Coulomb}
    \varphi^{\mathrm{(M)}}(r) = \frac{Q}{r}.
\end{equation}
However, plugging $\bm{E}^{\mathrm{(M)}}$ and $\bm{D}^{\mathrm{(M)}}$ into the field-energy integral~\eqref{energy_SR} (recall $\varkappa^{-2} = 0$ in the Maxwell-Maxwell equations) gives $\mathcal{E}^{\mathrm{(M)}} = \infty$, due to $\av{\bm{E}^{\mathrm{(M)}}}^2 = Q^2/r^4$ not being integrable around the location $r=0$ of the particle. We say that a charged particle has \textbf{infinite self-energy} under the conventional Maxwell theory. Also note that $\varphi^{(\mathrm{M})}(0)$ is undefined.

Now let us investigate the solution $\bm{E}^{\mathrm{(BLTP)}}$, $\bm{B}^{\mathrm{(BLTP)}}$ etc.~of the Maxwell-BLTP equations ($0<\varkappa<\infty$) corresponding to $\rho$ and $\bm{j}$ as in~\eqref{static_rho_j}. We carry out the calculations in details in order to have a point of comparison for our work in the next section. We omit the superscript (BLTP) from the unknowns until the end. Given that the equations are now of higher order, one can expect more degrees of freedom as in the case of Maxwell theory, but there will again be restrictions coming from physically motivated impositions: First of all, in parallel to the Maxwell-Maxwell case, we restrict the search to time-independent solutions that vanish at spatial infinity. The Maxwell equations then imply $\nabla\times\bm{E} = \nabla\times\bm{H} = \bm{0}$. The equation $\nabla\cdot\bm{B} = 0$ in conjunction with the constitutive relation of $\bm{H}$ in terms of $\bm{B}$ gives $\nabla\cdot\bm{H} = 0$, and thus $\bm{H} = \bm{0}$. Now $\bm{B}$ is determined by the constitutive relation $\bm{0} = \bm{B} - \varkappa^{-2}\Delta\bm{B}$ and the Maxwell equation $\nabla\cdot\bm{B} = 0$, and, although there are non-zero solutions $\bm{B}$ to this system, we have to consider $\bm{B} = \bm{0}$ as the only physically meaningful one because otherwise we would be dealing with the unphysical situation of a magnetic monopole. Next, the equations for $\bm{D}$ reduce to $\nabla\cdot\bm{D} = 4\pi Q \delta_{\bm{0}}$ and $\nabla\times\bm{D} = \bm{0}$, the latter as a consequence of $\nabla\times\bm{E} = \bm{0}$ and $\bm{D} = \bm{E} - \varkappa^{-2}\Delta\bm{E}$. This yields as before
\begin{equation*}
    \bm{D}^{\mathrm{(BLTP)}} = \frac{Q}{r^2}\bm{e_r}.
\end{equation*}
Since $\nabla\times\bm{E} = \bm{0}$, the electric field is again conservative, at least in simply connected regions of spacetime (though we will find this to be true globally): $\bm{E} = -\nabla\varphi$ for a scalar potential $\varphi:\bbR^3\setminus\{\bm{0}\}\longrightarrow\bbR$ that can be re-imagined as a function of $r$ due to spherical symmetry of the sources $\rho$ and $\bm{j}$. In particular $\bm{E} = -\varphi'(r) \bm{e_r}$. Plugging this into the constitutive relation~\eqref{BLTP_vacuum} of $\bm{D}$ in terms of $\bm{E}$ and using the standard formula for the Laplacian in spherical coordinates, we obtain the following equation for $\varphi$:
\begin{equation} \label{phi_eq_minkowski}
    \frac{Q}{r^2} = -\varphi' -\varkappa^{-2}\left(- \varphi''' - \frac{2\varphi''}{r} + \frac{2\varphi'}{r^2}\right).
\end{equation}
It is most easily solved by introducing a new unknown $\chi$ that we call the \textbf{electric-field relative-deviation scalar} or simply \textbf{electric deviation}:
\begin{equation} \label{chi_SR}
    \chi(r) := \frac{1}{Q}r^2\varphi'(r) + 1.
\end{equation}
Then~\eqref{phi_eq_minkowski} transforms into the homogeneous ODE
\begin{equation*}
    \frac{1}{r^2} = \frac{1-\chi}{r^2} - \varkappa^{-2}\left( -\frac{\chi''}{r^2} + \frac{2\chi'}{r^3} \right) \quad\Longleftrightarrow\quad \chi'' = \frac{2}{r}\chi' + \varkappa^2\chi,
\end{equation*}
whose general solution is $\chi(r) = C_1(1+\varkappa r)e^{-\varkappa r} + C_2(1-\varkappa r)e^{\varkappa r}$.
We must take $C_2 = 0$ in order for $\bm{E} = -\varphi'\bm{e_r} = \big(Q(1-\chi)/r^2\big) \bm{e_r}$ to vanish at infinity. Then the corresponding $\varphi$ that vanishes at infinity is defined by
\begin{equation*}
    \varphi(r) = \int_\infty^r \frac{Q\big(\chi(s)-1\big)}{s^2} \ \rmd s = \int_r^\infty \frac{Q\big(1 - C_1(1+\varkappa s)e^{-\varkappa s}\big)}{s^2} \ \rmd s = \frac{Q\big(1 - C_1 e^{-\varkappa r}\big)}{r},
\end{equation*}
with $C_1$ allowed as a free parameter. The value $C_1 = 0$ would yield $\varphi^{\mathrm{(BLTP)}} = \varphi^{\mathrm{(M)}}$, that is, the solution to the Maxwell-Maxwell system also solves the Maxwell-BLTP system, a fact that will remain true later in the context of GR. However, the choice $C_1 = 1$ is the natural one to make here because it renders $\varphi$ bounded around $r=0$: With this choice, we find the everywhere-smooth potential
\begin{equation} \label{phi_BLTP}
    \varphi^{\mathrm{(BLTP)}}(r) = \frac{Q(1 - e^{-\varkappa r})}{r} \quad\text{for } r>0, \qquad \varphi^{\mathrm{(BLTP)}}(0) = \lim_{r\to 0} \frac{Q(1 - e^{-\varkappa r})}{r} = Q\varkappa.
\end{equation}
Consequently, the electric field
\begin{equation*}
\bm{E}^{\mathrm{(BLTP)}} = \frac{Q\big(1 - (1+\varkappa r)e^{-\varkappa r}\big)}{r^2}\bm{e_r}
\end{equation*}
also remains bounded at small $r$ (though the vector $\bm{E}^{\mathrm{(BLTP)}}(\bm{0})$ itself is not defined). We also point out that the relative deviation of $\bm{E}^{\mathrm{(BLTP)}}$ compared to the Coulomb field $\bm{E}^{\mathrm{(M)}}$ is precisely our quantity $\chi$, justifying the name we chose for it:
\begin{equation} \label{relative_error}
    \frac{\av{\bm{E}^{\mathrm{(BLTP)}}-\bm{E}^{\mathrm{(M)}}}}{\av{\bm{E}^{\mathrm{(M)}}}} = \frac{Q\big(1-(1+\varkappa r)e^{-\varkappa r}\big)/r^2 - Q/r^2}{Q/r^2} = (1+\varkappa r)e^{-\varkappa r} = \chi(r).
\end{equation}
In particular, the analogous scalar $\chi$ corresponding to the Coulomb potential would be identically null:
\begin{equation*}
    \chi^{(\mathrm{M})}(r) := \frac{1}{Q}r^2\big(\varphi^{(\mathrm{M})}\big)'(r) + 1 = 0.
\end{equation*}
Finally, we investigate the electric-field energy~\eqref{energy_SR} of our static particle. Using $\bm{D} = (Q/r^2)\bm{e_r}$ and $\bm{E} = -\varphi'(r)\bm{e_r}$, we find
\begin{equation*}
   \bm{E}\cdot\bm{D} = -\frac{Q\varphi'}{r^2}, \quad |\bm{E}|^2 = (\varphi')^2, \quad (\nabla\cdot\bm{E})^2 = \left( -\varphi'' - \frac{2\varphi'}{r} \right)^2 
\end{equation*}
(the last of these equations comes from the standard formula for the divergence in spherical coordinates). This culminates in a finite value for the energy:
\begin{align*}
    \mathcal{E}^{\mathrm{(BLTP)}} &= \frac{1}{8\pi}\iiint_{\bbR^3} \left( 2\bm{E}\cdot\bm{D} - |\bm{E}|^2 - \varkappa^{-2}(\nabla\cdot\bm{E})^2 \right) \ \rmd\bm{x} \\
    &= \frac{4\pi}{8\pi} \int_0^\infty \left( -2Q\frac{\varphi'(r)}{r^2} - \big(\varphi'(r)\big)^2 - \varkappa^{-2} \left( -\varphi''(r) - \frac{2\varphi'(r)}{r} \right)^2 \right) \ \rmd r \\
    &= \frac{Q^2}{2}\int_0^\infty \frac{1}{r^2}\bigg( 1 - \left( 1+2\varkappa r + 2\varkappa^2 r^2 \right)e^{-2\varkappa r} \bigg) \ \rmd r \\
    &= \frac{Q^2}{2\varkappa}\int_0^\infty \frac{\varkappa^2}{s^2}\bigg( 1 - \left( 1+2s + 2s^2 \right)e^{-2s} \bigg) \ \rmd s \\
    &= \frac{Q^2\varkappa}{2}\left(\frac{-1+(1+s)e^{-2s}}{s}\right)\bigg|_0^\infty \\
    &= \frac{Q^2\varkappa}{2} < \infty.
\end{align*}

\begin{remark}
    Assuming for a moment that BLTP theory is valid in nature, it is possible to make a cursory guess for the range of the parameter $\varkappa$ by looking at what happens at a sub-atomic distance from a charged particle. Considering that the classical Coulomb's law of electric attraction/repulsion (which indirectly follows from the statement that a point charge's electric field is given by $\bm{E}^{\mathrm{(M)}}$) has been experimentally verified down to sub-atomic distances, experimentalists would only find evidence for BLTP theory at a distance of about $10^{-10}$ m or smaller away from a point charge. Since the relative deviation~\eqref{relative_error} between $\bm{E}^{\mathrm{(BLTP)}}$ and $\bm{E}^{\mathrm{(M)}}$ is of about 1\% at $r = 10^{-10}$ m when $\varkappa = 6.64 \cdot 10^{10}$ $\mathrm{m}^{-1}$ and decreases exponentially as $r$ increases, we will assume the ballpark estimate
    \begin{equation} \label{kappa_estimate1}
        \varkappa > 10^{11} \ \mathrm{m}^{-1}
    \end{equation}
    whenever we do calculations involving real-life constants. On the other hand, at a distance as small as $10^{-16}$ m --- safely inside atomic nuclei and therefore within the quantum realm, --- one can expect that BLTP theory, as a classical theory of electromagnetism, would have already made its presence clear. Since~\eqref{relative_error} is of about 50\% at $r = 10^{-16}$ m when $\varkappa = 1.68\cdot 10^{16}$ $\mathrm{m}^{-1}$, we will also assume
    \begin{equation} \label{kappa_estimate2}
        \varkappa < 10^{16} \ \mathrm{m}^{-1}.
    \end{equation}
\end{remark}

\subsection{Covariant formulation of electromagnetism in SR}
\label{subsec_SR_covariant}

As is well known, the equations of electromagnetism assume a simpler, manifestly covariant form when written in the language of differential forms and other tensor fields. We briefly explain this formulation now, as it will be needed to treat the GR case next.

Let source terms $\rho,\bm{j}$ and their associated fields $\bm{E},\bm{D},\bm{B},\bm{H}$ and field energy and momenta $\frake$, $\frakg$, $\frakS$ and $\frakT$, satisfying $\frakS = c^2\frakg$ and $\mathfrak{T}_{ij} = \mathfrak{T}_{ji}$, be observed in an inertial coordinate system $(x^\mu)$ to solve the EM equations presented in the previous subsection, either with or without the $\varkappa^{-1}$ terms. Now consider a 1-form field $J$ (called the \textbf{four-current} 1-form), two 2-form fields $F,M$ (called \textbf{Faraday}'s and \textbf{Maxwell}'s tensors, among many other names in the literature) and a symmetric $(2,0)$-tensor field $T$ (called the \textbf{stress-energy} tensor) which are given in these coordinates by
\begin{equation} \label{J_components}
    J = -c\rho \ \rmd x^0 + \sum_{k=1}^3 j_k \ \rmd x^k,
\end{equation}
\begin{equation} \label{F_components}
    (F_{\mu\nu}) = \left(\begin{array}{rrrr}
        0 & -E_1 & -E_2 & -E_3 \\
        E_1 & 0 & B_3 & -B_2 \\
        E_2 & -B_3 & 0 & B_1 \\
        E_3 & B_2 & -B_1 & 0
    \end{array}\right), \quad (M_{\mu\nu}) = \left(\begin{array}{rrrr}
        0 & H_1 & H_2 & H_3 \\
        -H_1 & 0 & D_3 & -D_2 \\
        -H_2 & -D_3 & 0 & D_1 \\
        -H_3 & D_2 & -D_1 & 0
    \end{array}\right),
\end{equation}
\begin{equation} \label{Tmunu_components}
    T^{00} = \frake, \quad T^{0j} = \frac{\mathfrak{S}_j}{c} = c\mathfrak{g}_j = T^{j0}, \quad T^{jk} = \mathfrak{T}_{jk}, \quad j,k=1,2,3,
\end{equation}
where $\bm{j} = (j_1,j_2,j_3)$, $\bm{E} = (E_1,E_2,E_3)$ etc. We remark that the tensors $F$ and $M$ themselves have dimension $\dimQ$; indeed, each of their components has dimension $\dimMM^{\frac12}\dimLL^{-\frac12}\dimTT^{-1}$, as mentioned before, while each of the basis tensors $\rmd x^\mu\otimes\rmd x^\nu$ for $\mathcal{T}^0_2(\bbR^4)$ has dimension $\dimLL^2$, yielding for example
\begin{equation*}
    \dim(F) = \dim\left( E_1 \ \rmd x^0\otimes\rmd x^1 + \cdots \right) = \dimMM^{\frac12}\dimLL^{-\frac12}\dimTT^{-1} \dimLL^2 = \dimQ.
\end{equation*}
The reader is now invited to check that the continuity equation~\eqref{continuity}, Maxwell equations~\eqref{Maxwell}, vacuum law~\eqref{BLTP_vacuum} and field-energy and --momenta conservation laws~\eqref{field_energy_conservation} and~\eqref{field_momentum_conservation} can be written as follows:
\begin{align}
    \label{conservation_cov}
        &\delta J = 0, \\
    \label{maxwell_cov}
        &\rmd F = 0, \quad \rmd M = \dfrac{4\pi}{c} \Star J, \\
    \label{vacuum_cov}
        &M = \Star \big(F + \varkappa^{-2} \rmd\delta F\big), \\
    \label{divergence_cov}
        &\partial_\nu T^{\mu\nu} = -\dfrac{1}{c} \tensor{F}{^\mu_\lambda} J^\lambda.
\end{align}
The tensors $J$, $F$, $M$ and $T$ satisfying the above tensorial equations (the first three of which are covariant) should be regarded as the intrinsic constructs of electromagnetism, while the sources and fields from the previous subsection are merely observer-dependent functions that constitute their components in a given inertial coordinate system.

\begin{remark}
    The reader should remain aware that different sign conventions exist in the literature for the metric signature, the components of $F$ and $M$, and the definitions of $\Star$ and $\delta$. Usually a choice of placement of minus signs in the definitions of these objects is adopted in order for certain standard formulas to always remain the same. As an example: For a test particle of charge $q$, 4-velocity vector $U$ and 4-momentum vector $P$, equations~\eqref{dpdt} and~\eqref{dEdt} can be collected together as
    \begin{equation} \label{lorentz_covariant}
        \frac{\rmd P_\mu}{\rmd\tau} = q F_{\mu\nu} U^\nu, \quad \mu=0,1,2,3,
    \end{equation}
    where $\tau$ denotes the particle's proper time. The appearance of this formula is almost universal in the physics literature. But note that what features on its left side are the \textit{index-lowered} components of $P$, whose definition depends on the metric. In particular, if the $({+}\,{-}\,{-}\,{-})$ metric signature were to be used instead of ours, then $(F_{\mu\nu})$ would need to be defined as the negative of what we introduced above, in order to keep~\eqref{lorentz_covariant} unaltered.
\end{remark}

Next we present the \textit{Lagrangian formulation} of both the Maxwell-Maxwell and the Maxwell-BLTP systems of tensor equations. For this purpose, a \textbf{four-potential} $A$ must be introduced: The equation $\rmd F = 0$ is equivalent (at least on simply connected regions of spacetime) to the existence of a 1-form field $A$ such that
\begin{equation} \label{FdA}
    F = \rmd A.
\end{equation}
It can be easily checked that this relation is equivalent to equations~\eqref{field_potentials} for $\bm{E}$ and $\bm{B}$ in terms of the potentials $\varphi$ and $\bm{A} = (A_1,A_2,A_3)$ when $A$ is defined in the inertial coordinate system $(x^\mu)$ by
\begin{equation} \label{A_SR}
    A = -\varphi \ \rmd x^0 + \sum_{k=1}^3 A_k \ \rmd x^k.
\end{equation}
Then the remaining equations of EM theory can be put together as one equation for $A$: We take a $\Star$ on both sides of the $M$ equation~\eqref{maxwell_cov} and replace $M$ in it by the expression given in~\eqref{vacuum_cov} to find the fourth-order equation
\begin{equation} \label{eq_A_J}
    \frac{1}{c}J - \frac{1}{4\pi}\delta\rmd A - \frac{\varkappa^{-2}}{4\pi} \delta\rmd\delta\rmd A = 0.
\end{equation}
It is this equation that arises from an action principle expressed in terms of the dynamical variable $A$:

\begin{proposition} \label{prop_lagrangian}
    Let a 1-form field $J$ satisfying $\delta J = 0$ be given. For an arbitrary 1-form field $A$, consider the Lagrangian 4-form
    \begin{equation} \label{lagrangian_4form}
        \mathcal{L}_{\mathrm{EM}}(A) = \frac{1}{c} A\wedge\Star J - \frac{1}{8\pi} (\rmd A)\wedge (\Star\rmd A) - \frac{\varkappa^{-2}}{8\pi} (\delta\rmd A)\wedge(\Star\delta\rmd A),
    \end{equation}
    and define the action functional $\mathcal{S}(A;U) = \int_U \mathcal{L}_{\mathrm{EM}}(A)$ for arbitrary open sets $U\subseteq\bbR^4$. Then equation~\eqref{eq_A_J} is equivalent to the Euler-Lagrange equations of $\mathcal{S}$, that is, in an inertial coordinate system,
    \begin{equation} \label{euler_lagrange}
        \frac{\partial\ell_{\mathrm{EM}}}{\partial A_\xi} - \frac{\partial}{\partial x^\mu} \left(\frac{\partial\ell_{\mathrm{EM}}}{\partial(\partial_\mu A_\xi)}\right) + \frac{\partial^2}{\partial x^\mu\partial x^\nu} \left(\frac{\partial\ell_{\mathrm{EM}}}{\partial(\partial_\mu\partial_\nu A_\xi)}\right) = 0, \quad \xi= 0,1,2,3,
    \end{equation}
    where $\mathcal{L}_{\mathrm{EM}}(A) = \ell_{\mathrm{EM}}(A,\partial A, \partial^2 A) \ \mathrm{vol}_\eta$ expresses $\mathcal{L}_{\mathrm{EM}}$ as a scalar $\ell_{\mathrm{EM}}$ multiplied by the canonical volume form $\mathrm{vol}_\eta = \rmd x^0\wedge\cdots\wedge\rmd x^3 =: \rmd^4 x$ of Minkowski space, with $\ell_{\mathrm{EM}}$ considered as a function of the components of $A$ and of their derivatives up to second order.
\end{proposition}

\begin{proof}
    Using the definitions~\eqref{definition_hodge} and~\eqref{definition_product_pforms} of the Hodge star and of the inner product $\langle\ast,\ast\rangle_\eta$ on $p$-forms, we find that $\mathcal{L}_{\mathrm{EM}} = \ell_{\mathrm{EM}} \ \mathrm{vol}_\eta$ for the scalar
    \begin{equation} \label{proof_ell}
        \ell_{\mathrm{EM}} = \frac{1}{c} A_\alpha J^\alpha - \frac{1}{16\pi} (\rmd A)_{\alpha\beta}(\rmd A)^{\alpha\beta} - \frac{\varkappa^{-2}}{8\pi} (\delta\rmd A)_\alpha (\delta\rmd A)^\alpha = \frac{1}{c}\ell_{\mathrm{EM}}^{(1)} -\frac{1}{16\pi} \ell_{\mathrm{EM}}^{(2)} -\frac{\varkappa^{-2}}{8\pi}\ell_{\mathrm{EM}}^{(3)},
    \end{equation}
    where each $\ell_{\mathrm{EM}}^{(j)}$ only depends on the $(j-1)$-st--order derivatives of the components $A_\alpha$. Now we need to write $\ell_{\mathrm{EM}}$ explicitly in terms of these components. For $\ell_{\mathrm{EM}}^{(1)} = A_\alpha J^\alpha$, the job is already done. For $\ell_{\mathrm{EM}}^{(2)}$, we use the coordinate expression~\eqref{d_coordinates} for the exterior derivative:
    \begin{align*}
        \ell_{\mathrm{EM}}^{(2)} &= (\rmd A)_{\alpha\beta}(\rmd A)^{\alpha\beta} \\
        &= (\rmd A)_{\alpha\beta} \eta^{\alpha\rho}\eta^{\beta\sigma}(\rmd A)_{\rho\sigma} \\
        &= \big( 
        \partial_\alpha A_\beta - \partial_\beta A_\alpha \big) \eta^{\alpha\rho}\eta^{\beta\sigma} \big( \partial_\rho A_\sigma - \partial_\sigma A_\rho \big) \\
        &= 2\eta^{\alpha\rho}\eta^{\beta\sigma}(\partial_\alpha A_\beta) \big( \partial_{\rho} A_\sigma - \partial_\sigma A_\rho \big).
    \end{align*}
    For $\ell_{\mathrm{EM}}^{(3)}$, we also use the coordinate expression~\eqref{delta_coordinates} for the codifferential, noting that $\av{\eta} = 1$ and $\partial_\xi\eta_{\mu\nu} = 0$:
    \begin{align*}
        \ell_{\mathrm{EM}}^{(3)} &= (\delta\rmd A)_\alpha (\delta\rmd A)^\alpha \\
        &= \eta_{\alpha\rho}(\delta\rmd A)^\rho (\delta\rmd A)^\alpha \\
        &= \eta_{\alpha\rho} \big( 
        -\partial_\lambda(\rmd A)^{\lambda\rho} \big)\big( 
        -\partial_\theta(\rmd A)^{\theta\alpha} \big) \\
        &= \eta_{\alpha\rho} \partial_\lambda \big( \eta^{\lambda\sigma}\eta^{\rho\tau}(\rmd A)_{\sigma\tau} \big) \partial_\theta \big( \eta^{\theta\phi}\eta^{\alpha\chi}(\rmd A)_{\phi\chi} \big) \\
        &= \eta_{\alpha\rho}\eta^{\lambda\sigma}\eta^{\rho\tau}\eta^{\theta\phi}\eta^{\alpha\chi} \partial_\lambda\big( \partial_\sigma A_\tau - \partial_\tau A_\sigma \big) \partial_\theta\big( \partial_\phi A_\chi - \partial_\chi A_\phi \big) \\
        &= 2\eta_{\alpha\rho}\eta^{\lambda\sigma}\eta^{\rho\tau}\eta^{\theta\phi}\eta^{\alpha\chi} \partial_\lambda\big( \partial_\sigma A_\tau \big) \partial_\theta\big( \partial_\phi A_\chi - \partial_\chi A_\phi \big) \\
        &= 2\eta^{\lambda\sigma}\eta^{\theta\phi}\eta^{\alpha\chi} (\partial_\lambda\partial_\sigma A_\alpha)(\partial_\theta\partial_\phi A_\chi) - 2\eta^{\lambda\sigma}\eta^{\alpha\phi}\eta^{\theta\chi} (\partial_\lambda\partial_\sigma A_\alpha)(\partial_\theta\partial_\phi A_\chi).
    \end{align*}
    Then, for each fixed $\mu,\nu,\xi$, we differentiate the above expressions with respect to $A_\xi$, $\partial_\mu A_\xi$ and $\partial_\mu\partial_\nu A_\xi$ respectively. The use of the Kronecker Delta $\delta^\ast_\ast$ helps with writing this calculation in a compact form, due to
    \begin{equation*}
        \frac{\partial A_\alpha}{\partial A_\xi} = \delta^\xi_\alpha, \quad \frac{\partial (\partial_\rho A_\alpha)}{\partial (\partial_\mu A_\xi)} = \delta^\mu_\rho \delta^\xi_\alpha, \quad \frac{\partial (\partial_\rho\partial_\sigma A_\alpha)}{\partial (\partial_\mu\partial_\nu A_\xi)} = \delta^\mu_\rho \delta^\nu_\sigma \delta^\xi_\alpha.
    \end{equation*}
    We find
    \begin{align*}
        \frac{\partial\ell_{\mathrm{EM}}^{(1)}}{\partial A_\xi} &= \frac{\partial}{\partial A_\xi} \big( A_\alpha J^\alpha \big) = \delta_\alpha^\xi J^\alpha = J^\xi, \\
        \frac{\partial\ell_{\mathrm{EM}}^{(2)}}{\partial(\partial_\mu A_\xi)} &= 2\frac{\partial}{\partial(\partial_\mu A_\xi)} \bigg( \eta^{\alpha\rho}\eta^{\beta\sigma}(\partial_\alpha A_\beta) (\partial_{\rho} A_\sigma - \partial_\sigma A_\rho) \bigg) \\
        &= 2\eta^{\alpha\rho}\eta^{\beta\sigma} \delta^\mu_\alpha\delta^\xi_\beta(\partial_{\rho} A_\sigma - \partial_\sigma A_\rho) + 2\eta^{\alpha\rho}\eta^{\beta\sigma}(\partial_\alpha A_\beta) ( \delta^\mu_\rho\delta^\xi_\sigma - \delta^\mu_\sigma\delta^\xi_\rho ) \\
        &= 2\eta^{\mu\rho}\eta^{\xi\sigma} (\partial_{\rho} A_\sigma - \partial_\sigma A_\rho) + 2\eta^{\alpha\mu}\eta^{\beta\xi}\partial_\alpha A_\beta - 2\eta^{\alpha\xi}\eta^{\beta\mu}\partial_\alpha A_\beta \\
        &= 2\partial^\mu A^\xi - 2\partial^\xi A^\mu + 2\partial^\mu A^\xi - 2\partial^\xi A^\mu \\
        &= 4\partial^\mu A^\xi - 4\partial^\xi A^\mu, \\
        \frac{\partial\ell_{\mathrm{EM}}^{(3)}}{\partial(\partial_\mu\partial_\nu A_\xi)} &= 2\frac{\partial}{\partial(\partial_\mu\partial_\nu A_\xi)} \bigg( \eta^{\lambda\sigma}\eta^{\theta\phi}\eta^{\alpha\chi} (\partial_\lambda\partial_\sigma A_\alpha)(\partial_\theta\partial_\phi A_\chi) - \eta^{\lambda\sigma}\eta^{\alpha\phi}\eta^{\theta\chi} (\partial_\lambda\partial_\sigma A_\alpha)(\partial_\theta\partial_\phi A_\chi) \bigg) \\
        &= 2\eta^{\lambda\sigma}\eta^{\theta\phi}\eta^{\alpha\chi} \delta^\mu_\lambda \delta^\nu_\sigma \delta^\xi_\alpha(\partial_\theta\partial_\phi A_\chi) + 2\eta^{\lambda\sigma}\eta^{\theta\phi}\eta^{\alpha\chi} (\partial_\lambda\partial_\sigma A_\alpha) \delta^\mu_\theta \delta^\nu_\phi \delta^\xi_\chi \\
        &\qquad - 2\eta^{\lambda\sigma}\eta^{\alpha\phi}\eta^{\theta\chi} \delta^\mu_\lambda \delta^\nu_\sigma \delta^\xi_\alpha(\partial_\theta\partial_\phi A_\chi) - 2\eta^{\lambda\sigma}\eta^{\alpha\phi}\eta^{\theta\chi} (\partial_\lambda\partial_\sigma A_\alpha)\delta^\mu_\theta \delta^\nu_\phi \delta^\xi_\chi \\
        &= 2\eta^{\mu\nu}\eta^{\theta\phi}\eta^{\xi\chi}\partial_\theta\partial_\phi A_\chi + 2\eta^{\lambda\sigma}\eta^{\mu\nu}\eta^{\alpha\xi} \partial_\lambda\partial_\sigma A_\alpha - 2\eta^{\mu\nu}\eta^{\xi\phi}\eta^{\theta\chi}\partial_\theta\partial_\phi A_\chi - 2\eta^{\lambda\sigma}\eta^{\alpha\nu}\eta^{\mu\xi} \partial_\lambda\partial_\sigma A_\alpha \\
        &= 2\eta^{\mu\nu}\partial^\phi \partial_\phi A^\xi + 2\eta^{\mu\nu}\partial^\sigma \partial_\sigma A^\xi - 2\eta^{\mu\nu}\partial^\chi \partial^\xi A_\chi - 2\eta^{\mu\xi}\partial^\sigma \partial_\sigma A^\nu.
    \end{align*}
    The coordinate derivatives of the last two of these expressions, as featured in~\eqref{euler_lagrange}, are
    \begin{align*}
        \frac{\partial}{\partial x^\mu}\left(\frac{\partial\ell_{\mathrm{EM}}^{(2)}}{\partial(\partial_\mu A_\xi)}\right) &= 4\partial_\mu \big( \partial^\mu A^\xi - \partial^\xi A^\mu \big) \\
        &= 4\square A^\xi - 4\partial^\xi \mathrm{div}(A), \\
        \frac{\partial^2}{\partial x^\mu\partial x^\nu}\left(\frac{\partial\ell_{\mathrm{EM}}^{(3)}}{\partial(\partial_\mu\partial_\nu A_\xi)}\right) &= 2\partial_\mu\partial_\nu \big( \eta^{\mu\nu}\partial^\phi \partial_\phi A^\xi + \eta^{\mu\nu}\partial^\sigma \partial_\sigma A^\xi - \eta^{\mu\nu}\partial^\chi \partial^\xi A_\chi - \eta^{\mu\xi}\partial^\sigma \partial_\sigma A^\nu \big) \\
        &= 2\partial^\nu\partial_\nu\partial^\phi\partial_\phi A^\xi + 2\partial^\nu\partial_\nu\partial^\sigma\partial_\sigma A^\xi - 2\partial^\nu\partial_\nu\partial^\xi\partial^\chi A_\chi - 2\partial^\xi\partial_\nu\partial^\sigma\partial_\sigma A^\nu \\
        &= 2\square^2 A^\xi + 2\square^2 A^\xi - 2\partial^\xi\square \mathrm{div}(A) - 2\partial^\xi\square \mathrm{div}(A) \\
        &= 4\square^2 A^\xi - 4\partial^\xi\square\mathrm{div}(A),
    \end{align*}
    where we used the symbols $\mathrm{div}(A) = \partial_\lambda A^\lambda = \partial^\lambda A_\lambda$ and $\square = \partial_\lambda\partial^\lambda = \partial^\lambda\partial_\lambda$. Thus the Euler-Lagrange equation~\eqref{euler_lagrange} becomes
    \begin{equation} \label{EL_proof_goal}
        \frac{1}{c} J^\xi - \frac{1}{4\pi} \big( \partial^\xi \mathrm{div}(A) - \square A^\xi \big) - \frac{\varkappa^{-2}}{4\pi} \big( \square^2 A^\xi - \partial^\xi\square\mathrm{div}(A) \big) = 0.
    \end{equation}
    Finally, we need to compare this to the coordinate expression of the target equation~\eqref{eq_A_J}:
    \begin{equation} \label{eq_A_J_2}
        0 = \frac{1}{c}J - \frac{1}{4\pi}\delta\rmd A - \frac{\varkappa^{-2}}{4\pi} \delta\rmd\delta\rmd A = \frac{1}{c}B_{(1)} - \frac{1}{4\pi}B_{(2)} -\frac{\varkappa^{-2}}{4\pi}B_{(3)},
    \end{equation}
    whose terms are found to be
    \begin{align*}
        B_{(1)}^\xi &= J^\xi, \\
        B_{(2)}^\xi &= (\delta\rmd A)^\xi = -\partial_\lambda\big( \eta^{\lambda\rho}\eta^{\xi\sigma} ( \partial_\rho A_\sigma - \partial_\sigma A_\rho) \big) = -\partial^\rho\partial_\rho A^\xi + \partial^\rho \partial^\xi A_\rho = -\square A^\xi + \partial^\xi\mathrm{div}(A), \\
        B_{(3)}^\xi &= (\delta\rmd\delta\rmd A)^\xi \\
        &= -\square (\delta\rmd A)^\xi + \partial^\xi \mathrm{div}(\delta\rmd A) \\
        &= -\square (\delta\rmd A)^\xi + \partial^\xi \partial_\lambda(\delta\rmd A)^\lambda \\
        &= -\square\big( -\square A^\xi + \partial^\xi\mathrm{div}(A) \big) + \partial^\xi \partial_\lambda\big( -\square A^\lambda + \partial^\lambda\mathrm{div}(A) \big) \\
        &= \square^2 A^\xi - \partial^\xi\square\mathrm{div}(A) - \partial^\xi\square\mathrm{div}(A) + \partial^\xi\square\mathrm{div}(A) \\
        &= \square^2 A^\xi - \partial^\xi\square\mathrm{div}(A).
    \end{align*}
    Then~\eqref{eq_A_J_2} is immediately seen to also be equivalent to~\eqref{EL_proof_goal}.
\end{proof}

\begin{remark} \label{rem_BLTP_generalization}
    The Lagrangian 4-form~\eqref{lagrangian_4form}, which looks like
    \begin{equation} \label{lagrangian_4form_F}
        \mathcal{L}_{\mathrm{EM}}(A) = \frac{1}{c} A\wedge\Star J - \frac{1}{8\pi} F\wedge\Star F - \frac{\varkappa^{-2}}{8\pi} (\delta F)\wedge(\Star\delta F)
    \end{equation}
    when written in terms of $F = \rmd A$, can be used to explain why BLTP theory is a \textit{natural} Lagrangian generalization of the conventional Maxwell theory of electromagnetism: If one wishes to generalize the latter while retaining the \textit{linear} nature of the field equations, one needs to take the Lagrangian given by the first two terms in~\eqref{lagrangian_4form_F} and add to it a new term that is \textit{quadratic} in quantities related to $F$. But there are no quadratic invariants of $F$ itself left to use apart from the term $F\wedge\Star F$ that is already present. One might consider using $F\wedge F$, but this expression would not contribute anything to the Euler-Lagrange equations on account of it being an exact differential, namely of $A\wedge F$. Thus one is led to consider the next best thing: quadratic invariants involving \textit{derivatives} of $F$, such as the term $(\delta F)\wedge (\Star\delta F)$ that appears in~\eqref{lagrangian_4form_F}.
\end{remark}

We close this section by deriving a stress-energy tensor $(T^{\mu\nu})$ from which the expression for $\frake$ claimed in~\eqref{e_SR} can be extracted. As mentioned before, formulas for the fields' (energy/momenta)-(density/flux) quantities are derived by rewriting the right sides of~\eqref{field_energy_conservation} and~\eqref{field_momentum_conservation} into the form $\partial_t(\ast) + \nabla\cdot(\bm{\ast})$. In tensorial language, however, what this amounts to is a rewriting of the right side of the equivalent equation~\eqref{divergence_cov} into the form $\partial_\nu T^{\mu\nu}$ for some symmetric tensor $T$.

\begin{proposition} \label{prop_Tmunu}
    Let a 1-form field $J$ and two 2-form fields $F$ and $M$ satisfy Maxwell's equations~\eqref{maxwell_cov} and the vacuum law~\eqref{vacuum_cov}. Let an inertial coordinate system be fixed. Then the relation
    \begin{equation*}
        \partial_\nu T^{\mu\nu} = -\frac{1}{c}\tensor{F}{^\mu_\lambda} J^\lambda    
    \end{equation*}
    holds for the symmetric $(2,0)$-tensor $T$ whose components are defined by
    \begin{equation} \label{Tmunu_SR}
        T^{\mu\nu} = T^{\mu\nu}_{\mathrm{(M)}} + \varkappa^{-2} T^{\mu\nu}_{\mathrm{(BLTP)}} 
    \end{equation}
    for
    \begin{align*}
        &4\pi T^{\mu\nu}_{\mathrm{(M)}} = F^{\alpha\mu}\tensor{F}{_\alpha^\nu} - \frac{\eta^{\mu\nu}}{4} F_{\alpha\beta}F^{\alpha\beta}, \\
        &4\pi T^{\mu\nu}_{\mathrm{(BLTP)}} = F^{\alpha\mu}\tensor{(\rmd\delta F)}{_\alpha^\nu} + F^{\alpha\nu}\tensor{(\rmd\delta F)}{_\alpha^\mu} - \frac{\eta^{\mu\nu}}{2}F_{\alpha\beta}(\rmd\delta F)^{\alpha\beta} - (\delta F)^\mu(\delta F)^\nu + \frac{\eta^{\mu\nu}}{2} (\delta F)^\alpha (\delta F)_\alpha.
    \end{align*}
\end{proposition}

\begin{proof}
    Equation~\eqref{eq_A_J}, recast in terms of the tensor $F$ given by $F=\rmd A$, becomes
    \begin{equation*}
        \frac{4\pi}{c} J^\lambda = (\delta F)^\lambda + \varkappa^{-2} (\delta\rmd\delta F)^\lambda = -\partial_\nu F^{\nu\lambda} - \varkappa^{-2}\partial_\nu (\rmd\delta F)^{\nu\lambda},
    \end{equation*}
    where we used formula~\eqref{delta_coordinates} in inertial coordinates to find $\delta\psi$ for the 2-forms $\psi = F$ and $\psi = \rmd\delta F$. Hence, we need to find a symmetric tensor $(T^{\mu\nu})$ satisfying
    \begin{equation*}
        4\pi\partial_\nu T^{\mu\nu} = -\frac{4\pi}{c} \tensor{F}{^\mu_\lambda} J^\lambda = \tensor{F}{^\mu_\lambda}\partial_\nu F^{\nu\lambda} + \varkappa^{-2} \tensor{F}{^\mu_\lambda}\partial_\nu (\rmd\delta F)^{\nu\lambda}.
    \end{equation*}
    This will be achieved in the form~\eqref{Tmunu_SR} if we can find separate symmetric tensors $(T^{\mu\nu}_{\mathrm{(M)}})$ and $(T^{\mu\nu}_{\mathrm{(BLTP)}})$ with
    \begin{equation*}
        4\pi \partial_\nu T^{\mu\nu}_{\mathrm{(M)}} = \tensor{F}{^\mu_\lambda}\partial_\nu F^{\nu\lambda}, \quad 4\pi \partial_\nu T^{\mu\nu}_{\mathrm{(BLTP)}} = \tensor{F}{^\mu_\lambda}\partial_\nu (\rmd\delta F)^{\nu\lambda}.
    \end{equation*}
    We look for $(T^{\mu\nu}_{\mathrm{(M)}})$ first. Using ``integration by parts'' to move the $\nu$-derivative over to the other term of the expression $\tensor{F}{^\mu_\lambda}\partial_\nu F^{\nu\lambda}$, we have
    \begin{align}
        \tensor{F}{^\mu_\lambda}\partial_\nu F^{\nu\lambda} &= \partial_\nu \big( \tensor{F}{^\mu_\lambda}F^{\nu\lambda} \big) - F^{\nu\lambda} \partial_\nu\tensor{F}{^\mu_\lambda} \nonumber \\
        &= \partial_\nu U_{(1)}^{\mu\nu} - \eta^{\mu\alpha} F^{\nu\lambda} \partial_\nu F_{\alpha\lambda} \qquad \text{where}\quad\boxed{U_{(1)}^{\mu\nu} := \tensor{F}{^\mu_\lambda}F^{\nu\lambda}}. \label{SET_proof1}
    \end{align}
    By switching the positions of the upper and lower copies of the index $\lambda$, we realize that $U_{(1)}$ is symmetric:
    \begin{equation*}
        U_{(1)}^{\mu\nu} = \tensor{F}{^\mu_\lambda}F^{\nu\lambda} = \tensor{F}{^\mu^\lambda}\tensor{F}{^\nu_\lambda} = U_{(1)}^{\nu\mu}.
    \end{equation*}
    So we just need to study the rest of the expression~\eqref{SET_proof1}. Using the coordinate formula~\eqref{d_coordinates} for the equation $\rmd F = 0$, which reads $\partial_\nu F_{\alpha\lambda} - \partial_{\alpha} F_{\nu\lambda} + \partial_\lambda F_{\nu\alpha} = 0$, we can prove that $F^{\nu\lambda}\partial_\nu F_{\alpha\lambda} = F^{\nu\lambda}\partial_{\alpha} F_{\nu\lambda}/2$ --- indeed, first using the antisymmetry of $(F^{\alpha\beta})$ and $(F_{\alpha\beta})$ and then relabeling indices on the second term, we have:
    \begin{equation*}
        F^{\nu\lambda}\partial_\nu F_{\alpha\lambda} = F^{\nu\lambda}\partial_{\alpha} F_{\nu\lambda} - F^{\nu\lambda}\partial_\lambda F_{\nu\alpha} = F^{\nu\lambda}\partial_{\alpha} F_{\nu\lambda} - F^{\lambda\nu}\partial_\lambda F_{\alpha\nu} = F^{\nu\lambda}\partial_{\alpha} F_{\nu\lambda} -F^{\nu\lambda}\partial_\nu F_{\alpha\lambda},
    \end{equation*}
    and now moving $F^{\nu\lambda}\partial_\nu F_{\alpha\lambda}$ to the left side yields the claimed identity. So let us apply this on the leftover term in~\eqref{SET_proof1}:
    \begin{align*}
        -\eta^{\mu\alpha} F^{\nu\lambda} \partial_\nu F_{\alpha\lambda} &= -\frac{1}{2}\eta^{\mu\alpha} F^{\nu\lambda}\partial_\alpha F_{\nu\lambda} \\
        &= -\frac{1}{4}\eta^{\mu\alpha} F^{\nu\lambda}\partial_\alpha F_{\nu\lambda} -\frac{1}{4}\eta^{\mu\alpha} F^{\nu\lambda}\partial_\alpha F_{\nu\lambda} \\
        &= -\frac{1}{4}\eta^{\mu\alpha} F^{\nu\lambda}\partial_\alpha F_{\nu\lambda} -\frac{1}{4}\eta^{\mu\alpha} F_{\nu\lambda}\partial_\alpha F^{\nu\lambda} \\
        &= -\frac{1}{4}\eta^{\mu\alpha}\partial_\alpha \big( 
        F^{\nu\lambda}F_{\nu\lambda} \big) \\
        &= \partial_\nu U^{\mu\nu}_{(2)} \qquad \text{where}\quad\boxed{U^{\mu\nu}_{(2)} := -\frac{1}{4}\eta^{\mu\nu} F^{\alpha\beta} F_{\alpha\beta}}.
    \end{align*}
    Thus we have proved $\tensor{F}{^\mu_\lambda}\partial_\nu F^{\nu\lambda} = \partial_\nu \big( U^{\mu\nu}_{(1)} + U^{\mu\nu}_{(2)} \big)$ for symmetric expressions $U_{(j)}$. Defining $4\pi T^{\mu\nu}_{\mathrm{(M)}}$ as $U^{\mu\nu}_{(1)} + U^{\mu\nu}_{(2)}$ gives the first term of~\eqref{Tmunu_SR} as claimed in the proposition.
    
    It remains to derive $T^{\mu\nu}_{\mathrm{(BLTP)}}$. The beginning of the computation is analogous to the above:
    \begin{align*}
        \tensor{F}{^\mu_\lambda}\partial_\nu (\rmd\delta F)^{\nu\lambda} &= \partial_\nu \big( \tensor{F}{^\mu_\lambda} (\rmd\delta F)^{\nu\lambda} \big) - (\rmd\delta F)^{\nu\lambda} \partial_\nu\tensor{F}{^\mu_\lambda} \\
        &= \partial_\nu V^{\mu\nu}_{(1)} - \eta^{\mu\alpha}(\rmd\delta F)^{\nu\lambda}\partial_\nu F_{\alpha\lambda} \qquad \text{where}\quad\boxed{V^{\mu\nu}_{(1)} := \tensor{F}{^\mu_\lambda} (\rmd\delta F)^{\nu\lambda}} \\ 
        &= \partial_\nu V^{\mu\nu}_{(1)} -\frac{1}{2}\eta^{\mu\alpha}(\rmd\delta F)^{\nu\lambda}\partial_\alpha F_{\nu\lambda}.
    \end{align*}
    Note that $V_{(1)}$ is \textit{not} symmetric, but we choose to keep it as a piece of our desired tensor because its counterpart with $\mu$ and $\nu$ exchanged will appear later. Next we ``integrate by parts'' again in the leftover term, this time with respect to $\partial_\alpha$:
    \begin{align*}
        -\frac{1}{2}\eta^{\mu\alpha}(\rmd\delta F)^{\nu\lambda}\partial_\alpha F_{\nu\lambda} &= -\frac{1}{2}\eta^{\mu\alpha}\partial_\alpha \big( (\rmd\delta F)^{\nu\lambda} F_{\nu\lambda} \big) + \frac{1}{2}\eta^{\mu\alpha} F_{\nu\lambda} \partial_\alpha(\rmd\delta F)^{\nu\lambda} \\
        &= \partial_\nu V^{\mu\nu}_{(2)} + \frac{1}{2} F_{\nu\lambda}\partial^\mu (\rmd\delta F)^{\nu\lambda} \qquad \text{where}\quad\boxed{V^{\mu\nu}_{(2)} := -\frac{1}{2}\eta^{\mu\nu} (\rmd\delta F)^{\alpha\beta} F_{\alpha\beta}}.
    \end{align*}
    To continue with the leftover term, we apply formula~\eqref{d_coordinates} for $\rmd$ (but with raised indices, which is allowed in an inertial coordinate system) and move the $\nu$ derivative over:
    \begin{align*}
        \frac{1}{2} F_{\nu\lambda}\partial^\mu (\rmd\delta F)^{\nu\lambda} &= \frac{1}{2} F_{\nu\lambda} \partial^\mu \big( \partial^\nu (\delta F)^\lambda - \partial^\lambda (\delta F)^\nu \big) \\
        &= F_{\nu\lambda}\partial^\mu \partial^\nu (\delta F)^\lambda \\
        &= \partial^\nu \big( F_{\nu\lambda} \partial^\mu (\delta F)^\lambda \big) - \big(\partial^\mu(\delta F)^\lambda\big)\big( \partial^\nu F_{\nu\lambda} \big) \\
        &= \eta^{\nu\alpha}\partial_\alpha \big( F_{\nu\lambda} \partial^\mu (\delta F)^\lambda \big) - \big(\partial^\mu(\delta F)^\lambda\big)\big( \partial_\nu \tensor{F}{^\nu_\lambda} \big) \\
        &= \eta^{\alpha\nu}\partial_\nu \big( F_{\alpha\lambda} \partial^\mu (\delta F)^\lambda \big) + \big(\partial^\mu(\delta F)^\lambda\big)(\delta F)_\lambda \\
        &= \partial_\nu \big( \tensor{F}{^\nu_\lambda} \partial^\mu (\delta F)^\lambda \big) + \frac{1}{2}\big(\partial^\mu(\delta F)^\lambda\big)(\delta F)_\lambda + \frac{1}{2}\big(\partial^\mu(\delta F)_\lambda\big)(\delta F)^\lambda \\
        &= \partial_\nu W^{\mu\nu} + \frac{1}{2}\partial^\mu\big( (\delta F)^\alpha(\delta F)_\alpha \big) \qquad \text{where}\quad\boxed{ W^{\mu\nu} := \tensor{F}{^\nu_\lambda} \partial^\mu (\delta F)^\lambda } \\
        &= \partial_\nu W^{\mu\nu} + \frac{1}{2} \eta^{\mu\nu}\partial_\nu \big( (\delta F)^\alpha(\delta F)_\alpha \big) \\
        &= \partial_\nu W^{\mu\nu} + \partial_\nu V^{\mu\nu}_{(3)} \qquad \text{where}\quad\boxed{V^{\mu\nu}_{(3)} := \frac{1}{2} \eta^{\mu\nu}\big( (\delta F)^\alpha(\delta F)_\alpha \big)}.
    \end{align*}
    We have now achieved the decomposition
    \begin{equation} \label{SET_proof2}
        \tensor{F}{^\mu_\lambda}\partial_\nu (\rmd\delta F)^{\nu\lambda} = \partial_\nu \big( V^{\mu\nu}_{(1)} + V^{\mu\nu}_{(2)} + V^{\mu\nu}_{(3)} + W^{\mu\nu}\big),
    \end{equation}
    but there are parts of it that are still not symmetric. We need to work with $W^{\mu\nu}$ a little more:
    \begin{align*}
        W^{\mu\nu} &= \tensor{F}{^\nu_\lambda} \partial^\mu (\delta F)^\lambda \\
        &= \tensor{F}{^\nu_\lambda} (\rmd\delta F)^{\mu\lambda} + \tensor{F}{^\nu_\lambda} \partial^\lambda (\delta F)^\mu \\
        &= V^{\mu\nu}_{(4)} + \partial^\lambda\big(\tensor{F}{^\nu_\lambda} (\delta F)^\mu\big) - (\delta F)^\mu\partial^\lambda \tensor{F}{^\nu_\lambda} \qquad \text{where}\quad\boxed{ V^{\mu\nu}_{(4)} := \tensor{F}{^\nu_\lambda} (\rmd\delta F)^{\mu\lambda} } \\
        &= V^{\mu\nu}_{(4)} + \partial_\lambda\big(F^{\nu\lambda} (\delta F)^\mu\big) - (\delta F)^\mu\partial_\lambda F^{\nu\lambda} \\
        &= V^{\mu\nu}_{(4)} + \partial_\lambda\big(F^{\nu\lambda} (\delta F)^\mu\big) - (\delta F)^\mu(\delta F)^\nu \\
        &= V^{\mu\nu}_{(4)} + V^{\mu\nu}_{(5)} + \partial_\lambda\big(F^{\nu\lambda} (\delta F)^\mu\big) \qquad \text{where}\quad\boxed{V^{\mu\nu}_{(5)} := - (\delta F)^\mu(\delta F)^\nu}.
    \end{align*}
    We recognize that $V_{(4)}$ is the promised counterpart of $V_{(1)}$, that is, $V_{(1)}^{\mu\nu}+V_{(4)}^{\mu\nu}$ is symmetric. Thus~\eqref{SET_proof2} finally turns into
    \begin{equation*}
        \tensor{F}{^\mu_\lambda}\partial_\nu (\rmd\delta F)^{\nu\lambda} = \partial_\nu \bigg( V^{\mu\nu}_{(1)} + V^{\mu\nu}_{(2)} + V^{\mu\nu}_{(3)} + V^{\mu\nu}_{(4)} + V^{\mu\nu}_{(5)} + \partial_\lambda\big(F^{\nu\lambda} (\delta F)^\mu \big) \bigg) = \partial_\nu \sum_{j=1}^5 V^{\mu\nu}_{(j)},
    \end{equation*}
    where $\partial_\nu\partial_\lambda\big(F^{\nu\lambda} (\delta F)^\mu \big)$ vanished due to symmetry of $\partial_\nu\partial_\lambda$ combined with antisymmetry of $F^{\nu\lambda}$. Then we can define $4\pi T^{\mu\nu}_{\mathrm{(BLTP)}}$ as the sum of the $V_{(j)}$ terms, which is a symmetric expression that yields the formula claimed in the proposition.
\end{proof}

By now plugging in the components of $F$ as in~\eqref{F_components} into formula~\eqref{Tmunu_SR} for $(T^{\mu\nu})$, we finally obtain exactly the claimed formula~\eqref{e_SR} for $T^{00} = \frake$ (see~\eqref{Tmunu_components}). Indeed, the individual terms of~\eqref{Tmunu_SR} for $\mu=\nu=0$ come out as
\begin{align*}
    &F^{\alpha0}\tensor{F}{_\alpha^0} = \av{\bm{E}}^2, & &- \dfrac{\eta^{00}}{4} F_{\alpha\beta}F^{\alpha\beta} = -\dfrac{\av{\bm{E}}^2}{2} + \dfrac{\av{\bm{B}}^2}{2}, \\
    &2\varkappa^{-2}F^{\alpha0}\tensor{(\rmd\delta F)}{_\alpha^0} = -2\varkappa^{-2}\bm{E}\cdot\square\bm{E}, & &\varkappa^{-2}\dfrac{\eta^{00}}{2}F_{\alpha\beta}(\rmd\delta F)^{\alpha\beta} = \varkappa^{-2}\big(\bm{E}\cdot\square\bm{E} - \bm{B}\cdot\square\bm{B}\big), \\
    &- \varkappa^{-2}(\delta F)^0(\delta F)^0 = -\varkappa^{-2}(\nabla\cdot\bm{E})^2, & &\varkappa^{-2}\dfrac{\eta^{00}}{2} (\delta F)^\alpha (\delta F)_\alpha = \varkappa^{-2} \left(\dfrac{(\nabla\cdot\bm{E})^2}{2} - \dfrac{1}{2}\left| \nabla\times\bm{B} - \dfrac{1}{c}\dfrac{\partial\bm{E}}{\partial t} \right|^2 \right),
\end{align*}
whose sum is the right side of~\eqref{e_SR}.

\section{Electromagnetism in General Relativity}
\label{sec_EM_GR}

The covariant formulation of electromagnetism in Special Relativity applies promptly to General Relativity. In this section we elaborate on this topic, also explaining how the EM fields couple to gravity via the Einstein equations. Our goal is to derive the GR equations for the spacetime of a static point charge under the Maxwell-Maxwell and Maxwell-BLTP theories of electromagnetism, but we initially let the source term $J$ be general in the discussion.

\subsection{Covariant formulation of electromagnetism in GR}
\label{subsec_GR_intro}

Let $\mathcal{M}$ be an oriented 4-dimensional Lorentzian manifold with a metric $g$. Let a 1-form field $J\in\extprod{1}{\mathcal{M}}$ be given satisfying $\delta J = 0$. Then the \textbf{Maxwell equations} and \textbf{vacuum law} for two unknown 2-form fields $F,M\in \extprod{2}{\mathcal{M}}$ with source term $J$ are by definition the same covariant equations~\eqref{maxwell_cov} and~\eqref{vacuum_cov} as in the SR case, with $\varkappa = \infty$ yielding the conventional Maxwell theory and $0 < \varkappa < \infty$ the BLTP theory. As before, the equation $\rmd F = 0$ implies the local existence of a 1-form field $A\in\extprod{1}{\mathcal{M}}$ such that $F = \rmd A$, and it is now easy to check that the remaining EM equations are given by the Euler-Lagrange equation of the same Lagrangian $\mathcal{L}_{\mathrm{EM}} = \mathcal{L}_{\mathrm{EM}}(A)$ given in~\eqref{lagrangian_4form}. For this purpose, the proof of Proposition~\ref{prop_lagrangian} simply needs to be adapted with the inclusion of factors of $\sqrt{\av{g}}$ and its inverse at the appropriate places, that is, in the Euler-Lagrange equation~\eqref{euler_lagrange} itself and everywhere where the coordinate expression~\eqref{delta_coordinates} of $\delta$ was used. Note how we have been careful in that proof to only apply~\eqref{d_coordinates} to tensors with lower indices and~\eqref{delta_coordinates} to those with upper indices, as would be required in a general manifold. We skip the details of this proof.

\subsection{The Hilbert stress-energy tensor}
\label{subsec_GR_energy}

Since our EM theories (with either $\varkappa = \infty$ or $0<\varkappa<\infty$) are derived from a Lagrangian $\mathcal{L}_{\mathrm{EM}}$, the standard procedure of varying $\mathcal{L}_{\mathrm{EM}}$ with respect to the inverse metric $(g^{\mu\nu})$ can be applied to produce a canonical definition for the \textit{covariant} stress-energy tensor --- that is, $(T_{\mu\nu})$ as opposed to $(T^{\mu\nu})$ --- that acts as the source of gravity in GR. More precisely: Fix a 2-form $F = \rmd A\in\extprod{2}{\mathcal{M}}$ that solves our EM equations, and now consider the Lagrangian $\mathcal{L}_{\mathrm{EM}}$ given in~\eqref{lagrangian_4form_F} to be a function of the inverse metric $g^{-1}\in\mathcal{T}^2_0(\mathcal{M})$. In particular, components of the differential-form fields $A$ and $F$ with lower indices are considered independent of $g^{-1}$, while those with at least one raised index depend on it via the operation of index raising. Given an arbitrary open set $U\subseteq\mathcal{M}$, define the \textbf{Einstein-Hilbert action}
\begin{equation} \label{EH}
    \mathcal{S}_{\text{EH}}(g^{-1};U) := \int_U \left(\frac{c^4}{16\pi G} R \ \mathrm{vol}_g + \mathcal{L}_{\mathrm{EM}}\right) = \int_U \left( \frac{c^4}{16\pi G} R + \ell_{\mathrm{EM}} \right) \sqrt{\av{g}} \ \rmd^4 x,
\end{equation}
where $G$ and $c$ are Newton's gravitational constant and the speed of light, and where the Ricci scalar $R$ is as given in~\eqref{Ricci_Christoffel}.
Consider also a differentiable one-parameter family of inverse metrics $g^{-1}_{(s)}\in\mathcal{T}^2_0(\mathcal{M})$, all coinciding along the boundary of $U$, with the parameter $s$ varying in some small interval $(-\eps,\eps)$. We denote $g_{(0)}^{-1}$ simply by $g^{-1}$. For a quantity $f = f(g^{-1},\partial g^{-1},\partial^2 g^{-1})$ that depends on the components $g^{\mu\nu}$ of an arbitrary inverse metric and on their first-- and second-order derivatives $\partial_\xi g^{\mu\nu}$, $\partial_{\rho}\partial_\xi g^{\mu\nu}$ in some coordinate system, one defines the \textbf{variational derivative}
\begin{equation*}
    \delta f := \frac{\rmd}{\rmd s}\bigg|_{s=0} f\big( g^{-1}_{(s)}, \partial g^{-1}_{(s)}, \partial^2 g^{-1}_{(s)} \big) = \frac{\partial f}{\partial g^{\mu\nu}} \delta g^{\mu\nu} + \frac{\partial f}{\partial(\partial_\xi g^{\mu\nu})} \partial_\xi(\delta g^{\mu\nu}) + \frac{\partial f}{\partial(\partial_\rho\partial_\xi g^{\mu\nu})} \partial_\rho\partial_\xi(\delta g^{\mu\nu}).
\end{equation*}
The Chain Rule was used for the last equality; note that $\delta$ commutes with coordinate derivatives $\partial_\xi$. Then it is well known (see~\cite{straumann}, p.~85) that the variation of the first term of~\eqref{EH} produces an integrand containing the Einstein tensor, calculated with respect to the $s=0$ metric $g_{(0)} = g$. More precisely:
\begin{equation*}
    \delta\int_U \frac{c^4}{16\pi G}R\sqrt{\av{g}} \ \rmd^4 x = \frac{c^4}{16\pi G}\int_U \left( R_{\mu\nu} - \frac{1}{2}Rg_{\mu\nu} \right) (\delta g^{\mu\nu}) \sqrt{\av{g}} \ \rmd^4 x
\end{equation*}
Consequently, the Einstein equations~\eqref{EFE} follow from the action principle $\delta \mathcal{S}_{\mathrm{EM}} = 0$ if one defines the so-called \textbf{Hilbert stress-energy tensor} $T\in \mathcal{T}_2^0(\mathcal{M})$ by the formula
\begin{equation} \label{Hilbert_Tmunu}
    \delta \int_U \ell_{\mathrm{EM}}\sqrt{\av{g}} \ \rmd^4 x =: -\frac{1}{2}\int_U T_{\mu\nu} \ (\delta g^{\mu\nu}) \sqrt{\av{g}} \ \rmd^4 x.
\end{equation}
This object automatically comes out \textit{symmetric} ($T_{\mu\nu} = T_{\nu\mu}$) and \textit{divergence-free} ($\nabla^\nu T_{\mu\nu} = 0$), independently from the Einstein equations --- see~\cite{straumann}, p.~91. We now check that, for BLTP theory, it generalizes the special-relativistic $(T^{\mu\nu})$ that we had defined in~\eqref{Tmunu_SR}:

\begin{proposition} \label{prop_Hilbert_Tmunu}
    The components $T_{\mu\nu}$ of the Hilbert stress-energy tensor defined by~\eqref{Hilbert_Tmunu} are given in any coordinate system by
    \begin{equation} \label{Tmunu_GR}
        T_{\mu\nu} = \frac{1}{c}T_{\mu\nu}^{(1)} + \frac{1}{4\pi}T_{\mu\nu}^{(2)} + \frac{\varkappa^{-2}}{4\pi}T_{\mu\nu}^{(3)}
    \end{equation}
    for
    \begin{align*}
        &T_{\mu\nu}^{(1)} = g_{\mu\nu} J^\alpha A_\alpha - J_\mu A_\nu - J_\nu A_\mu, \\
        &T_{\mu\nu}^{(2)} = \tensor{F}{^\alpha_\mu}F_{\alpha\nu} - \frac{g_{\mu\nu}}{4} F_{\alpha\beta}F^{\alpha\beta}, \\
        &T_{\mu\nu}^{(3)} = \tensor{F}{^\alpha_\mu}(\rmd\delta F)_{\alpha\nu} + \tensor{F}{^\alpha_\nu}(\rmd\delta F)_{\alpha\mu} - \frac{g_{\mu\nu}}{2}F_{\alpha\beta}(\rmd\delta F)^{\alpha\beta} - (\delta F)_\mu(\delta F)_\nu + \frac{g_{\mu\nu}}{2} (\delta F)_\alpha (\delta F)^\alpha.
    \end{align*}
\end{proposition}

Apart from $T^{(1)}$, which is present due to the energy and momentum associated with the motion of the charges modeled by $J$, we see that the remaining terms are exactly as in~\eqref{Tmunu_SR}, only with lowered indices and with $g$ replacing $\eta$.

\begin{proof}
    We start by finding formulas for derivatives of metric-related terms with respect to the inverse-metric components. They will be used matter-of-factly throughout this proof. To begin, we differentiate the relation $g_{\tau\theta}g^{\theta\rho} = \delta_\tau^\rho$ to get
    \begin{equation*}
        0 = \frac{\partial g_{\tau\theta}}{\partial g^{\mu\nu}}g^{\theta\rho} + g_{\tau\theta}\frac{\partial g^{\theta\rho}}{\partial g^{\mu\nu}} = \frac{\partial g_{\tau\theta}}{\partial g^{\mu\nu}}g^{\theta\rho} + g_{\tau\theta}\delta^{\theta}_\mu\delta^{\rho}_\nu = \frac{\partial g_{\tau\theta}}{\partial g^{\mu\nu}}g^{\theta\rho} + g_{\tau\mu}\delta^\rho_\nu,
    \end{equation*}
    and we contract with $g_{\rho\sigma}$ to find
    \begin{equation} \label{delgdelg}
        0 = \frac{\partial g_{\tau\theta}}{\partial g^{\mu\nu}}\delta^\theta_\sigma + g_{\tau\mu}g_{\nu\sigma} = \frac{\partial g_{\tau\sigma}}{\partial g^{\mu\nu}} + g_{\tau\mu}g_{\nu\sigma} \quad\Longrightarrow\quad \frac{\partial g_{\tau\sigma}}{\partial g^{\mu\nu}} = -g_{\tau\mu}g_{\nu\sigma}.
    \end{equation}
    Next, using the well-known \textit{Jacobi formula} $\partial_{g_{\mu\nu}}(\det g) = (\det g) g^{\nu\mu}$ for the derivative of a determinant, together with the Chain Rule and~\eqref{delgdelg}, we compute
    \begin{equation*}
        \frac{\partial(\det g)}{\partial g^{\mu\nu}} = \frac{\partial(\det g)}{\partial g_{\alpha\beta}}\frac{\partial g_{\alpha\beta}}{\partial g^{\mu\nu}} = -(\det g) g^{\beta\alpha}g_{\alpha\mu}g_{\nu\beta} = -(\det g) \delta^\beta_\mu g_{\nu\beta} = -(\det g) g_{\nu\mu}.
    \end{equation*}
    Remembering that $\av{g} = -\det g$ and also considering $g_{\nu\mu} = g_{\mu\nu}$, we rewrite this identity as
    \begin{equation} \label{deldetg}
        \frac{\partial\av{g}}{\partial g^{\mu\nu}} = -\av{g} g_{\mu\nu},
    \end{equation}
    which immediately gives
    \begin{equation} \label{delsqrtdetg}
        \frac{\partial\sqrt{\av{g}}}{\partial g^{\mu\nu}} = \frac{1}{2\sqrt{\av{g}}} \frac{\partial\av{g}}{\partial g^{\mu\nu}} = -\frac{\sqrt{\av{g}}}{2} g_{\mu\nu}.
    \end{equation}
    Then note, again via the Chain Rule, that
    \begin{equation*}
        \frac{\partial\sqrt{\av{g}}}{\partial x^\xi} = \frac{\partial\sqrt{\av{g}}}{\partial g^{\sigma\tau}}\frac{\partial g^{\sigma\tau}}{\partial x^\xi} = -\frac{\sqrt{\av{g}}}{2} g_{\sigma\tau} \partial_\xi g^{\sigma\tau} = \Gamma^\sigma_{\sigma\xi}\sqrt{\av{g}}.
    \end{equation*}
    In the last step we used the identity
    \begin{equation} \label{special_christoffel}
        \Gamma^\sigma_{\sigma\xi} = -\frac{1}{2} g_{\sigma\tau} (\partial_\xi g^{\sigma\tau}),
    \end{equation}
    whose proof follows from the definition~\eqref{Ricci_Christoffel} for the Christoffel symbols:
    \begin{equation*} \label{special_christoffel}
        \Gamma^\sigma_{\sigma\xi} = \frac{g^{\sigma\lambda}}{2} (\partial_\sigma g_{\lambda\xi} + \partial_\xi g_{\lambda\sigma} - \partial_\lambda g_{\sigma\xi}) = \frac{1}{2}g^{\sigma\lambda}(\partial_\xi g_{\sigma\lambda}) = -\frac{1}{2}g_{\sigma\lambda}(\partial_\xi g^{\sigma\lambda}),
    \end{equation*}
    with the last step being a consequence of the product rule applied to $0 = \partial_\xi (4) = \partial_\xi (g_{\sigma\lambda}g^{\sigma\lambda})$.
    
    Having noted down these identities, we can begin expanding $T_{\mu\nu}$ from its definition~\eqref{Hilbert_Tmunu}:
    \begin{align*}
        \int_U T_{\mu\nu} (\delta g^{\mu\nu})\sqrt{\av{g}} \ \rmd^4 x &= -2 \delta\int_U \ell_{\mathrm{EM}}\sqrt{\av{g}} \ \rmd^4 x \\
        &= -2 \int_U \bigg( \ell_{\mathrm{EM}}\delta\big(\sqrt{\av{g}}\big) + (\delta\ell_{\mathrm{EM}})\sqrt{\av{g}} \bigg) \ \rmd^4 x  \\
        &= -2\int_U \bigg( \ell_{\mathrm{EM}} \frac{\partial\sqrt{\av{g}}}{\partial g^{\mu\nu}} (\delta g^{\mu\nu}) + \frac{\partial\ell_{\mathrm{EM}}}{\partial g^{\mu\nu}} (\delta g^{\mu\nu})\sqrt{\av{g}} + \frac{\partial\ell_{\mathrm{EM}}}{\partial(\partial_\xi g^{\mu\nu})} \partial_\xi(\delta g^{\mu\nu})\sqrt{\av{g}}\bigg) \ \rmd^4 x \\
        &= -2\int_U \bigg( \ell_{\mathrm{EM}} \frac{\partial\sqrt{\av{g}}}{\partial g^{\mu\nu}} + \frac{\partial\ell_{\mathrm{EM}}}{\partial g^{\mu\nu}} \sqrt{\av{g}} - \partial_\xi\left( \frac{\partial\ell_{\mathrm{EM}}}{\partial(\partial_\xi g^{\mu\nu})}\sqrt{\av{g}} \right) \bigg) (\delta g^{\mu\nu}) \ \rmd^4 x \\
        &= -2\int_U \bigg( \frac{1}{\sqrt{\av{g}}}\ell_{\mathrm{EM}} \frac{\partial\sqrt{\av{g}}}{\partial g^{\mu\nu}} + \frac{\partial\ell_{\mathrm{EM}}}{\partial g^{\mu\nu}} \ - \frac{1}{\sqrt{\av{g}}}\partial_\xi\left( \frac{\partial\ell_{\mathrm{EM}}}{\partial(\partial_\xi g^{\mu\nu})}\sqrt{\av{g}} \right) \bigg) (\delta g^{\mu\nu}) \sqrt{\av{g}} \ \rmd^4 x,
    \end{align*}
    where integration by parts in the $\xi$ variable was used once, as well as the fact that $\ell_{\mathrm{EM}}$ only depends on metric-component derivatives up to order 1. Comparing the integrands of the expressions at the start and the end, and considering the arbitrariness of the variation $\delta g_{\mu\nu}$, we find the explicit definition of $T_{\mu\nu}$, valid in any coordinate system:
    \begin{align*}
        T_{\mu\nu} &= -\frac{2}{\sqrt{\av{g}}}\ell_{\mathrm{EM}} \frac{\partial\sqrt{\av{g}}}{\partial g^{\mu\nu}} -2 \frac{\partial\ell_{\mathrm{EM}}}{\partial g^{\mu\nu}} 
        + \frac{2}{\sqrt{\av{g}}} \partial_\xi\left( \frac{\partial\ell_{\mathrm{EM}}}{\partial(\partial_\xi g^{\mu\nu})}\sqrt{\av{g}} \right) \\
        &= -\frac{2}{\sqrt{\av{g}}}\ell_{\mathrm{EM}}\left( 
        -\frac{\sqrt{\av{g}}}{2}g_{\mu\nu} \right) - 2\frac{\partial\ell_{\mathrm{EM}}}{\partial g^{\mu\nu}} + \frac{2}{\sqrt{\av{g}}}\left( \partial_\xi\left(\frac{\partial\ell_{\mathrm{EM}}}{\partial(\partial_\xi g^{\mu\nu})}\right)\sqrt{\av{g}} + \frac{\partial\ell_{\mathrm{EM}}}{\partial(\partial_\xi g^{\mu\nu})} \big(\partial_\xi \sqrt{\av{g}}\big) \right) \\
        &= g_{\mu\nu}\ell_{\mathrm{EM}} - 2\frac{\partial\ell_{\mathrm{EM}}}{\partial g^{\mu\nu}} + 2\partial_\xi\left( \frac{\partial\ell_{\mathrm{EM}}}{\partial(\partial_\xi g^{\mu\nu})} \right) + 2\Gamma^\sigma_{\sigma\xi} \frac{\partial\ell_{\mathrm{EM}}}{\partial(\partial_\xi g^{\mu\nu})}.
    \end{align*}
    Now consider the coordinate expression of our Lagrangian scalar $\ell_{\mathrm{EM}}$, given in~\eqref{proof_ell}:
    \begin{equation*}
        \ell_{\mathrm{EM}} = \frac{1}{c} J^\alpha A_\alpha -\frac{1}{16\pi} g^{\alpha\zeta}g^{\beta\eta} F_{\alpha\beta} F_{\zeta\eta} - \frac{\varkappa^{-2}}{8\pi} g_{\alpha\zeta} (\delta F)^\alpha (\delta F)^\zeta = \frac{1}{c}\ell_{\mathrm{EM}}^{(1)} - \frac{1}{16\pi} \ell_{\mathrm{EM}}^{(2)} - \frac{\varkappa^{-2}}{8\pi}\ell_{\mathrm{EM}}^{(3)}.
    \end{equation*}
    Define accordingly
    \begin{equation} \label{Tmunu_allj}
        T_{\mu\nu} = \frac{1}{c}\widetilde{T}_{\mu\nu}^{(1)} - \frac{1}{16\pi}\widetilde{T}_{\mu\nu}^{(2)} - \frac{\varkappa^{-2}}{8\pi}\widetilde{T}_{\mu\nu}^{(3)}
    \end{equation}
    for
    \begin{equation} \label{Tmunuj}
        \widetilde{T}_{\mu\nu}^{(j)} = g_{\mu\nu}\ell_{\mathrm{EM}}^{(j)} - 2\frac{\partial\ell_{\mathrm{EM}}^{(j)}}{\partial g^{\mu\nu}} + 2\partial_\xi\left( \frac{\partial\ell_{\mathrm{EM}}^{(j)}}{\partial(\partial_\xi g^{\mu\nu})} \right) + 2\Gamma^\sigma_{\sigma\xi} \frac{\partial\ell_{\mathrm{EM}}^{(j)}}{\partial(\partial_\xi g^{\mu\nu})}, \quad j=1,2,3.
    \end{equation}
    The $T^{(j)}_{\mu\nu}$ terms as stated in~\eqref{Tmunu_GR} are mere multiples of the ones with tildes, which are used in this proof simply for convenience:
    \begin{equation} \label{relation_Ttilde}
        T^{(1)}_{\mu\nu} = \widetilde{T}^{(1)}_{\mu\nu}, \quad T^{(2)}_{\mu\nu} = -\frac{1}{4}\widetilde{T}^{(2)}_{\mu\nu}, \quad T^{(3)}_{\mu\nu} = -\frac{1}{2}\widetilde{T}^{(3)}_{\mu\nu}.
    \end{equation}
    We now compute $\widetilde{T}^{(j)}_{\mu\nu}$ for $j=2$, $j=3$ and $j=1$ in this order. To begin, since $\ell_{\mathrm{EM}}^{(2)}$ only features $g^{\mu\nu}$ terms as opposed to $\partial_\xi g^{\mu\nu}$ terms, equation~\eqref{Tmunuj} simplifies down to
    \begin{align*}
        \widetilde{T}_{\mu\nu}^{(2)} &= g_{\mu\nu}\ell_{\mathrm{EM}}^{(2)} - 2\frac{\partial\ell_{\mathrm{EM}}^{(2)}}{\partial g^{\mu\nu}} \\
        &= \left( g_{\mu\nu}g^{\alpha\zeta}g^{\beta\eta} F_{\alpha\beta} F_{\zeta\eta} - 2\frac{\partial(g^{\alpha\zeta}g^{\beta\eta})}{\partial g^{\mu\nu}} F_{\alpha\beta} F_{\zeta\eta} \right) \\
        &= g_{\mu\nu}F_{\alpha\beta} F^{\alpha\beta} -2 \big(\delta^\alpha_\mu\delta^\zeta_\nu g^{\beta\eta} + g^{\alpha\zeta}\delta^\beta_\mu\delta^\eta_\nu\big) F_{\alpha\beta} F_{\zeta\eta} \\
        &= g_{\mu\nu}F_{\alpha\beta} F^{\alpha\beta} -2 g^{\beta\eta} F_{\mu\beta} F_{\nu\eta} -2 g^{\alpha\zeta} F_{\alpha\mu} F_{\zeta\nu} \\
        &= g_{\mu\nu}F_{\alpha\beta} F^{\alpha\beta} -2 g^{\beta\eta} F_{\beta\mu} F_{\eta\nu} -2 g^{\alpha\zeta} F_{\alpha\mu} F_{\zeta\nu} \\
        &= g_{\mu\nu}F_{\alpha\beta} F^{\alpha\beta} - 2 \tensor{F}{^\eta_\mu} F_{\eta\nu} -2 \tensor{F}{^\zeta_\mu} F_{\zeta\nu} \\
        &= g_{\mu\nu}F_{\alpha\beta} F^{\alpha\beta} -4 \tensor{F}{^\alpha_\mu} F_{\alpha\nu}.
    \end{align*}
    Considering~\eqref{relation_Ttilde}, this gives $T^{(2)}_{\mu\nu}$ exactly as stated in the proposition. The $j=3$ term now will require the bulk of the work, seeing as how the $\delta$ operator depends on terms of both types $g^{\mu\nu}$ and $\partial_\xi g^{\mu\nu}$ (we remind the reader that each $\partial_\xi g^{\mu\nu}$ is considered a variable that is independent from each $g^{\mu\nu}$). First we calculate the following derivatives:
    \begin{align*}
        \frac{\partial \ell_{\mathrm{EM}}^{(3)}}{\partial g^{\mu\nu}} &= \frac{\partial}{\partial g^{\mu\nu}}\big( g_{\alpha\zeta} (\delta F)^\alpha (\delta F)^\zeta \big) \\
        &= \frac{\partial g_{\alpha\zeta}}{g^{\mu\nu}} (\delta F)^\alpha (\delta F)^\zeta + g_{\alpha\zeta} \frac{\partial (\delta F)^\alpha}{\partial g^{\mu\nu}} (\delta F)^\zeta + g_{\alpha\zeta} (\delta F)^\alpha \frac{\partial (\delta F)^\zeta}{\partial g^{\mu\nu}} \\
        &= -g_{\alpha\mu}g_{\nu\zeta}(\delta F)^\alpha(\delta F)^\zeta + \frac{\partial (\delta F)^\alpha}{\partial g^{\mu\nu}} (\delta F)_\alpha + \frac{\partial (\delta F)^\zeta}{\partial g^{\mu\nu}} (\delta F)_\zeta \\
        &= -(\delta F)_\mu (\delta F)_\nu + 2\frac{\partial (\delta F)^\alpha}{\partial g^{\mu\nu}} (\delta F)_\alpha, \\
        \frac{\partial \ell_{\mathrm{EM}}^{(3)}}{\partial (\partial_\xi g^{\mu\nu})} &= \frac{\partial}{\partial(\partial_\xi g^{\mu\nu})} \big( g_ {\alpha\zeta} (\delta F)^\alpha (\delta F)^\zeta \big) \\
        &= g_{\alpha\zeta} \frac{\partial (\delta F)^\alpha}{\partial(\partial_\xi g^{\mu\nu})} (\delta F)^\zeta + g_{\alpha\zeta} (\delta F)^\alpha \frac{\partial (\delta F)^\zeta}{\partial(\partial_\xi g^{\mu\nu})} \\
        &= \frac{\partial (\delta F)^\alpha}{\partial(\partial_\xi g^{\mu\nu})} (\delta F)_\alpha + \frac{\partial (\delta F)^\zeta}{\partial(\partial_\xi g^{\mu\nu})} (\delta F)_\zeta \\
        &= 2\frac{\partial (\delta F)^\alpha}{\partial(\partial_\xi g^{\mu\nu})} (\delta F)_\alpha,
    \end{align*}
    which we plug into the definition~\eqref{Tmunuj} of $\widetilde{T}^{(3)}_{\mu\nu}$ to obtain
    \begin{align}
        \label{Tmunu3} \widetilde{T}_{\mu\nu}^{(3)} &= g_{\mu\nu} g_{\alpha\zeta} (\delta F)^\alpha (\delta F)^\zeta - 2\left( -(\delta F)_\mu (\delta F)_\nu + 2 \frac{\partial (\delta F)^\alpha}{\partial g^{\mu\nu}} (\delta F)_\alpha \right) + 2\partial_\xi\left( 2\frac{\partial (\delta F)^\alpha}{\partial(\partial_\xi g^{\mu\nu})} (\delta F)_\alpha \right) \\
        \nonumber &\qquad + 2 \Gamma^\sigma_{\sigma\xi} \left( 2\frac{\partial (\delta F)^\alpha}{\partial(\partial_\xi g^{\mu\nu})} (\delta F)_\alpha\right) \nonumber \\
        &= 2(\delta F)_\mu (\delta F)_\nu + 4 \frac{\partial (\delta F)^\alpha}{\partial(\partial_\xi g^{\mu\nu})}\partial_\xi(\delta F)_\alpha \label{Tmunu3_part1} \\
        &\qquad + \left( g_{\mu\nu} (\delta F)^\alpha - 4\frac{\partial (\delta F)^\alpha}{\partial g^{\mu\nu}} + 4\Gamma^\sigma_{\sigma\xi} \frac{\partial (\delta F)^\alpha}{\partial(\partial_\xi g^{\mu\nu})} + 4\partial_\xi\left( \frac{\partial (\delta F)^\alpha}{\partial(\partial_\xi g^{\mu\nu})} \right) \right) (\delta F)_\alpha. \label{Tmunu3_part2}
    \end{align}
    To further expand this, we need a formula for $(\delta F)^\alpha$ showing its explicit dependence on $g^{\mu\nu}$ and $\partial_\xi g^{\mu\nu}$ terms. We have, from~\eqref{delta_coordinates},
    \begin{align}
        (\delta F)^\alpha &= -\frac{1}{\sqrt{\av{g}}} \partial_\lambda \big( \sqrt{\av{g}} F^{\lambda\alpha} \big) \label{deltaF_alpha_initial} \\
        &= -\frac{1}{\sqrt{\av{g}}} \big( \Gamma^\sigma_{\sigma\lambda}\sqrt{\av{g}} F^{\lambda\alpha} + \sqrt{\av{g}} \partial_\lambda F^{\lambda\alpha} \big) \nonumber \\
        &= -\Gamma^{\sigma}_{\sigma\lambda} F^{\lambda\alpha} - \partial_\lambda F^{\lambda\alpha} \label{deltaF_alpha_forlater} \\
        &= \frac{1}{2} g_{\sigma\tau}(\partial_\lambda g^{\sigma\tau}) F^{\lambda\alpha} - \partial_\lambda F^{\lambda\alpha} \nonumber \\
        &= \frac{1}{2} g_{\sigma\tau}(\partial_\lambda g^{\sigma\tau}) g^{\lambda\zeta}g^{\alpha\eta} F_{\zeta\eta} - \partial_\lambda (g^{\lambda\zeta}g^{\alpha\eta} F_{\zeta\eta})  \nonumber \\
        &= \frac{1}{2} g_{\sigma\tau}(\partial_\lambda g^{\sigma\tau}) g^{\lambda\zeta}g^{\alpha\eta} F_{\zeta\eta} - (\partial_\lambda g^{\lambda\zeta})g^{\alpha\eta} F_{\zeta\eta} - g^{\lambda\zeta} (\partial_\lambda g^{\alpha\eta}) F_{\zeta\eta} - g^{\lambda\zeta}g^{\alpha\eta} (\partial_\lambda F_{\zeta\eta}). \label{deltaF_alpha}
    \end{align}
    Now we must differentiate this with respect to $g^{\mu\nu}$ and $\partial_\xi g^{\mu\nu}$ for fixed values of $\mu,\nu,\xi$ and plug the results into~\eqref{Tmunu3}.
    
    Here we must pause to talk about a subtle, often overlooked issue that arises in such computations: the use of the symmetry $g^{\mu\nu} = g^{\nu\mu}$ \textit{before} derivatives with respect to $g^{\mu\nu}$ are taken. It usually leads to a non-symmetric $(T^{\mu\nu})$, prompting many authors to simply symmetrize the end result without further explanation. We wish to clarify what the problem is and show why this last-step symmetrization fixes it. As an illustration, suppose that we have a symmetric function $f = f(u,w)$ --- that is, $f(u,w) = f(w,u)$, which implies $\partial_u f(u,w) = \partial_w f(w,u)$ --- and we wish to find $(\partial_u f)(u,u)$. This symbol means that first $f(u,w)$ must be differentiated with respect to $u$, and then the arguments $(u,w)$ must be set equal to each other. Let $f$ be given as a sum of many symmetric terms, but imagine that we, much like a freshman Calculus student, mistakenly set $u=w$ into some of them (but maybe not \textit{all} of them) before finding the $u$-derivative, then we proceeded to differentiate and set $u=w$ into the remaining parts. That is, letting $f(u,w) = g(u,w) + h(u,w)$ for symmetric $g$ and $h$, we originally wished to find
    \begin{equation} \label{fuv_1}
        (\partial_u f)(u,u) = (\partial_u g)(u,u) + (\partial_u h)(u,u),
    \end{equation}
    but instead we have first set $u=w$ into the $g$ part of $f$ (but not into the $h$ part), then differentiated, and then set $u=w$ into the $h$ part, producing the different result
    \begin{equation} \label{fuv_2}
        \partial_u \big( g(u,u) + h(u,w) \big)\bigg|_{u=w} = (\partial_u g)(u,u) + (\partial_w g)(u,u) + (\partial_u h)(u,w)\bigg|_{u=w} = 2(\partial_u g)(u,u) + (\partial_u h)(u,u).
    \end{equation}
    It turns out that our mistake is easy to fix without having to throw away our work: The desired expression~\eqref{fuv_1} can be recovered from the expression~\eqref{fuv_2} that we currently have by adding to it the $w$ derivative \textit{of the same expression} $g(u,u) + h(u,w)$, then dividing by 2 and setting $u=w$:
    \begin{align*}
        \frac{\partial_u \big( g(u,u) + h(u,w) \big) + \partial_w \big( g(u,u) + h(u,w) \big)}{2}\bigg|_{u=w} &= \frac{2(\partial_u g)(u,u) + (\partial_u h)(u,w) + 0 + (\partial_w h)(u,w)}{2}\bigg|_{u=w} \\
        &= \frac{2(\partial_u g)(u,u) + (\partial_u h)(u,w) + (\partial_u h)(w,u)}{2}\bigg|_{u=w} \\
        &= (\partial_u g)(u,u) + (\partial_u h)(u,u).
    \end{align*}
    For a similar reason, when taking derivatives of formula~\eqref{deltaF_alpha} for $(\delta F)^\alpha$ with respect to each $g^{\mu\nu}$ or $\partial_\xi g^{\mu\nu}$ term, we will \textit{not} necessarily obtain the correct answer, due to the fact that we have already used the property $g^{\alpha\beta} = g^{\beta\alpha}$ at several different places (but not all possible places) when calculating~\eqref{deltaF_alpha}. We used it, for example, in the identity~\eqref{special_christoffel} for $\Gamma^\sigma_{\sigma\lambda}$. However, noting that $(\delta F)^\alpha$ as given in~\eqref{deltaF_alpha_initial} is manifestly symmetric as a function of all the pairs of variables $(g^{\alpha\beta},g^{\beta\alpha})$ and their coordinate derivatives, the remedy is just as explained above: To fix $\partial_{g^{\mu\nu}}(\delta F)^\alpha$ or $\partial_{\partial_\xi g^{\mu\nu}}(\delta F)^\alpha$, we must add to it the expression $\partial_{g^{\nu\mu}}(\delta F)^\alpha$ or $\partial_{\partial_\xi g^{\nu\mu}}(\delta F)^\alpha$, which has the roles of $\mu$ and $\nu$ reversed, then divide the result by 2, and then finish applying the symmetry property if needed. This operation will be denoted by the self-explanatory symbol $[+(\mu\leftrightarrow\nu),\olddiv 2,g=g^{\mathrm{T}}]$; for example,
    \begin{equation*}
        (3g^{\alpha\beta}+4g^{\beta\alpha})(F_{\alpha\mu}F_{\beta\nu}) \quad [+(\mu\leftrightarrow\nu),\olddiv 2,g=g^{\mathrm{T}}] \quad = \frac{7g^{\alpha\beta}(F_{\alpha\mu}F_{\beta\nu} + F_{\alpha\nu}F_{\beta\mu})}{2}.
    \end{equation*}
    To save the reader's printer ink, we write this symbol only at the end of long chains of equalities, but it is implied on all previous lines too.
    
    Now we return to the calculation. Differentiating~\eqref{deltaF_alpha} with respect to $g^{\mu\nu}$, we have
    \begin{align}
        \frac{\partial (\delta F)_\alpha}{\partial g^{\mu\nu}} &= \frac{1}{2} \frac{\partial g_{\sigma\tau}}{\partial g^{\mu\nu}} (\partial_\lambda g^{\sigma\tau}) g^{\lambda\zeta}g^{\alpha\eta} F_{\zeta\eta} + \frac{1}{2} g_{\sigma\tau}(\partial_\lambda g^{\sigma\tau}) \delta^\lambda_\mu\delta^\zeta_\nu g^{\alpha\eta} F_{\zeta\eta} + \frac{1}{2} g_{\sigma\tau}(\partial_\lambda g^{\sigma\tau}) g^{\lambda\zeta}\delta^\alpha_\mu\delta^\eta_\nu F_{\zeta\eta} \nonumber \\
        &\qquad - (\partial_\lambda g^{\lambda\zeta}) \delta^\alpha_\mu\delta^\eta_\nu F_{\zeta\eta} - \delta^\lambda_\mu\delta^\zeta_\nu (\partial_\lambda g^{\alpha\eta}) F_{\zeta\eta} - \delta^\lambda_\mu\delta^\zeta_\nu g^{\alpha\eta} (\partial_\lambda F_{\zeta\eta}) - g^{\lambda\zeta}\delta^\alpha_\mu\delta^\eta_\nu (\partial_\lambda F_{\zeta\eta}) \nonumber \\
        &= -\frac{1}{2} g_{\sigma\mu}g_{\tau\nu} (\partial_\lambda g^{\sigma\tau}) g^{\lambda\zeta}g^{\alpha\eta}F_{\zeta\eta} + \frac{1}{2} g_{\sigma\tau}(\partial_\mu g^{\sigma\tau}) g^{\alpha\eta} F_{\nu\eta} + \frac{1}{2} g_{\sigma\tau}(\partial_\lambda g^{\sigma\tau}) g^{\lambda\zeta}\delta^\alpha_\mu F_{\zeta\nu} \nonumber \\
        &\qquad - (\partial_\lambda g^{\lambda\zeta}) \delta^\alpha_\mu F_{\zeta\nu} - (\partial_\mu g^{\alpha\eta}) F_{\nu\eta} - g^{\alpha\eta} (\partial_\mu F_{\nu\eta}) - g^{\lambda\zeta}\delta^\alpha_\mu (\partial_\lambda F_{\zeta\nu}) \nonumber \\
        &= -\frac{1}{2} g_{\sigma\mu}g_{\tau\nu} (\partial_\lambda g^{\sigma\tau}) F^{\lambda\alpha} - \Gamma^{\sigma}_{\sigma\mu} g^{\alpha\eta} F_{\nu\eta} -\Gamma^\sigma_{\sigma\lambda} \delta^\alpha_\mu \tensor{F}{^\lambda_\nu} - \partial_\mu(g^{\alpha\eta} F_{\nu\eta}) -\delta^\alpha_\mu \partial_\lambda ( g^{\lambda\zeta} F_{\zeta\nu}) \nonumber \\
        &= -\frac{1}{2} g_{\sigma\mu}g_{\tau\nu} (\partial_\lambda g^{\sigma\tau}) F^{\lambda\alpha} + \Gamma^{\sigma}_{\sigma\mu} g^{\alpha\eta} F_{\eta\nu} -\Gamma^\sigma_{\sigma\lambda} \delta^\alpha_\mu \tensor{F}{^\lambda_\nu} + \partial_\mu(g^{\alpha\eta} F_{\eta\nu}) -\delta^\alpha_\mu \partial_\lambda ( g^{\lambda\zeta} F_{\zeta\nu}) \nonumber \\
        &= \frac{1}{2}(\partial_\lambda g_{\mu\nu}) F^{\lambda\alpha} + \partial_\mu \tensor{F}{^\alpha_\nu} +\Gamma^\sigma_{\sigma\mu} \tensor{F}{^\alpha_\nu} - \delta^\alpha_\mu \big( \partial_\lambda \tensor{F}{^\lambda_\nu} + \Gamma^\sigma_{\sigma\lambda} \tensor{F}{^\lambda_\nu} \big) \quad [+(\mu\leftrightarrow\nu),\olddiv 2,g=g^{\mathrm{T}}], \label{deltaF_alpha_munu}
    \end{align}
    where the change in the first term on the last step is explained by
    \begin{equation*}
        \partial_\lambda g_{\mu\nu} = \partial_\lambda \big( 
        g_{\sigma\mu}g_{\tau\nu}g^{\sigma\tau} \big) = (\partial_\lambda g_{\sigma\mu}) g_{\tau\nu}g^{\sigma\tau} + g_{\sigma\mu} (\partial_\lambda g_{\tau\nu}) g^{\sigma\tau} + g_{\sigma\mu}g_{\tau\nu} (\partial_\lambda g^{\sigma\tau}) = 2\partial_\lambda g_{\mu\nu} + g_{\sigma\mu}g_{\tau\nu} (\partial_\lambda g^{\sigma\tau}).
    \end{equation*}
    Similarly, differentiating~\eqref{deltaF_alpha} with respect to $\partial_\xi g^{\mu\nu}$, we have
    \begin{align}
        \frac{\partial (\delta F)_\alpha}{\partial(\partial_\xi g^{\mu\nu})} &= \frac{1}{2}g_{\sigma\tau} \delta^\xi_\lambda\delta^\sigma_\mu\delta^\tau_\nu g^{\lambda\zeta}g^{\alpha\eta} F_{\zeta\eta} - \delta^\xi_\lambda\delta^\lambda_\mu\delta^\zeta_\nu g^{\alpha\eta} F_{\zeta\eta} - g^{\lambda\zeta}\delta^\xi_\lambda\delta^\alpha_\mu\delta^\eta_\nu F_{\zeta\eta} \nonumber \\
        &= \frac{1}{2} g_{\mu\nu} g^{\xi\zeta} g^{\alpha\eta} F_{\zeta\eta} - \delta^\xi_\mu g^{\alpha\eta}F_{\nu\eta} - g^{\xi\zeta}\delta^\alpha_\mu F_{\zeta\nu} \nonumber \\
        &= \frac{1}{2} g_{\mu\nu} F^{\xi\alpha} + \delta^\xi_\mu \tensor{F}{^\alpha_\nu} - \delta^\alpha_\mu \tensor{F}{^\xi_\nu} \quad [+(\mu\leftrightarrow\nu),\olddiv 2,g=g^{\mathrm{T}}]. \label{deltaF_alpha_munuxi}
    \end{align}
    What is left is to plug~\eqref{deltaF_alpha_forlater},~\eqref{deltaF_alpha_munu} and~\eqref{deltaF_alpha_munuxi} back into formula~\eqref{Tmunu3} for $\widetilde{T}_{\mu\nu}^{(3)}$. The long term in parentheses on line~\eqref{Tmunu3_part2} simplifies considerably, and the symmetrization operation, carried over from the previous steps, turns out to not matter for this part:
    \begin{align*}
        g_{\mu\nu} &(\delta F)^\alpha - 4\frac{\partial (\delta F)^\alpha}{\partial g^{\mu\nu}} + 4\Gamma^\sigma_{\sigma\xi} \frac{\partial (\delta F)^\alpha}{\partial(\partial_\xi g^{\mu\nu})} + 4\partial_\xi\left( \frac{\partial (\delta F)^\alpha}{\partial(\partial_\xi g^{\mu\nu})} \right) \\
        &= -g_{\mu\nu}\Gamma^\sigma_{\sigma\lambda} F^{\lambda\alpha} -g_{\mu\nu}\partial_\lambda F^{\lambda\alpha} \\
        &\qquad - 2(\partial_\lambda g_{\mu\nu})F^{\lambda\alpha} - 4\partial_\mu \tensor{F}{^\alpha_\nu} - 4\Gamma^\sigma_{\sigma\mu}\tensor{F}{^\alpha_\nu} + 4\delta^\alpha_\mu \big(\partial_\lambda \tensor{F}{^\lambda_\nu} + \Gamma^\sigma_{\sigma\lambda}\tensor{F}{^\lambda_\nu}\big) \\
        &\qquad +2\Gamma^\sigma_{\sigma\xi}g_{\mu\nu} F^{\xi\alpha} + 4\Gamma^\sigma_{\sigma\xi}\delta^\xi_\mu \tensor{F}{^\alpha_\nu} - 4\Gamma^\sigma_{\sigma\xi}\delta^\alpha_\mu \tensor{F}{^\xi_\nu} \\
        &\qquad + 2 (\partial_\xi g_{\mu\nu}) F^{\xi\alpha} + 2 g_{\mu\nu} (\partial_\xi F^{\xi\alpha}) + 4\delta^\xi_\mu (\partial_\xi\tensor{F}{^\alpha_\nu}) - 4\delta^\alpha_\mu (\partial_\xi\tensor{F}{^\xi_\nu}) \\
        &= -g_{\mu\nu}\Gamma^\sigma_{\sigma\lambda} F^{\lambda\alpha} -g_{\mu\nu}\partial_\lambda F^{\lambda\alpha} \\
        &\qquad - 2(\partial_\lambda g_{\mu\nu})F^{\lambda\alpha} - 4\partial_\mu \tensor{F}{^\alpha_\nu} - 4\Gamma^\sigma_{\sigma\mu}\tensor{F}{^\alpha_\nu} + 4\delta^\alpha_\mu (\partial_\lambda \tensor{F}{^\lambda_\nu}) + 4\delta^\alpha_\mu\Gamma^\sigma_{\sigma\lambda}\tensor{F}{^\lambda_\nu} \\
        &\qquad +2\Gamma^\sigma_{\sigma\xi}g_{\mu\nu} F^{\xi\alpha} + 4\Gamma^\sigma_{\sigma\mu}\tensor{F}{^\alpha_\nu} - 4\delta^\alpha_\mu\Gamma^\sigma_{\sigma\xi} \tensor{F}{^\xi_\nu} \\
        &\qquad + 2 (\partial_\xi g_{\mu\nu}) F^{\xi\alpha} + 2 g_{\mu\nu} (\partial_\xi F^{\xi\alpha}) + 4 \partial_\mu\tensor{F}{^\alpha_\nu} - 4\delta^\alpha_\mu (\partial_\xi\tensor{F}{^\xi_\nu}). \\
        &= g_{\mu\nu} \big( -\Gamma^\sigma_{\sigma\lambda}F^{\lambda\alpha} - \partial_\lambda F^{\lambda\alpha} + 2\Gamma^\sigma_{\sigma\xi} F^{\xi\alpha} + 2\partial_\xi F^{\xi\alpha} \big) \\
        &\qquad + 4\delta^\alpha_\mu \big( \partial_\lambda \tensor{F}{^\lambda_\nu} + \Gamma^\sigma_{\sigma\lambda}\tensor{F}{^\lambda_\nu} - \Gamma^\sigma_{\sigma\xi}\tensor{F}{^\xi_\nu} - \partial_\xi \tensor{F}{^\xi_\nu} \big) \\
        &\qquad - 2(\partial_\lambda g_{\mu\nu})F^{\lambda\alpha} - 4\partial_\mu \tensor{F}{^\alpha_\nu} - 4\Gamma^\sigma_{\sigma\mu}\tensor{F}{^\alpha_\nu} + 4\Gamma^\sigma_{\sigma\mu}\tensor{F}{^\alpha_\nu} + 2 (\partial_\xi g_{\mu\nu}) F^{\xi\alpha} + 4 \partial_\mu\tensor{F}{^\alpha_\nu} \\
        &= g_{\mu\nu} \big( \Gamma^\sigma_{\sigma\lambda}F^{\lambda\alpha} + \partial_\lambda F^{\lambda\alpha}\big) \\
        &= -g_{\mu\nu}(\delta F)^\alpha \quad [+(\mu\leftrightarrow\nu),\olddiv 2,g=g^{\mathrm{T}}] \\
        &= -g_{\mu\nu}(\delta F)^\alpha.
    \end{align*}
    Continuing our simplification of~\eqref{Tmunu3}, the second term on line~\eqref{Tmunu3_part1} can be rearranged with help of the antisymmetry of $(F^{\mu\nu})$ and of the coordinate expression~\eqref{d_coordinates} for the $\rmd$ operator, and this time the symmetrization operation really is important:
    \begin{align*}
        4 \frac{\partial (\delta F)^\alpha}{\partial(\partial_\xi g^{\mu\nu})}\partial_\xi(\delta F)_\alpha &= 2g_{\mu\nu} F^{\xi\alpha}\partial_\xi(\delta F)_\alpha + 4\delta^\xi_\mu\tensor{F}{^\alpha_\nu}\partial_\xi(\delta F)_\alpha - 4\delta^\alpha_\mu\tensor{F}{^\xi_\nu}\partial_\xi(\delta F)_\alpha \\
        &= 2g_{\mu\nu} F^{\xi\alpha}\partial_\xi(\delta F)_\alpha + 4\tensor{F}{^\alpha_\nu}\partial_\mu(\delta F)_\alpha - 4\tensor{F}{^\xi_\nu}\partial_\xi(\delta F)_\mu \\
        &= g_{\mu\nu} F^{\xi\alpha}\partial_\xi(\delta F)_\alpha + g_{\mu\nu} F^{\xi\alpha}\partial_\xi(\delta F)_\alpha + 4\tensor{F}{^\alpha_\nu}\partial_\mu(\delta F)_\alpha - 4\tensor{F}{^\alpha_\nu}\partial_\alpha(\delta F)_\mu \\
        &= g_{\mu\nu} F^{\xi\alpha}\partial_\xi(\delta F)_\alpha - g_{\mu\nu} F^{\xi\alpha}\partial_\alpha(\delta F)_\xi + 4\tensor{F}{^\alpha_\nu}\big(\partial_\mu(\delta F)_\alpha - \partial_\alpha(\delta F)_\mu\big) \\
        &= g_{\mu\nu}F^{\xi\alpha}(\rmd\delta F)_{\xi\alpha} - 4\tensor{F}{^\alpha_\nu}(\rmd\delta F)_{\alpha\mu} \quad [+(\mu\leftrightarrow\nu),\olddiv 2,g=g^{\mathrm{T}}] \\
        &= g_{\mu\nu}F^{\xi\alpha}(\rmd\delta F)_{\xi\alpha} - 2\tensor{F}{^\alpha_\nu}(\rmd\delta F)_{\alpha\mu} - 2\tensor{F}{^\alpha_\mu}(\rmd\delta F)_{\alpha\nu}.
    \end{align*}
    Going back to~\eqref{Tmunu3}, we have proved
    \begin{equation*}
        \widetilde{T}_{\mu\nu}^{(3)} = -2\tensor{F}{^\alpha_\mu}(\rmd\delta F)_{\alpha\nu} - 2\tensor{F}{^\alpha_\nu}(\rmd\delta F)_{\alpha\mu} +g_{\mu\nu}F^{\alpha\beta}(\rmd\delta F)_{\alpha\beta} +2 (\delta F)_\mu(\delta F)_\nu -g_{\mu\nu} (\delta F)^\alpha (\delta F)_\alpha,
    \end{equation*}
    and using~\eqref{relation_Ttilde} we finally arrive at precisely the $T^{(3)}_{\mu\nu}$ that we wanted.
    
    The $j=1$ term, unimportant for us, is all that is left: Given $\ell_{\mathrm{EM}}^{(1)} = J^\alpha A_\alpha = g^{\alpha\beta} J_\beta A_\alpha$, we get
    \begin{align*}
        T^{(1)}_{\mu\nu} = \widetilde{T}^{(1)}_{\mu\nu} &= g_{\mu\nu}\ell_{\mathrm{EM}}^{(1)} - 2\frac{\partial\ell_{\mathrm{EM}}^{(1)}}{\partial g^{\mu\nu}} \\
        &= g_{\mu\nu}J^\alpha A_\alpha - 2\delta^\alpha_\mu\delta^\beta_\nu J_\beta A_\alpha \\
        &= g_{\mu\nu}J^\alpha A_\alpha - 2 J_\nu A_\mu \quad [+(\mu\leftrightarrow\nu),\olddiv 2,g=g^{\mathrm{T}}] \\
        &= g_{\mu\nu}J^\alpha A_\alpha -  J_\mu A_\nu -  J_\nu A_\mu. \qedhere
    \end{align*}
\end{proof}

\begin{remark}
    Our formula~\eqref{Tmunu_GR} for $(T_{\mu\nu})$ is the same as the one that appears in the original paper~\cite{podolsky} by Podolsky as well as the ones in the recent works~\cite{zayats} and~\cite{gratus}. The appendix of the latter presents a different derivation of this formula, not relying so heavily on component calculations, but we do point out that their end result (formula (128) on p.~28) contains an error: The term $F_{cb}\partial^a\partial_a \tensor{F}{^b_d}$ is not symmetrized with respect to the pair of indices $(c,d)$.
\end{remark}

\subsection{The E-M-BLTP system for a static point charge}
\label{subsec_GR_pointcharge}

From now on, we assume that a single point particle of electric charge $Q$, obeying either Maxwell-Maxwell or Maxwell-BLTP theory, sits statically alone in the Universe, and we consider the electrovacuum spacetime $\mathcal{M}$ ``around it.'' On account of its symmetries, this must be a static, spherically symmetric 4-manifold whose topological boundary (not considered part of $\mathcal{M}$) is a line representing the ``singularity worldline'' of the particle. It is well-known (see~\cite{parry}) that a system of coordinates $(ct,r,\theta,\phi) \in \bbR\times (0,\infty) \times [0,\pi] \times [0,2\pi)$ can be found in which the metric of such a spacetime takes a diagonal form, with the $\rmd(ct)^2$ and $\rmd r^2$ terms accompanied by two unknown functions of the so-called \textbf{area-radius} coordinate $r$, and the $\rmd\theta^2$ and $\rmd\phi^2$ terms figuring as in the usual spherical volume element. We choose to write it in the form
\begin{equation} \label{metric}
    g = -\frac{e^{2\alpha(r)}}{\beta(r)} \ \rmd (ct)\otimes\rmd(ct) + \beta(r) \ \rmd r\otimes\rmd r + r^2\big( \rmd\theta\otimes\rmd\theta + \sin^2\theta \ \rmd\phi\otimes\rmd\phi \big)
\end{equation}
for undetermined functions $\alpha$ and $\beta$, with a requirement that
\begin{equation} \label{beta_nohorizons}
    \beta(r) > 0 \quad\text{for all } r>0.
\end{equation}
By convention, the unknowns $\alpha,\beta$ are dimensionless, with the differentials $\rmd(ct)^2$, $\rmd r^2$ both carrying the dimension of $\dimL^2$, while the angles $\theta$, $\phi$ and differentials $\rmd\theta$, $\rmd\phi$ are dimensionless, with their coefficients satisfying $\dim(r^2) = \dim(r^2\sin\theta) = \dimL^2$. We will use the symbols $t,r,\theta,\phi$ as indices in tensor equations written in the $(ct,r,\theta,\phi)$ coordinates; for example, the first metric coefficient $-\beta^{-1}e^{2\alpha}$ is denoted by $g_{tt}$ instead of $g_{(ct)(ct)}$ or $g_{00}$.

We assume these coordinates to be valid \textit{globally} on $\mathcal{M}$. This means, in particular, that the particle is modeled as a \textbf{naked singularity}, not hidden behind any event horizons (where $\beta$ would take the value 0). We also assume the physically motivated condition of \textbf{asymptotic flatness} of the spacetime, which for a metric of the form~\eqref{metric} implies, among a stronger set of conditions (see~\cite{straumann}, p.~528), that
\begin{equation} \label{asymp_flat}
    \lim_{r\to\infty} \alpha(r) = 0, \quad \beta(r) = 1 + O(r^{-1}) \quad\text{as } r\to\infty.
\end{equation}
An interpretation of these conditions is that, away from the influence of the point particle, the metric approximates the Minkowski metric expressed in spherical coordinates, which has $\alpha \equiv 0$ and $\beta \equiv 1$. In particular, the following limit exists:
\begin{equation} \label{ADM}
    M_{\mathrm{ADM}} := \frac{c^2}{G} \lim_{r\to \infty} \frac{r}{2}\left( 1 - \frac{1}{\beta(r)} \right).
\end{equation}
With this, $\beta$ takes the form
\begin{equation*}
    \beta(r) = \left( 1 - \frac{2GM_{\mathrm{ADM}}}{c^2 r} + O(r^{-2}) \right)^{-1} \quad\text{as } r\to\infty.
\end{equation*}
The constant~\eqref{ADM} is called the spacetime's \textbf{ADM mass}, representing the singularity's active gravitational mass as measured by a distant observer. For this reason, physically relevant spacetimes should satisfy $M_{\mathrm{ADM}} > 0$. We also define the spacetime's \textbf{bare mass} as the limit of the same quantity on the other end of the domain $(0,\infty)$, in case it exists:
\begin{equation} \label{bare}
    M_{\mathrm{bare}} := \frac{c^2}{G} \lim_{r\to 0^+} \frac{r}{2}\left( 1 - \frac{1}{\beta(r)} \right).
\end{equation}
If $M_{\mathrm{bare}}$ is defined, it is not hard to see that condition~\eqref{beta_nohorizons} implies $M_{\mathrm{bare}} \leq 0$. A physical interpretation for this constant, however, is not immediately apparent, nor is it clear whether $\av{M_{\text{bare}}}$ should have any relation to the inertial mass of the particle. We will see in equations~\eqref{massofaparticle} and~\eqref{baremassofaparticle} that physically meaningful solutions to our problem will have $\av{M_{\text{bare}}} \gg M_{\mathrm{ADM}}$.

Next we describe the electromagnetic quantities in $\mathcal{M}$. The Faraday tensor $F\in\extprod{2}{\mathcal{M}}$ also has to obey staticity and spherical symmetry, properties which, by definition, are expressed in terms of vanishing Lie derivatives: The equation $\mathcal{L}_X F = 0$ must hold for each of the Killing vector fields $X\in\frakX{\mathcal{M}}$ associated to the properties of $\mathcal{M}$ being static and spherically symmetric. This condition imposes (see~\cite{carroll}, p.~254) that the only nonzero components of $F$ in our coordinates are $F_{tr} = -F_{rt}$ (which must be a general function of $r$) and $F_{\theta\phi} = -F_{\phi\theta}$ (which must have the form $\psi(r)\sin\theta$ for $\psi$ being a general function of $r$). However, a nonzero $F_{\theta\phi}$ would signify the presence of magnetic fields, which is unphysical for a single point source. Thus we assume that only the components $F_{tr} = -F_{rt}$ are not identically zero. We choose to call $F_{tr} =: e^\alpha\varphi'$ for an unknown \textbf{electric potential} $\varphi:(0,\infty)\longrightarrow\bbR$, and therefore
\begin{equation} \label{F_GR}
    F = e^{\alpha(r)}\varphi'(r) \ \rmd (ct)\wedge\rmd r,
\end{equation}
while $\varphi$ is defined so as to vanish at infinity:
\begin{equation} \label{varphi_GR}
    \varphi(r) = -\int_r^\infty \varphi'(s) \ \rmd s = -\int_r^\infty e^{-\alpha(s)} F_{tr}(s) \ \rmd s.
\end{equation}
See subsection A.1 in the appendix for an explanation of why the factor $e^\alpha$ is present in~\eqref{F_GR}. As it turns out, the Einstein and Maxwell equations are most easily written not in terms of $\varphi$ or $\varphi'$, but rather in terms of the same \textbf{electric-deviation scalar} as introduced in~\eqref{chi_SR} in the context of SR:
\begin{equation} \label{chi_GR}
    \chi(r) := \frac{1}{Q}r^2\varphi'(r) + 1,
\end{equation}
which is dimensionless: Since $\dimQ = \dim(F) = \dim(\varphi') \dim(\rmd(ct)\wedge\rmd r) = \dim(\varphi') \dimLL^2$, we get $\dim(\chi) = \dim(r^2/Q)\dim(\varphi') = \dimLL^2 \dimQ^{-1} \dim(\varphi') = \dimone$. In particular~\eqref{F_GR} becomes
\begin{equation} \label{F_GR2}
    F = Q e^{\alpha(r)} \frac{\chi(r)-1}{r^2} \ \rmd(ct)\wedge\rmd r.
\end{equation}
We can now formulate the equations of the problem:

\begin{proposition} \label{prop_main_system}
    Let $g$ and $F$ be the metric~\eqref{metric} and Faraday tensor~\eqref{F_GR2}, depending on sufficiently differentiable unknowns $\alpha,\beta,\chi: (0,\infty)\longrightarrow\bbR$. Then the system of the EM equations~\eqref{maxwell_cov},~\eqref{vacuum_cov} and the Einstein equations~\eqref{EFE} with stress-energy tensor given by~\eqref{Tmunu_GR} is equivalent to
    \begin{equation} \label{embltp}
        \left\{\begin{array}{l}
        \alpha' = -\varkappa^{-2}\dfrac{GQ^2}{c^4}\dfrac{(\chi')^2}{r^3}, \\ \\
        \beta' = \dfrac{\beta(1-\beta)}{r} + \dfrac{GQ^2}{c^4}\dfrac{\beta^2(1-\chi^2) - \varkappa^{-2}\beta(\chi')^2}{r^3}, \\ \\
        \varkappa^{-2}\dfrac{\rmd}{\rmd r}\left( \dfrac{\chi' e^\alpha}{r^2\beta} \right) = \dfrac{\chi e^\alpha}{r^2},
        \end{array}\right.
    \end{equation}
    which we call the \textbf{Einstein-Maxwell-BLTP (E-M-BLTP) system}.
\end{proposition}

\begin{proof}
    Deriving these equations is ``merely'' a matter of plugging the above $F$ into~\eqref{maxwell_cov},~\eqref{vacuum_cov},~\eqref{EFE}. We carry out this long calculation in subsection A.2 of the appendix.
\end{proof}

Next we need to establish an intrinsic way to measure the electric-field energy in our spacetime. We choose a canonical definition of the energy contained in a constant-time slice $\{t = t_0\}$ of a static spacetime: the conserved \textbf{Noether charge} corresponding to the Killing vector field $X = \partial_{t}$ that is associated to the property of staticity. To define this concept, consider the 1-form $P\in\extprod{1}{\mathcal{M}}$ given by
\begin{equation} \label{noether_P}
    P(Y) := T(X,Y) \quad\text{for all } Y\in\frakX{\mathcal{M}},
\end{equation}
where $T = (T_{\mu\nu})$ is the stress-energy tensor~\eqref{Tmunu_GR}. We compute the divergence of $P$ using~\eqref{ddelta_cov} and the Leibniz rule for the component expression $P_\mu = T_{\mu\nu}X^\nu$:
\begin{equation*}
    -\Star\rmd\Star P = -\delta P = \nabla^\mu P_\mu = (\nabla^\mu T_{\mu\nu})X^\nu + T_{\mu\nu}(\nabla^\mu X^\nu).
\end{equation*}
The first term on the right vanishes because $T$ is divergence-free. Using the symmetry of $T$ and the formula $0 = (\mathcal{L}_X g)^{\mu\nu} = \nabla^\mu X^\nu - \nabla^\nu X^\mu$ for the Lie derivative of the metric, this equation turns into
\begin{equation*}
    -\Star\rmd\Star P = \frac{1}{2}T_{\mu\nu}(\nabla^\mu X^\nu - \nabla^\nu X^\mu) = \frac{1}{2}T_{\mu\nu}(\mathcal{L}_X g)^{\mu\nu} = 0.
\end{equation*}
In particular $\rmd\Star P = 0$, and now Stokes' Theorem implies that the integral $\int_{\{t=t_0\}} \Star P$ is constant as a function of $t_0$:
\begin{equation*}
    \int_{\{t=t_1\}} \Star P - \int_{\{t=t_0\}} \Star P = \int_{\partial\{ t_0<t<t_1 \}} \Star P = \int_{\{ t_0<t<t_1\}} \rmd\Star P = 0.
\end{equation*}
Strictly speaking, one needs to consider in this calculation a \textit{bounded} region of the spacetime and then let it grow out towards spatial infinity. Hence this result only holds when $T$ is assumed to vanish sufficiently rapidly at spatial infinity, which indeed will be true in our case. We skip these details; we merely want to use this calculation as justification for the following:

\begin{definition}
    Let $P\in\extprod{1}{\mathcal{M}}$ be as in~\eqref{noether_P}. Then the \textbf{electric-field energy} in $\mathcal{M}$ is defined by the integral (where $t_0\in\bbR$ is arbitrary)
    \begin{equation} \label{energy_GR}
        \mathcal{E} := -\int_{\{t=t_0\}} \Star P = \frac{Q^2}{2} \int_0^\infty \frac{e^{\alpha}}{r^2} \left( 1 - \chi^2 - \varkappa^{-2}\frac{(\chi')^2}{\beta} \right) \ \rmd r.
    \end{equation}
\end{definition}
The explicit calculation of $\Star P$ to derive the above expression in terms of $\alpha,\beta,\chi$ will be shown in the appendix, subsection A.3, where the reason for the choice of the normalization constant $-1$ in front of the integral will also be explained.

\begin{remark}
    Using the Maxwell-Maxwell equations (that is, letting $\varkappa=\infty$), system~\eqref{embltp} becomes much simpler:
    \begin{equation*}
        \alpha' = 0, \qquad \beta' = \frac{\beta(1-\beta)}{r} + \frac{GQ^2}{c^4}\frac{\beta^2(1-\chi^2)}{r^3}, \qquad \chi = 0.
    \end{equation*}
    Its asymptotically flat solution is the so-called \textbf{Reissner-Weyl-Nordström} or \textbf{RWN} spacetime:
    \begin{equation} \label{RWN}
        \alpha^{(\mathrm{RWN})} \equiv 0, \qquad \beta^{(\mathrm{RWN})}_M(r) = \left(1 - \frac{2GM}{c^2r}+\frac{GQ^2}{c^4r^2}\right)^{-1}, \qquad \chi^{(\mathrm{RWN})} \equiv 0
    \end{equation}
    for an arbitrary parameter $M\in\bbR$ that is quickly seen to correspond to this spacetime's ADM mass according to~\eqref{ADM}. Using~\eqref{chi_GR}, the electric potential~\eqref{varphi_GR} is found to be the same Coulomb potential $\varphi^{(\mathrm{M})}$ as in~\eqref{Coulomb}. Furthermore, the condition~\eqref{beta_nohorizons} for the absence of event horizons in this case is equivalent to $GM < Q^2$, which holds --- with \textit{lots} of room to spare --- for the values of $M$ and $Q$ corresponding to typical atomic particles. Thus a lonely point charge under the conventional Einstein-Maxwell theory is indeed modeled as a \textit{naked singularity} of a static, spherically symmetric spacetime.
    
    We pointed out in Section~\ref{sec_EM_SR} that the Coulomb potential $\varphi^{\mathrm{(M)}}$ solves the SR equations of BLTP theory for any $0<\varkappa\leq\infty$. We remark now that the analogous fact is true in GR as well: The RWN solution~\eqref{RWN} solves system~\eqref{embltp}  for any $0<\varkappa\leq\infty$. This is trivial to check: Since all derivatives of $\chi$ in the system come accompanied by a $\varkappa^{-2}$ factor, we obtain one and the same reduction of the equations if we either set $\varkappa^{-2} = 0$ (yielding the system whose solution is RWN) or set $\chi \equiv 0$ (which holds for the RWN solution). Nevertheless, RWN must be considered an undesirable solution to the E-M-BLTP system for our purposes because of the divergence of its energy integral and electric potential at $r=0$:
    \begin{equation*}
        \mathcal{E} = \frac{Q^2}{2}\int_0^\infty \frac{1}{r^2}(1 - 0 - 0) \ \rmd r = \frac{Q^2}{2}\int_0^\infty \frac{\rmd r}{r^2} = \infty, \quad \varphi^{(\mathrm{M})}(0) = \frac{Q}{r}\bigg|_{r=0} = \infty.
    \end{equation*}

    A recent result by Cuzinatto et al.~\cite{cuzinatto} shows that, in a sense made precise in their work, RWN is in fact the only physically relevant \textit{exterior solution} of system~\eqref{embltp} when the spacetime contains a black hole --- that is, if the expression of the metric~\eqref{metric} is only valid in a domain of the form $r\in (r_0,\infty)$ for an $r_0>0$ where $\beta(r_0) = 0$. So it is remarkable that we will be able to find other solutions in the naked-singularity case.
    
\end{remark}

Our goal in this paper is to show that the E-M-BLTP system in the case $0<\varkappa<\infty$ admits solutions whose values for $\mathcal{E}$ and $\varphi(0)$ are finite. We will also investigate the value of $M_{\mathrm{ADM}}$ for these solutions and find conditions for it to be positive. Given that $\chi^{\text{(RWN)}} \equiv 0$, we will restrict the search to only those solutions that obey
\begin{equation} \label{chi_infty}
    \lim_{r\to\infty} \chi(r) = 0,
\end{equation}
so that the laws of electromagnetism away from the point charge approximate the ones observed in the real world, that is, the conventional Maxwell theory. In fact, our solutions $(\alpha,\beta,\chi)$ will be such that the quantities $\alpha - \alpha^{(\mathrm{RWN})} = \alpha$, $\beta - \beta^{(\mathrm{RWN})}$ and $\chi - \chi^{(\mathrm{RWN})} = \chi$ decay exponentially for large $r$. There are also other necessary properties of $\chi$ that can be deduced now from the desired behavior of our solution; we state them in a Lemma:

\begin{lemma} \label{lemma_sign_chi}
    Let $(\alpha,\beta,\chi)$ solve~\eqref{embltp} for $0<\varkappa<\infty$ and obey $\beta(r)>0$ for all $r>0$, $\chi(\infty) = 0$ and $\av{\varphi(0)} < \infty$, where $\varphi\in C\big([0,\infty)\big)$ is related to $\chi$ as in~\eqref{chi_GR}. Then
    \begin{equation} \label{chi_zero}
        \chi(0) = 1
    \end{equation}
    and 
    \begin{equation} \label{signs_chi}
        0 < \chi(r) < 1, \quad \chi'(r) < 0 \quad\text{for all } r>0.
    \end{equation}
\end{lemma}

\begin{proof}
    Since
    \begin{equation*}
        \varphi(r) - \varphi(0) = \int_0^r \varphi'(s) \ \rmd s = \int_0^r \frac{\chi(s)-1}{s^2} \ \rmd s,
    \end{equation*}
    the condition $\av{\varphi(0)} < \infty$ already gives~\eqref{chi_zero} --- in fact, we see that $\chi(s)-1$ needs to decay faster than $s$ as $s\to 0^+$ in order for this integral to be finite. Now define $\overline{\chi} := e^\alpha\chi'/(r^2\beta)$, so that the $\chi$ equation in~\eqref{embltp} can be interpreted as a first-order system for the unknowns $\chi$ and $\overline{\chi}$:
    \begin{equation*}
        \chi' = \frac{r^2\beta}{e^{\alpha}}\overline{\chi}, \quad \overline{\chi}' = \varkappa^2\frac{e^\alpha}{r^2}\chi.
    \end{equation*}
    Considering that $\beta>0$, these equations imply that, at each $r>0$, we have $\chi(r) > 0$ if and only if $\overline{\chi}'(r) > 0$, and the same also holds when exchanging the roles of $\chi$ and $\overline{\chi}$. This allows us to rule out certain scenarios involving the mutual signs of $\chi$ and $\overline{\chi}$:
    \begin{itemize}
        \item If $\chi(r_0) > 0$ and $\overline{\chi}(r_0) > 0$ at some $r_0$, then both $\chi$ and $\overline{\chi}$ would be increasing at $[r_0,\infty)$ (potentially blowing up before $r=\infty$), contradicting $\chi(\infty) = 0$. If $\chi(r_0) < 0$ and $\overline{\chi}(r_0) < 0$ at some $r_0$, the same contradiction would arise, only with decreasing behavior.
        \item If either $\chi$ or $\overline{\chi}$ (but not both) vanished at some $r_0$, we would fall into the previous case at any point $r_1>r_0$, ultimately also leading to a contradiction.
        \item If $\chi(r_0) = \overline{\chi}(r_0) = 0$ at some $r_0$, the existence-and-uniqueness theorem would show that $\chi$ must be null everywhere, contradicting $\chi(0) = 1$.
    \end{itemize}
    Therefore, $\chi(r)$ and $\overline{\chi}(r)$ must have opposite, nonzero signs at all $r>0$, and, taking into account that $\chi(0) = 1 > 0$ and $\mathrm{sgn}(\chi') = \mathrm{sgn}(\overline{\chi})$, we have the proof that $\chi > 0$ and $\chi' < 0$ everywhere in $(0,\infty)$. Since $\chi(0) = 1$ and $\chi(\infty) = 0$, with $\chi$ a strictly decreasing function, it follows as well that $0 < \chi < 1$ everywhere in $(0,\infty)$.
\end{proof}

\section{GR solutions for a static point charge}
\label{sec_main}

This essentially self-contained section includes the statement and proof of our main result, Theorem~\ref{main_theorem}.

\subsection{Statement of the existence theorem} \label{subsec_thm}

Before anything, we must define a special function that will appear in our theorem. As is well-known, given parameters $a,b\in\bbC$, the corresponding \textbf{confluent hypergeometric function of the second kind}, also known as \textbf{Tricomi's function} and denoted by $U = U(a,b;\ast): \bbC\setminus \{0\}\longrightarrow\bbC$, is by definition the only solution $w = w(z)$ of \textbf{Kummer's equation}
\begin{equation} \label{Kummer}
    z w'' + (b-z) w' - aw = 0
\end{equation}
satisfying the following condition:
\begin{equation} \label{U_polyrate}
    \lim_{\Re z\to\infty } z^a w(z) = 1.
\end{equation}
Depending on the values of the parameters $a$ and $b$, it may be a multi-valued function whose principal branch is normally taken to be that where $z = re^{i\theta}$ for $-\pi < \theta < \pi$. For our work, we will only need to consider this function for a real argument $z>0$ and real parameters $a,b$, in which case $U(a,b;z)$ is well-defined and real-valued. Also well-known is the fact that $U$ has an integral representation in a special case:
\begin{equation} \label{U_integral_rep}
    U(a,b;z) = \frac{1}{\Gamma(a)}\int_0^\infty e^{-tz}t^{a-1}(1+t)^{b-a-1} \ \rmd t \quad\text{for } \Re a > 0, \ \Re z > 0,
\end{equation}
where $\Gamma$ denotes the usual Gamma function. This formula is easily proved via the Laplace Transform of Kummer's equation~\eqref{Kummer}. See~\cite{abramowitz} for details. We now make the following definition:

\begin{definition}
    For all reals $\sfh\geq 0$, we denote by $U_{\sfh}$ the function
    \begin{equation} \label{def_Un}
        U_{\sfh}: (0,\infty)\longrightarrow\bbR, \quad U_{\sfh}(s) := U(1-\sfh,2,s).
    \end{equation}
    One can promptly check that $U_0(s) = 1/s$, $U_1(s) = 1$, $U_2(s) = s-2$ and that, for $\sfh\in\bbN$, the function $U_{\sfh}$ is a polynomial of degree $\sfh-1$. However, for most $\sfh\notin\bbN$, it cannot be written in terms of elementary functions.
\end{definition}

The following properties of $U_{\sfh}$ follow directly from~\eqref{U_polyrate} and~\eqref{U_integral_rep}:
\begin{equation} \label{Un_asymp}
    \lim_{s\to\infty} s^{1-\sfh} U_{\sfh}(s) = 1 \quad\text{for all } \sfh\geq 0, \text{ and}
\end{equation}
\begin{equation} \label{Un_integral_rep}
    \Gamma(1-\sfh) \, U_{\sfh}(s) = \int_0^\infty e^{-ts} t^{-\sfh} (1+t)^\sfh \ \rmd t \stackrel{ ts \longmapsto t }{=} \frac{1}{s} \int_0^\infty e^{-t} t^{-\sfh} (t+s)^\sfh \ \rmd t \quad\text{for all } s > 0, \ 0\leq\sfh<1.
\end{equation}
For our work in this section, we will also need the fact that $U_\sfh(4\sfh) > 0$ for all $\sfh > 0$. For this we offer a ``proof by plotting'' --- see Figure~\ref{fig_U4}. From this fact, we will be able to prove later that in fact $U_\sfh(s) > 0$ for all $s\geq 4\sfh$.

\begin{figure}[t!]
    \centering
    \begin{subfigure}[t]{0.4\textwidth}
        \centering
        \includegraphics[scale=.75]{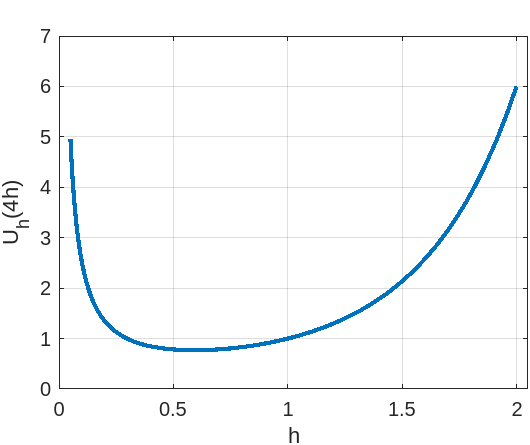}
        \caption{$U_\sfh(4\sfh)$ plotted for $0.05 < \sfh < 2.00$. If this were plotted from $\sfh = 0$, a blow-up towards $\infty$ would be present around $\sfh = 0$, because $U_0(s) = 1/s$.}
    \end{subfigure}
    \quad
    \begin{subfigure}[t]{0.4\textwidth}
        \centering
        \includegraphics[scale=.64]{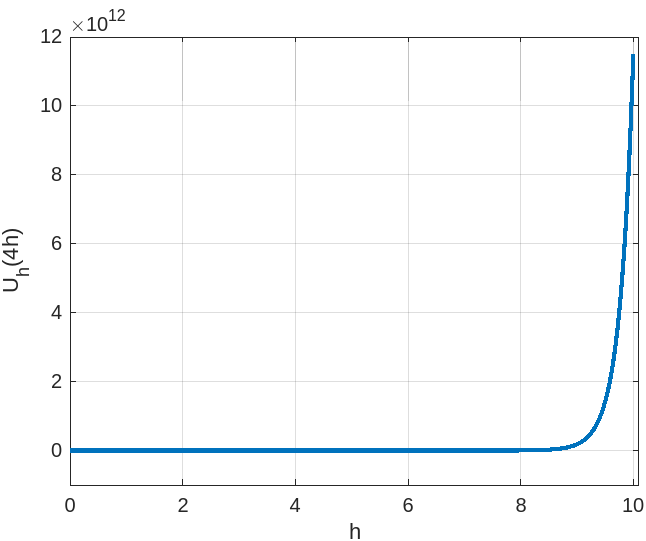}
        \caption{$U_\sfh(4\sfh)$ plotted for $0.05 < \sfh < 10$, showing that the expression grows unboundedly with $\sfh$ after attaining its positive minimum between $\sfh = 0.5$ and $\sfh = 1$.}
    \end{subfigure}
    \caption{$U_\sfh(4\sfh)$ is strictly positive for all $\sfh > 0$.}
    \label{fig_U4}
\end{figure}

\begin{remark}
    Our solutions to the E-M-BLTP system will depend on a constant $\sfh\geq 0$, related to $\av{M_{\text{bare}}}$, that is fixed \textit{a priori}. Formula~\eqref{Un_integral_rep} will then enable us to produce \textit{quantitative} estimates for the metric coefficients and electric potential when $0\leq \sfh<1$, whereas if $\sfh\geq 1$ our statements will only be \textit{qualitative} as a consequence of~\eqref{Un_asymp}. More-concrete statements for all values of $\sfh$ can likely be obtained by making use of other representation formulas for Tricomi's function, but we will leave this to a future paper if we ever find a need for such statements. For now, we are content with having a better grasp on our spacetimes only for $0\leq \sfh < 1$, since, as we shall see, the physically meaningful values of $\sfh$ for the spacetime of a point charge are within this range --- in fact between 0 and about $10^{-18}$.
\end{remark}

We can now state our result:

\begin{theorem} \label{main_theorem}
    Let $G,c,\varkappa$ be Newton's gravitational constant, the speed of light, and the BLTP constant of dimension $\emph{\texttt{[Length]}}^{-1}$, respectively. Also fix a finite constant $M_{\mathrm{bare}} \leq 0$ with dimension $\emph{\texttt{[Mass]}}$. Then there exists a threshold $\sfg_0 > 0$ such that, if a constant $Q\neq 0$ of dimension $\emph{\texttt{[Charge]}}$ is such that
    \begin{equation} \label{def_g}
        \sfg := \frac{GQ^2\varkappa^2}{c^4} < \sfg_0,
    \end{equation}
    then a solution to the Einstein-Maxwell-BLTP system of a point particle of charge $Q$ exists with a naked singularity having bare mass $M_{\mathrm{bare}}$ and with finite values for both its total field energy and its electric potential at the singularity. More precisely: A static, spherically symmetric, asymptotically flat spacetime $\mathcal{M}$ exists equipped with:
    \begin{itemize}
    
        \item a global coordinate system such that $\mathcal{M} = \big\{ (ct,r,\theta,\phi): \, t\in\bbR, \ r>0, \ \theta\in [0,\pi], \ \phi\in [0,2\pi) \big\}$;

        \item a Lorentzian metric of the form $\rmd s^2 = -\beta(r)^{-1} e^{2\alpha(r)}\ \rmd (ct)^2 + \beta(r) \ \rmd r^2 + r^2\big( \rmd\theta^2 + \sin^2\theta \ \rmd\phi^2 \big)$, for coefficients $\alpha,\beta: (0,\infty) \longrightarrow \bbR$ satisfying
        \begin{equation} \label{alphabeta_thm}
            \lim_{r\to\infty} \alpha(r) = 0, \quad \beta(r) = 1 + O(r^{-1}) \quad\text{as } r\to\infty
        \end{equation}
        and
        \begin{equation} \label{beta_thm}
            \beta(r) > 0 \quad\text{for all } r>0;
        \end{equation}

        \item a current 1-form $J\in\extprod{1}{\overline{\mathcal{M}}}$, defined on $\overline{\mathcal{M}} = \big\{ (ct,r,\theta,\phi): \, t\in\bbR, \ r\geq 0, \ \theta\in [0,\pi], \ \phi\in [0,2\pi) \big\}$ and expressed in $(ct,r,\theta,\phi)$ coordinates by
        \begin{equation} \label{J_thm}
            J = -Qc \, \delta_{\bm{0}}(r) \frac{e^{2\alpha}}{r^2\beta\sin\theta} \ \rmd (ct),
        \end{equation}
        where the Dirac Delta $\delta_{\bm{0}}$ is such that $\int_{B_R(0)} \delta_{\bm{0}} \ \rmd^3\bm{x} = 1$ for all balls $B_R(0) = \big\{ (r,\theta,\phi): \, 0\leq r < R, \ \theta\in [0,\pi], \ \phi\in [0,2\pi) \big\}$, with $\rmd^3\bm{x} = r^2\sin\theta \ \rmd r\wedge\rmd\theta\wedge\rmd\phi$;

        \item a 2-form field $F\in\extprod{2}{\mathcal{M}}$ defined by $F = -e^{\alpha(r)}\varphi'(r) \ \rmd(ct)\wedge\rmd r$ for an electric potential $\varphi: (0,\infty) \longrightarrow \bbR$ satisfying 
        \begin{equation} \label{varphi_thm}
            \lim_{r\to\infty} \varphi(r) = \lim_{r\to\infty} \chi(r) = 0, \quad\text{where } \chi(r) = \frac{1}{Q}r^2\varphi'(r) + 1;
        \end{equation}

        \item a stress-energy tensor $(T_{\mu\nu})$ defined in terms of $F$ as in~\eqref{Tmunu_GR};
        
    \end{itemize}
    such that the system of Einstein-Maxwell-BLTP equations is verified, namely
    \begin{equation*}
        R_{\mu\nu} - \frac{1}{2}Rg_{\mu\nu} = \frac{8\pi G}{c^4} T_{\mu\nu}, \quad \rmd F = 0, \quad \rmd M = \frac{4\pi}{c}\Star J, \quad M = \Star F + \varkappa^{-2} \Star\rmd\delta F,
    \end{equation*}
    and such that the following properties hold:
    \begin{equation} \label{bare_thm}
         \frac{c^2}{G} \lim_{r\to 0^+} \frac{r}{2}\left( 1 - \frac{1}{\beta(r)} \right) = M_{\mathrm{bare}},
    \end{equation}
    \begin{equation} \label{energy_finite_thm}
        0 < \mathcal{E} := \frac{Q^2}{2} \int_0^\infty \frac{e^{\alpha}}{r^2} \left( 1 - \chi^2 - \varkappa^{-2}\frac{(\chi')^2}{\beta} \right) \ \rmd r < \infty,
    \end{equation}
    \begin{equation} \label{varphi_finite_thm}
        0 < \av{\varphi(0)} < \infty.
    \end{equation}
    In addition, let constants $\sfh$, $M_{\mathrm{ADM}}$, $\rng{\mathcal{E}}$ and a function $\rng{\varphi}$, of respective dimensions $\emph{\texttt{[1]}}$, $\emph{\texttt{[Mass]}}$, $\emph{\texttt{[Energy]}}$ and $\emph{\texttt{[Charge]}}$, be defined as follows:
    \begin{equation} \label{newdefinition_n}
        \sfh := \frac{\varkappa G}{c^2}\av{M_{\mathrm{bare}}},
    \end{equation}
    \begin{equation} \label{newdefinition_M}
         M_{\mathrm{ADM}} := \frac{c^2}{G} \lim_{r\to\infty} \frac{r}{2}\left( 1 - \frac{1}{\beta(r)} \right),
    \end{equation}
    \begin{equation} \label{Ebolinha}
        \rng{\mathcal{E}} := Q^2\varkappa \left(\frac{1}{2} - \frac{1}{4\sfh} - \frac{U_\sfh'(4\sfh)}{U_\sfh(4\sfh)}\right) \text{ if } \sfh>0, \quad \rng{\mathcal{E}} := Q^2\varkappa\lim_{\sfh\to 0^+} \left(\frac{1}{2} - \frac{1}{4\sfh} - \frac{U_\sfh'(4\sfh)}{U_\sfh(4\sfh)}\right) = \frac{Q^2\varkappa}{2} \text{ if } \sfh=0,
    \end{equation}
    \begin{equation} \label{phibolinha}
        \rng{\mathcal{\varphi}}(r) := \left\{\begin{array}{ll}
            \dfrac{Q}{r} - \dfrac{Q\,U_\sfh(4\sfh+2\varkappa r)}{2\sfh \, U_\sfh(4\sfh)}\left(1 - \dfrac{2GM_{\mathrm{bare}}}{c^2 r}\right)e^{-\varkappa r} &\quad\text{if } \sfh > 0, \\ \\
            \dfrac{Q(1 - e^{-\varkappa r})}{r} &\quad\text{if } \sfh = 0.
        \end{array}\right. 
    \end{equation}
    Then the solution described above satisfies the following estimates at all $r>0$, for $\sfh$-dependent constants $C_1',C_j$:
    \begin{equation} \label{end_alpha}
        \av{ e^{\alpha(r)} - 1 } \leq \left\{\begin{array}{ll}
            \sfg C_1 \left( e + \dfrac{e}{2\varkappa r} \right)^{\delta} \ln\left( e + \dfrac{e}{2\varkappa r} \right) e^{-2\varkappa r} &\text{if } \sfh = 0 \text{ (where } 0<\delta\ll 1 \text{),} \\ \\
            \sfg C_1 \exp\left(\dfrac{C_1'}{\sfh^2}\right) (4\sfh+2\varkappa r)^{-2}(1+4\sfh+2\varkappa r)^{2+2\sfh}e^{-2\varkappa r} &\text{if } \sfh > 0,
        \end{array}\right.
    \end{equation}
    \begin{equation} \label{end_beta1}
        \av{ \beta(r) - \beta^{\mathrm{(RWN)}}_{M_{\mathrm{ADM}}}(r) } \leq \sfg C_2 r(4\sfh+2\varkappa r)^{-1} (1+4\sfh+2\varkappa r)^{2+2\sfh} e^{-2 \varkappa r} \beta^{\mathrm{(RWN)}}_{M_{\mathrm{ADM}}}(r),
    \end{equation}
    \begin{equation} \label{end_beta2}
        \av{ \beta(r) - \left( 1 - \frac{2GM_{\mathrm{bare}}}{c^2 r} \right)^{-1} } \leq \sfg C_3 (\varkappa r)^3(4\sfh+2\varkappa r)^{-3}(1+4\sfh+2\varkappa r)^{-1},
    \end{equation}
    \begin{equation} \label{end_varphi1}
        \av{\varphi(r)-\frac{Q}{r}} \leq \av{Q} C_4 (\varkappa r)^{-1}(1+4\sfh+2\varkappa r)^{\frac12+\sfh}e^{-\varkappa r},
    \end{equation}
    \begin{equation} \label{end_varphi2}
        \av{\varphi(r)- \rng{\varphi}(r) } \leq \sfg \av{Q} C_5 (1+4\sfh+2\varkappa r)^{-\frac12+\sfh} e^{-\varkappa r},
    \end{equation}
    \begin{equation} \label{end_const_potential}
        \av{\varphi(0) - \frac{2\rng{\mathcal{E}}}{Q\varkappa} } \leq \sfg \av{Q} C_6,
    \end{equation}
    \begin{equation} \label{end_const_energy}
         \av{ \mathcal{E} - \rng{\mathcal{E}} } \leq \sfg Q^2\varkappa C_7,
    \end{equation}
    \begin{equation} \label{end_const_mass}
        \av{ M_{\mathrm{ADM}} - \big( M_{\mathrm{bare}} + c^{-2}\rng{\mathcal{E}} \big) } \leq \frac{\sfg Q^2\varkappa}{c^2} C_8.
    \end{equation}    
    Finally, if $0\leq\sfh<1$, the following values of the constants $\sfg_0$ and $C_1',C_j$ may be chosen:
    \begin{multline} \label{numerical_C}
        \begin{array}{|l||c|c|c|c|c|c|c|c} \hline
            & \sfg_0 & C_1 & C_1' & C_2 & C_3 & C_4 & C_5 & \cdots \\ \hline \hline
            \sfh = 0 & 2\cdot 10^{-4} & 1.57 & \text{---} & 13.9 & 389 & 2.14 & 1240 & \cdots \\ \hline
            \sfh\in (0,1) & 1.17\cdot 10^{-17} & 58.5 & 2.68\cdot 10^{-14} & 1910 & 7.41\cdot 10^{6} & 45 & 2.59\cdot 10^{16} & \cdots \\ \hline
        \end{array} \\
        \begin{array}{cc|c|c|} \hline
            \cdots & C_6 & C_7 & C_8 \\ \hline \hline
            \cdots & 1240 & 623 & 620 \\ \hline
            \cdots & 5.80\cdot 10^{16} & 2.91\cdot 10^{16} & 2.90\cdot 10^{16} \\ \hline
        \end{array}
    \end{multline}
\end{theorem}

Before the proof, we make some comments on the significance of the theorem's hypotheses and conclusions:

\begin{itemize}

\item \textbf{Approximate solution formulas.} Estimates~\eqref{end_alpha},~\eqref{end_beta2} and~\eqref{end_varphi2} provide expressions of $r$ that approximate the true values of $e^{\alpha(r)}$, $\beta(r)$ and $\varphi(r)$. The corresponding error terms are small due to the presence of $\sfg$ in them, and they decay exponentially for large $r$ in the case of~\eqref{end_alpha} and~\eqref{end_varphi2}. These error terms are also uniformly bounded over all $r\geq 0$, except for the one for $\av{e^\alpha - 1}$ when $\sfh=0$, which blows up at $r=0$ at the slow rate of $r^{-\delta}\ln(e+e/r)$, reflecting the fact that $e^\alpha$ diverges at $r=0$ in the zero-bare-mass case.

\item \textbf{Comparison with RWN.} Estimates~\eqref{end_alpha},~\eqref{end_beta1} and~\eqref{end_varphi1} show how $e^{\alpha(r)}$, $\beta(r)$ and $\varphi(r)$ compare to their RWN counterparts. As expected, the respective error terms decay exponentially for large $r$ (where RWN is expected to accurately approximate the true solution), while not necessarily being small for small $r$ (where RWN becomes singular, contrary to the true solution). Also note there is no $\sfg$ in the error term for~\eqref{end_varphi1}; this is not a typo.

\item \textbf{Finiteness of EM constants and ADM mass.} Estimates~\eqref{end_const_potential},~\eqref{end_const_energy} and~\eqref{end_const_mass} provide good approximations for the values of the electric potential at $r=0$, the total field energy, and the ADM mass, also showing that they are finite.

\item \textbf{Blow-up of $\bm{g_{tt}}$ near the singularity.} Information about the first metric coefficient $g_{tt}(r) = -\beta(r)^{-1}e^{2\alpha(r)}$ is immediately derivable from the estimates~\eqref{end_alpha} and~\eqref{end_beta2} on $e^{\alpha(r)}$ and $\beta(r)$. Interestingly, and contrary to what happens to $g_{rr}(r) = \beta(r)$, we find that $g_{tt}(0)$ is infinite no matter what the value of $\sfh$ is: In the zero-bare-mass case $\sfh = 0$, we have $\beta(0) \approx 1$ but $e^{\alpha(0)} = \infty$, while in the negative-bare-mass case $\sfh > 0$ we have $e^{\alpha(0)} \approx 1$ but $\beta(0)^{-1} = \infty$.

\item \textbf{Formula for $\bm{J}$.} Equation~\eqref{J_thm} expresses the general-relativistic form of the 4-current $J$ corresponding to a static point charge. An explanation for why it looks like this can be found in the appendix, inside of the proof of Proposition~\ref{prop_charge}, which is the only place where we truly need it.

\item \textbf{Physical $\bm{\sfg}$ is small.} Let us assume that the typical value of $\av{Q}$ for a charged atomic particle is between one and ten times the elementary charge:
\begin{equation*}
    1.52\cdot 10^{-4} \ \mathrm{kg}^{\frac12}\cdot\mathrm{m}^{\frac32}\cdot\mathrm{s}^{-1} < \av{Q} < 1.52\cdot 10^{-3} \ \mathrm{kg}^{\frac12}\cdot\mathrm{m}^{\frac32}\cdot\mathrm{s}^{-1}
\end{equation*}
Then, using our crude postulated bounds~\eqref{kappa_estimate1} and~\eqref{kappa_estimate2} for $\varkappa$ as well as the known values
\begin{equation*}
    G = 6.67\cdot 10^{-11} \ \mathrm{kg}^{-1}\cdot\mathrm{m}^3\cdot\mathrm{s}^{-2}, \quad c = 3.00\cdot 10^{8} \ \mathrm{m}\cdot\mathrm{s}^{-1},
\end{equation*}
correct to three significant digits, we calculate that the physical value of $\sfg = GQ^2\varkappa^2/c^4$ is indeed very small and falls within the size of the threshold $\sfg_0$ that is given in the theorem for the case $0\leq\sfh<1$:
\begin{equation} \label{range_g}
    1.90\cdot 10^{-30} < \sfg < 1.90\cdot 10^{-18}.
\end{equation}

\item \textbf{Physical $\bm{\sfh}$ is small if $\bm{M_{\mathrm{ADM}} > 0}$.} Let us analyze the constant $\rng{\mathcal{E}}$ defined in~\eqref{Ebolinha}. Conventional mathematical software often includes the capability to compute Tricomi's function $U(a,b;z)$ --- for example, in MATLAB it is implemented under the name \texttt{KummerU}. In particular, the expression
\begin{equation*}
    \frac{\rng{\mathcal{E}}}{Q^2\varkappa} = \frac{1}{2} - \frac{1}{4\sfh} - \frac{U_\sfh'(4\sfh)}{U_\sfh(4\sfh)}, \quad \sfh > 0,
\end{equation*}
can be plotted without much effort --- we note the general property $(\rmd/\rmd z) U(a,b;z) = -a \ U(a+1,b+1;z)$ of Tricomi's function, which implies that also the derivative term $U_\sfh'(4\sfh) = (\sfh-1) \, U(2-\sfh,3,4\sfh)$ can be easily managed by a computer. The plots in Figure~\ref{fig_Ebolinha} reveal that $\rng{\mathcal{E}} \leq Q^2\varkappa/2$ for all $\sfh\geq 0$.
\begin{figure}[t!]
    \centering
    \begin{subfigure}[t]{0.4\textwidth}
        \centering
        \includegraphics[scale=.6]{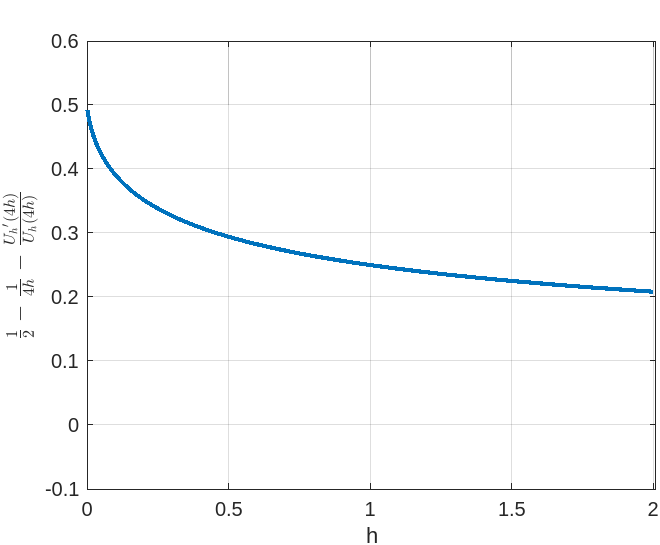}
        \caption{$\rng{\mathcal{E}}/(Q^2\varkappa)$ plotted for $0\leq\sfh\leq 2$. The exact values $1/2$, $1/4$ and $5/24$ at $\sfh = 0$, $1$ and $2$, respectively, are explicitly computable from the formulas $U_0(s) = 1/s$, $U_1(s) = 1$ and $U_2(s) = s-2$.}
    \end{subfigure}
    \quad
    \begin{subfigure}[t]{0.4\textwidth}
        \centering
        \includegraphics[scale=.6]{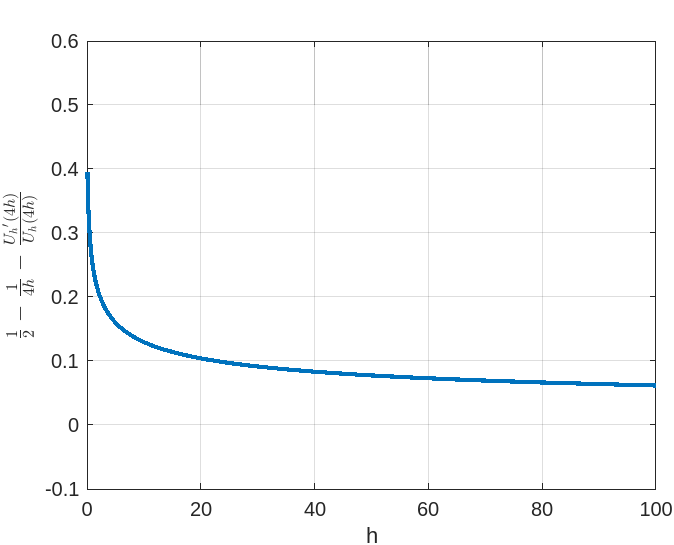}
        \caption{$\rng{\mathcal{E}}/(Q^2\varkappa)$ plotted for $0\leq\sfh\leq 100$, showing that the expression continues to decrease as $\sfh$ grows.}
    \end{subfigure}
    \caption{$\rng{\mathcal{E}}/(Q^2\varkappa)$ stays below $1/2$ for all $\sfh > 0$.}
    \label{fig_Ebolinha}
\end{figure}
As a consequence of this and of estimate~\eqref{end_const_mass}, we find
\begin{equation*}
    c^2 M_{\mathrm{ADM}} \leq c^2 M_{\mathrm{bare}} + \rng{\mathcal{E}} + \sfg Q^2\varkappa C_8 \leq c^2 M_{\mathrm{bare}} + \frac{Q^2\varkappa}{2} + \sfg Q^2\varkappa C_8.
\end{equation*}
If we wish for our spacetime to have a positive ADM mass, it is therefore necessary that the bare mass satisfy
\begin{equation} \label{ineq_baremass_baremass_baremass}
    M_{\mathrm{bare}} > -\frac{Q^2\varkappa}{c^2} \left(\frac{1}{2} + \sfg C_8\right) \approx -\frac{Q^2\varkappa}{2 c^2},
\end{equation}
where we used the fact that $\sfg C_8$ is small compared to $1/2$ (we will see in our proof that $\sfg_0 \ll C_j$ for each $j=1,\ldots,8$). Using the definitions~\eqref{def_g} and~\eqref{newdefinition_n} of $\sfg$ and $\sfh$, inequality~\eqref{ineq_baremass_baremass_baremass} turns into
\begin{equation} \label{ineq_hhh}
    \sfh \leq \sfg \left(\frac{1}{2} + \sfg C_8\right) \approx \frac{\sfg}{2}.
\end{equation}
Hence the physically relevant values of $\sfh$ are indeed small, of the order of at most $10^{-18}$ according to~\eqref{range_g}.

On the other hand, the constants $\sfg_0$ and $C_j$ given in the second line of table~\eqref{numerical_C} are valid not only for $0\leq\sfh\leq 10^{-18}$, but also in the whole range $0<\sfh\leq 1$. This renders them far from optimal when $\sfh$ is as small as $10^{-18}$. At this point we have to mention that our upcoming proofs could easily be adapted to find better values for these constants for small, positive $\sfh$ (they would come out much closer to the ones on the first line of the table, corresponding to the $\sfh = 0$ case), but we opted to leave out this improvement, lest we obscure the arguments with too many calculations.

\item \textbf{Physically relevant range for $\bm{|M_{\mathrm{bare}}|}$.} Let us consider typical values for the mass of an atomic particle, represented in our spacetime by the ADM mass~\eqref{newdefinition_M}:
\begin{equation} \label{massofaparticle}
    10^{-30} \ \mathrm{kg} < M_{\mathrm{ADM}} < 10^{-25} \ \mathrm{kg}.
\end{equation}
Since this assigns a positive value for $M_{\mathrm{ADM}}$, the upper bound~\eqref{ineq_hhh} holds, which makes $\sfh$ very small. The first plot in Figure~\ref{Ebolinha} then shows that $\rng{\mathcal{E}} \approx Q^2\varkappa/2$. We plug this into estimate~\eqref{end_const_mass} and find
\begin{equation} \label{seriously_stop_reading_my_labels}
    \av{ \frac{c^2M_{\mathrm{bare}}}{Q^2\varkappa} - \left( \frac{c^2M_{\mathrm{ADM}}}{Q^2\varkappa} - \frac{1}{2} \right) } \approx \av{ \frac{c^2M_{\mathrm{bare}}}{Q^2\varkappa} - \left( \frac{c^2M_{\mathrm{ADM}}}{Q^2\varkappa} - \frac{\rng{\mathcal{E}}}{Q^2\varkappa} \right) } \leq \sfg C_8 < 0.0551,
\end{equation}
where we used the physical upper bound~\eqref{range_g} on $\sfg$ and the value of $C_8$ from~\eqref{numerical_C}. Moreover, with~\eqref{massofaparticle} and our values for the constants $Q,c,\varkappa$, one checks that $c^2 M_{\mathrm{ADM}}/(Q^2\varkappa) < 3.90\cdot 10^{-12}$ is negligible in comparison to $1/2$ and can practically be left out of~\eqref{seriously_stop_reading_my_labels}. This gives us an approximate range
\begin{equation*}
    -0.556 < \frac{c^2M_{\mathrm{bare}}}{Q^2\varkappa} < -0.444.
\end{equation*}
By then applying to this our known values/ranges for the physical constants, we finally arrive at
\begin{equation} \label{baremassofaparticle}
    1.13 \cdot 10^{-14} \ \mathrm{kg} < \av{ M_{\mathrm{bare}} } < 1.43 \cdot 10^{-7} \ \mathrm{kg}.
\end{equation}
Hence the physically meaningful values for $\av{ M_{\mathrm{bare}} }$ in the E-M-BLTP spacetimes are seen to be several orders of magnitude larger than the particle's inertial mass $M_{\mathrm{ADM}}$.

\end{itemize}

The rest of this section is dedicated to proving Theorem~\ref{main_theorem}. We begin with some simplifications followed by auxiliary results, some with entire subsections dedicated to them. The proof is then concluded in Subsection~\ref{subsec_proof}. The core of the argument will be to treat the problem as a perturbation of the $G=0$ equations, which is a valid strategy because, for a point charge, gravity can be said to be much weaker than electromagnetic effects --- a fact expressed by the smallness of the dimensionless constant $\sfg$.

To begin simplifying the E-M-BLTP system, we point out that the unknown $\alpha$ does not appear on the right side of the $\alpha$ and $\beta$ equations in~\eqref{embltp}. To add to this, if we expand the derivative on the left side of the $\chi$ equation, cancel the factor $e^\alpha$ from both sides and use the $\alpha$ and $\beta$ equations to simplify everything, the result is also an equation from which $\alpha$ is absent:
\begin{equation} \label{chi_improved}
    \frac{1}{\varkappa^{-2}} \frac{\rmd}{\rmd r}\left( \frac{\chi'}{r^2\beta} \right) = \frac{GQ^2}{c^4 \varkappa^{4}} \frac{(\chi')^3}{r^5\beta} + \frac{\chi}{r^2}.
\end{equation}
Hence $\beta$ and $\chi$ can be solved for independently from $\alpha$, and the $\alpha$ equation can subsequently be integrated once $\chi$ is known. We may therefore focus our initial attention only on the unknowns $\beta$ and $\chi$. Expanding the derivative on the left side of~\eqref{chi_improved} with help of the $\beta$ equation and then simplifying the result, we find the following second-order equation for $\chi$:
\begin{equation} \label{chi_2ndorder}
    \chi'' = \frac{3-\beta}{r} \chi' + \varkappa^{2}\beta\chi + \frac{GQ^2}{c^4} \frac{\beta(1-\chi^2)\chi'}{r^3}.
\end{equation}
At this point, like in many other situations involving spherical symmetry in GR, it is better to start considering in place of $\beta$ the spacetime's \textit{mass function} $\mu$ --- except that we will actually consider a different, closely related function $\nu$. Both of these quantities are defined next:

\begin{definition} \label{def_mass_function}
    The spacetime's \textbf{mass function} in the coordinates $(ct,r,\theta,\phi)$ is defined by
    \begin{equation} \label{definition_mu}
        \mu(r) := \frac{c^2r}{2G} \left(1 - \frac{1}{\beta(r)}\right) \quad\Longleftrightarrow\quad \beta(r) = \left( 1 - \frac{2 G \, \mu(r)}{c^2r} \right)^{-1} = \frac{c^2r}{c^2r - 2G \, \mu(r)}.
    \end{equation}
    It has physical significance: An observer who interprets the coordinate $r$ as measuring the distance from the singularity will find that $\mu(r)$ measures the amount of mass contained inside the sphere of radius $r$, at least in the Newtonian regime where $G\mu(r) \ll c^2 r$ (see~\cite{parry}). Building on this, we define the \textbf{re-scaled mass function} by absorbing the factor of $G$ in~\eqref{definition_mu} into it:
    \begin{equation} \label{definition_nu}
        \nu(r) := G \, \mu(r) = \frac{c^2r}{2} \left(1 - \frac{1}{\beta(r)}\right) \quad\Longleftrightarrow\quad \beta(r) = \left( 1 - \frac{2\nu(r)}{c^2r} \right)^{-1} = \frac{c^2r}{c^2r - 2\nu(r)}.
    \end{equation}
    Note the dimensions
    \begin{equation*}
        \dim(\mu) = \dimM = \dimMM, \quad \dim(\nu) = \dimM\cdot\dim(G) = \dimLL^3 \dimTT^{-2}.
    \end{equation*}
\end{definition}

\begin{remark}
    One reason to consider $\nu$ instead of $\mu$ in our calculations has to do with our proof method: We will later decompose $\nu$ in the form $\rng{\nu} + G\xi$, where $\rng{\nu}$ solves the first equation of the E-M-BLTP system for $G=0$, while $\xi$ is a new unknown such that $G$ is supposed to be small in comparison to $\nn{\xi}$ (for a norm that will be defined in our proof). If we instead applied the same idea to $\mu$, that is, if we considered the decomposition $\mu = \rng{\mu} + G\rho$ for a new unknown $\rho$, we would find that $\rho = \xi/G$ is too large for the proof to work. Another way of thinking about this is that, by writing our equations in terms of $\nu = G\mu$ instead of $\mu$, we ensure that our proof method is a perturbation only with respect to the parameter $G$ that appears in Einstein's equations $G_{\mu\nu} = (8\pi G/c^4) T_{\mu\nu}$, and not with respect to the $G$ that appears in the definition~\eqref{definition_mu} of the mass function.
\end{remark}

The differential equations~\eqref{embltp},~\eqref{chi_2ndorder} for $\beta$ and $\chi$, when rewritten in terms of $\nu$ and $\chi$, are checked to become
\begin{equation} \label{system_nu_chi_dim}
    \left\{\begin{array}{l}
        \nu' = \dfrac{GQ^2}{c^2}\left(\dfrac{(1+\chi)(1-\chi)}{2r^2} - \varkappa^{-2}\dfrac{(c^2r-2\nu)(\chi')^2}{2c^2r^3}\right), \\ \\
        \chi'' = \dfrac{2(c^2r-3\nu)}{r(c^2r-2\nu)}\chi' + \dfrac{\varkappa^2 c^2r}{c^2r-2\nu}\chi + \dfrac{GQ^2}{c^2}\dfrac{(1+\chi)(1-\chi)\chi'}{r^2(c^2r-2\nu)},
    \end{array}\right.
\end{equation}
while the definitions~\eqref{bare} and~\eqref{ADM} of the bare and ADM masses give
\begin{equation} \label{nu_and_masses}
    \nu(0) = G M_{\mathrm{bare}}, \quad \nu(\infty) = G M_{\mathrm{ADM}}.
\end{equation}
Next, since there are three independent parameters that will remain constant in this proof, namely $Q,c,\varkappa$, they can be used to define characteristic scales of mass, length and time (and therefore also of any other dimension) for our problem:
\begin{equation*}
    \dimM = \mathrm{dim}(Q^2c^{-2}\varkappa), \quad \dimL = \mathrm{dim}(\varkappa^{-1}), \quad \dimT = \mathrm{dim}(c^{-1}\varkappa^{-1}).
\end{equation*}
(The only reason why we waited so long to finally mention this simplification is that there is no $\varkappa$ in the conventional Maxwell theory, which we also wished to expand upon for comparison with BLTP theory). Using these scales, all variables and constants in the theorem can be made dimensionless: We define $\widetilde{r} = \varkappa r$ as well as $\widetilde{f}(\widetilde{r}) = f(\varkappa^{-1}\widetilde{r})$ for each of the already dimensionless functions $f=\alpha,\beta,\chi$. Similarly, we set $\widetilde{\varphi}(\widetilde{r}) = Q^{-1} \varphi(\varkappa^{-1} \widetilde{r})$, as well as $\widetilde{\mu}(\widetilde{r}) = Q^{-2}c^2\varkappa^{-1}\mu(\varkappa^{-1}\widetilde{r})$ and analogously for $\widetilde{\nu}$. Then the equations obeyed by $\widetilde{\alpha}(\widetilde{r})$ etc.~are the same as those for the original functions, except that the constants $Q,c,\varkappa$ get replaced by 1 and the constant $G$ gets replaced by its corresponding re-scaled parameter $\widetilde{G} = \sfg$ --- see~\eqref{def_g}. To recover the original objects from the dimensionless ones, they simply need to be multiplied by the only combination of $Q,c,\varkappa$ that produces their correct dimension; for example, for the field energy $\mathcal{E}$, measured in units of $\dimMM\dimLL^2\dimTT^{-2}$, the following relation holds:
\begin{equation*}
    \mathcal{E} = \dimM \dimL^2 \dimT^{2} \widetilde{\mathcal{E}} = (Q^2c^{-2}\varkappa)(\varkappa^{-1})^2(c^{-1}\varkappa^{-1})^{-2} \widetilde{\mathcal{E}} = Q^2\varkappa \, \widetilde{\mathcal{E}}.
\end{equation*}
We stop writing the tildes from now on, but they are implied above all constants, functions and their arguments in everything that follows. The claims of the theorem will only be proved for these dimensionless variables, and the reader is tasked with recovering the original version of these claims by applying the procedure just explained. We also note that definition~\eqref{newdefinition_n} of the constant $\sfh$ is now written as
\begin{equation} \label{nu_zero}
    \sfh = \sfg \av{M_{\mathrm{bare}}} = -\sfg M_{\mathrm{bare}} = -\nu(0),
\end{equation}
where we used~\eqref{nu_and_masses}. Putting together the differential equations~\eqref{system_nu_chi_dim} and boundary conditions~\eqref{nu_zero},~\eqref{chi_zero} and~\eqref{chi_infty} for the re-scaled mass function $\nu$ and the relative-deviation scalar $\chi$, we arrive at the following boundary-value problem, which constitutes the core of the whole proof:

\begin{definition} \label{def_mdproblem} The \textbf{mass/deviation problem of parameters} $\bm{(\sfg,\sfh)}$, where $\sfg>0$ and $\sfh\geq 0$ are given \textit{a priori}, consists of finding a solution $(\nu,\chi)\in C^1\big((0,\infty)\big)\times C^2\big((0,\infty)\big)$ to the equations
\begin{equation} \label{system_nu_chi}
    \left\{\begin{array}{l}
        \nu' = \sfg\left(\dfrac{(1+\chi)(1-\chi)}{2r^2} - \dfrac{(r-2\nu)(\chi')^2}{2r^3}\right), \\ \\
        \chi'' = \dfrac{2(r-3\nu)}{r(r-2\nu)}\chi' + \dfrac{r}{r-2\nu}\chi + \dfrac{\sfg(1+\chi)(1-\chi)\chi'}{r^2(r-2\nu)}
    \end{array}\right.
\end{equation}
such that
\begin{equation} \label{conditions_nu_chi}
    \nu(0) = -\sfh, \quad \chi(0) = 1, \quad \chi(\infty) = 0.
\end{equation}
\end{definition}

We now dedicate three subsections to solving this problem. The rest of the claims in the theorem will follow easily afterwards. We will be able to prove that, for all $\sfh\geq 0$, a solution exists as long as $\sfg$ is suitably small. The idea will be to look for solutions in the form $\nu = \rng{\nu} + \sfg\xi$ and $\chi = \rng{\chi} + \sfg\psi$, with $(\rng{\nu},\rng{\chi})$ being a solution to the problem with parameter $\sfg=0$, and with $(\xi,\psi)$ being a new pair of unknowns.

\subsection{Integral lemmas}
\label{subsec_lemmas}

Here we collect three auxiliary results on the asymptotic behavior of certain integrals.

\begin{lemma} \label{lemma_U}
    For all $0\leq\sfh<1$, we have
    \begin{equation} \label{estimate_U}
        \frac{0.232}{(1-\sfh) \, \Gamma(1-\sfh)} s^{-1}(1+s)^\sfh \leq U_{\sfh}(s) \leq \frac{1.37}{(1-\sfh) \, \Gamma(1-\sfh)} s^{-1}(1+s)^\sfh \quad\text{for all } s\geq 4\sfh,
    \end{equation}
    where $U_\sfh$ is the function defined in~\eqref{def_Un}.
\end{lemma}

\begin{proof}
    Given the integral representation~\eqref{Un_integral_rep} for $U_{\sfh}$, valid precisely for $0\leq\sfh<1$, we will be done if we can bracket the expression
    \begin{equation}
        \Gamma(1-\sfh) \, U_{\sfh}(s) \, s(1+s)^{-\sfh} = \int_0^\infty e^{-t} t^{-\sfh} \left(\frac{t+s}{1+s}\right)^\sfh \ \rmd t
    \end{equation}
    between multiples of $(1-\sfh)^{-1}$. Using the fact that $t+s \leq 1+s$ if $t\leq 1$, and $t+s \leq t+ts$ if $t\geq 1$, we find the upper bound
    \begin{multline*}
        \Gamma(1-\sfh) \, U_{\sfh}(s) \, s(1+s)^{-\sfh} \leq \int_0^1 e^{-t} t^{-\sfh} \ \rmd t + \int_1^\infty e^{-t} t^{-\sfh} t^\sfh \ \rmd t \\
        < \int_0^1 t^{-\sfh} \ \rmd t + \int_1^\infty e^{-t} \ \rmd t = \frac{1}{1-\sfh} + e^{-1} < \frac{1.37}{1-\sfh}.
    \end{multline*}
    To find a lower bound, we ignore values $t>1$ and use the fact that $(t+s)/(1+s)$ is increasing in $s$ for fixed $t\leq 1$. Thus, considering that $s\geq 4\sfh$ and $\sfh<1$, we have
    \begin{equation}
        \frac{t+s}{1+s} \geq \frac{t+4\sfh }{1+4\sfh } \geq \frac{4\sfh }{1+4\sfh } > \frac{4\sfh }{5},
    \end{equation}
    which gives
    \begin{equation*}
        \Gamma(1-\sfh) \, U_{\sfh}(s) \, s(1+s)^{-\sfh} \geq \left( 
        \min_{0\leq \sfh<1} \left(\frac{4\sfh }{5}\right)^\sfh \right) \int_0^1 e^{-t} t^{-\sfh} \ \rmd t > e^{-5/(4e)}e^{-1}\int_0^1 t^{-\sfh} \ \rmd t > \frac{0.232}{1-\sfh}. \qedhere
    \end{equation*}
\end{proof}

\begin{lemma} \label{lemma_int_quali}
    Let $c,d,k\in\bbR$ with $c\neq -1$. Let $f:(0,\infty)\longrightarrow\bbR$ be a continuous, strictly positive function such that
    \begin{equation} \label{characteristic_exponents}
        \lim_{r\to 0^+} r^{-c} f(r) \in (0,\infty), \quad \lim_{r\to\infty} r^{-d} f(r) \in (0,\infty),
    \end{equation}
    and consider the following integrals for $r>0$:
    \begin{align*}
        I_0(r) &:= \int_0^r f(s) e^{ks} \ \rmd s, \quad\ \text{well-defined when } c>-1, \\
        I_\infty(r) &:= \int_r^\infty f(s) e^{ks} \ \rmd s, \quad\text{well-defined when } (k<0) \text{ or } (k=0, \, d<-1).
    \end{align*}
    Then the following description of the behavior of $I_0(r)$ and $I_\infty(r)$ holds (see Definition~\ref{def_equivalent} for the symbol $\simeq$):
    \begin{equation*}
        \begin{array}{lll}
            \text{(a) If } c>-1,\, k<0, &\text{then} & \ I_0(r) \simeq r^{1+c}(1+r)^{-1-c}. \\
            \text{(b) If } c>-1,\, k=0,\, d\neq -1, &\text{then} & \ I_0(r) \simeq r^{1+c}(1+r)^{-c+d}. \\
            \text{(c) If } c>-1,\, k>0, &\text{then} & \ I_0(r) \simeq r^{1+c}(1+r)^{-1-c+d}e^{kr}. \\
            \text{(d) If } c<-1,\, k<0, &\text{then} & I_\infty(r) \simeq r^{1+c}(1+r)^{-1-c+d}e^{kr}. \\
            \text{(e) If } c<-1,\, k=0,\, d<-1, &\text{then} & I_\infty(r) \simeq r^{1+c}(1+r)^{-c+d}. \\
            \text{(f) If } c>-1,\, k<0, &\text{then} & I_\infty(r) \simeq (1+r)^{d}e^{kr}. \\
            \text{(g) If } c>-1,\, k=0,\, d<-1, &\text{then} & I_\infty(r) \simeq (1+r)^{1+d}. \\
        \end{array}
    \end{equation*}
\end{lemma}

\begin{proof}
    Condition~\eqref{characteristic_exponents} quickly implies $f(r) \simeq r^c(1+r)^{-c+d}$. Therefore, instead of $I_0$ and $I_\infty$, it is enough to prove the lemma for the integrals
    \begin{equation*}
        J_0(r) := \int_0^r s^c(1+s)^{-c+d}e^{ks} \ \rmd s, \quad J_\infty(r) := \int_r^\infty s^c(1+s)^{-c+d}e^{ks} \ \rmd s.
    \end{equation*}
    Calling $g_1(r),\ldots,g_7(r)$ the functions of $r$ that appear on the right side of the $\simeq$ symbol in items (a)--(g) of the lemma, we can prove each of the claims by showing that the limits $\lim_{r\to 0^+}$ and $\lim_{r\to\infty}$ of the ratios $J_0/g_m$, $m=1,2,3$, and $J_\infty/g_m$, $m=4,5,6,7$, are finite and nonzero. In each case, these limits can be found either directly or with help of the L'Hôpital Rule when applicable. We only show items (a) and (d) as examples, the others being completely analogous:
    \begin{itemize}
    
        \item[(a)] Since the integrand $r^c(1+r)^{-c+d}e^{kr}$ is integrable near $\infty$ due to $k<0$, the limit $J_0(\infty)$ is some positive number, just like $g_1(\infty)$ for the function $g_1(r) = r^{1+c}(1+r)^{-1-c}$. Thus $J_0(\infty)/g_1(\infty) \in (0,\infty)$. As for the $\lim_{r\to 0^+}$ limit, the L'Hôpital Rule is applicable due to both $J_0(0)$ and $g_1(0)$ being null (note that $r^{1+c}$ vanishes at $r=0$), and it gives
        \begin{align*}
            \lim_{r\to 0^+} \frac{J_0(r)}{g_1(r)} &= \lim_{r\to 0^+} \frac{J_0'(r)}{g_1'(r)} \\
            &= \lim_{r\to 0^+} \frac{r^c(1+r)^{-c+d}e^{kr}}{r^{1+c}(1+r)^{-1-c}\big( (1+c)r^{-1} + (-1-c)(1+r)^{-1} \big)} \\
            &= \lim_{r\to 0^+} \frac{(1+r)^d e^{kr}}{(1+c)(1+r)^{-1}\big( 1 - r(1+r)^{-1} \big)} \\
            &= \frac{1}{1+c} \in (0,\infty).
        \end{align*}
    
        \item[(d)] In this case, with $g_4(r) = r^{1+c}(1+r)^{-1-c+d}e^{kr}$, we have $g_4(0) = J_\infty(0) = \infty$ and $g_4(\infty) = J_\infty(\infty) = 0$. Hence L'Hôpital is applicable to find $\lim_{r\to a}$ at both extremes $a = 0^+$ and $a=\infty$:
        \begin{align*}
            \lim_{r\to a} \frac{J_\infty(r)}{g_4(r)} &= \lim_{r\to a} \frac{J_\infty'(r)}{g_4'(r)} \\
            &= \lim_{r\to a} \frac{-r^c(1+r)^{-c+d}e^{kr}}{r^{1+c}(1+r)^{-1-c+d}e^{kr}\big( (1+c)r^{-1} + (-1-c+d)(1+r)^{-1} + k \big)} \\
            &= \lim_{r\to a} \frac{-(1+r)}{(1+c) + (-1-c+d)r(1+r)^{-1} + kr},
        \end{align*}
        which equals $-1/(1+c) \in (0,\infty)$ at $a=0^+$ and $-1/k \in (0,\infty)$ at $a=\infty$. \qedhere
    \end{itemize}
\end{proof}

The next lemma provides bounds for 14 different integrals that will show up at various points of our proofs.

\begin{lemma} \label{lemma_int_quanti}
    Let $\sfh\geq0$ and $0<\eps\leq 1/2$ be real constants. Define, for $r>0$,
    \begin{equation*}
        p(r) = 4\sfh +2r, \quad q(r) = 1+4\sfh +2r.
    \end{equation*}
    Then there exist constants $K_{1},\ldots,K_{14} > 0$, possibly depending on $\sfh$ and/or $\eps$, such that for all $r>0$
    \begin{align}
        &\int_0^r s^2 p(s)^{-1} q(s)^\sfh e^{-s} \ \rmd s \leq K_1 \ r^3p(r)^{-1}q(r)^{-2}, \label{lemma_quanti_K1} \\
        &\int_0^r s^2 p(s)^{-2} q(s)^{2\sfh +2\eps} e^{-2s} \ \rmd s \leq K_2 \ r^3p(r)^{-2}q(r)^{-1}, \label{lemma_quanti_K2} \\
        &\int_0^r s p(s)^{-1} q(s)^{-2} \ \rmd s \leq K_3 \ r^2p(r)^{-1}q(r)^{-1}, \label{lemma_quanti_K3} \\
        &\int_0^r s^4 p(s)^{-3} q(s)^{-2+\eps} \ \rmd s \leq K_4 \ r^5p(r)^{-3}q(r)^{-2+\eps}, \label{lemma_quanti_K4} \\
        &\int_0^r s^5 p(s)^{-4} q(s)^{-2+\eps} \ \rmd s \leq K_5 \ r^6p(r)^{-4}q(r)^{-2+\eps}, \label{lemma_quanti_K5} \\
        &\int_0^r s^3 p(s)^{-1} q(s)^{-2-2\sfh } e^{2s} \ \rmd s \leq K_6 \ r^4p(r)^{-1}q(r)^{-3-2\sfh } e^{2r}, \label{lemma_quanti_K6} \\
        &\int_r^\infty s^{-2}q(s)^{\frac32+\sfh}e^{-s} \ \rmd s \leq K_7 \ r^{-1} q(r)^{\frac12+\sfh}e^{-r}, \label{lemma_quanti_K7} \\
        &\int_r^\infty s^{-2}q(s)^{3+2\sfh}e^{-2s} \ \rmd s \leq K_8 \ r^{-1} q(r)^{2+2\sfh}e^{-2r}, \label{lemma_quanti_K8} \\
        &\int_r^\infty p(s)^{-2} q(s)^{1+2\sfh +\eps} e^{-2s} \ \rmd s \leq K_9 \ p(r)^{-1}q(r)^{2\sfh +\eps} e^{-2r}, \label{lemma_quanti_K9} \\
        &\int_r^\infty s p(s)^{-3} q(s)^{1+2\sfh +\eps} e^{-2s} \ \rmd s \leq K_{10} \ p(r)^{-1}q(r)^{2\sfh +\eps} e^{-2r}, \label{lemma_quanti_K10} \\
        &\int_r^\infty s^2 p(s)^{-2} q(s)^{-\frac12 + \sfh} e^{-s} \ \rmd s \leq K_{11} \ q(r)^{-\frac12+\sfh} e^{-r}, \label{lemma_quanti_K11} \\
        &\int_r^\infty s p(s)^{-2} q(s)^{1+2\sfh} e^{-2s} \ \rmd s \leq K_{12} \ p(r)^{-2}q(r)^{2+2\sfh} e^{-2r}, \label{lemma_quanti_K12} \\
        &\int_r^\infty s^2 p(s)^{-1} q(s)^\sfh e^{-s} \ \rmd s \leq K_{13} \ q(r)^{1+\sfh} e^{-r}, \label{lemma_quanti_K13} \\
        &\int_r^\infty s^2 p(s)^{-1} q(s)^\sfh e^{-s} \ \rmd s \geq K_{14} \ q(r)^{1+\sfh} e^{-r}. \label{lemma_quanti_K14}
    \end{align}
    Furthermore, if $0\leq\sfh<1$, the above inequalities hold for the following values of the constants $K_j$:
    \begin{equation} \label{numerical_K}
        \begin{array}{|l||c|c|c|c|c|c|c|c|c|c|c|c|c|c|} \hline
            & K_1 & K_2 & K_3 & K_4 & K_5 & K_6 & K_7 & K_8 & K_9 & K_{10} & K_{11} & K_{12} & K_{13} & K_{14} \\ \hline \hline
            \sfh = 0 & 4.65 & 2.37 & 1.00 & 1/\eps & 1/\eps & 1.30 & 2.00 & 2.00 & 1.00 & 0.500 & 0.250 & 0.350 & 0.500 & 0.250 \\ \hline
            \sfh\in (0,1) & 91.4 & 625 & 1.00 & 1/\eps & 1/\eps & 5.00 & 5.00 & 5.00 & 1.00 & 0.500 & 0.300 & 0.380 & 0.750 & 0.0169 \\ \hline
        \end{array}
    \end{equation}
\end{lemma}

\begin{proof}
    The existence of all but one of the constants $K_j$ is a direct consequence of Lemma~\ref{lemma_int_quali}. Indeed, considering
    \begin{equation*}
        p(s) \simeq \left\{\begin{array}{ll}
            s & \text{in case } \sf\sfh = 0, \\
            1+s & \text{in case } \sfh>0,
        \end{array}\right. \qquad q(s) \simeq 1+s,
    \end{equation*}
    we see that
    \begin{itemize}
        \item \eqref{lemma_quanti_K1} and~\eqref{lemma_quanti_K2} follow from item (a) of Lemma~\ref{lemma_int_quali},
        \item \eqref{lemma_quanti_K3},~\eqref{lemma_quanti_K4} and~\eqref{lemma_quanti_K5} follow from item (b),
        \item \eqref{lemma_quanti_K6} follows from item (c),
        \item \eqref{lemma_quanti_K7} and~\eqref{lemma_quanti_K8} follow from item (d),
        \item \eqref{lemma_quanti_K9} and~\eqref{lemma_quanti_K10} follow from item (d) if $\sf\sfh = 0$ and item (f) if $\sfh>0$, and
        \item \eqref{lemma_quanti_K11},~\eqref{lemma_quanti_K12},~\eqref{lemma_quanti_K13} and~\eqref{lemma_quanti_K14} follow from item (f), with the exception of~\eqref{lemma_quanti_K12} in case $\sfh = 0$, which will be proven below.
    \end{itemize}
    But finding explicit values for these constants, which we only wish to do for $0\leq\sfh<1$, requires more work. We analyze each integral individually. A helpful general remark, to be used repeatedly in what follows, is that, for $a,b,c,d\in\bbR$,
    \begin{equation*}
        \frac{\rmd}{\rmd r} \big( r^a p(r)^b q(r)^c e^{dr} \big) = r^a p(r)^b q(r)^c e^{dr} \big( a\,r^{-1} + 2b\,p(r)^{-1} + 2c\,q(r)^{-1} + d \big).
    \end{equation*}
    We also make the following disclaimers about the values to be obtained:
    \begin{itemize}
        \item most of them are in no way optimal;
        \item for some of them, there will clearly be no need to restrict the value of $\sfh$ to $[0,1)$ or $\eps$ to $(0,1/2]$, other than that these just happen to be the ranges that we wish to consider;
        \item for all of them, the proof of the $0<\sfh<1$ case would also apply to the $\sfh = 0$ case; however, we provide separate derivations for when $\sfh = 0$ because we wish to derive better constants for this special case, at least when a better derivation is indeed available.
    \end{itemize}
    \begin{itemize}
    
        \item[($\bm{K_1}$)] For $s\leq r$, one can check that $p(s) \geq (s/r)p(r)$. We use this to estimate $p(s)^{-1}$ from above:
        \begin{equation*}
            \int_0^r s^2 p(s)^{-1} q(s)^\sfh  e^{-s} \ \rmd s \leq p(r)^{-1} \int_0^r s^2 \left(\frac{s}{r}\right)^{-1} q(s)^\sfh  e^{-s} \ \rmd s = rp(r)^{-1} \int_0^r sq(s)^\sfh  e^{-s} \ \rmd s. 
        \end{equation*}
        Thus we will be done if we can find $K_1$ such that
        \begin{equation} \label{mario}
            \int_0^r sq(s)^\sfh  e^{-s} \ \rmd s \leq K_1 r^2q(r)^{-2} \quad\text{for all } r>0.
        \end{equation}
        For $\sfh = 0$, the integral on the left is explicitly computable and we can choose
        \begin{equation*}
            \max_{r>0} r^{-2}q(r)^2 \int_0^r se^{-s} \ \rmd s = \max_{r>0} r^{-2}(1+2r)^2 \big(1 - (1+r)e^{-r} \big) < 4.65 =: K_1.
        \end{equation*}
        The existence of this maximum is clear --- the apparent $r^{-2}$ singularity for small $r$ is canceled by the term $1-(1+r)e^{-r} = O(r^2)$, --- however, the explicit value $4.65$ was found with computer assistance. This will also be the case for many of the other inequalities below in the $\sfh = 0$ case.
        
        For $0<\sfh<1$, noting that both sides of~\eqref{mario} vanish at $r=0$, we can instead compare their derivatives: We will be done if we can find $K_1$ such that
        \begin{equation*}
            rq(r)^\sfh e^{-r} \leq K_1 \frac{\rmd}{\rmd r}\big( r^2q(r)^{-2} \big) = K_1\cdot 2(1+4\sfh ) rq(r)^{-3} \quad\Longleftrightarrow\quad K_1 \geq \max_{r>0}\frac{q(r)^{3+\sfh}e^{-r}}{2+8\sfh}.
        \end{equation*}
        With a little bit of Calculus, this maximum is found to be exactly $(2+8\sfh)^{-1}(6+2\sfh )^{3+\sfh}e^{\sfh-5/2}$, which in turn is checked to be increasing in $\sfh$ and therefore bounded above by its value $10^{-1}8^4e^{-\frac32} < 91.4 =: K_1$ at $\sfh = 1$.
    
        \item[($\bm{K_2}$)] We start by estimating $p(s)^{-2}$ like we did $p(s)^{-1}$ above, as well as bounding $q(s)^{2\eps}$ by $q(s)$ (using the fact that $\eps\leq 1/2$):
        \begin{equation*}
            \int_0^r s^2 p(s)^{-2} q(s)^{2\sfh +2\eps} e^{-2s} \ \rmd s \leq r^2p(r)^{-2} \int_0^r q(s)^{1+2\sfh } e^{-2s} \ \rmd s. 
        \end{equation*}
        Thus we will be done if we can find $K_2$ such that
        \begin{equation} \label{luigi}
            \int_0^r q(s)^{1+2\sfh } e^{-2s} \ \rmd s \leq K_2 rq(r)^{-1} \quad\text{for all } r>0.
        \end{equation}
        For $\sfh = 0$, the integral on the left is explicitly computable and we can choose
        \begin{equation*}
            \max_{r>0} r^{-1}(1+2r) \big(1 - (1+r)e^{-2r} \big) < 2.37 =: K_2.
        \end{equation*}
        For $0<\sfh<1$, we again compare derivatives in~\eqref{luigi}: We will be done if we can find $K_2$ such that
        \begin{equation*}
            q(r)^{1+2\sfh }e^{-2r} \leq K_2 \frac{\rmd}{\rmd r}\big( rq(r)^{-1} \big) = K_2 (1+4\sfh ) q(r)^{-2} \quad\Longleftrightarrow\quad K_2 \geq \max_{r>0}\frac{q(r)^{3+2\sfh }e^{-2r}}{1+4\sfh }.
        \end{equation*}
        This maximum is found to be $(1+4\sfh )^{-1}(3+2\sfh )^{3+2\sfh }e^{2\sfh -2}$, which, for $0<\sfh<1$, is bounded above by $5^{-1}5^5e^{0} = 625 =: K_2$ similarly to the previous item.
    
        \item[($\bm{K_3}$)] We bound $p(s)^{-1}$ as before, but this time the leftover integral can be computed explicitly for any $\sfh \geq 0$:
        \begin{multline*}
            \int_0^r s p(s)^{-1} q(s)^{-2} \ \rmd s \leq \int_0^r s \left(\frac{s}{r}\right)^{-1} p(r)^{-1} q(s)^{-2} \ \rmd s = rp(r)^{-1} \int_0^r q(s)^{-2} \ \rmd s \\
           = rp(r)^{-1} \frac{1}{1+4\sfh } rq(r)^{-1} \leq r^2p(r)^{-1}q(r)^{-1}.
        \end{multline*}
    
        \item[($\bm{K_4}$)] We use not only $p(s)^{-3} \leq (s/r)^{-3}p(r)^{-3}$ like before, but also the analogous bound for $q(s)^{-2+\eps}$ --- note that the exponent $-2+\eps$ is negative. This gives
        \begin{multline*}
            \int_0^r s^4 p(s)^{-3} q(s)^{-2+\eps} \ \rmd s \leq p(r)^{-3}q(r)^{-2+\eps} \int_0^r s^4 \left(\frac{s}{r}\right)^{-3}\left(\frac{s}{r}\right)^{-2+\eps} \ \rmd s \\
           = r^{5-\eps} p(r)^{-3}q(r)^{-2+\eps} \int_0^r s^{-1+\eps} \ \rmd s = \frac{1}{\eps} r^5 p(r)^{-3}q(r)^{-2+\eps}.
        \end{multline*}
        Remark: The L'Hôpital Rule reveals that $\lim_{r\to\infty} \big( r^5 p(r)^{-3} q(r)^{-2+\eps} \big)^{-1} \int_0^r s^4 p(s)^{-3} q(s)^{-2+\eps} \ \rmd s = 1/\eps$, proving that the value $K_4 = 1/\eps$ is optimal. The fact that it blows up when $\eps$ is taken close to 0 is what will explain the unfortunate fact that, in Theorem~\ref{thm_nuchi} ahead, the value of $\sfg_0 = \sfg_0(\eps)$ becomes worse (smaller) when a better (smaller) $\eps$ is chosen.
    
        \item[($\bm{K_5}$)] The same proof as for $K_4$ above applies here, only with different exponents.
        
        \item[($\bm{K_6}$)] Again we begin by using estimates for $p(s)^{-1}$ and $q(s)^{-2}$ like before:
        \begin{equation*}
            \int_0^r s^3p(s)^{-1}q(s)^{-2-2\sfh }e^{2s} \ \rmd s \leq r^3p(r)^{-1}q(r)^{-2}\int_0^r q(s)^{-2\sfh }e^{2s}.
        \end{equation*}
        Now we will be done if we can find $K_6$ such that
        \begin{equation} \label{bowser}
            \int_0^r q(s)^{-2\sfh }e^{2s} \ \rmd s \leq K_6rq(r)^{-1-2\sfh } e^{2r} \quad\text{for all } r>0.
        \end{equation}
        If $\sfh = 0$, the explicit computation of the integral leads to the choice
        \begin{equation*}
            \max_{r>0} (2r^{-1})(1+2r)\big( 1 - e^{-2r} \big) < 1.30 =: K_6.
        \end{equation*}
        If $0<\sfh<1$, we compare derivatives in~\eqref{bowser}: It suffices to find $K_6$ such that
        \begin{align*}
            q(r)^{-2\sfh } e^{2r} &\leq K_6 \frac{\rmd}{\rmd r}\big( rq(r)^{-1-2\sfh }e^{2r} \big) \\
            &= K_6 rq(r)^{-1-2\sfh }e^{2r}\big( r^{-1} - 2(1+2\sfh )q(r)^{-1} + 2 \big) \\
            &= q(r)^{-2\sfh }e^{2r} \cdot K_6 q(r)^{-2}\big( q(r) - 2(1+2\sfh )r + 2rq(r) \big) \\
            \Longleftrightarrow K_6 &\geq \max_{r>0} q(r)^2 \big( q(r) - 2(1+2\sfh )r + 2rq(r) \big)^{-1} \\
            &= \max_{r>0} \frac{(1+4\sfh )^2 + (4+16\sfh)r + 4r^2}{1+4\sfh  + (2+4\sfh )r + 4r^2}.
        \end{align*}
        We estimate this maximum by comparing the coefficients of the distinct powers of $r$: Using $\sfh<1$, we have $(1+4\sfh )^2 < 5(1+4\sfh )$ and $4+16\sfh < (10/3)(2+4\sfh ) < 5(2+4\sfh )$, so we can choose $K_6 := 5$.
    
        \item[($\bm{K_7}$)] The idea is again to compare the derivatives of both sides of the desired inequality, but now what we have to prove looks different because the integral is of the form $\int_r^\infty$ instead of $\int_0^r$: We claim it suffices to find $K_7>0$ such that
        \begin{equation} \label{toad}
            \frac{\rmd}{\rmd r}\bigg( r^{-1}q(r)^{\frac12+\sfh}e^{-r} \bigg) \leq -\frac{1}{K_7} r^{-2}q(r)^{\frac32+\sfh}e^{-r}.
        \end{equation}
        The reason why this is enough is that, if~\eqref{toad} is known, then
        \begin{equation*}
            r^{-1}q(r)^{\frac12+\sfh}e^{-r} = - \int_r^\infty \frac{\rmd}{\rmd s}\bigg( s^{-1}q(s)^{\frac12+\sfh}e^{-s} \bigg) \ \rmd s \geq \frac{1}{K_7}\int_r^\infty s^{-2}q(s)^{-\frac32+\sfh}e^{-s} \ \rmd s
        \end{equation*}
        follows, which is the desired inequality. To prove~\eqref{toad}, we estimate:
        \begin{align*}
            \frac{\rmd}{\rmd r}\bigg( r^{-1}q(r)^{\frac12+\sfh}e^{-r} \bigg) &= r^{-1}q(r)^{\frac12+\sfh}e^{-r} \bigg( -r^{-1} + (1+2\sfh)q(r)^{-1} - 1 \bigg) \\
            &= -r^{-2}q(r)^{\frac32+\sfh}e^{-r} \frac{1+4\sfh + 2(1+\sfh)r + 2r^2}{(1+4\sfh+2r)^2}.
        \end{align*}
        Thus we can define $K_7$ by
        \begin{equation*}
            \frac{1}{K_7} = \inf_{r\geq 0} f_\sfh(r), \quad\text{where } f_\sfh(r) := \frac{1+4\sfh + 2(1+\sfh)r + 2r^2}{(1+4\sfh+2r)^2}.
        \end{equation*}
        For $\sfh = 0$, one checks that $f_0$ is decreasing in $r$ and thus bounded below by its limit $1/2$ at $r=\infty$, which yields $K_7 = 2$. For $0<\sfh<1$, one checks that $f_\sfh$ attains the global minimum $m_\sfh := (1+6\sfh-\sfh^2)(2+12\sfh+16\sfh^2)^{-1}$ at some finite value of $r\geq 0$, and furthermore $m_\sfh$ is decreasing in $\sfh$; therefore $K_7 = m_1^{-1} = 5$ works.
    
        \item[($\bm{K_8}$)] Analogously to the above case, we compute
        \begin{align*}
            \frac{\rmd}{\rmd r}\bigg( r^{-1}q(r)^{2+2\sfh}e^{-2r} \bigg) &= r^{-1}q(r)^{2+2\sfh}e^{-2r} \bigg( -r^{-1} + 4(1+\sfh)q(r)^{-1} - 2 \bigg) \\
            &= -r^{-2}q(r)^{3+2\sfh}e^{-2r} \frac{1+4\sfh+4\sfh r + 4r^2}{(1+4\sfh+2r)^2},
        \end{align*}
        which leads to the ansatz
        \begin{equation*}
            \frac{1}{K_8} = \inf_{r\geq 0} f_\sfh(r), \quad\text{where now } f_\sfh(r) := \frac{1+4\sfh + 4\sfh r + 4r^2}{(1+4\sfh+2r)^2}.
        \end{equation*}
        This time, for any $0\leq\sfh\leq 1$, the global minimum $m_\sfh := (1+4\sfh-\sfh^2)(2+10\sfh+8\sfh^2)^{-1}$ of $f_\sfh$ is attained at a finite $r\geq 0$ and is also decreasing with $\sfh$, again leading to the values $K_8 = 2$ when $\sfh = 0$ and $K_8 = 5$ otherwise.
        
        \item[($\bm{K_9}$)] We proceed similarly to the previous two cases, except that now a simpler estimate is available in the middle of the calculation:
        \begin{align*}
            \frac{\rmd}{\rmd r}\bigg( p(r)^{-1}q(r)^{2\sfh +\eps}e^{-2r} \bigg) &= p(r)^{-1}q(r)^{2\sfh +\eps}e^{-2r} \bigg( -2p(r)^{-1} + 2(2\sfh +\eps)q(r)^{-1} - 2 \bigg) \\
            &= -2p(r)^{-2}q(r)^{1+2\sfh +\eps}e^{-2r}\bigg( 1 - (2\sfh +\eps) \frac{4\sfh +2r}{(1+4\sfh +2r)^2} \bigg) \\
            &\leq -2p(r)^{-2}q(r)^{1+2\sfh +\eps}e^{-2r}\bigg( 1 - \frac{\eps+2\sfh }{1+4\sfh } \bigg) \\
            &\leq -p(r)^{-2}q(r)^{1+2\sfh +\eps}e^{-2r},
        \end{align*}
        where we used $\eps \leq 1/2$ in the last step. This proves that the value $K_9 = 1$ is good enough.
    
        \item[($\bm{K_{10}}$)] Use $sp(s)^{-1} \leq 1/2$:
        \begin{equation*}
            \int_r^\infty s p(s)^{-3} q(s)^{1+2\sfh +\eps} e^{-2s} \ \rmd s \leq \frac{1}{2}\int_r^\infty p(s)^{-2} q(s)^{1+2\sfh +\eps} e^{-2s} \ \rmd s.
        \end{equation*}
        This reduces the estimate to that of the previous case, only with an extra factor of $1/2$ included. Thus $K_{10} = K_9/2 = 1/2$ is good enough.
    
        \item[($\bm{K_{11}}$)] We start by using $s^2p(s)^{-2} \leq 1/4$:
        \begin{equation*}
            \int_r^\infty s^2 p(s)^{-2} q(s)^{-\frac12 + \sfh} e^{-s} \ \rmd s \leq \frac{1}{4} \int_r^\infty q(s)^{-\frac12 + \sfh} e^{-s} \ \rmd s.
        \end{equation*}
        Then we will be done if we can find $K_{11}>0$ such that
        \begin{equation} \label{coffeetime}
            \int_r^\infty q(s)^{-\frac12 + \sfh} e^{-s} \ \rmd s \leq 4K_{11} \, q(r)^{-\frac12+\sfh} e^{-r}.
        \end{equation}
        For $\sfh \leq 1/2$, the factor $q(s)^{-\frac12+\sfh}$ in the integrand is decreasing and thus bounded above by $q(r)^{-\frac12+\sfh}$:
        \begin{equation*}
              \int_r^\infty q(s)^{-\frac12 + \sfh} e^{-s} \ \rmd s \leq  q(r)^{-\frac12 + \sfh}\int_r^\infty e^{-s} \ \rmd s = q(r)^{-\frac12 + \sfh}e^{-r} \quad\Longrightarrow\quad K_{11} = \frac{1}{4} \text{ works.}
        \end{equation*}
        For $\sfh > 1/2$, we integrate by parts once and then use the fact that $q(s)^{-\frac32+\sfh}$ is decreasing and $-1+2\sfh > 0$:
        \begin{align*}
            \int_r^\infty q(s)^{-\frac12 + \sfh} e^{-s} \ \rmd s &= q(r)^{-\frac12+\sfh}e^{-r} +  (-1+2\sfh)\int_r^\infty q(s)^{-\frac32+\sfh} e^{-s} \ \rmd s \\
            &\leq  q(r)^{-\frac12+\sfh}e^{-r} +  (-1+2\sfh) q(r)^{-\frac32+\sfh} \int_r^\infty e^{-s} \ \rmd s \\
            &= \big( 1 + (-1+2\sfh)q(r)^{-1} \big)  q(r)^{-\frac12+\sfh}e^{-r} \\
            &\leq \left( 1 + \frac{-1+2\sfh}{1+4\sfh} \right)  q(r)^{-\frac12+\sfh}e^{-r} \\
            &< \frac{6}{5}q(r)^{-\frac12+\sfh}e^{-r} \quad\Longrightarrow\quad K_{11} = \frac{3}{10} \text{ works}.
        \end{align*}
    
        \item[($\bm{K_{12}}$)] We start by using $p(s)^{-2} \leq p(r)^{-2}$ and $s \leq q(s)/2$. Then we also note that $q(s)^{\frac12+\sfh}e^{-s}$ is decreasing --- its derivative reads $-2q(s)^{\frac12+\sfh}e^{-s} (\sfh+s)/(1+4\sfh+2s) < 0$ --- and therefore bounded above by its value at $s=r$:
        \begin{equation*}
            \int_r^\infty s p(s)^{-2} q(s)^{1+2\sfh} e^{-2s} \ \rmd s \leq \frac{1}{2} p(r)^{-2} \int_r^\infty q(s)^{2+2\sfh} e^{-2s} \ \rmd s \leq \frac{1}{2}p(r)^{-2} q(r)^{\frac12+\sfh}e^{-r} \int_r^\infty q(s)^{\frac32+\sfh} e^{-s} \ \rmd s.
        \end{equation*}
        Thus we will be done if we can find $K_{12}$ such that
        \begin{equation} \label{teatime}
            \int_r^\infty q(s)^{\frac32+\sfh} e^{-s} \ \rmd s \leq 2K_{12} \, q(r)^{\frac32+\sfh} e^{-r}.
        \end{equation}
        After integrating by parts twice, the integral on the left reduces to one for which we already have a bound in~\eqref{coffeetime} above:
        \begin{align*}
            \int_r^\infty q(s)^{\frac32+\sfh} e^{-s} \ \rmd s &= q(r)^{\frac32+\sfh} e^{-r} + (3+2\sfh)\int_r^\infty q(s)^{\frac12+\sfh} e^{-s} \ \rmd s \\
            &= q(r)^{\frac32+\sfh} e^{-r} + (3+2\sfh)q(r)^{\frac12+\sfh} e^{-r} + (3+2\sfh)(1+2\sfh)\int_r^\infty q(s)^{-\frac12+\sfh} e^{-s} \ \rmd s \\
            &\leq q(r)^{\frac32+\sfh} e^{-r} + (3+2\sfh)q(r)^{\frac12+\sfh} e^{-r} + 4K_{11}(3+2\sfh)(1+2\sfh) q(r)^{-\frac12+
            \sfh}e^{-r} \\
            &= \big( 1 + (3+2\sfh)q(r)^{-1} + 4K_{11}(3+2\sfh)(1+2\sfh)q(r)^{-2} \big) q(r)^{\frac32+\sfh} e^{-r} \\
            &\leq \left( 1 + \frac{3+2\sfh}{1+4\sfh} + 4K_{11}\frac{(3+2\sfh)(1+2\sfh)}{(1+4\sfh)^2} \right) q(r)^{\frac32+\sfh} e^{-r}.
        \end{align*}
        If $\sfh = 0$, using $K_{11} = 1/4$, the expression in parenthesis becomes $7$. For all other values of $\sfh$, using $K_{11} = 3/10$ instead, it is bounded above by its limit as $\sfh\to 0^+$, which is now $38/5$. Therefore~\eqref{teatime} implies that $K_{12} = 7/2$ can be taken for $\sfh = 0$ and $K_{12} = 19/5$ can be taken for $0<\sfh<1$.
    
        \item[($\bm{K_{13}}$)] If $\sfh = 0$, the desired inequality becomes
        \begin{equation*}
            \frac{1}{2}\int_r^\infty se^{-s} \ \rmd s \leq K_{13} \ (1+2r)e^{-r},
        \end{equation*}
        whose optimal constant is easily seen to be $\sup_{r>0} (1+2r)^{-1}(1+r)/2 = 1/2 =: K_{13}$. If $0<\sfh<1$, we bound a factor of $s$ by $q(s)/2$ and the expression $sp(s)^{-1}$ by $1/2$, then we integrate by parts until a negative coefficient is reached in front of the integral, at which point it may get ignored:
        \begin{align*}
            \int_r^\infty s^2 p(s)^{-1} q(s)^\sfh  e^{-s} \ \rmd s &\leq \frac{1}{4} \int_r^\infty q(s)^{1+\sfh }e^{-s} \ \rmd s \\
            & = \frac{1}{4}\left( q(r)^{1+\sfh }e^{-r} + 2(1+\sfh )\int_r^\infty q(s)^\sfh  e^{-s} \ \rmd s \right) \\
            & = \frac{1}{4}q(r)^{1+\sfh } e^{-r} + \frac{1+\sfh }{2}\left( q(r)^\sfh e^{-r} + 2\sfh \int_r^\infty q(s)^{\sfh-1}e^{-s} \ \rmd s \right) \\
            & = \frac{1}{4}q(r)^{1+\sfh } e^{-r} + \frac{1+\sfh }{2}q(r)^{\sfh} e^{-r} + \sfh(1+\sfh )\bigg( q(r)^{\sfh-1}e^{-r} \\
            &\qquad + 2(\sfh-1)\int_r^\infty q(s)^{\sfh-2} e^{-s} \ \rmd s \bigg) \\
            & \leq \frac{1}{4}q(r)^{1+\sfh } e^{-r} + \frac{1+\sfh }{2}q(r)^{\sfh} e^{-r} + \sfh(1+\sfh )q(r)^{\sfh-1}e^{-r} \\
            & = \left( \frac{1}{4} + \frac{1+\sfh }{2}q(r)^{-1} + \sfh(1+\sfh )q(r)^{-2} \right) q(r)^{1+\sfh } e^{-r} \\
            & \leq \left( \frac{1}{4} + \frac{1+\sfh }{2(1+4\sfh )} + \frac{\sfh(1+\sfh )}{(1+4\sfh )^2} \right) q(r)^{1+\sfh } e^{-r} \\
            & \leq \frac{3}{4} q(r)^{1+\sfh } e^{-r}.
        \end{align*}
    
        \item[($\bm{K_{14}}$)] This is the only item in the lemma whose inequality sign goes the other way. If $\sfh = 0$, replace $\sup$ with $\inf$ in the derivation of $K_{13}$ above to find $K_{14} := 1/4$. If $0<\sfh<1$, we bound the $q(s)^\sfh$ term below by $q(r)^\sfh$, while the problematic $s^2$ term can be dealt with by first throwing out the $(r,r+2\sfh)$ piece of the integration domain and then changing variables to $t = s-2\sfh$:
        \begin{align*}
            \int_r^\infty s^2p(s)^{-1}q(s)^\sfh e^{-s} \ \rmd s &\geq  q(r)^\sfh  \int_{r+2\sfh }^\infty s^2p(s)^{-1}e^{-s} \ \rmd s \\
            &= q(r)^\sfh  \int_{r}^\infty (2\sfh +t)^2p(2\sfh +t)^{-1}e^{-(2\sfh +t)} \ \rmd t \\
            &= e^{-2\sfh } q(r)^\sfh  \int_{r}^\infty \frac{(2\sfh +t)^2}{8\sfh+2t}e^{-t} \ \rmd t \\
            &\geq \frac{e^{-2\sfh }}{4} q(r)^\sfh  \int_r^\infty (2\sfh +t)e^{-t} \ \rmd t \\
            &= \frac{e^{-2\sfh }}{4} q(r)^\sfh  (1+2\sfh +r)e^{-r} \\
            &= \frac{e^{-2\sfh }}{4} q(r)^{1+\sfh } \frac{1+2\sfh +r}{1+4\sfh +2r}e^{-r} \\
            &\geq \frac{e^{-2}}{8} q(r)^{1+\sfh } e^{-r} \\
            &> 0.0169 \ q(r)^{1+\sfh } e^{-r}. \qedhere
        \end{align*}
        
    \end{itemize}
\end{proof}

\subsection{The mass/deviation problem without gravity} \label{subsec_nuchi_g0}

Although the parameter $\sfg$ defined in~\eqref{def_g} is positive, having an explicit solution of the mass/deviation problem for $\sfg=0$ in our hands will be essential for our subsequent study of it as a perturbation problem. So we derive this solution now:

\begin{theorem} \label{thm_chi_zerog}
    Let $\sfh\geq 0$. Then there exist functions $(\nu,\chi)\in C^1\big((0,\infty)\big)\times C^2\big((0,\infty)\big)$ solving the mass/deviation problem~\eqref{system_nu_chi},~\eqref{conditions_nu_chi} of parameters $(\sfg,\sfh) = (0,\sfh)$ and such that there are constants $X_0,\ldots,X_5 > 0$, depending on $\sfh$, with
    \begin{align}
        &\av{\chi(r)} = \chi(r) \geq X_0 \, (1+4\sfh +2r)^{1+\sfh } e^{-r}, \label{X0} \\
        &\av{\chi(r)} = \chi(r) \leq X_1 \, (1+4\sfh +2r)^{1+\sfh } e^{-r}, \label{X1} \\
        &\av{\chi'(r)} = -\chi'(r) \leq X_2 \, r^2(4\sfh +2r)^{-1}(1+4\sfh +2r)^{\sfh} e^{-r}, \label{X2} \\
        &\av{1-\chi(r)} = 1-\chi(r) \leq X_3 \, r^3(4\sfh +2r)^{-1}(1+4\sfh +2r)^{-2}, \label{X3} \\
        &\av{\upsilon(r)} = \upsilon(r) \leq X_4 \, r^4(4\sfh +2r)^{-1}(1+4\sfh +2r)^{-2-\sfh} e^{r}, \label{X4} \\
        &\av{\upsilon'(r)} \leq X_5 \, r^3(4\sfh +2r)^{-1}(1+4\sfh +2r)^{-1-\sfh} e^{r} \label{X5}
    \end{align}
    for all $r>0$, where
    \begin{equation} \label{upsilon}
        \upsilon(r) := \chi(r) \int_0^r \frac{t^3 \ \rmd t}{(2\sfh +t)\chi(t)^2}
    \end{equation}
    is called the \textbf{conjugate solution} to $\chi$. Furthermore, if $0\leq\sfh<1$, the above inequalities hold for the following values of the constants $X_j$:
    \begin{equation} \label{numerical_X}
        \begin{array}{|l||c|c|c|c|c|c|} \hline
            & X_0 & X_1 & X_2 & X_3 & X_4 & X_5 \\ \hline \hline
            \sfh = 0  & 0.500 & 1.00 & 2.00 & 9.29 & 4.00 & 2.60 \\ \hline
            \sfh\in (0,1) & 1.14\cdot 10^{-3} & 8.93 & 11.9 & 1.09\cdot 10^3 & 6.88\cdot 10^7 & 2.87\cdot 10^6 \\ \hline
        \end{array}
    \end{equation}
\end{theorem}

\begin{proof}
    The $\nu$ equation in~\eqref{system_nu_chi} with $\sfg$ set equal to 0 becomes trivial, and its boundary condition (the first condition in~\eqref{conditions_nu_chi}) implies $\nu \equiv -\sfh$. The rest of the problem is now a boundary-value problem for $\chi$:
    \begin{equation} \label{system_zerog}
        \chi''(r) = \frac{2(r+3\sfh)}{r(r+2\sfh )}\chi'(r) + \frac{r}{r+2\sfh }\chi(r), \quad \chi(0) = 1, \quad \chi(\infty) = 0.
    \end{equation}
    In the case $\sfh = 0$, there is not much to do: Elementary formulas for the solution $\chi$ and its conjugate $\upsilon$ are available, as can be promptly checked:
    \begin{equation} \label{chi_sol_gn0}
        \chi(r) = (1+r)e^{-r}, \quad \upsilon(r) = \frac{ (1+r)e^{-r} - (1-r)e^r }{2}.
    \end{equation}
    Note that this $\chi$ is the same electric-deviation scalar~\eqref{relative_error} from section~\ref{sec_EM_SR}, which makes sense because setting $\sfg = 0$ turns General Relativity into Special Relativity, while $\sfh = 0$ means $\nu \equiv 0$ and thus $\beta \equiv 1$, that is, the metric becomes Minkowski. Also note the derivative $\chi'$ for future use:
    \begin{equation} \label{chi_prime_sol_gn0}
        \chi'(r) = -re^{-r}.
    \end{equation}
    The desired inequalities~\eqref{X0}--\eqref{X5} for these functions become
    \begin{align*}
        &(1+r)e^{-r} \geq X_0 \, (1+2r) e^{-r}, \\
        &(1+r)e^{-r} \leq X_1 \, (1+2r) e^{-r}, \\
        &re^{-r} \leq X_2 \, r^2(2r)^{-1} e^{-r}, \\
        &1-(1+r)e^{-r} \leq X_3 \, r^3(2r)^{-1}(1+2r)^{-2}, \\
        &\frac{(1+r)e^{-r} - (1-r)e^r}{2} \leq X_4 \, r^4(2r)^{-1}(1+2r)^{-2} e^r, \\
        &\frac{re^r - re^{-r}}{2} \leq X_5 \, r^3(2r)^{-1}(1+2r)^{-1} e^r.
    \end{align*}
    Here the values $X_0 = 1/2$, $X_1 = 1$, $X_2 = 2$ are easily seen to be optimal, while $X_3,X_4,X_5$ can be defined and numerically estimated as in the proof of Lemma~\ref{lemma_int_quanti}, that is, by bounding the ratio between the left and right sides of each inequality:
    \begin{align}
        X_3 &= \max_{r>0} \frac{1-(1+r)e^{-r}}{r^3(2r)^{-1}(1+2r)^{-2}} < 9.29, \label{proof_X3_sup} \\
        X_4 &= \sup_{r>0} \frac{\big( (1+r)e^{-r} - (1-r)e^r \big)}{2\big( r^4(2r)^{-1}(1+2r)^{-2} e^r \big)} = 4, \label{proof_X4_sup} \\
        X_5 &= \max_{r>0} \frac{re^r - re^{-r}}{2\big( r^3(2r)^{-1}(1+2r)^{-1} e^r \big)} < 2.60. \label{proof_X5_sup}
    \end{align}
    This concludes the $\sfh = 0$ case. Now let $\sfh>0$ for the remainder of the proof. The non-trivial step in the study of problem~\eqref{system_zerog} is to realize that the transformations
    \begin{equation*}
        \overline{\chi}(r) := -\frac{2\sfh  \, \chi'(r)}{r^2}, \quad \overline{\overline{\chi}}(s) := e^{-2\sfh } \, \overline{\chi}\left( \frac{s}{2} - 2\sfh  \right)e^{s/2}
    \end{equation*}
    reduce the ODE to Kummer's equation~\eqref{Kummer} of parameters $a=1-\sfh$ and $b=2$, ultimately leading to the ansatz
    \begin{equation} \label{chi_formula_zerog}
        \chi(r) := \frac{1}{2\sfh  \, U_{\sfh}(4\sfh )}\int_r^\infty t^2 U_{\sfh}(4\sfh +2t)e^{-t} \ \rmd t.
    \end{equation}
    We do not have to show proof of this reduction; for mathematical correctness, it suffices to prove now that~\eqref{chi_formula_zerog} is indeed a solution to~\eqref{system_zerog}. To this end, first note that $\chi$ is well-defined: We have already noted that $U_\sfh(4\sfh)$ is never zero, and the integrability of $t^2 U_{\sfh}(4\sfh +2t)e^{-t}$ around $t=\infty$ is granted by the $e^{-t}$ term, which decays faster than the polynomial-like term $t^2U_{\sfh}(4\sfh +2t) \sim t^2 t^{\sfh-1}$ grows (see~\eqref{Un_asymp}). This also implies $\chi(\infty) = 0$. Next, to check that $\chi$ solves the differential equation in~\eqref{system_zerog}, we find an expression for $\chi$ not involving any integrals. Consider the auxiliary function
    \begin{equation*}
        V(t) = U_{\sfh}(4\sfh +2t)e^{-t}.    
    \end{equation*}
    The Kummer equation~\eqref{Kummer} for $U_{\sfh}$ yields, after some simple algebra, the following second-order equation for $V$:
    \begin{equation} \label{V_equation}
        (2\sfh +t)V''(t) + 2V'(t) = tV(t).
    \end{equation}
    Then we use integration by parts to express a multiple of $\chi$ in terms of $V$ and $V'$ without having to use any integral signs in the end:
    \begin{align}
        -2\sfh  \, U_{\sfh}(4\sfh ) \chi(r) &= - \int_r^\infty t^2 V(t) \ \rmd t \label{aux_chi_aux} \\
        &= -\int_r^\infty t \big( (2\sfh +t)V''(t) + 2V'(t) \big) \ \rmd t \notag \\
        &= -(t^2+2\sfh t)V'(t)\bigg|_{t=r}^\infty + \int_r^\infty (2\sfh +t)V'(t) \ \rmd t - 2tV(t)\bigg|_{t=r}^\infty + \int_r^\infty 2V(t) \ \rmd t \notag \\
        &= (r^2+2\sfh r)V'(r) + (2\sfh +t)V(t)\bigg|_{t=r}^\infty - \int_r^\infty 2V(t) \ \rmd t + 2rV(r) + \int_r^\infty 2V(t) \ \rmd t \notag \\
        &= (r^2+2\sfh r)V'(r) - 2\sfh V(r). \label{aux_chi}
    \end{align}
    Differentiating this two times and using the $V$ equation~\eqref{V_equation}, we also find
    \begin{equation*}
        -2\sfh  \, U_{\sfh}(4\sfh ) \chi'(r) = r^2 V(r), \quad -2\sfh  \, U_{\sfh}(4\sfh ) \chi''(r) = 2r V(r) + r^2 V'(r),
    \end{equation*}
    and now the $\chi$ equation~\eqref{system_zerog} is easy to verify:
    \begin{align*}
        -2\sfh  \, &U_{\sfh}(4\sfh ) \left( \chi'' - \frac{2(r+3\sfh)}{r(r-2\sfh )}\chi' - \frac{r}{r+2\sfh }\chi \right) \\
        &= 2rV + r^2 V' - \frac{2(r+3\sfh )}{r(r+2\sfh )}r^2V - \frac{r}{r+2\sfh }\big( (r^2+2\sfh r)V' - 2\sfh V \big) \\
        &= \left( 2r - \frac{2r(r+3\sfh )}{r+2\sfh } + \frac{2\sfh r}{r+2\sfh } \right)V + \left(r^2 - \frac{r(r^2+2\sfh r)}{r+2\sfh }\right)V' \\
        &= 0.
    \end{align*}
    To finish checking the boundary conditions on $\chi$, a different integral expression is more convenient:
    \begin{equation} \label{chi_alternative}
        \chi(r) = 1 - \frac{1}{2\sfh  \, U_{\sfh}(4\sfh )}\int_0^r t^2 U_{\sfh}(4\sfh +2t)e^{-t} \ \rmd t, \quad r>0.
    \end{equation}
    This expression follows immediately from our definition~\eqref{chi_formula_zerog} of $\chi$ after using the fact that
    \begin{equation*}
        \int_0^\infty t^2 U_{\sfh}(4\sfh +2t)e^{-t} \ \rmd t = 2\sfh  \, U_{\sfh}(4\sfh ),    
    \end{equation*}
    which in turn is obtained from~\eqref{aux_chi_aux} and ~\eqref{aux_chi} when setting $r=0$. Then plugging $r=0$ into~\eqref{chi_alternative} is allowed because the integrand is not singular at $t=0$, and we find $\chi(0) = 1$.
    
    Hence a solution $\chi$ to~\eqref{system_zerog} exists also for all $\sfh >0$. We also point out that, since the ODE~\eqref{system_zerog} arises from the $\chi$ equation from~\eqref{embltp} with $\alpha=0$ and $\beta=r/(r+2\sfh)$, our previous Lemma~\ref{lemma_sign_chi} is applicable and yields
    \begin{equation} \label{signs_chi2}
        0 < \chi(r) < 1, \quad \chi'(r) < 0 \quad\text{for all } r>0.
    \end{equation}
    We will use these properties shortly. Next we tackle the estimates~\eqref{X0}--\eqref{X5}. For this we need lower and upper bounds for the quantity $U_{\sfh}(4\sfh +2t)/\big(2\sfh  \, U_{\sfh}(4\sfh )\big)$ featuring in the solution formulas proved for $\chi$. Lemma~\ref{lemma_U} provides bounds for $U_{\sfh}(s)$, $s\geq 4\sfh$, in the $0\leq \sfh < 1$ regime by bracketing $U_{\sfh}(s)$ between constant multiples of the function $s^{-1}(1+s)^\sfh$. We claim that this bracketing holds in general (but for different constants): For $\sfh > 0$, there exist $\sfh$-dependent constants $M_1,M_2>0$ such that
    \begin{equation} \label{estimate_U_general}
        M_1 s^{-1}(1+s)^\sfh \leq U_{\sfh}(s) \leq M_2 s^{-1}(1+s)^\sfh \quad\text{for all } s\geq 4\sfh .
    \end{equation}
    To prove this, we check the sufficient conditions described after Definition~\ref{def_equivalent} of the equivalence relation $\simeq$. First note that the two limits
    \begin{equation*}
        \lim_{s\to a} \frac{U_{\sfh}(s)}{s^{-1}(1+s)^\sfh}, \quad a = 4\sfh ^+ \text{ and } a=\infty,
    \end{equation*}
    are finite and nonzero --- for $a=4\sfh ^+$ this is due to $U_{\sfh}(4\sfh ) > 0$, while for $a=\infty$ it follows from the
    asymptotic property~\eqref{Un_asymp}. Moreover, given that $\chi'(r)$ is found from~\eqref{chi_formula_zerog} as
    \begin{equation} \label{chi_prime_zerog}
        \chi'(r) = -\frac{U_{\sfh}(4\sfh +2r)}{2\sfh \,U_{\sfh}(4\sfh )} r^2e^{-r},
    \end{equation}
    and considering that we already know $U_{\sfh}(4\sfh )>0$ and $\chi'(r)<0$ (see~\eqref{signs_chi2}), the strict positivity of $U_{\sfh}(4\sfh+2r)$ for all $r>0$ follows, that is, $U_{\sfh}(s) > 0$ for all $s\geq 4\sfh$. Then $U_{\sfh}(s) \simeq s^{-1}(1+s)^\sfh$ for $s\in[4\sfh ,\infty)$, that is, bounds of the form~\eqref{estimate_U_general} hold as claimed. They imply, for all $t>0$,
    \begin{multline} \label{Ufrac_bound1}
        \frac{U_{\sfh}(4\sfh +2t)}{2\sfh \,U_{\sfh}(4\sfh )} \geq \frac{1}{2\sfh }\frac{M_1 (4\sfh +2t)^{-1}(1+4\sfh +2t)^{\sfh}}{M_2 (4\sfh )^{-1}(1+4\sfh )^{\sfh} } \\
        = \frac{2M_1}{(1+4\sfh )^{\sfh}M_2} (4\sfh +2t)^{-1}(1+4\sfh +2t)^{\sfh} =: N_1 (4\sfh +2t)^{-1}(1+4\sfh +2t)^{\sfh}
    \end{multline}
    and analogously
    \begin{equation} \label{Ufrac_bound2}
        \frac{U_{\sfh}(4\sfh +2t)}{2\sfh \,U_{\sfh}(4\sfh )} \leq \frac{2M_2}{(1+4\sfh )^{\sfh}M_1} (4\sfh +2t)^{-1}(1+4\sfh +2t)^{\sfh} =: N_2 (4\sfh +2t)^{-1}(1+4\sfh +2t)^{\sfh}.
    \end{equation}
    Also recall that in Lemma~\eqref{lemma_U} we derived the values
    \begin{equation*}
        M_1 = \frac{0.232}{(1-\sfh) \, \Gamma(1-\sfh)}, \quad M_2 = \frac{1.37}{(1-\sfh) \, \Gamma(1-\sfh)} \quad\text{for } 0\leq \sfh <1,
    \end{equation*}
    with which we now deduce
    \begin{equation} \label{Ufrac_bound3}
        N_1 = \frac{2\cdot 0.232}{1.37(1+4\sfh )^{\sfh}} \geq \frac{2\cdot 0.232}{1.37\cdot 5} > 0.0677, \quad N_2 = \frac{2\cdot 1.37}{0.232(1+4\sfh )^{\sfh}} < 11.9 \quad\text{for } 0\leq \sfh <1.
    \end{equation}
    Now we simply have to use~\eqref{Ufrac_bound1} and~\eqref{Ufrac_bound2} to prove~\eqref{X0}--\eqref{X5} for all $\sfh>0$. Some of the inequalities from Lemma~\eqref{lemma_int_quanti}, employing the constants $K_j$, will be needed. Thus each $X_j$ will be found in terms of the constants $K_j$ and $N_1,N_2$, and furthermore, as the reader's calculator will check, the explicit values of $X_j$ claimed in~\eqref{Ufrac_bound3} for the case $0<\sfh<1$ follow by subbing in the explicit ones for $K_j$ (see~\eqref{numerical_K}) and $N_1,N_2$ (see~\eqref{Ufrac_bound3} just above).
    
    \begin{itemize}
    
        \item[($\bm{X_{0}}$)] Use the definition~\eqref{chi_formula_zerog} of $\chi$ to see that $X_0 := N_1K_{14}$ works:
        \begin{multline*}
            \chi(r) = \frac{1}{2\sfh \,U_{\sfh}(4\sfh )}\int_r^\infty t^2U_{\sfh}(4\sfh +2t)e^{-t} \ \rmd t \geq N_1 \int_r^\infty t^2(4\sfh +2t)^{-1}(1+4\sfh +2t)^{\sfh} e^{-t} \ \rmd t \\
            \geq N_1K_{14} (1+4\sfh +2r)^{1+\sfh }e^{-r}.
        \end{multline*}
    
        \item[($\bm{X_{1}}$)] Analogously to the above item, we find that $X_1 := N_2K_{13}$ works.
    
        \item[($\bm{X_{2}}$)] The derivative $\chi'(r)$ was given in~\eqref{chi_prime_zerog}, immediately yielding $X_2 := N_2$:
        \begin{equation}
            \chi'(r) = -\frac{U_{\sfh}(4\sfh +2r)}{2\sfh  \, U_{\sfh}(4\sfh )}r^2e^{-r} \geq -N_2r^2(4\sfh +2r)^{-1}(1+4\sfh +2r)^{\sfh}e^{-r}.
        \end{equation}
    
        \item[($\bm{X_{3}}$)] This time we use the alternative formula~\eqref{chi_alternative} for $\chi$ to find $X_3 := N_2K_{1}$:
        \begin{multline*}
            1 - \chi(r) = \frac{1}{2\sfh  \, U_{\sfh}(4\sfh )}\int_0^r t^2U_{\sfh}(4\sfh +2t)e^{-t} \ \rmd t \leq N_2\int_0^r t^2(4\sfh +2t)^{-1}(1+4\sfh +2t)^{\sfh} e^{-t} \ \rmd t \\
            \leq N_2K_1 r^3(4\sfh +2r)^{-1}(1+4\sfh +2r)^{-2}.
        \end{multline*}
    
        \item[($\bm{X_{4}}$)] We use~\eqref{X0} and~\eqref{X1} to bound the $\chi$ factors in formula~\eqref{upsilon} for $\upsilon(r)$, ultimately arriving at the value $X_4 := 2X_1K_6/X_0^2$:
        \begin{align*}
            \upsilon(r) &= \chi(r)\int_0^r \frac{t^3 \ \rmd t}{(2\sfh +t)\chi(t)^2 } \\
            &\leq \frac{2X_1}{X_0^2} (1+4\sfh +2r)^{1+\sfh }e^{-r} \int_0^r \frac{t^3 e^{2t} \ \rmd t}{(4\sfh +2t)(1+4\sfh +2t)^{2+2\sfh }} \\
            &\leq \frac{2X_1}{X_0^2} (1+4\sfh +2r)^{1+\sfh }e^{-r} K_6 r^4 (4\sfh +2r)^{-1}(1+4\sfh +2r)^{-3-2\sfh }e^{2r} \\
            &= \frac{2X_1K_6}{X_0^2} r^4 (4\sfh +2r)^{-1} (1+4\sfh +2r)^{-2-\sfh} e^r.
        \end{align*}
    
        \item[($\bm{X_{5}}$)] We differentiate $\upsilon(r)$ and bound away with the help of~\eqref{X0} and~\eqref{X2}, ultimately arriving at the value $X_5 := X_2K_6/(4X_0^2) + 2/X_0$:
        \begin{align*}
            \av{\upsilon'(r)} &= \av{\chi'(r)\int_0^r \frac{t^3 \ \rmd t}{(2\sfh +t)\chi(t)^2 } + \frac{r^3}{(2\sfh +r)\chi(r)}} \\
            &\leq 2\av{\chi'(r)}\int_0^r \frac{t^3 \ \rmd t}{(4\sfh +2t)\chi(t)^2} + \frac{2r^3}{(4\sfh +2r)\chi(r)} \\
            &\leq \frac{2X_2}{X_0^2} r^2(4\sfh +2r)^{-1}(1+4\sfh +2r)^{\sfh}e^{-r} \int_0^r \frac{t^3 e^{2t} \ \rmd t}{(4\sfh +2t)(1+4\sfh +2t)^{2+2\sfh }} \\
            &\qquad + \frac{2}{X_0}\frac{r^3}{(4\sfh +2r)(1+4\sfh +2r)^{1+\sfh }e^{-r}} \\
            &\leq \frac{2X_2}{X_0^2} r^2(4\sfh +2r)^{-1}(1+4\sfh +2r)^{\sfh}e^{-r} K_6 r^4 (4\sfh +2r)^{-1}(1+4\sfh +2r)^{-3-2\sfh }e^{2r} \\
            &\qquad + \frac{2}{X_0} r^3(4\sfh +2r)^{-1}(1+4\sfh +2r)^{-1-\sfh}e^r \\
            &= r^3(4\sfh +2r)^{-1}(1+4\sfh +2r)^{-1-\sfh}e^r\left( \frac{2X_2K_6}{X_0^2} r^3(4\sfh +2r)^{-1}(1+4\sfh +2r)^{-2} + \frac{2}{X_0} \right) \\
            &\leq r^3(4\sfh +2r)^{-1}(1+4\sfh +2r)^{-1-\sfh}e^r\left( \frac{X_2K_6}{4X_0^2} + \frac{2}{X_0} \right). \qedhere
        \end{align*}
        
    \end{itemize}
\end{proof}

\begin{remark}
    The conjugate solution $\upsilon$ defined in~\eqref{upsilon} also solves the second-order ODE~\eqref{system_zerog} and is independent from $\chi$ at all $r>0$, in the sense that the Wronskian of $\chi$ and $\upsilon$ never vanishes over this range. In fact, $\upsilon$ is the solution that results from the reduction-of-order technique applied to $\chi$. After quickly checking that $\upsilon(0) = 0$ and $\upsilon(\infty) = \infty$, it can be deduced that $\chi$ as given in~\eqref{chi_formula_zerog} is in fact the \textit{unique} solution of the boundary-value problem~\eqref{system_zerog}.
\end{remark}

To close this subsection, we prove a formula for the solution to the \textit{inhomogeneous} version of the differential equation~\eqref{system_zerog} for $\chi$. It will be required when tackling the $\sfg>0$ problem, and it is here that the notion~\eqref{upsilon} of the conjugate $\upsilon$ to $\chi$ is needed.

\begin{corollary} \label{cor_inhomog_zerog}
    Let $\sfh\geq 0$. Let $\chi$ and $\upsilon$ be as in Theorem~\ref{thm_chi_zerog}, and suppose that a function $H\in C^0\big((0,\infty)\big)$ satisfies
    \begin{equation} \label{thm_inhomog_omega}
        \av{H(r)} \leq C r^x (4\sfh +2r)^y (1+4\sfh +2r)^z  e^{-r} \quad\text{for all } r>0,
    \end{equation}
    for some $C>0$ and $x,y,z\in\bbR$ with $x>2$,  $x+y\in (-2,1)$ and $x+y+z\neq \sfh$. Then the formula
    \begin{equation} \label{solution_inhomog}
        \psi(r) = -\chi(r)\int_0^r \frac{s+2\sfh }{s^3} \upsilon(s)H(s) \ \rmd s
        - \upsilon(r) \int_r^\infty \frac{s+2\sfh }{s^3} \chi(s)H(s) \ \rmd s
    \end{equation}
    defines a solution to the inhomogeneous boundary-value problem
    \begin{equation} \label{inhomog_bvp}
        \left\{\begin{array}{l}
            \psi''(r) = \dfrac{2(r+3\sfh )}{r(r+2\sfh )}\psi'(r) + \dfrac{r}{r+2\sfh }\psi(r) + H(r), \\ \\
            \psi(0) = \psi(\infty) = 0.
        \end{array}\right.
    \end{equation}
    We also write down its derivative:
    \begin{equation} \label{solution_prime_inhomog}
        \psi'(r) = -\chi'(r)\int_0^r \frac{s+2\sfh }{s^3} \upsilon(s)H(s) \ \rmd s
        - \upsilon'(r) \int_r^\infty \frac{s+2\sfh }{s^3} \chi(s)H(s) \ \rmd s.
    \end{equation}
\end{corollary}

\begin{proof}
    Assuming that $\psi$ is well-defined, it is immediate to check formula~\eqref{solution_prime_inhomog} and also that it solves the given inhomogeneous equation. In doing this, it helps to first check that $\upsilon$ solves the homogeneous one. These checks are skipped. The non-trivial claim is that $\psi$ is well-defined and satisfies the claimed null boundary conditions. We prove these properties separately for the two terms that compose it in formula~\eqref{solution_inhomog}.
    
    Due to~\eqref{thm_inhomog_omega} and the estimate~\eqref{X4} for $\upsilon$, there is a constant that we also call $C$ --- we let the value of the arbitrary constant $C$ get updated as needed along this proof --- such that the first integrand in formula~\eqref{solution_inhomog} obeys
    \begin{equation*}
        \av{ \frac{s+2\sfh }{s^3}\upsilon(s)H(s) } \leq C s^{1+x} (4\sfh +2s)^y (1+4\sfh +2s)^{-2-{\sfh}+z}.
    \end{equation*}
    On the right side of this inequality, we see a function $f(s)$ which has positive, finite values for both of the limits $\lim_{s\to 0^+} s^{-c}f(s)$ and $\lim_{s\to \infty} s^{-d}f(s)$, where
    \begin{equation*}
        c = \left\{\begin{array}{ll}
            1+x+y > -1 & \text{in case } \sfh = 0, \\
            1+x > 3 > -1 & \text{in case } \sfh > 0,
        \end{array}\right. \quad d = -1-\sfh+x+y+z \neq -1.
    \end{equation*}
    Then Lemma~\ref{lemma_int_quali} (item (b)) applies and yields the existence of $C>0$ such that
    \begin{multline*}
        \av{ \int_0^r \frac{s+2\sfh }{s^3}\upsilon(s)H(s) \ \rmd s } \leq C\int_0^r s^{1+x} (4\sfh +2s)^y (1+4\sfh +2s)^{-2-{\sfh}+z} \ \rmd s \\
        \leq \left\{\begin{array}{ll}
            C r^{2+x+y}(1+r)^{-2-{\sfh}+z} & \text{in case } \sfh = 0, \\
            C r^{2+x}(1+r)^{-2-{\sfh}+y+z} & \text{in case } \sfh > 0.
        \end{array}\right.
    \end{multline*}
    Now each of these two expressions is bounded above in its corresponding case by a multiple of the function $r^{2+x}(4\sfh +2r)^y(1+4\sfh +2r)^{-2-{\sfh}+z}$. Thus, also using estimate~\eqref{X1} for $\chi$, we find for all $\sfh\geq 0$
    \begin{multline*}
        \av{\chi(r)\int_0^r \frac{s+2\sfh }{s^3} \upsilon(s)H(s) \ \rmd s} \leq C (1+4\sfh +2r)^{1+\sfh }e^{-r} r^{2+x}(4\sfh +2r)^y(1+4\sfh +2r)^{-2-{\sfh}+z} \\
       = C r^{2+x}(4\sfh +2r)^y (1+4\sfh +2r)^{-1+z} e^{-r},
    \end{multline*}
    which is everywhere bounded, vanishes at $r=0$ (due to $x+y > -2$ if $\sfh = 0$, and due to $x>-2$ if $\sfh>0$), and also clearly vanishes at $r=\infty$, as needed.
    
    We proceed similarly for the second integrand in formula~\eqref{solution_inhomog}. This time (details skipped), it will be item (d) for $\sfh=0$ and item (f) for $\sfh > 0$ from Lemma~\ref{lemma_int_quali} that will help in ultimately deriving
    \begin{equation*}
        \av{\upsilon(r)\int_r^\infty \frac{s+2\sfh }{s^3} \upsilon(s)H(s) \ \rmd s} \leq C r^4(4\sfh +2r)^{-2+x+y}(1+4\sfh +2r)^{-2-2\sfh +z} e^{-r},
    \end{equation*}

    which is everywhere bounded, vanishes at $r=0$ (due to $x+y > -2$ if $\sfh = 0$), and also clearly vanishes at $r=\infty$, as needed.
\end{proof}

\begin{remark}
    The solution formula~\eqref{solution_inhomog} was found, of course, through the variation-of-parameters technique applied to the fundamental solutions $\chi$ and $\upsilon$ of the homogeneous ODE. This fact can be used to prove the \textit{uniqueness} of the solution to the boundary-value problem~\eqref{inhomog_bvp}.
\end{remark}

\subsection{The mass/deviation problem with gravity} \label{subsec_nuchi_g}

We now prove the existence of a solution to the mass/devi\-ation problem~\eqref{system_nu_chi},~\eqref{conditions_nu_chi} of parameters $\sfg>0$ and $\sfh\geq 0$ when $\sfg$ is small enough. A free parameter $\eps\in (0,1/2]$ is part of this result's statement, and we immediately see that choosing a smaller $\eps$ improves the power of the estimates to be obtained, which feature $\eps$ as an exponent on the right side, but this comes at the expense of having to make $\sfg$ smaller. For simplicity, we will fix the value $\eps=1/2$ when computing explicit bounds in the case $0\leq\sfh<1$.

\begin{theorem} \label{thm_nuchi}
    Let $\sfh\geq 0$ and $0<\eps\leq 1/2$. Then there exist constants $W_1,\ldots,W_4,w_1,\ldots,w_3 > 0$ and a threshold value $\sfg_0>0$ (all depending on $\sfh$ and/or $\eps$) such that, for all $0<\sfg<\sfg_0$, a solution $(\nu,\chi)\in C^1\big((0,\infty)\big)\times C^2\big((0,\infty)\big)$ to the mass/deviation problem~\eqref{system_nu_chi},~\eqref{conditions_nu_chi} of parameters $(\sfg,\sfh)$ exists satisfying the following inequalities for all $r>0$:
    \begin{align}
        &\av{1 + \chi(r)} = 1+\chi(r) \leq W_1, \label{tomato1} \\
        &\av{1 - \chi(r)} = 1 - \chi(r) \leq W_2\, r^3(4\sfh +2r)^{-1} (1+4\sfh +2r)^{-2}, \label{tomato2} \\
        &\av{\chi(r)} = \chi(r) \leq W_3\, (1+4\sfh +2r)^{1+\sfh+\eps}e^{-r}, \label{tomato3} \\
        &\av{\chi'(r)} = -\chi'(r) \leq W_4\, r^2(4\sfh +2r)^{-1}(1+4\sfh +2r)^{\sfh+\eps} e^{-r},  \label{tomato4} \\
        &\av{\nu(r)-\rng\nu(r)} \leq \sfg w_1 r^2 (4\sfh+2r)^{-1} (1+4\sfh+2r)^{-1}, \label{tomato5} \\
        &\av{\chi(r)-\rng\chi(r)} \leq \sfg w_2 r^4(4\sfh+2r)^{-2}(1+4\sfh+2r)^{-1+\sfh+\eps}e^{-r}, \label{tomato6} \\
        &\av{\chi'(r)-{\rng\chi}'(r)} \leq \sfg w_3 \, r^3(4\sfh+2r)^{-2}(1+4\sfh+2r)^{\sfh+\eps}e^{-r}, \label{tomato7}
    \end{align}
    where $(\rng\nu,\rng\chi)$ is the solution corresponding to parameters $(0,\sfh)$. Furthermore, if $0\leq\sfh<1$ and $\eps = 1/2$, the constants mentioned above can be taken as
    \begin{equation} \label{numerical_W}
        \begin{array}{|l||c|c|c|c|c|c|c|c|} \hline
            & \sfg_0 & W_1 & W_2 & W_3 & W_4 & w_1 & w_2 & w_3 \\ \hline \hline
            \sfh = 0 & 2\cdot 10^{-4} & 3.04 & 10.3 & 1.07 & 2.50 & 24.3 & 4950 & 2230 \\ \hline 
            \sfh\in (0,1) & 1.17\cdot 10^{-17} & 504 & 1120 & 9.00 & 12.4 & 4.63\cdot 10^5 & 8.62\cdot 10^{16} & 8.97\cdot 10^{15} \\ \hline 
        \end{array}
    \end{equation}
\end{theorem}

\begin{proof}
    Let $\sfh\geq 0$ and $0<\eps\leq 1/2$ be fixed. We first define $\sfg_0$ without giving any indication of why it will work; this will become clear later in the proof. Let $K_{1},\ldots,K_{14} > 0$ be the constants from Lemma~\ref{lemma_int_quanti} (some of which depend on $\sfh$ and/or $\eps$) and $X_0,\ldots,X_5 > 0$ the constants from Theorem~\ref{thm_chi_zerog} (all of which depend on $\sfh$). Let also
    \begin{equation} \label{def_Y}
        Y := \sup_{r>0} (1+4\sfh +2r)^{1+\sfh+\eps}e^{-r} = \left\{\begin{array}{ll}
            (2+2\sfh +2\eps)^{1+\sfh+\eps} e^{\sfh-\eps-1/2} &\text{if } \sfh \leq 1/2 + \eps, \\
            (1+4\sfh )^{1+\sfh+\eps} &\text{if } \sfh \geq 1/2 + \eps.
        \end{array}\right.
    \end{equation}
    Next, choose arbitrary numbers $Z_1,Z_2,Z_3>0$ under the constraint
    \begin{equation} \label{choice_Z1}
        Z_1 \leq \frac{1}{2},
    \end{equation}
    define/choose constants $W_j$ and \textit{weights} $w_j$ such that
    \begin{equation} \label{big_W_14}
        W_1 := 1 + X_1 Y + \dfrac{YZ_2}{16}, \quad W_2 := X_3 + \dfrac{YZ_2}{2}, \quad W_3 := X_1 + \dfrac{Z_2}{16}, \quad W_4 := X_2 + \dfrac{Z_3}{2},
    \end{equation}
    \begin{equation} \label{weight_w1}
        w_1 \geq \frac{K_3W_1W_2}{2} + \frac{3K_3W_4^2 Y^2}{8},
    \end{equation}
    \begin{equation} \label{big_W_5}
        W_5 := 4W_1W_2W_4 + 8w_1(2W_3 + W_4),
    \end{equation}
    \begin{equation} \label{weight_w23}
        w_2 \geq \dfrac{K_4W_5X_1X_4}{4} + \dfrac{K_9W_5X_1X_4}{2}, \quad w_3 \geq \dfrac{K_4W_5X_2X_4}{32} + \frac{K_9W_5X_1X_5}{2},
    \end{equation}
    and let constants $T_j,U_j,V_j,T,U,V$ be given by
    \begin{equation*}
        \begin{array}{ll}
            T_1 := \dfrac{K_2W_4^2}{4}, & T := 16W_1W_2W_4 + 16(2w_1 + T_1)(2W_3+W_4), \\ \\
            U_1 := \dfrac{K_2W_3}{2}, & U := 8W_3W_4Y^2 + 2w_1 + 16U_1(2W_3+W_4), \\ \\
            V_1 := \dfrac{3K_2W_4}{8}, & V := 4W_1W_2 + 8w_1 + 16V_1(2W_3+W_4), \\ \\
            T_2 := TX_1X_4\left(\dfrac{K_5}{8} + \dfrac{K_{10}}{2}\right), & T_3 := T\left(\dfrac{K_5X_2X_4}{16} + K_{10}X_1X_5\right), \\ \\
            U_2 := UX_1X_4\left(\dfrac{K_5}{8} + \dfrac{K_{10}}{2}\right), & U_3 := U\left(\dfrac{K_5X_2X_4}{16} + K_{10}X_1X_5\right), \\ \\
            V_2 := VX_1X_4\left(\dfrac{K_5}{8} + \dfrac{K_{10}}{2}\right), & V_3 := V\left(\dfrac{K_5X_2X_4}{16} + K_{10}X_1X_5\right).
        \end{array}
    \end{equation*}
    Finally, select also a real number
    \begin{equation} \label{choice_L}
        0 < L < 1.
    \end{equation}
    We choose $\sfg_0 > 0$ so as to satisfy
    \begin{align}
        \sfg_0 &\leq \min\left\{ \,\frac{Z_1}{w_1}, \,\frac{Z_2}{w_2}, \,\frac{Z_3}{w_3}\, \right\} \text{ and} \label{g0_1} \\
        \sfg_0 &\leq \frac{L}{3}\min\left\{ \, \frac{1}{T_1} \, , \frac{w_1}{U_1w_2} \, , \frac{w_1}{V_1w_3} \, , \frac{w_2}{T_2w_1} \, , \frac{1}{U_2} \, , \frac{w_2}{V_2w_3} \, , \frac{w_3}{T_3w_1} \, , \frac{w_3}{U_3w_2} \, , \frac{1}{V_3} \, \right\}. \label{g0_2}
    \end{align}
    When $\eps = 1/2$ and $0\leq\sfh <1$, the explicit values claimed in the theorem statement for the $W_j$ constants come from the following choices, which respect all the restrictions imposed above and have been brute-forced by computer to be close to optimal, in the sense that they yield nearly as large a $\sfg_0$ as possible:
    \begin{equation*}
        Z_1 = \frac{1}{2}, \quad Z_2 = Z_3 = 1, \quad L = 0.999, \quad w_1,w_2,w_3 \text{ as in~\eqref{numerical_W}}.
    \end{equation*}
    
    Now let $0<\sfg<\sfg_0$ be fixed and let us find a solution to mass/deviation problem. We define new unknowns $(\xi,\psi)$ by the relations
    \begin{equation} \label{def_xi_psi}
        \nu(r) = \rng{\nu}(r) + \sf\sfg\xi(r), \quad \chi(r) = \rng{\chi}(r) + \sf\sfg\psi(r),
    \end{equation}
    where
    \begin{equation} \label{bolinha_bolinha}
        \rng{\nu}(r) = -\sfh, \quad \rng{\chi}(r) = \frac{1}{2\sfh  \, U_{\sfh}(4\sfh )}\int_r^\infty t^2 U_{\sfh}(4\sfh +2t)e^{-t} \ \rmd t
    \end{equation}
    is the solution corresponding to the parameters $(0,\sfh)$ found in the previous subsection. In order for $\nu$ and $\chi$ to satisfy the boundary conditions listed in~\eqref{conditions_nu_chi}, we will require null conditions
    \begin{equation} \label{boundary_xi_psi}
        \xi(0) = 0, \quad \psi(0) = 0, \quad \psi(\infty) = 0,
    \end{equation}
    since $\rng{\nu}$ and $\rng{\xi}$ already satisfy~\eqref{conditions_nu_chi}. Plugging~\eqref{def_xi_psi} into the $\nu$ and $\chi$ equations from~\eqref{system_nu_chi}, then simplifying out the terms corresponding to the $\sfg=0$ equations, we arrive at the following equations for $\xi$ and $\psi$:
    \begin{equation} \label{xi_equation_aux}
        \xi' = \frac{(1-\rng{\chi}-\sfg\psi)(1+\rng{\chi}+\sfg\psi)}{2r^2} - \frac{(r+2\sfh -2\sfg\xi)(\rng{\chi}'+\sfg\psi')^2}{2r^3},
    \end{equation}
    \begin{equation} \label{psi_equation_aux}
        \psi'' = \frac{2(r+2\sfh )}{r(r+3\sfh )}\psi' + \frac{r}{r+2\sfh }\psi + \frac{(1+\rng{\chi}+\sfg\psi)(1-\rng{\chi}-\sfg\psi)(\rng{\chi}'+\sfg\psi')}{r^2(r+2\sfh -2\sfg\xi)} + \frac{ 2\xi\big( r(\rng{\chi}+\sfg\psi) - (\rng{\chi}'+\sfg\psi') \big) }{(r+2\sfh )(r+2\sfh -2\sfg\xi)}.
    \end{equation}
    Taking advantage of the smallness of $\sfg$, we will obtain a solution $(\xi,\psi,\omega)$ (where $\omega := \psi'$) to the boundary-value problem composed of~\eqref{xi_equation_aux} ,~\eqref{psi_equation_aux} and~\eqref{boundary_xi_psi} by following a procedure that is standard in problems involving a small parameter: We will look for a fixed point, in an appropriate complete metric space $\mathcal{B}$, of the map
    \begin{equation} \label{F_map}
        \mathcal{F}:\mathcal{B}\longrightarrow\mathcal{B}, \quad (\xi,\psi,\omega)\longmapsto(\Xi,\Psi,\Omega),
    \end{equation}
    that is defined by letting the functions $(\Xi,\Psi,\Omega)$ solve~\eqref{xi_equation_aux} ,~\eqref{psi_equation_aux} ,~\eqref{boundary_xi_psi} and $\Psi' = \Omega$, but with $(\xi,\psi,\omega)$ plugged into all the terms that come accompanied by $\sfg$ in the equations. That is, given a ``guess'' $(\xi,\psi,\omega)$, we let $(\Xi,\Psi,\Omega)$ solve
    \begin{equation} \label{xi_equation}
        \Xi' = \frac{(1-\rng{\chi}-\sfg\psi)(1+\rng{\chi}+\sfg\psi)}{2r^2} - \frac{(r+2\sfh -2\sfg\xi)(\rng{\chi}'+\sfg\psi')^2}{2r^3},
    \end{equation}
    \begin{equation} \label{psi_equation}
        \Psi'' = \frac{2(r+2\sfh )}{r(r+3\sfh )}\Psi' + \frac{r}{r+2\sfh }\Psi + \frac{(1+\rng{\chi}+\sfg\psi)(1-\rng{\chi}-\sfg\psi)(\rng{\chi}'+\sfg\omega)}{r^2(r+2\sfh -2\sfg\xi)} + \frac{ 2\Xi\big( r(\rng{\chi}+\sfg\psi) - (\rng{\chi}'+\sfg\omega) \big) }{(r+2\sfh )(r+2\sfh -2\sfg\xi)},
    \end{equation}
    \begin{equation*}
        \Omega = \Psi'
    \end{equation*}
    with $\Xi(0) = \Psi(0) = \Psi(\infty) = 0$. This solution is easy to write out: The unknown $\Xi$ is immediately found by integrating the right side of~\eqref{xi_equation} --- which does not include any $\Xi$, $\Psi$ or $\Omega$, --- while $\Psi$ and $\Omega$ are given by the formulas~\eqref{solution_inhomog} and~\eqref{solution_prime_inhomog} in Corollary~\ref{cor_inhomog_zerog}, whose prerequisites will be checked to be met:
    \begin{align}
        \Xi(r) &:= \int_0^r \left( \frac{(1+\rng{\chi}+\sfg\psi)(1-\rng{\chi}-\sfg\psi)}{2s^2} - \frac{(s+2\sfh -2\sfg\xi)(\rng{\chi}'+\sfg\psi')^2}{2s^3} \right) \ \rmd s, \label{Xi_map} \\
        \Psi(r) &:= -\rng{\chi}(r)\int_0^r \frac{s+2\sfh }{s^3} \rng{\upsilon}H \ \rmd s - \rng{\upsilon}(r)\int_r^\infty \frac{s+2\sfh }{s^3} \rng{\chi}H \ \rmd s, \label{Psi_map} \\
        \Omega(r) &:= -\rng{\chi}'(r)\int_0^r \frac{s+2\sfh }{s^3} \rng{\upsilon}H \ \rmd s - \rng{\upsilon}'(r)\int_r^\infty \frac{s+2\sfh }{s^3} \rng{\chi}H \ \rmd s, \label{Omega_map}
    \end{align}
    where $\rng{\upsilon}$ is the conjugate of $\rng{\chi}$ (see~\eqref{upsilon}) and we abbreviated as $H = H(s) = H(\rng{\chi},\rng{\chi}',\xi,\psi,\omega;s)$ the inhomogeneous term
    \begin{equation} \label{i_have_so_many_labels}
        H(s) := \frac{(1+\rng{\chi}+\sfg\psi)(1-\rng{\chi}-\sfg\psi)(\rng{\chi}'+\sfg\omega)}{s^2(s+2\sfh -2\sfg\xi)} + \frac{ 2\Xi\big( s(\rng{\chi}+\sfg\psi) - (\rng{\chi}'+\sfg\omega) \big) }{(s+2\sfh )(s+2\sfh -2\sfg\xi)}
    \end{equation}
    of the ODE~\eqref{psi_equation}. The argument $s$ where all functions on the right sides of~\eqref{Xi_map},~\eqref{Psi_map},~\eqref{Omega_map} and~\eqref{i_have_so_many_labels} are evaluated has been omitted to save space. Then it is clear that, if an ordered triple of continuous functions $(\xi,\psi,\omega):(0,\infty)\longrightarrow\bbR^3$ is a fixed point of $\mathcal{F}$ in $\mathcal{B}$, then we have $\xi\in C^1\big((0,\infty)\big)$, $\psi\in C^2\big((0,\infty)\big)$, $\omega = \psi'$ and the pair $(\xi,\psi)$ solves the differential equations~\eqref{xi_equation} and~\eqref{psi_equation}. The existence of a fixed point will be proved via the Banach Fixed-Point Theorem, while the desired boundary conditions~\eqref{boundary_xi_psi} for $\xi$ and $\psi$ will follow automatically from the properties of the metric space $\mathcal{B}$, to be defined shortly. With this fixed point, the functions $(\nu,\chi)$ given by~\eqref{def_xi_psi} will be the solution for the problem in the theorem as we wanted, with the desired inequalities~\eqref{tomato1}--\eqref{tomato7} following as collateral damage of the proof.
    
    We now give the definition of $\mathcal{B}$. For the constant $\eps$ that has already been fixed since the start, consider the following weighted $L^\infty$ norms defined for appropriate functions $f:(0,\infty)\longrightarrow\bbR$:
    \begin{align}
        \nn{f}_1 &:= \sup_{r>0} r^{-2}(4\sfh +2r)(1+4\sfh +2r) \av{f(r)}, \label{norm1} \\
        \nn{f}_2 &:= \sup_{r>0} r^{-4}(4\sfh +2r)^2(1+4\sfh +2r)^{1-\sfh-\eps} e^r \av{f(r)}, \label{norm2} \\
        \nn{f}_3 &:= \sup_{r>0} r^{-3}(4\sfh +2r)^2(1+4\sfh +2r)^{-{\sfh}-\eps} e^r \av{f(r)}. \label{norm3}
    \end{align}
    Then, for the weights $w_j$ chosen in~\eqref{weight_w1} and~\eqref{weight_w23}, construct the following norm for ordered triples of such functions:
    \begin{equation} \label{norm}
        \nn{(\xi,\psi,\omega)} := \max\left\{ \frac{\nn{\xi}_1}{w_1}, \frac{\nn{\psi}_2}{w_2}, \frac{\nn{\omega}_3}{w_3} \right\}.
    \end{equation}
    We define the Banach space
    \begin{equation*}
        \mathcal{X} := \big\{ (\xi,\psi,\omega):(0,\infty)\longrightarrow\bbR^3 : \ \xi,\psi,\omega \text{ are measurable functions and } \nn{(\xi,\psi,\omega)} < \infty \big\}
    \end{equation*}
    and the closed subset $\mathcal{B} := \big\{ (\xi,\psi,\omega)\in\mathcal{X}: \ \nn{(\xi,\psi,\omega)} \leq 1 \}$. What is left now is to prove that the map $\mathcal{F}$ satisfies $\mathcal{F}(\mathcal{B}) \subseteq \mathcal{B}$ and is a strict contraction in the sense of the statement of the Banach Fixed-Point Theorem. Proving these claims only requires lengthy but elementary estimates; the non-trivial work has already been completed in Lemma~\ref{lemma_int_quanti} and in the choice of the norms $\nn{\ast}_j$, which have been carefully crafted (especially with the presence of the exponent $-\eps<0$) to make everything work. From now on, as in Lemma~\ref{lemma_int_quanti}, we use the following notation for brevity:
    \begin{equation} \label{pee_and_queueueue}
        p(r) := 4\sfh +2r, \quad q(r) := 1+4\sfh +2r,
    \end{equation}
    and we may omit the argument $(r)$ from $p(r)$ and $q(r)$ when it is clear from context (it will always be either $r$ or $s$). Also note, for repeated future use, that
    \begin{equation} \label{obvious}
        rp(r)^{-1} \leq \frac{1}{2}, \quad rq(r)^{-1} \leq \frac{1}{2}, \quad r+2\sfh = \frac{1}{2}p(r), \quad \frac{1}{r+2\sfh } = 2p(r)^{-1}.
    \end{equation}
    
    \textbf{Claim 1:} $\mathcal{F}$ as given by~\eqref{F_map},~\eqref{Xi_map},~\eqref{Psi_map} and~\eqref{Omega_map} maps $\mathcal{B}$ into itself.
    
    \textbf{Proof:} Let $(\xi,\psi,\omega)\in\mathcal{B}$; in particular, from the definitions~\eqref{norm1},~\eqref{norm2} and~\eqref{norm3} of the $\nn{\ast}_j$ norms, we have
    \begin{equation} \label{sleepy}
        \av{\xi(r)} \leq w_1 \, r^2 p^{-1} q^{-1}, \quad \av{\psi(r)} \leq w_2 \, r^4p^{-2}q^{-1+\sfh+\eps}e^{-r}, \quad \av{\omega(r)} \leq w_3 \, r^3p^{-2}q^{\sfh+\eps}e^{-r}.
    \end{equation}
    Note that the inequalities~\eqref{tomato5}--\eqref{tomato7} claimed in the theorem will follow from here once the existence of the fixed point is shown, because the solution $(\rng{\nu}+\sfg\xi,\rng{\chi}+\sfg\psi)$ of the mass/deviation problem will be such that $(\xi,\psi,\psi')\in\mathcal{B}$.
    
    Next we derive bounds for some recurring expressions that appear in the formulas for $\Xi$, $\Psi$ and $\Omega$. First, using $q \geq 1$ and $\sfg w_1 \leq Z_1 \leq 1/2$ (see~\eqref{g0_1} and~\eqref{choice_Z1}), we prove
    \begin{equation*}
        \sfg\av{\xi(r)} \leq \sfg w_1 \, r^2p^{-1}q^{-1} \leq Z_1 \, r^2p^{-1} \leq \frac{1}{2}(2\sfh +r)^2p^{-1} = \frac{p}{8},
    \end{equation*}
    which implies
    \begin{equation} \label{positivity_beta}
        r + 2\sfh - 2\sfg\xi(r) = \frac{p}{2} - 2\sfg\xi(r) \geq \frac{p}{4} > 0 \quad\text{for all } r>0
    \end{equation}
    and also yields the following useful bounds:
    \begin{equation} \label{anker1}
        r + 2\sfh  - 2\sfg\xi(r) \leq r+2\sfh  + 2\sfg\av{\xi(r)} \leq \frac{p}{2} + \frac{2p}{8} = \frac{3p}{4}, \quad\text{and } \frac{1}{r + 2\sfh  - 2\sfg\xi(r)} \leq 4p^{-1}.
    \end{equation}
    Using the estimates~\eqref{X1},~\eqref{X2} and~\eqref{X3} for $\rng{\chi}$, $1-\rng{\chi}$ and $\rng{\chi}'$, as well as the definition~\eqref{def_Y} of the constant $Y$ and the restrictions $\sfg w_j \leq Z_j$ from~\eqref{g0_1}, we also obtain:
    \begin{align}
        \label{globetrotter1} \av{1+\rng{\chi}(r)+\sfg\psi(r)} &\leq 1 + X_1 \, q^{1+\sfh }e^{-r} + \sfg w_2 \, r^4p^{-2}q^{-1+\sfh+\eps}e^{-r} \\
        \nonumber &\leq 1 + q^{1+\sfh+\eps}e^{-r} \big( X_1 \, q^{-\eps} + \sfg w_2 \, r^4p^{-2}q^{-2} \big) \\
        \nonumber &\leq 1 + Y \left( X_1 + \frac{Z_2}{16} \right) \\
        \nonumber &= W_1, \\
        \label{globetrotter2} \av{1-\rng{\chi}(r)-\sfg\psi(r)} &\leq \av{1-\rng{\chi}(r)} + g\av{\psi(r)} \\ 
        \nonumber &\leq X_3 \, r^3p^{-1}q^{-2} + \sfg w_2 \, r^4p^{-2}q^{-1+\sfh+\eps}e^{-r} \\
        \nonumber &\leq r^3p^{-1}q^{-2} \big( X_3 + \sfg w_2 \, rp^{-1}q^{1+\sfh+\eps} e^{-r} \big) \\
        \nonumber &\leq r^3p^{-1}q^{-2} \left( X_3 + \frac{Z_2 Y}{2} \right) \\
        \nonumber &= W_2 \, r^3p(r)^{-1}q(r)^{-2}, \\
        \label{globetrotter3} \av{\rng{\chi}(r)+\sfg\psi(r)} &\leq X_1 \, q^{1+\sfh }e^{-r} + \sfg w_2 \, r^4p^{-2}q^{-1+\sfh+\eps}e^{-r} \\
        \nonumber &\leq q^{1+\sfh+\eps}e^{-r} \big( X_1 \, q^{-\eps} + \sfg w_2 \, r^4p^{-2}q^{-2} \big) \\
        \nonumber &\leq q^{1+\sfh+\eps}e^{-r} \left( X_1 + \frac{Z_2}{16} \right) \\
        \nonumber &= W_3 \, q(r)^{1+\sfh+\eps}e^{-r}, \\
        \label{globetrotter4} \av{\rng{\chi}'(r)+\sfg\omega(r)} &\leq X_2 \, r^2p^{-1}q^{\sfh}e^{-r} + \sfg w_3 \, r^3p^{-2}q^{\sfh+\eps}e^{-r} \\
        \nonumber &\leq r^2 p^{-1} q^{\sfh+\eps} e^{-r} \big( X_2 \, q^{-\eps} + \sfg w_3 \, rp^{-1} \big) \\
        \nonumber &\leq r^2 p^{-1} q^{\sfh+\eps} e^{-r} \left( X_2 + \frac{Z_3}{2} \right) \\
        \nonumber &= W_4 \, r^2 p(r)^{-1} q(r)^{\sfh+\eps} e^{-r},
    \end{align}
    which directly yield the rest of the desired inequalities~\eqref{tomato1}--\eqref{tomato4} once we establish that the solution $(\xi,\psi,\omega)$ does indeed exist. Now we are ready to estimate $\Xi$, $\Psi$ and $\Omega$. The goal is to prove $\nn{(\Xi,\Psi,\Omega)} \leq 1$, which, according to~\eqref{norm}, is achieved by showing $\nn{\Xi}_1 \leq w_1$, $\nn{\Psi}_2 \leq w_2$ and $\nn{\Omega}_3 \leq w_3$. Lemma~\eqref{lemma_int_quanti} and the simple relations~\eqref{obvious} will also be needed throughout this process. We have, for all $r>0$,
    \begin{align*}
        \av{\Xi(r)} &\leq \int_0^r \left( \frac{\av{1 + \rng{\chi} + \sfg\psi} \, \av{1 - \rng{\chi} - \sfg\psi}}{2s^2} + \frac{\av{s+2\sfh -2\sfg\xi} \, \av{\rng{\chi}'+\sfg\omega}^2}{2s^3} \right) \ \rmd s \\
        &\leq \int_0^r \bigg( \frac{1}{2}s^{-2}W_1 W_2 \, s^3p^{-1}q^{-2} + \frac{1}{2}s^{-3} \frac{3p}{4} \big( W_4 \, s^2 p^{-1} q^{\sfh+\eps} e^{-s} \big)^2 \bigg) \ \rmd s \\
        &\leq \int_0^r \bigg( \frac{W_1W_2}{2}s p^{-1}q^{-2} + \frac{3W_4^2}{8} \, s p^{-1} q^{2\sfh +2\eps} e^{-2s} \bigg) \ \rmd s \\
        &= \int_0^r \bigg( \frac{W_1W_2}{2} + \frac{3W_4^2}{8} \, q^{2+2\sfh +2\eps} e^{-2s} \bigg) s p^{-1} q^{-2} \ \rmd s \\
        &\leq \left( \frac{W_1W_2}{2} + \frac{3W_4^2 Y^2}{8}\right) \int_0^r s p(s)^{-1} q(s)^{-2} \ \rmd s \\
        &\leq \left( \frac{W_1W_2}{2} + \frac{3W_4^2 Y^2}{8}\right) K_3 \, r^2 p(r)^{-1} q(r)^{-1} \\
        &\leq w_1 \, r^2 p(r)^{-1} q(r)^{-1},
    \end{align*}
    where we used the condition~\eqref{weight_w1} on $w_1$ in the end. Given the definition~\eqref{norm1} of $\nn{\ast}_1$, this proves $\nn{\Xi}_1 \leq w_1$. Next we need a bound for the expression $H$ defined in~\eqref{i_have_so_many_labels}, for which we will also use the freshly proved bound for $\av{\Xi}$ (note that $H$ includes a $\Xi$ term in it):
    \begin{align*}
        \av{H(s)} &\leq \frac{\av{1+\rng{\chi}+\sfg\psi} \, \av{1-\rng{\chi}-\sfg\psi} \, \av{\rng{\chi}'+\sfg\omega}}{s^2(s+2\sfh -2\sfg\xi)} + \frac{ 2\av{\Xi} \, \big( s \av{\rng{\chi}+\sfg\psi} + \av{\rng{\chi}'+\sfg\omega} \big) }{(s+2\sfh )(s+2\sfh -2\sfg\xi)} \\
        &\leq s^{-2}4p^{-1} W_1 W_2 \, s^3p^{-1}q^{-2} W_4 \, s^2p^{-1}q^{\sfh+\eps}e^{-s} \\
        &\qquad + 2\cdot 2p^{-1} 4p^{-1} w_1 \, s^2p^{-1}q^{-1} \big( s W_3 \, q^{1+\sfh+\eps} e^{-s} + W_4 \, s^2p^{-1}q^{\sfh+\eps} e^{-s} \big) \\
        &= 4W_1W_2W_4 \, s^3p^{-3}q^{-2+\sfh+\eps} e^{-s} + 16w_1 \, s^3p^{-3} q^{\sfh+\eps}e^{-s} \big( W_3 + W_4 \, s p^{-1} q^{-1} \big) \\
        &\leq 4W_1W_2W_4 \, s^3p^{-3}q^{\sfh+\eps} e^{-s} + 16w_1 \, s^3p^{-3} q^{\sfh+\eps}e^{-s} \left( W_3 + \frac{W_4}{2}\right) \\
        &= W_5 \, s^3p(s)^{-3}q(s)^{\sfh+\eps}e^{-s}.
    \end{align*}
    The conditions of Corollary~\ref{cor_inhomog_zerog} on the exponents $x=3$, $y=-3$ and $z=\sfh+\eps$ that we just obtained here are satisfied; this is why the small $\eps>0$ was introduced in our definition of the norms, otherwise $x+y+z$ would have been equal to $\sfh$ and the corollary would not have been applicable. Then formula~\eqref{solution_inhomog} does indeed give a zero-boundary-value solution to the $\Psi$ equation~\eqref{Psi_map}, and, using the above bound for $H$ and the estimates~\eqref{X1} and~\eqref{X4} for $\rng{\chi}$ and $\rng{\upsilon}$, we obtain the following bounds on the integrals that appear in the formulas~\eqref{Psi_map} and~\eqref{Omega_map} for $\Psi$ and $\Omega$:
    \begin{align}
        \av{ \int_0^r \frac{s+2\sfh }{s^3} \rng{\upsilon}(s)H(s) \ \rmd s } &\leq \int_0^r s^{-3}\frac{1}{2}p X_4 \, s^4 p^{-1}q^{-2-{\sfh}} e^s W_5 \, s^3p^{-3}q^{\sfh+\eps}e^{-s} \ \rmd s \label{apples} \\
        &= \frac{X_4W_5}{2} \int_0^r s^{4}p(s)^{-3}q(s)^{-2+\eps} \ \rmd s \nonumber \\
        &\leq \frac{X_4W_5K_4}{2} r^{5}p(r)^{-3}q(r)^{-2+\eps}, \nonumber \\
        \av{ \int_r^\infty \frac{s+2\sfh }{s^3} \rng{\chi}(s)H(s) \ \rmd s } &\leq \int_r^\infty s^{-3}\frac{1}{2}p X_1 \, q^{1+\sfh } e^{-s} W_5 \, s^3p^{-3}q^{\sfh+\eps}e^{-s} \ \rmd s \label{oranges} \\
        &= \frac{X_1W_5}{2} \int_r^\infty p(s)^{-2}q(s)^{1+2\sfh +\eps} e^{-2s} \ \rmd s \nonumber \\
        &\leq \frac{X_1W_5K_9}{2} p(r)^{-1}q(r)^{2\sfh +\eps}e^{-2r}. \nonumber
    \end{align}
    Therefore
    \begin{align*}
        \av{\Psi(r)} &\leq \av{\rng{\chi}(r)} \av{ \int_0^r \frac{s+2\sfh }{s^3} \rng{\upsilon}(s)H(s) \ \rmd s } + \av{\rng{\upsilon}(r)} \av{ \int_r^\infty \frac{s+2\sfh }{s^3} \rng{\chi}(s)H(s) \ \rmd s } \\ 
        &\leq X_1 \, q^{1+\sfh } e^{-r} \frac{X_4W_5K_4}{2} r^{5}p^{-3}q^{-2+\eps} + X_4 \, r^4 p^{-1}q^{-2-{\sfh}} e^r \frac{X_1W_5K_9}{2} p^{-1}q^{2\sfh +\eps}e^{-2r} \\
        &= \frac{X_1X_4W_5}{2}\big( K_4 r^5p^{-3}q^{-1+\sfh+\eps}e^{-r} + K_9 r^4p^{-2}q^{-2+\sfh+\eps}e^{-r} \big) \\
        &= \frac{X_1X_4W_5}{2} r^4p^{-2}q^{-1+\sfh+\eps} e^{-r} \big( K_4 r\,p^{-1} + K_9 q^{-1} \big) \\
        &\leq \frac{X_1X_4W_5}{2} r^4p^{-2}q^{-1+\sfh+\eps} e^{-r} \left( \frac{K_4}{2} + K_9 \right) \\
        &\leq w_2 \, r^4p(r)^{-2}q(r)^{-1+\sfh+\eps} e^{-r}, \\
        \av{\Omega(r)} &\leq \av{\rng{\chi}'(r)} \av{ \int_0^r \frac{(s+2\sfh )\rng{\upsilon}(s)H(s)}{s^3} \ \rmd s } + \av{\rng{\upsilon}'(r)} \av{ \int_r^\infty \frac{(s+2\sfh )\rng{\chi}(s)H(s)}{s^3} \ \rmd s } \\ 
        &\leq X_2 \, r^2 p^{-1} q^{\sfh} e^{-r} \frac{X_4W_5K_4}{2} r^{5}p^{-3}q^{-2+\eps} + X_5 \, r^3 p^{-1}q^{-1-\sfh} e^r \frac{X_1W_5K_9}{2} p^{-1}q^{2\sfh +\eps}e^{-2r} \\
        &= \frac{W_5}{2}\big( X_2X_4K_4 r^7p^{-4}q^{-2+\sfh+\eps}e^{-r} + X_1X_5K_9 r^3p^{-2}q^{-1+\sfh+\eps}e^{-r} \big) \\
        &= \frac{W_5}{2} r^3p^{-2}q^{\sfh+\eps} e^{-r} \big( X_2X_4K_4 r^4 p^{-2}q^{-2} + X_1X_5K_9 q^{-1} \big) \\
        &\leq \frac{W_5}{2} r^3p^{-2}q^{\sfh+\eps} e^{-r} \left( \frac{X_2X_4K_4}{16} + X_1X_5K_9 \right) \\
        &\leq w_3 \, r^3p(r)^{-2}q(r)^{\sfh+\eps} e^{-r},
    \end{align*}
    where we used the conditions~\eqref{weight_w23} on $w_2$ and $w_3$. From here, using the definitions~\eqref{norm2} and~\eqref{norm3} of the norms $\nn{\ast}_2$ and $\nn{\ast}_3$, it is immediate that $\nn{\Psi}_2 \leq w_2$ and $\nn{\Omega}_3 \leq w_3$.
    
    \textbf{Claim 2:} $\mathcal{F}:\mathcal{B} \longrightarrow \mathcal{B}$ is a strict contraction with Lipschitz constant $L$.
    
    \textbf{Proof:} Recall that $L$ was chosen in~\eqref{choice_L} to be smaller than 1. Let $(\xi,\psi,\omega),(\overline{\xi},\overline{\psi},\overline{\omega})\in\mathcal{B}$ be given. We must prove
    \begin{equation} \label{contraction}
        \nn{ (\Xi - \overline{\Xi},\Psi - \overline{\Psi},\Omega - \overline{\Omega}) } \leq L\nn{ (\xi - \overline{\xi},\psi - \overline{\psi},\omega - \overline{\omega}) }
    \end{equation}
    for $(\Xi,\Psi,\Omega) = \mathcal{F}(\xi,\psi,\omega)$ and $(\overline{\Xi},\overline{\Psi},\overline{\Omega}) = \mathcal{F}(\overline{\xi},\overline{\psi},\overline{\omega})$ --- we cannot resist the urge to use the cumbersome symbol $\overline{\Xi}$. To show~\eqref{contraction}, it is enough to prove
    \begin{align}
        \nn{\Xi-\overline{\Xi}}_1 &\leq \sfg T_1 \nn{\xi-\overline{\xi}}_1 + \sfg U_1\nn{\psi-\overline{\psi}}_2 + \sfg V_1\nn{\omega-\overline{\omega}}_3, \label{Xi_contraction} \\
        \nn{\Psi-\overline{\Psi}}_2 &\leq \sfg T_2 \nn{\xi-\overline{\xi}}_1 + \sfg U_2\nn{\psi-\overline{\psi}}_2 + \sfg V_2\nn{\omega-\overline{\omega}}_3, \label{Psi_contraction} \\
        \nn{\Omega-\overline{\Omega}}_3 &\leq \sfg T_3 \nn{\xi-\overline{\xi}}_1 + \sfg U_3\nn{\psi-\overline{\psi}}_2 + \sfg V_3\nn{\omega-\overline{\omega}}_3,\label{Psi_prime_contraction}
    \end{align}
    since then it follows immediately from the condition~\eqref{g0_2} on $\sfg_0$ that
    \begin{align*}
        \frac{\nn{\Xi-\overline{\Xi}}_1}{w_1} &\leq \frac{L}{3}\left( \frac{\nn{\xi-\overline{\xi}}_1}{w_1} + \frac{\nn{\psi-\overline{\psi}}_2}{w_2} + \frac{\nn{\omega-\overline{\omega}}_3}{w_3} \right) \\
        &\leq L \, \max\left\{ \frac{\nn{\xi-\overline{\xi}}_1}{w_1}, \frac{\nn{\psi-\overline{\psi}}_2}{w_2}, \frac{\nn{\omega-\overline{\omega}}_3}{w_3} \right\} \\
        &= L\nn{ (\xi - \overline{\xi},\psi - \overline{\psi},\omega - \overline{\omega}) }
    \end{align*}
    and similarly for $\nn{\Psi-\overline{\Psi}}_2/w_2$ and $\nn{\Omega-\overline{\Omega}}_3/w_3$, proving~\eqref{contraction}.
    
    So let us show~\eqref{Xi_contraction} first. Let the integrand of~\eqref{Xi_map} be viewed as a function $F(\rng{\chi},\rng{\chi}',\xi,\psi,\omega;s)$ for each fixed value of $s$:
    \begin{multline*}
        \Xi(r) = \int_0^r F\big(\rng{\chi}(s),\rng{\chi}'(s),\xi(s),\psi(s),\omega(s);s\big) \ \rmd s, \\
        F(u,v,x,y,z;s) := \frac{(1+u+\sfg y)(1-u-\sfg y)}{2s^2} - \frac{(s+2\sfh -2\sfg x)(v+\sfg z)^2}{2s^3}.
    \end{multline*}
    The Mean-Value Theorem gives, for arbitrary arguments $(u,v,x,y,z;s)$ and $(u,v,\overline{x},\overline{y},\overline{z};s)$,
    \begin{equation} \label{MVT}
        \av{F(u,v,x,y,z;s)-F(u,v,\overline{x},\overline{y},\overline{z};s)} \leq \sup_{x,y,z}\left| \frac{\partial F}{\partial x} \right| \av{x-\overline{x}} + \sup_{x,y,z}\left| \frac{\partial F}{\partial y} \right| \av{y-\overline{y}} + \sup_{x,y,z}\left| \frac{\partial F}{\partial z} \right| \av{z-\overline{z}}.
    \end{equation}
    These partial derivatives are quickly computed as
    \begin{equation*}
        \frac{\partial F}{\partial x} = \frac{\sfg(v+\sfg z)^2}{s^3}, \quad \frac{\partial F}{\partial y} = -\frac{\sfg (u+\sfg y)}{s^2}, \quad \frac{\partial F}{\partial z} = -\frac{\sfg(s+2\sfh -2\sfg x)(v+\sfg z)}{s^3},
    \end{equation*}
    and, when evaluated at the arguments $\big(\rng{\chi}(s),\rng{\chi}'(s),\xi(s),\psi(s),\omega(s);s\big)$, are bounded by
    \begin{align*}
        \sup_{(\xi,\psi,\omega)\in\mathcal{B}}\left| \frac{\partial F}{\partial x} \right| &= \sfg s^{-3} \sup_{(\xi,\psi,\omega)\in\mathcal{B}} \av{ \rng{\chi}' + \sfg\omega }^2 \leq \sfg W_4^2 \, s p^{-2} q^{2\sfh +2\eps} e^{-2s}, \\
        \sup_{(\xi,\psi,\omega)\in\mathcal{B}}\left| \frac{\partial F}{\partial y} \right| &= \sfg s^{-2} \, \sup_{(\xi,\psi,\omega)\in\mathcal{B}} \av{ \rng{\chi} + \sfg\psi } \leq \sfg W_3 \, s^{-2} q^{1+\sfh+\eps}e^{-s}, \\
        \sup_{(\xi,\psi,\omega)\in\mathcal{B}}\left| \frac{\partial F}{\partial z} \right| &=
        \sfg s^{-3} \, \sup_{(\xi,\psi,\omega)\in\mathcal{B}} \av{s+2\sfh  - 2\sfg\xi}\,\av{\rng{\chi}' + \sfg\omega} \leq  \frac{3\sfg W_4}{4} s^{-1} q^{\sfh+\eps} e^{-s}.
    \end{align*}
    The presence of $\sfg$ in each of these bounds is what will grant us the contraction property. The other terms in~\eqref{MVT} are bounded using the definition of the norms $\nn{\ast}_j$:
    \begin{align}
        \av{\xi(s)-\overline{\xi}(s)} &\leq s^2 p(s)^{-1}q(s)^{-1} \nn{\xi-\overline{\xi}}_1, \label{xi_difference} \\
        \av{\psi(s)-\overline{\psi}(s)} &\leq s^4 p(s)^{-2}q(s)^{-1+\sfh+\eps}e^{-s} \nn{\psi-\overline{\psi}}_2, \label{psi_difference} \\
        \av{\omega(s)-\overline{\omega}(s)} &\leq s^3 p(s)^{-2}q(s)^{\sfh+\eps}e^{-s} \nn{\omega-\overline{\omega}}_3. \label{omega_difference}
    \end{align}
    Therefore we obtain from~\eqref{Xi_map} and~\eqref{MVT}
    \begin{align}
        \av{ \Xi(r) - \overline{\Xi}(r) } &\leq
        \int_0^r \bigg( \sfg W_4^2 \, s p^{-2} q^{2\sfh +2\eps} e^{-2s} s^2 p^{-1}q^{-1} \nn{\xi-\overline{\xi}}_1 \nonumber \\
        &\qquad + \sfg W_3 \, s^{-2} q^{1+\sfh+\eps}e^{-s} s^4 p^{-2}q^{-1+\sfh+\eps}e^{-s} \nn{\psi-\overline{\psi}}_2 \nonumber \\
        &\qquad + \frac{3\sfg W_4}{4} s^{-1} q^{\sfh+\eps} e^{-s} s^3 p^{-2}q^{\sfh+\eps}e^{-s} \nn{\omega-\overline{\omega}}_3 \bigg) \ \rmd s \nonumber \\
        & = \sfg \int_0^r \bigg( W_4^2 s^3 p^{-3} q^{-1+2\sfh +2\eps} e^{-2s} \nn{\xi-\overline{\xi}}_1 + W_3 s^2 p^{-2} q^{2\sfh +2\eps} e^{-2s} \nn{\psi-\overline{\psi}}_2 \nonumber \\
        &\qquad + \frac{3W_4}{4} s^2 p^{-2} q^{2\sfh +2\eps} e^{-2s} \nn{\omega-\overline{\omega}}_3 \bigg) \ \rmd s \nonumber \\
        & = \sfg\int_0^r s^2p^{-2}q^{2\sfh +2\eps}e^{-2s} \left( W_4^2 s p^{-1}q^{-1}\nn{\xi-\overline{\xi}}_1 + W_3\nn{\psi-\overline{\psi}}_2 + \frac{3W_4}{4}\nn{\omega-\overline{\omega}}_3 \right) \ \rmd s \nonumber \\
        & \leq \sfg\left( \frac{W_4^2}{2}\nn{\xi-\overline{\xi}}_1 + W_3\nn{\psi-\overline{\psi}}_2 + \frac{3W_4}{4}\nn{\omega-\overline{\omega}}_3 \right) \int_0^r s^2p(s)^{-2}q(s)^{2\sfh +2\eps}e^{-2s} \ \rmd s \nonumber \\
        & \leq \sfg\left( \frac{W_4^2}{2}\nn{\xi-\overline{\xi}}_1 + W_3\nn{\psi-\overline{\psi}}_2 + \frac{3W_4}{4}\nn{\omega-\overline{\omega}}_3 \right)K_2 r^3p(r)^{-2}q(r)^{-1} \label{XXi_difference}
    \end{align}
    which then gives, using the definition~\eqref{norm1} of $\nn{\ast}_1$,
    \begin{align*}
        \nn{\Xi-\overline{\Xi}}_1 &\leq \left( \frac{\sfg W_4^2K_2}{2}\nn{\xi-\overline{\xi}}_1 + \sfg W_3K_2\nn{\psi-\overline{\psi}}_2 + \frac{3\sfg W_4K_2}{4}\nn{\omega-\overline{\omega}}_3 \right) \sup_{r>0} rp(r)^{-1} \\
        &= \sfg T_1 \nn{\xi-\overline{\xi}}_1 + \sfg U_1\nn{\psi-\overline{\psi}}_2 + \sfg V_1\nn{\omega-\overline{\omega}}_3,
    \end{align*}
    as claimed. To conclude the proof of claim 2, we show~\eqref{Psi_contraction} and~\eqref{Psi_prime_contraction} using the same strategy, noting that formulas~\eqref{Psi_map} and~\eqref{Omega_map} yield
    \begin{align}
        \av{\Psi(r)-\overline{\Psi}(r)} &\leq \av{\rng{\chi}(r)}\int_0^r \frac{s+2\sfh }{s^3} \av{\rng{\upsilon}} \big| 
        H-\overline{H} \big| \ \rmd s + \av{\rng{\upsilon}(r)} \int_r^\infty \frac{s+2\sfh }{s^3} \av{\rng{\chi}} \big| H - \overline{H} \big| \ \rmd s, \label{psi_diff_aux} \\
        \av{\Omega(r)-\overline{\Omega}(r)} &\leq \av{\rng{\chi}'(r)}\int_0^r \frac{s+2\sfh }{s^3} \av{\rng{\upsilon}} \big| 
        H-\overline{H} \big| \ \rmd s + \av{\rng{\upsilon}'(r)} \int_r^\infty \frac{s+2\sfh }{s^3} \av{\rng{\chi}} \big| H - \overline{H} \big| \ \rmd s.  \label{omega_diff_aux}
    \end{align}
    In here we will first need an estimate for 
    \begin{equation*}
        \av{H-\overline{H}} := \av{H(\rng{\chi}, \rng{\chi}', \xi, \psi, \omega,\Xi;s)-H(\rng{\chi}, \rng{\chi}', \overline{\xi}, \overline{\psi}, \overline{\omega}, \overline{\Xi};s)},
    \end{equation*}
    where $H$, defined in~\eqref{i_have_so_many_labels}, is now viewed as the function
    \begin{equation*}
        H(u,v,x,y,z,w;s) := \frac{(1+u+\sfg y)(1-u-\sfg y)(v+\sfg z)}{s^2(s+2\sfh -2\sfg x)} + \frac{2w\big(s(u+\sfg y)-(v+\sfg z)\big)}{(s+2\sfh )(s+2\sfh -2\sfg x)}.
    \end{equation*}
    Inequality~\eqref{MVT} applies to $H$ with an extra term added for its $w$ variable:
    \begin{multline} \label{MVT2}
        \av{H(u,v,x,y,z,w;s)-H(u,v,\overline{x},\overline{y},\overline{z},\overline{w};s)} \\
        \leq \sup_{x,y,z,w}\left| \frac{\partial H}{\partial x} \right| \av{x-\overline{x}} + \sup_{x,y,z,w}\left| \frac{\partial H}{\partial y} \right| \av{y-\overline{y}} + \sup_{x,y,z,w}\left| \frac{\partial H}{\partial z} \right| \av{z-\overline{z}} + \sup_{x,y,z,w}\left| \frac{\partial H}{\partial w} \right| \av{w-\overline{w}}.
    \end{multline}
    When evaluated at the arguments $\big(\rng{\chi}(s),\rng{\chi}'(s),\xi(s),\psi(s),\omega(s),\Xi(s);s\big)$, the partials of $H$ are checked to be bounded as follows:
    \begin{align*}
        \mathop{\sup_{(\xi,\psi,\omega)\in\mathcal{B}}}_{(\Xi,\psi,\omega)\in\mathcal{B}} \left| \frac{\partial H}{\partial x} \right| &\leq 32\sfg \bigg( W_1W_2W_4 + 2w_1( 2W_3 + W_4) \bigg) s^3 p(s)^{-4} q(s)^{\sfh+\eps} e^{-s}, \\
        \mathop{\sup_{(\xi,\psi,\omega)\in\mathcal{B}}}_{(\Xi,\psi,\omega)\in\mathcal{B}} \left| \frac{\partial H}{\partial y} \right| &\leq 2\sfg \bigg( 4W_3W_4Y^2 + w_1 \bigg) p(s)^{-2}q(s), \\
        \mathop{\sup_{(\xi,\psi,\omega)\in\mathcal{B}}}_{(\Xi,\psi,\omega)\in\mathcal{B}} \left| \frac{\partial H}{\partial z} \right| &\leq 4\sfg \bigg( W_1W_2 + 2w_1 \bigg) sp(s)^{-2}q(s)^{-1}, \\
        \mathop{\sup_{(\xi,\psi,\omega)\in\mathcal{B}}}_{(\Xi,\psi,\omega)\in\mathcal{B}} \left| \frac{\partial H}{\partial w} \right| &\leq 8\bigg( 2W_3 + W_4 \bigg) s p(s)^{-2} q(s)^{1+\sfh+\eps} e^{-s}.
    \end{align*}
    The inequalities~\eqref{xi_difference},~\eqref{psi_difference} and~\eqref{omega_difference} for $\av{\xi-\overline{\xi}}$, $\av{\psi-\overline{\psi}}$ and $\av{\omega-\overline{\omega}}$ are needed again in~\eqref{MVT2}, and this time we also need~\eqref{XXi_difference} for $\av{\Xi-\overline{\Xi}}$, which reads
    \begin{equation*}
        \av{ \Xi(s) - \overline{\Xi}(s) } \leq \sfg\bigg( 2T_1\nn{\xi-\overline{\xi}}_1 + 2U_1\nn{\psi-\overline{\psi}}_2 + 2V_1\nn{\omega-\overline{\omega}}_3 \bigg) s^3 p(s)^{-2}q(s)^{-1}.
    \end{equation*}
    Although $\av{\partial_w H}$ is not bounded by a multiple of $\sfg$ like the other partials of $H$, we see that the corresponding term $\av{\Xi-\overline{\Xi}}$ is. Consequently, all terms in~\eqref{MVT2} are indeed small and we find
    \begin{align*}
        \av{H-\overline{H}} &\leq 32\sfg \bigg( W_1W_2W_4 + 2w_1( 2W_3 + W_4) \bigg) \nn{\xi-\overline{\xi}}_1 s^5 p^{-5} q^{-1+\sfh+\eps} e^{-s} \\
        &\qquad + 2\sfg\bigg( 4W_3W_4Y^2 + w_1 \bigg) \nn{\psi-\overline{\psi}}_2 s^4 p^{-4}q^{\sfh+\eps} e^{-s} \\
        &\qquad+ 4\sfg \bigg( W_1W_2 + 2w_1 \bigg) \nn{\omega-\overline{\omega}}_3 s^4 p^{-4}q^{-1+\sfh+\eps} e^{-s} \\
        &\qquad+ 16\sfg\bigg( 2W_3 + W_4 \bigg) \bigg( T_1\nn{\xi-\overline{\xi}}_1 + U_1\nn{\psi-\overline{\psi}}_2 + V_1\nn{\omega-\overline{\omega}}_3 \bigg) \, s^4 p^{-4} q^{\sfh+\eps} e^{-s} \\
        &\leq \sfg\bigg( T\nn{\xi-\overline{\xi}}_1 + U\nn{\psi-\overline{\psi}}_2 + V\nn{\omega-\overline{\omega}}_3 \bigg) s^4p(s)^{-4}q(s)^{\sfh+\eps} e^{-s}.
    \end{align*}
    We can now promptly estimate the pair of integrals that appeared in~\eqref{psi_diff_aux} and~\eqref{omega_diff_aux}. Analogously to~\eqref{apples} and~\eqref{oranges}, the result is
    \begin{align*}
        \int_0^r \frac{s+2\sfh }{s^3} \av{\rng{\upsilon}(s)} \big| 
        H(s)-\overline{H}(s) \big| \ \rmd s &\leq \frac{g X_4}{2} \big( T\nn{\xi-\overline{\xi}}_1 + U\nn{\psi-\overline{\psi}}_2 + V\nn{\omega-\overline{\omega}}_3 \big) K_5 \, r^6 p(r)^{-4} q(r)^{-2+\eps}, \\
        \int_r^\infty \frac{s+2\sfh }{s^3} \av{\rng{\chi}(s)} \big| 
        H(s)-\overline{H}(s) \big| \ \rmd s &\leq \frac{\sfg X_1}{2} \big( T\nn{\xi-\overline{\xi}}_1 + U\nn{\psi-\overline{\psi}}_2 + V\nn{\omega-\overline{\omega}}_3 \big) K_{10} \, p(r)^{-1} q(r)^{2\sfh +\eps} e^{-2r}.
    \end{align*}
    To save space, let us temporarily write $(T\cdots)$ for $T\nn{\xi-\overline{\xi}}_1 + U\nn{\psi-\overline{\psi}}_2 + V\nn{\omega-\overline{\omega}}_3$. Then~\eqref{psi_diff_aux} and~\eqref{omega_diff_aux} become
    \begin{align*}
        \av{\Psi(r)-\overline{\Psi}(r)} &\leq X_1q^{1+\sfh }e^{-r}\frac{\sfg X_4}{2}(T\cdots)K_5r^6p^{-4}q^{-2+\eps} + X_4r^4p^{-1}q^{-2-{\sfh}}e^r\frac{\sfg X_1}{2}SK_{10}p^{-1}q^{2\sfh +\eps}e^{-2r} \\
        &= \frac{\sfg X_1X_4(T\cdots)}{2} r^4p^{-2}q^{-1+\sfh+\eps}e^{-r} \bigg( K_5 r^2p^{-2} + K_{10} q^{-1} \bigg) \\
        &\leq \frac{\sfg X_1X_4(T\cdots)}{2}\left( \frac{K_5}{4} + K_{10} \right) r^4p^{-2}q^{-1+\sfh+\eps}e^{-r} \\
        &= \bigg( \sfg T_2 \nn{\xi-\overline{\xi}}_1 + \sfg U_2\nn{\psi-\overline{\psi}}_2 + \sfg V_2\nn{\omega-\overline{\omega}}_3\bigg) r^4p^{-2}q^{-1+\sfh+\eps}e^{-r}, \\
        \av{\Omega(r)-\overline{\Omega}(r)} &\leq X_2r^2p^{-1}q^{\sfh}e^{-r}\frac{\sfg X_4}{2}(T\cdots)K_5r^6p^{-4}q^{-2+\eps} + X_5r^3p^{-1}q^{-1-\sfh}e^r\frac{\sfg X_1}{2}SK_{10}p^{-1}q^{2\sfh +\eps}e^{-2r} \\
        &= \frac{g(T\cdots)}{2} r^3p^{-2}q^{\sfh+\eps}e^{-r} \bigg( X_2X_4K_5 r^{-5}p^{-3}q^{-2} + X_1X_5K_{10} q^{-1} \bigg) \\
        &\leq \frac{g(T\cdots)}{2}\left( \frac{X_2X_4K_5}{32} + X_1X_5K_{10} \right) r^3p^{-2}q^{\sfh+\eps}e^{-r} \\
        &= \bigg( \sfg T_3 \nn{\xi-\overline{\xi}}_1 + \sfg U_3\nn{\psi-\overline{\psi}}_2 + \sfg V_3\nn{\omega-\overline{\omega}}_3\bigg) r^3p^{-2}q^{\sfh+\eps}e^{-r}.
    \end{align*}
    Using definitions~\eqref{norm2} and~\eqref{norm3} of the norms $\nn{\ast}_2$ and $\nn{\ast}_3$, we arrive at the desired inequalities~\eqref{Psi_contraction} and~\eqref{Psi_prime_contraction}.
\end{proof}

\subsection{Proof of the existence theorem}
\label{subsec_proof}

Having completely solved the mass/deviation problem~\ref{def_mdproblem}, we are ready to finish proving Theorem~\ref{main_theorem}. The theorem initially gave us constants $G,c,\varkappa,M_{\mathrm{bare}}$, with which we construct $\sfh$ as in~\eqref{newdefinition_n}. For this $\sfh$ and for the value $\eps = 1/2$ (let us not bother with a general $\eps\in (0,1/2]$ anymore), let $\sfg_0$ be the threshold provided by Theorem~\ref{thm_nuchi}. Now let $Q\neq 0$ be chosen as in the main theorem, that is, in such a way that $\sfg < \sfg_0$, for $\sfg$ as in~\eqref{def_g}. Using $Q,c,\varkappa$, we define dimensionless variables as explained previously. Theorem~\ref{thm_nuchi} then gives us a solution $(\nu,\chi)$ to the mass/deviation problem of parameters $(\sfg,\sfh)$, which, as explained in subsection~\ref{subsec_thm}, has the consequence that the functions $\beta = r/(r-2\nu)$ (see~\eqref{definition_mu} and~\eqref{definition_nu}) and $\chi$ solve their respective differential equations in the E-M-BLTP system~\eqref{embltp} and satisfy the conditions
\begin{equation*}
    \frac{1}{\sfg} \lim_{r\to 0^+} \frac{r}{2}\left( 1 - \frac{1}{\beta(r)} \right) = M_{\mathrm{bare}}, \quad \chi(0) = 1, \quad \chi(\infty) = 0.
\end{equation*}
We continue to consider the notation of the proof of Theorem~\ref{thm_nuchi}, in particular the constants $w_j, W_j$, the definition~\eqref{def_xi_psi} of $\xi$ and $\psi$, as well as $\omega = \psi'$, and the functions $p$ and $q$ from~\eqref{pee_and_queueueue}. Given that $\beta = r/(r+2\sfh - 2g\xi)$, the desired positivity condition $\beta(r) > 0$ for all $r>0$ follows from property~\eqref{positivity_beta}. Next we integrate the $\alpha$ equation in~\eqref{embltp} to obtain the only solution for $\alpha(r)$ that vanishes at infinity:
\begin{equation} \label{stop_reading_my_labels}
    \alpha(r) = \sfg \int_r^\infty \frac{ (\chi')^2 }{s^3} \ \rmd s.
\end{equation}
Thus a solution $(\alpha,\beta,\chi)$ has been found with all desired boundary conditions. It remains to find constants $C_1'$, $C_j$ that verify the estimates~\eqref{end_alpha}--\eqref{end_const_mass}. For any $\sfh\geq 0$, we will be able to express them in terms of the $K_j$ from Lemma~\ref{lemma_int_quanti} and the $W_j$, $w_j$ from Theorem~\ref{thm_nuchi}. By then using their numerical values~\eqref{numerical_K},~\eqref{numerical_W} for the $0\leq\sfh<1$ case, one can check that the values~\eqref{numerical_C} claimed in the theorem for this case are obtained.

\begin{itemize}

    \item \textbf{Proof of~\eqref{end_alpha}}. The desired inequality reads
    \begin{equation} \label{endnew_alpha}
        \av{ e^{\alpha(r)} - 1 } \leq \left\{\begin{array}{ll}
            \sfg C_1 \left( e + \dfrac{e}{2r} \right)^{\delta} \ln\left( e + \dfrac{e}{2r} \right) e^{-2r} &\text{if } \sfh = 0 \text{ (where } 0<\delta\ll 1 \text{),} \\ \\
            \sfg C_1 \exp\left(\dfrac{C_1'}{\sfh^2}\right) p(r)^{-2}q(r)^{2+2\sfh}e^{-2r} &\text{if } \sfh > 0.
        \end{array}\right.
    \end{equation}
    We start by estimating the above definition~\eqref{stop_reading_my_labels} of $\alpha$ using the bound~\eqref{tomato4} for $\chi'$:
    \begin{equation} \label{alpha_initial_kiwi}
        \alpha(r) = \sfg \int_r^\infty \frac{ (\chi')^2 }{s^3} \ \rmd s \leq \sfg W_4^2 \int_r^\infty sp(s)^{-2}q(s)^{1+2\sfh}e^{-2s} \ \rmd s.
    \end{equation}
        The behavior of this integral as a function of $r$ now depends on whether $\sfh = 0$ or $\sfh > 0$. If $\sfh = 0$, which implies $p(s) = 2s$ and $q(s) = 1+2s$, we can integrate it directly with help of the \textbf{exponential-integral} special function, defined for all $s>0$ as
        \begin{equation*}
            E_1(s) := \int_s^\infty \frac{e^{-t}}{t} \ \rmd t
        \end{equation*}
        and sometimes denoted by $-\mathrm{Ei}(-s)$ in the literature. We find
        \begin{equation*}
            \alpha(r) \leq \frac{\sfg W_4^2}{4} \int_r^\infty s^{-1}(1+2s)e^{-2s} \ \rmd s = \frac{\sfg W_4^2}{4} \big( e^{-2r} + E_1(2r) \big).
        \end{equation*}
        Then the well-known inequality $E_1(s) \leq \ln(1 + 1/s)e^{-s}$ (proof: $f(s) := E_1(s) - \ln(1+1/s)e^{-s}$ vanishes at $s=\infty$ and satisfies $f'(s) > 0$ for all $s>0$) implies
        \begin{equation*}
            \alpha(r) \leq \frac{\sfg W_4^2}{4} e^{-2r} \ln\left( e + \frac{e}{2r} \right).
        \end{equation*}
        In particular, using $e^\alpha-1 \leq \alpha e^\alpha$, this proves the first line in~\eqref{endnew_alpha} for the constants $C_1 = W_4^2/4$ and $\delta = \sfg W_4^2/4$. With the values $\sfg \leq \sfg_0 = 0.0002$ and $W_4 = 2.50$ corresponding to $\sfh = 0$, we get $\delta \leq 0.0004 \ll 1$. If instead $\sfh > 0$, we apply Lemma~\ref{lemma_int_quanti} to the integral in~\eqref{alpha_initial_kiwi} and immediately find
        \begin{equation*}
            \alpha(r) \leq \sfg W_4^2 K_{12} p(r)^{-2}q(r)^{2+2\sfh}e^{-2r}.
        \end{equation*}
        In particular, again using $e^\alpha-1 \leq \alpha e^\alpha$, this proves
        \begin{equation*}
            \av{ e^{\alpha(r)} - 1 } \leq \sfg W_4^2 K_{12} p(r)^{-2}q(r)^{2+2\sfh}e^{-2r} \exp\left( \sfg W_4^2 K_{12} (4\sfh)^{-2}(1+4\sfh)^{2+2\sfh} \right),
        \end{equation*}
        where the argument of $\exp$ is a direct upper bound for $e^{\alpha(0)}$, which bounds $e^{\alpha(r)}$ for all $r>0$ since $\alpha' < 0$, as is clear from the $\alpha$ equation in~\eqref{embltp}. With this, the second line in~\eqref{endnew_alpha} follows with $C_1 = W_4^2 K_{12}$ and $C_1' = \sfg_0 W_4^2K_{12} 4^{-2}(1+4\sfh)^{2+2\sfh}$. If $0<\sfh<1$, one has to further bound the $(1+4\sfh)^{2+2\sfh}$ term in $C_1'$ by $5^4$ before plugging in the known values of $W_4$, $K_{12}$ and $\sfg_0$ to obtain values for $C_1$ and $C_1'$ as in the theorem.

    \item \textbf{Proof of~\eqref{end_varphi2},~\eqref{end_const_potential} and~\eqref{end_varphi1}}. The desired inequalities read
    \begin{align}
        \label{endnew_varphi2} &\av{\varphi(r)- \rng{\varphi}(r) } \leq \sfg C_5 q(r)^{-\frac12+\sfh} e^{-r}, \\
        \label{endnew_const_potential} &\av{\varphi(0) - 2\rng{\mathcal{E}} } \leq \sfg C_6, \\
        \label{endnew_varphi1} &\av{\varphi(r)-\frac{1}{r}} \leq C_4 r^{-1}q(r)^{\frac12+\sfh}e^{-r},
    \end{align}
    where $\rng{\varphi}$ is as defined in~\eqref{phibolinha} (but now in the dimensionless convention) and the electric potential is given according to~\eqref{varphi_thm} by
    \begin{equation} \label{phi_sem_bolinha}
        \varphi(r) = \int_r^\infty \frac{1-\chi}{s^2} \ \rmd s.
    \end{equation}
    Note that $\varphi$ is strictly positive and decreasing in $r$, as a consequence of $\chi<1$ (see Lemma~\ref{lemma_sign_chi}), and thus $\varphi(0)>0$ --- but we note that the original, dimensionful $\varphi$ from the theorem could still be negative because the sign of the constant $Q$ is not specified. Also note that $\rng{\varphi}$ is the $\sfg=0$ counterpart of $\varphi$:
    \begin{equation} \label{phi_bolinha}
        \int_r^\infty \frac{1-\rng\chi}{s^2} \ \rmd s = \frac{1}{r} - \frac{(s+2\sfh)\rng\chi'(s)}{s^3}\bigg|_{s=r}^\infty = \left\{\begin{array}{ll}
            \dfrac{1}{r} - \dfrac{U_\sfh(4\sfh+2r)}{2\sfh \, U_\sfh(4\sfh)}\left(1 - \dfrac{2\sfh}{r}\right)e^{-r} &\quad\text{if } \sfh > 0, \\ \\
            \dfrac{1 - e^{-r}}{r} &\quad\text{if } \sfh = 0.
        \end{array}\right\} = \rng\varphi(r).
    \end{equation}
    In this derivation, we first integrated $1/s^2$ and $\rng{\chi}/s^2$ separately, also using the $\rng\chi$ equation (that is,~\eqref{chi_improved} with $\beta = r/(r+2\sfh)$ and $\sfg = 0$); then we used formulas~\eqref{chi_prime_sol_gn0} and~\eqref{chi_prime_zerog} for $\rng\chi'(r)$ to get the final result. Consequently, for all $r>0$, the inequality~\eqref{tomato6} for $\chi-\rng\chi$ gives
    \begin{equation} \label{varphi_comparison}
        \av{\varphi(r)-\rng\varphi(r)} \leq \int_r^\infty \frac{\av{\chi-\rng\chi}}{s^2} \ \rmd s \leq \sfg w_2 \int_r^\infty s^2 p(s)^{-2} q(s)^{-\frac12+\sfh} e^{-s} \ \rmd s \leq \sfg w_2 K_{11} q(r)^{-\frac12+\sfh} e^{-r} ,
    \end{equation}
    proving~\eqref{endnew_varphi2} for the constant $C_5 = w_2 K_{11}$. The exact value of $\rng\varphi(0)$ when $\sfh=0$ is $1 = 2\rng{\mathcal{E}}$, while for $\sfh>0$ it can be found by Taylor-expanding the $U_\sfh(4\sfh+2r)$ and $e^{-r}$ terms in~\eqref{phi_bolinha} to a sufficiently high order in $r$:
    \begin{align*}
        \rng\varphi(r) &= \frac{1}{r} - \frac{1}{2\sfh\, U_\sfh(4\sfh)} \bigg( U_\sfh(4\sfh) + 2U_\sfh'(4\sfh) r + O(r^2) \bigg) \left(1+\frac{2\sfh}{r}\right) \bigg(1 - r + O(r^2)\bigg) \\
        &= 1 - \frac{1}{2\sfh} - \frac{2U_\sfh'(4\sfh)}{U_\sfh(4\sfh)} + O(r) \\
        &= 2\rng{\mathcal{E}} + O(r) \quad\text{as } r\to 0,
    \end{align*}
    where the definition~\eqref{Ebolinha} of $\rng{\mathcal{E}}$ was used. Therefore, for any $\sfh\geq 0$,
    \begin{equation} \label{phi_bolinha_energy}
        \rng\varphi(0) = 2\rng{\mathcal{E}},
    \end{equation}
    and now plugging $r=0$ into~\eqref{varphi_comparison} proves~\eqref{endnew_const_potential} with the constant $C_6 = w_2 K_{11} (1+4\sfh)^{-\frac12+\sfh}$ (which is further bounded by $w_2 K_{11} 5^{\frac12}$ when $0<\sfh<1$). We can also use~\eqref{phi_sem_bolinha} and estimate~\eqref{tomato3} for $\chi$ to obtain a comparison between $\varphi$ and the Coulomb potential $1/r$ (note that this bound will not feature the constant $\sfg$, as it does not compare any function $f$ with its $\sfg=0$ counterpart $\rng{f}$):
    \begin{equation*}
        \av{ \varphi(r) - \frac{1}{r} } = \av{ \int_r^\infty \frac{1-\chi - 1}{s^2} \ \rmd s } \leq \int_r^\infty \frac{\av{\chi}}{s^2} \ \rmd s \leq W_3 \int_r^\infty s^{-2}q(s)^{\frac32+\sfh}e^{-s} \ \rmd s \leq W_3K_7 r^{-1}q(r)^{\frac12+\sfh}e^{-r}.
    \end{equation*}
    This proves~\eqref{endnew_varphi1} for the constant $C_4 = W_3K_7$.

    \item \textbf{Proof of~\eqref{end_const_mass}}. The desired inequality reads
    \begin{equation} \label{endnew_const_mass}
        \av{ M_{\mathrm{ADM}} - \big( M_{\mathrm{bare}} + \rng{\mathcal{E}} \big) } \leq \sfg C_8.
    \end{equation}
    Start with the $\chi$ equation as originally written in the E-M-BLTP system~\eqref{embltp}, but expressed in terms of $\nu$ instead of $\beta$:
    \begin{equation*}
        \frac{\rmd}{\rmd r}\left( \frac{(r-2\nu)\chi' e^{\alpha}}{r^3} \right) = \frac{\chi e^\alpha}{r^2}.
    \end{equation*}
    Apply the $r$-derivative on the left side using the product rule for the factors $(r-2\nu)\chi'/r^3$ and $e^\alpha$, then use the $\alpha$ equation to compute $(e^\alpha)'$ and cancel $e^\alpha$ from the whole equation to find
    \begin{equation*}
        \frac{\rmd}{\rmd r}\left( \frac{(r-2\nu)\chi'}{r^3} \right) - \sfg \frac{(r-2\nu)(\chi')^3}{r^6} = \frac{\chi}{r^2}.
    \end{equation*}
    Add and subtract $1/r^2$ to the right side, multiply the whole equation by $\chi-1$, and rearrange some terms:
    \begin{equation*}
        \frac{1-\chi^2}{r^2} + (\chi-1)\frac{\rmd}{\rmd r}\left( \frac{(r-2\nu)\chi'}{r^3} \right) = \frac{1-\chi}{r^2} + \sfg \frac{(r-2\nu)(\chi-1)(\chi')^3}{r^6}.
    \end{equation*}
    Now integrate from 0 to $\infty$, using integration by parts on the second term on the left side:
    \begin{multline} \label{int_parts_1}
        \frac{(r-2\nu)(\chi-1)\chi'}{r^3}\bigg|_{r=0}^\infty + \int_0^\infty \left(\frac{1-\chi^2}{r^2} - \frac{(r-2\nu)(\chi')^2}{r^3} \right) \ \rmd r \\
       = \int_0^\infty \frac{1-\chi}{r^2} \ \rmd r + \sfg \int_0^\infty \frac{(r-2\nu)(\chi-1)(\chi')^3}{r^6} \ \rmd r.
    \end{multline}
    Here the first term on the left vanishes at both endpoints as a consequence of the bounds~\eqref{anker1},~\eqref{tomato2}, and~\eqref{tomato4} for $r-2\nu$, $\chi-1$ and $\chi'$, which imply that $(r-2\nu)(\chi-1)\chi'/r^3$ behaves like $r^2p(r)^{-1}q(r)^{-3/2+\sfh}e^{-r}$. Moreover, the integrand in the second term is exactly $2\nu'(r)/\sfg = 2\mu'(r)$ according to the $\nu$ equation in~\eqref{system_nu_chi} and the definition~\eqref{definition_nu} of $\nu$ in terms of $\mu$, while the first integral on the right is exactly $\varphi(0)$ according to~\eqref{phi_sem_bolinha}. Seeing as how $\mu(0) = M_{\mathrm{bare}}$ and $\mu(\infty) = M_{\mathrm{ADM}}$, we have thus proved
    \begin{equation} \label{hoka_two_two}
        2M_{\mathrm{ADM}} - 2M_{\mathrm{bare}} = \varphi(0) + \sfg \int_0^\infty \frac{(r-2\nu)(\chi-1)(\chi')^3}{r^6} \ \rmd r.
    \end{equation}
    Using the bounds mentioned just above for $r-2\nu$, $\chi-1$ and $\chi'$, we estimate this integral as
    \begin{align*}
        \av{ \int_0^\infty \frac{(r-2\nu)(\chi-1)(\chi')^3}{r^6} \ \rmd r } &\leq \frac{3 W_2 W_4^3}{4} \int_0^\infty r^3p(r)^{-3} q(r)^{-\frac12+3\sfh}e^{-3r} \ \rmd r \\
        &\leq \frac{3 W_2 W_4^3}{32} \int_0^\infty q(r)^{-\frac12+3\sfh}e^{-3r} \ \rmd r.
    \end{align*}
    Meanwhile, the $\varphi(0)$ term is approximately given by $2\rng{\mathcal{E}}$ with error bound $\sfg C_6$ given in~\eqref{endnew_const_potential}. With this, estimate~\eqref{hoka_two_two} implies~\eqref{endnew_const_mass} with the constant
    \begin{equation*}
        C_8 = \frac{1}{2}\left( C_6 + \frac{3 W_2 W_4^3}{32} \int_0^\infty q(r)^{-\frac12+3\sfh}e^{-3r} \ \rmd r \right).
    \end{equation*}
    The numerical values~\eqref{numerical_C} of $C_8$ in the case $0\leq\sfh<1$ are derived directly from here by first bounding this integral with computer help: For $\sfh=0$, the integral can be directly calculated to be less than $0.280$; for $0<\sfh<1$, it can be estimated above by its value at $\sfh=1$, that comes out smaller than $26.2$.

    \item \textbf{Proof of~\eqref{end_const_energy}}. The desired inequality reads
    \begin{equation} \label{endnew_const_energy}
         \av{ \mathcal{E} - \rng{\mathcal{E}} } \leq \sfg C_7,
    \end{equation}
    Reasoning similarly to how we did in the previous item, but this time not canceling any $e^\alpha$ terms, we prove an equation similar to~\eqref{int_parts_1}:
    \begin{equation*}
        \frac{e^\alpha(r-2\nu)(\chi-1)\chi'}{r^3}\bigg|_{r=0}^\infty + \int_0^\infty e^\alpha\left(\frac{1-\chi^2}{r^2} + \frac{(r-2\nu)(\chi')^2}{r^3} \right) \ \rmd r = \int_0^\infty \frac{e^\alpha(1-\chi)}{r^2} \ \rmd r.
    \end{equation*}
    Here one checks again that the first term vanishes at both endpoints, this time also making use of the known behavior of $e^\alpha$ --- in particular note that the already proved rate of divergence of $e^\alpha$ around $r=0$ when $\sfh = 0$, given in~\eqref{endnew_alpha}, is not fast enough to become a problem, since $\delta < 1$. Moreover, according to definition~\eqref{energy_finite_thm} of the field energy, we recognize $2\mathcal{E}$ in the second term. This means that we have proved a simpler formula for it:
    \begin{equation*}
        \mathcal{E} = \frac{1}{2}\int_0^\infty \frac{e^\alpha(1-\chi)}{r^2} \ \rmd r.
    \end{equation*}
    We also remark that~\eqref{phi_bolinha} and~\eqref{phi_bolinha_energy} together imply that the constant $\rng{\mathcal{E}}$ as defined in the theorem in~\eqref{Ebolinha} is the value of this expression in the $\sfg = 0$ case, that is, with $\alpha$ and $\chi$ replaced by $\rng{\alpha} = 0$ and $\rng\chi$. Indeed:
    \begin{equation}
        \frac{1}{2}\int_0^\infty \frac{1-\rng\chi}{r^2} \ \rmd r = \frac{\rng{\varphi}(0)}{2} = \frac{1}{2}\left(1 - \frac{1}{2\sfh} - \frac{2U_\sfh'(4\sfh)}{U_\sfh(4\sfh)}\right) = \rng{\mathcal{E}}.
    \end{equation}
    Consequently, we can compare $\mathcal{E}$ and $\rng{\mathcal{E}}$ as follows:
    \begin{equation} \label{UV162_umbrella}
        \av{ \mathcal{E} - \rng{\mathcal{E}} } = \frac{1}{2} \av{ \int_0^\infty \left( \frac{e^\alpha(1-\chi)}{r^2} - \frac{1-\rng\chi}{r^2} \right) \ \rmd r } \leq \frac{1}{2} \int_0^\infty \frac{(e^\alpha-1)(1-\chi)}{r^2} \ \rmd r + \frac{1}{2} \int_0^\infty \frac{\av{\rng\chi - \chi}}{r^2} \ \rmd r.
    \end{equation}
    Inequality~\eqref{varphi_comparison} with $r=0$ provides a bound for the second integral on the right:
    \begin{equation} \label{sagrotan}
        \frac{1}{2} \int_0^\infty \frac{\av{\rng\chi - \chi}}{r^2} \ \rmd r \leq \frac{\sfg w_2K_{11}}{2}(1+4\sfh)^{-\frac12+\sfh} = \frac{\sfg C_6}{2}.
    \end{equation}
    As for the first integral, it is necessary to consider the cases $\sfh = 0$ and $\sfh > 0$ separately. If $\sfh = 0$, estimates~\eqref{endnew_alpha} and~\eqref{tomato2} for $e^\alpha-1$ and $1-\chi$ yield
    \begin{align*}
        \frac{1}{2} \int_0^\infty \frac{(e^\alpha-1)(1-\chi)}{r^2} \ \rmd r &\leq \frac{1}{2}\int_0^\infty r^{-2} \sfg C_1 \ln\left(e+\frac{e}{2r}\right)\left(e+\frac{e}{2r}\right)^\delta e^{-2r} W_2 r^3p(r)^{-1}q(r)^{-2} \ \rmd r \\
        &= \frac{\sfg W_2 C_1}{4} \left(\frac{e}{2}\right)^\delta \int_0^\infty r^{-\delta} (1+2r)^{-2+\delta} \ln\left(e+\frac{e}{2r}\right) e^{-2r} \ \rmd r.
    \end{align*}
    Using $\delta < 0.0004$ (note that $r^{-\delta} (1+2r)^{-2+\delta}$ is increasing in $\delta$), a computer will estimate this integral as less than $0.553$. Using also the values of $W_2$, $C_1$ and $C_6$ corresponding to $\sfh = 0$, estimate~\eqref{UV162_umbrella} gives the value of $C_7$ as claimed in~\eqref{numerical_C} for this case. If $\sfh > 0$, estimates~\eqref{endnew_alpha} and~\eqref{tomato2} for $e^\alpha-1$ and $1-\chi$ yield
    \begin{align*}
        \frac{1}{2} \int_0^\infty \frac{(e^\alpha-1)(1-\chi)}{r^2} \ \rmd r &\leq \frac{1}{2}\int_0^\infty r^{-2} \sfg C_1 \exp(C_1'/\sfh^2) p(r)^{-2}q(r)^{2+2\sfh}e^{-2r} W_2 r^3p(r)^{-1}q(r)^{-2} \ \rmd r \\
        &= \frac{\sfg W_2 C_1 \exp(C_1'/\sfh^2)}{2} \int_0^\infty rp(r)^{-3}q(r)^{2n}e^{-2r} \ \rmd r \\
        &\leq \frac{\sfg W_2C_1 \exp(C_1'/\sfh^2)}{4} \int_0^\infty p(r)^{-2}q(r)^{2n}e^{-2r} \ \rmd r.
    \end{align*}
    Together with~\eqref{UV162_umbrella}, this provides the general value
    \begin{equation*}
        C_7 = \frac{C_6}{2} + \frac{W_2C_1 \exp(C_1'/\sfh^2)}{4} \int_0^\infty p(r)^{-2}q(r)^{2n}e^{-2r} \ \rmd r, \quad \sfh > 0.
    \end{equation*}
    For small $\sfh$, however, this is far from being useful, since the $\exp(C_1'/\sfh^2)$ term becomes large. We present now a different estimate, valid for $0<\sfh<1$, that yields a uniformly bounded value of $C_7$. The idea is to use a different estimate for $e^\alpha-1$, one which lets go of the finiteness at small $r$ present in~\eqref{endnew_alpha} in exchange for a bounded-in-$\sfh$ constant. We go back to~\eqref{alpha_initial_kiwi} and this time use $sp(s)^{-2} \leq s^{-1}/4$ and $q(s)^{1+2\sfh} \leq (5+2s)^3$ to find
    \begin{multline*}
        \alpha(r) \leq \frac{\sfg W_4^2}{4} \int_r^\infty s^{-1}(5+2s)^3 e^{-2s} \ \rmd s = \frac{\sfg W_4^2}{4} \bigg( 2(2r^2+17r+46)e^{-2r} + 125 E_1(2r) \bigg) \\
        \leq \frac{\sfg W_4^2}{4} \bigg( 2(2r^2+17r+46) + 125 \ln\left( 1 + \frac{1}{2r} \right) \bigg)e^{-2r}.
    \end{multline*}
    In particular, using the values of $\sfg < \sfg_0$ and $W_4$ for the $0<\sfh<1$ case,
    \begin{multline*}
        e^{\alpha(r)} \leq \Bigg( \exp\Big( 2(2r^2+17r+46)e^{-2r} + 125e^{-2r}\ln\big(1+\frac{1}{2r}\big) \Big) \Bigg)^{\sfg W_4^2/4} \leq \bigg( e^{92}\big( 1 + \frac{1}{2r} \big)^{125} \bigg)^{\sfg W_4^2/4} \\
        < 1.01 \left(1 + \frac{1}{2r}\right)^{5.63\cdot 10^{-14}},
    \end{multline*}
    so that the inequality $e^\alpha-1 \leq \alpha e^\alpha$ now yields
    \begin{equation*}
        e^{\alpha(r)}-1 \leq 38.9 \sfg \left( 1 + \frac{1}{2r} \right)^{\delta} \bigg( 2(2r^2+17r+46) + 125 \ln\left( 1 + \frac{1}{2r} \right) \bigg) e^{-2r}
    \end{equation*}
    for the small exponent $\delta = 5.63\cdot 10^{-14}$. We plug this into the first integral on the right side of~\eqref{UV162_umbrella} and again use~\eqref{sagrotan} on the second integral:
    \begin{align*}
        \av{\mathcal{E} - \rng{\mathcal{E}}} &\leq \frac{1}{2} \int_0^\infty \frac{(e^\alpha-1)(1-\chi)}{r^2} \ \rmd r + \frac{1}{2} \int_0^\infty \frac{\av{\rng\chi - \chi}}{r^2} \ \rmd r \\
        &\leq \frac{38.9\sfg W_2}{2}\int_0^\infty r^{-2} \left( 1 + \frac{1}{2r} \right)^{\delta} \bigg( 2(2r^2+17r+46) + 125 \ln\left( 1 + \frac{1}{2r} \right) \bigg) e^{-2r} r^3 p^{-1}q^{-2} \ \rmd r + \frac{\sfg C_6}{2} \\
        &\leq \frac{38.9\sfg W_2}{4}\int_0^\infty \left( 1 + \frac{1}{2r} \right)^{\delta} \bigg( 2(2r^2+17r+46) + 125 \ln\left( 1 + \frac{1}{2r} \right) \bigg) e^{-2r} \frac{1}{(1+2r)^2} \ \rmd r + \frac{\sfg C_6}{2},
    \end{align*}
    where $rp^{-1}q^{-2}$ was ultimately bounded by $(1/2)(1+2r)^{-2}$ to get rid of all dependency on $\sfh$. This proves~\eqref{endnew_const_energy} in the $0<\sfh<1$ case for the constant
    \begin{equation*}
        C_7 = \frac{38.9 W_2}{4}\int_0^\infty \left( 1 + \frac{1}{2r} \right)^{5.63\cdot 10^{-14}} \bigg( 2(2r^2+17r+46) + 125 \ln\left( 1 + \frac{1}{2r} \right) \bigg) \frac{e^{-2r}}{(1+2r)^2} \ \rmd r + \frac{C_6}{2} < 2.91\cdot 10^{16}.
    \end{equation*}

    \item \textbf{Proof of~\eqref{end_beta2} and~\eqref{end_beta1}}. The desired inequalities read
    \begin{align}
        \label{endnew_beta2} &\av{ \beta(r) - \left( 1 + \frac{2\sfh}{r} \right)^{-1} } \leq \sfg C_3 r^3p(r)^{-3}q(r)^{-1}, \\
        \label{endnew_beta1} &\av{ \beta(r) - \beta^{\mathrm{(RWN)}}_{M_{\mathrm{ADM}}}(r) } \leq \sfg C_2 rp(r)^{-1} q(r)^{2+2\sfh} e^{-2r} \beta^{\mathrm{(RWN)}}_{M_{\mathrm{ADM}}}(r).
    \end{align}
    Comparing $\beta = r/(r-2\nu) = r/(r+2\sfh-2\sfg\xi)$ to its counterpart in the $\sfg = 0$ case, which is $r/(r-2\rng\nu) = r/(r+2\sfh)$, and using the estimate~\eqref{tomato5} for $\nu-\rng\nu = \sfg\xi$ and the last inequality in~\eqref{anker1}, namely $1/(r+2\sfh-2\sfg\xi) \leq 4p(r)^{-1}$, we immediately derive
    \begin{equation}
        \av{ \beta - \frac{r}{r+2\sfh} } = \frac{2\sfg r\av{\xi}}{(r+2\sfh)(r+2\sfh-2\sfg\xi)} \leq 16\sfg w_1 r^3p(r)^{-3}q(r)^{-1},
    \end{equation}
    proving~\eqref{endnew_beta2} for the constant $C_3 = 16w_1$. In turn, comparing $\beta$ to the RWN coefficient
    \[ \beta^{\text{(RWN)}}_{M} = \left( 1 - \frac{2\sfg M}{r} + \frac{\sfg}{r^2} \right)^{-1} = \frac{r^2}{r^2 - 2\nu(\infty)r + \sfg} \]
    defined in~\eqref{RWN}, with the finite parameter $M = M_{\mathrm{ADM}} = \nu(\infty)/\sfg$, we find
    \begin{equation*}
        \beta - \beta^{\mathrm{(RWN)}}_M = \frac{r}{r-2\nu} - \frac{r^2}{r^2 - 2\nu(\infty)r + \sfg} = \frac{2 \beta^{\mathrm{(RWN)}}_M}{r+2\sfh - 2\sfg\xi} \left( \nu - \nu(\infty) + \frac{\sfg}{2r} \right).
    \end{equation*}
    Now this gives us the estimate
    \begin{equation} \label{beta_comp_temp}
        \av{\beta - \beta^{\mathrm{(RWN)}}_M} \leq 8 \beta^{\mathrm{(RWN)}}_M p(r)^{-1} \av{ \nu - \nu(\infty) + \frac{\sfg}{2r} }.
    \end{equation}
    To bound the term in absolute-value signs on the right, we go back to the $\nu$ equation in~\eqref{system_nu_chi} and integrate it from $\infty$ to find
    \begin{equation*}
        \nu(r) = \nu(\infty) - \frac{\sfg}{2}\int_r^\infty \left( \frac{1-\chi^2}{s^2} - \frac{(s-2\nu)(\chi')^2}{s^3} \right) \ \rmd s = \nu(\infty) - \frac{\sfg}{2r} + \frac{\sfg}{2}\int_r^\infty \left( \frac{\chi^2}{s^2} + \frac{(s-2\nu)(\chi')^2}{s^3} \right) \ \rmd s.
    \end{equation*}
    In particular
    \begin{align*}
        \av{ \nu(r) - \nu(\infty) + \frac{\sfg}{2r} } &\leq \frac{\sfg}{2} \int_r^\infty \left( \frac{\chi^2}{s^2} + \frac{(s-2\nu)(\chi')^2}{s^3} \right) \ \rmd s \\
        &\leq \frac{\sfg}{2} \int_r^\infty \left( W_3^2 s^{-2} q(s)^{3+2\sfh} e^{-2s} + \frac{3W_4^2}{4} sp(s)^{-1}q(s)^{1+2\sfh} e^{-2s} \right) \ \rmd r \\
        &\leq \frac{\sfg}{2} \int_r^\infty s^{-2} q(s)^{3+2\sfh} e^{-2s}\left( W_3^2 + \frac{3W_4^2}{4} s^3p(s)^{-1}q(s)^{-2} \right) \ \rmd r \\
        &\leq \frac{\sfg}{2} \left( W_3^2 + \frac{3W_4^2}{32} \right) \int_r^\infty s^{-2} q(s)^{3+2\sfh} e^{-2s} \ \rmd r \\
        &\leq \frac{\sfg}{2} \left( W_3^2 + \frac{3W_4^2}{32} \right) K_8 r^{-1} q(r)^{2+2\sfh} e^{-2r}.
    \end{align*}
    Applying this to the above estimate~\eqref{beta_comp_temp}, we have proved~\eqref{endnew_beta1} for the constant $C_2 = 4 K_8 ( W_3^2 + 3W_4^2/32 )$.

\end{itemize}

\section*{Final remarks}

In this paper, we initiated a potential study of the equations of motion for general-relativistic point charges under the classical, generalized theory of electromagnetism of BLTP, by rigorously proving that the single-particle, static case admits spacetime solutions with a finite electric-field energy and finite values for the electric potential and bare mass at the charge's location. In this sense, it can interpreted as an analogue of the Reissner-Weyl-Nordström solutions, but without the latter's pathological singularities. Our results provide a rigorous framework for embedding static point charges governed by BLTP electrodynamics into general relativity, setting the stage for future efforts to treat multi-particle systems and dynamic evolutions within this setting.

The author would like to thank his PhD advisors Michael Kiessling and Shadi Tahvildar-Zadeh for all the valuable discussions related to this project throughout his course of studies at Rutgers University.

\renewcommand{\theequation}{A.\arabic{equation}}
\renewcommand{\subsection}[1]{A.\arabic{subsection}. \textbf{#1.}}
\setcounter{equation}{0}
\section*{Appendix}

\setcounter{subsection}{1}

\small

\subsection{The electric potential in GR} \label{subsec_app_potential}

Here we provide justification for equation~\eqref{F_GR} relating the general-relativistic electric potential $\varphi$ and the Faraday tensor $F$ in the form $F_{tr} = e^\alpha\varphi'(r)$, where $e^\alpha = \sqrt{-g_{tt} g_{rr}}$. The idea is to invoke the Equivalence Principle to define the quantity $\varphi'(r)$ --- we remark that it is the derivative $\varphi'$ instead of the electric potential $\varphi$ itself that is accessible by an observer measuring the EM fields.

In SR, using spherical coordinates $(c\overline{t},\overline{r},\overline{\theta},\overline{\phi})$ constructed from an inertial system $(c\overline{t},\overline{x}^1,\overline{x}^2,\overline{x}^3)$, the tensor $F = \rmd A$~\eqref{FdA} for $A = -\varphi(\overline{r}) \ \rmd (c\overline{t})$ \eqref{A_SR} is expressed as
\begin{equation}
    F = -\varphi'(\overline{r}) \ \rmd \overline{r}\wedge\rmd(c\overline{t}) = \varphi'(\overline{r}) \ \rmd(c\overline{t})\wedge\rmd \overline{r}.
\end{equation}
This implies that $\varphi'(\overline{r}) = F_{tr}$. Now consider our general-relativistic spacetime $\mathcal{M}$ as in subsection~\ref{subsec_GR_pointcharge} with its coordinate system $(ct,r,\theta,\varphi)$. Loosely speaking, the Equivalence Principle states that the local laws of physics in GR around a given spacetime event reduce to those of SR if one uses \textbf{normal coordinates} at this event, which by definition means a coordinate system under which the metric and its first-order derivatives coincide with those of the Minkowski metric at this event. Thus, given an event $x_0\in\mathcal{M}$ with coordinates $(ct_0,r_0,\theta_0,\phi_0)$, we shall assume that the quantity $(\rmd \varphi/\rmd r)(r_0)$ that we wish to relate to the components of $F$ is given by $F_{\overline{t}\overline{r}}$, where $(c\overline{t},\overline{r},\overline{\theta},\overline{\phi})$ is a system of normal coordinates around $x_0$ for which the metric assumes the spherical form of the Minkowski metric. The tensor-transformation law for changing the components of $F$ from our coordinates $(x^\mu) = (ct,r,\theta,\phi)$ into the coordinates $(c\overline{t},\overline{r},\overline{\theta},\overline{\phi})$ gives a workable expression for this quantity:
\begin{equation} \label{rmk_phi0}
    F_{\overline{t}\overline{r}} = \left( \frac{\partial x^\zeta}{\partial\overline{t}} \right)\left(\frac{\partial x^\eta}{\partial\overline{r}} \right)F_{\zeta\eta} = \left( \frac{\partial t}{\partial\overline{t}} \right)\left(\frac{\partial r}{\partial\overline{r}} \right)F_{tr} + \left( \frac{\partial r}{\partial\overline{t}} \right)\left(\frac{\partial t}{\partial\overline{r}} \right)F_{rt} = \left( \frac{\partial t}{\partial\overline{t}} \frac{\partial r}{\partial\overline{r}} - \frac{\partial r}{\partial\overline{t}} \frac{\partial t}{\partial\overline{r}} \right) F_{tr}.
\end{equation}
(Let us assume $c=1$ here for brevity). Now suppose, if only for simplicity, that our coordinate change acts as $(t,r)\longmapsto(\overline{t},\overline{r})$, leaving the spherical angles $(\theta,\phi)$ untouched. Thus only the radial terms of the metric~\eqref{metric} transform into their corresponding terms of the spherical Minkowski metric:
\begin{equation} \label{rmk_phi}
    g_{tt}(r) \ \rmd t\otimes\rmd t + g_{rr}(r) \ \rmd r\otimes\rmd r = -\rmd \overline{t}\otimes\rmd t + \rmd r\otimes\rmd \overline{r}.
\end{equation}
This equality and all the upcoming ones are meant only to first order in $(\overline{t}-t_0,\overline{r}-r_0)$, which poses no problem for us because we wish to evaluate~\eqref{rmk_phi0} only at $(ct_0,r_0,\theta_0,\phi_0)$ anyway. Using
\begin{equation}
    \rmd t = \frac{\partial t}{\partial \overline{t}}\rmd \overline{t} + \frac{\partial t}{\partial\overline{r}}\rmd\overline{r}, \quad \rmd r = \frac{\partial r}{\partial \overline{t}}\rmd \overline{t} + \frac{\partial r}{\partial\overline{r}}\rmd\overline{r},
\end{equation}
the left side of~\eqref{rmk_phi} evaluates to
\begin{multline*}
    \left( g_{tt} \left(\frac{\partial t}{\partial\overline{t}}\right)^2 + g_{rr} \left(\frac{\partial r}{\partial\overline{t}}\right)^2 \right) \rmd \overline{t}\otimes\rmd \overline{t} + \left( g_{tt} \left(\frac{\partial t}{\partial\overline{r}}\right)^2 + g_{rr} \left(\frac{\partial r}{\partial\overline{r}}\right)^2 \right) \rmd\overline{r}\otimes\rmd\overline{r} \\
    + \left( g_{tt} \frac{\partial t}{\partial\overline{t}}\frac{\partial t}{\partial\overline{r}} + g_{rr} \frac{\partial r}{\partial\overline{t}}\frac{\partial r}{\partial\overline{r}} \right) \big( \rmd \overline{t}\otimes\rmd\overline{r} + \rmd\overline{r}\otimes\rmd \overline{t} \big).
\end{multline*}
Therefore~\eqref{rmk_phi} yields the three equations $-x^2 + y^2 = -1$, $-z^2+w^2 = 1$ and $-xz+yw = 0$ for the auxiliary quantities
\begin{equation*}
    x = \sqrt{\av{g_{tt}(r_0)}} \frac{\partial t}{\partial\overline{t}}\bigg|_{(t_0,r_0)}, \ y = \sqrt{g_{rr}(r_0)} \frac{\partial r}{\partial\overline{t}}\bigg|_{(t_0,r_0)}, \ z = \sqrt{\av{g_{tt}(r_0)}} \frac{\partial t}{\partial\overline{r}}\bigg|_{(t_0,r_0)}, \ w = \sqrt{g_{rr}(r_0)} \frac{\partial r}{\partial\overline{r}}\bigg|_{(t_0,r_0)}.
\end{equation*}
In particular note that
\begin{equation*}
    (xw-yz)^2 = (xw-yz)^2 - (-xz+yw)^2 = (x^2-y^2)(-z^2+w^2) = (1)(1) = 1 \quad \Longrightarrow \quad xw-yz = 1,
\end{equation*}
where, when taking the square root in the end, we also assumed that the transformation $(t,r)\longmapsto(\overline{t},\overline{r})$ must preserve orientation. Thus the quantity~\eqref{rmk_phi0} that expresses our desired $\varphi'(r)$ turns into
\begin{equation} \label{rmk_phi_end}
    F_{\overline{t}\overline{r}} = \left( \frac{\partial t}{\partial\overline{t}} \frac{\partial r}{\partial\overline{r}} - \frac{\partial r}{\partial\overline{t}} \frac{\partial t}{\partial\overline{r}} \right) F_{tr} = \left( \av{g_{tt}(r_0)}^{-1/2} x g_{rr}(r_0)^{-1/2} w - g_{rr}(r_0)^{-1/2} y \av{g_{tt}(r_0)}^{-1/2} z \right) F_{tr} = \av{g_{tt}(r_0) g_{rr}(r_0)}^{-1/2} F_{tr}.
\end{equation}
Using the form~\eqref{metric} of our metric, the final result is precisely $e^{-\alpha}F_{tr}$. This motivates our definition $\varphi' = e^{-\alpha}F_{tr}$ in the GR setting.

\setcounter{subsection}{2}
\subsection{The E-M-BLTP system for a static particle} \label{subsec_app_system}

Here we derive the E-M-BLTP system~\eqref{embltp} as stated in Proposition~\eqref{prop_main_system}.

Our metric~\eqref{metric} has the general form of a static, spherically symmetric metric
\begin{equation} \label{metric_philambda}
    g = -e^{2\Phi(r)} \ \rmd (ct)\otimes\rmd (ct) + e^{2\Lambda(r)} \ \rmd r\otimes\rmd r + r^2 \ \rmd\theta\otimes\rmd\theta + r^2\sin^2\theta \ \rmd\phi\otimes\rmd\phi
\end{equation}
for $e^{2\Phi} = \beta^{-1}e^{2\alpha}$ and $e^{2\Lambda} = \beta$. Thus the non-vanishing metric and inverse-metric components are
\begin{equation} \label{app_g_coeffs}
    \begin{array}{lclclcl}
        g_{tt} = -\beta^{-1}e^{2\alpha}, &\quad& g_{rr} = \beta, &\quad& g_{\theta\theta} = r^2, &\quad& g_{\phi\phi} = r^2\sin^2\theta, \\
        g^{tt} = -\beta e^{-2\alpha}, &\quad& g^{rr} = \beta^{-1}, &\quad& g^{\theta\theta} = r^{-2}, &\quad& g^{\phi\phi} = r^{-2}(\sin\theta)^{-2}, \\
    \end{array}
\end{equation}
and in particular
\begin{equation} \label{app_sqrt_g}
    \sqrt{\av{g}} = r^2e^\alpha\sin\theta.
\end{equation}
The Einstein tensor $(G_{\mu\nu})$ corresponding to~\eqref{metric_philambda} turns out (see~\cite{schutz}, p.~260) to be diagonal and given by
\begin{align}
    G_{tt} &= \frac{e^{2\Phi}}{r^2}\left(r(1-e^{-2\Lambda})\right)', \label{Gtt_philambda} \\
    G_{rr} &= -\frac{e^{2\Lambda}}{r^2}(1-e^{-2\Lambda}) + \frac{2\Phi'}{r}, \label{Grr_philambda} \\
    G_{\theta\theta} &= r^2e^{-2\Lambda}\left( \Phi'' + (\Phi')^2 - \Phi'\Lambda' + \frac{\Phi'-\Lambda'}{r} \right), \nonumber \\
    G_{\phi\phi} &= G_{\theta\theta}\sin^2\theta. \nonumber
\end{align}
The first two of these components are the ones we will need later. Written in terms of $\alpha = \Lambda+\Phi$ and $\beta = e^{2\Lambda}$, they become
\begin{equation} \label{Gmunu_initially}
    G_{tt} = \frac{e^{2\alpha}}{r\beta^2}\left( \frac{\beta'}{\beta} + \frac{\beta}{r}-\frac{1}{r} \right), \quad G_{rr} = \frac{1}{r}\left( -\frac{\beta'}{\beta} - \frac{\beta}{r} + \frac{1}{r} + 2\alpha' \right).
\end{equation}
The 2-form $F$ is given in terms of the electric-deviation scalar $\chi$ as in~\eqref{F_GR2}. The equation $\rmd F = 0$ comes for free using this expression. Two of the other equations defining our spacetime are the vacuum law~\eqref{vacuum_cov} and Maxwell's equation~\eqref{maxwell_cov} for $M$, repeated here:
\begin{equation} \label{appM}
    M = \Star(F + \varkappa^{-2}\rmd\delta F), \quad \rmd M = 0.
\end{equation}
Note that the $M$ equation above is homogeneous now because the 4-current $J$ is identically null on $\mathcal{M}$ --- the point charge's worldline is not considered part of the manifold. The remaining equations are the EFE (Einstein Field Equations)
\begin{equation} \label{appEFE}
    G_{\mu\nu} := R_{\mu\nu} - \frac{1}{2}Rg_{\mu\nu} = \frac{8\pi G}{c^4} T_{\mu\nu}, \quad \mu,\nu = t,r,\theta,\phi,
\end{equation}
where the Hilbert stress-energy tensor $(T_{\mu\nu})$ is defined in~\eqref{Tmunu_GR}. For future use, we compute the Hodge duals of $\rmd(ct)$ and $\rmd(ct)\wedge\rmd r$. According to definition~\eqref{definition_hodge},
\begin{align*}
    &\rmd(ct)\wedge\Star\rmd(ct) = \big\langle \rmd(ct),\rmd(ct)\big\rangle_g \ \sqrt{\av{g}} \ \rmd(ct)\wedge\rmd r\wedge\rmd\theta\wedge\rmd\phi, \\
    &\big(\rmd(ct)\wedge\rmd r\big)\wedge\Star\big(\rmd(ct)\wedge\rmd r\big) = \big\langle \rmd(ct)\wedge\rmd r, \rmd(ct)\wedge\rmd r \big\rangle_g \ \sqrt{\av{g}} \ \rmd(ct)\wedge\rmd r\wedge\rmd\theta\wedge\rmd\phi,
\end{align*}
while definition~\eqref{blades_inner_product} gives
\begin{equation*}
    \big\langle \rmd(ct),\rmd(ct)\big\rangle_g = g^{tt}, \quad \big\langle \rmd(ct)\wedge\rmd r, \rmd(ct)\wedge\rmd r \big\rangle_g = \det\left( \begin{array}{rr} g^{tt} & g^{tr} \\ g^{rt} & g^{rr} \end{array} \right) = g^{tt}g^{rr}.
\end{equation*}
Then, using~\eqref{app_g_coeffs} and ~\eqref{app_sqrt_g}, we quickly deduce
\begin{align}
    \Star\rmd(ct) &= -r^2\beta e^{-\alpha}\sin\theta \ \rmd r\wedge\rmd\theta\wedge\rmd\phi, \label{hodge_dt} \\
    \Star\big(\rmd(ct)\wedge\rmd r\big) &= -r^2 e^{-\alpha}\sin\theta \ \rmd\theta\wedge\rmd\phi. \label{hodge_dtdr} 
\end{align}
As is well known, the EFE do not form a completely independent system of equations. We leave to the end of this section a quick check that the complete information for our spacetime is given by the Maxwell equations~\eqref{appM} together with the $(t,t)$ and $(r,r)$ equations of the EFE~\eqref{appEFE}. Now we proceed to write out these equations in the $(t,r,\theta,\phi)$ coordinates. For this we need to note down the components of $F$, $\delta F$ and $\rmd\delta F$. The only nonzero components of $F$ are
\begin{equation*}
    F_{tr} = Qr^{-2}e^\alpha(\chi-1), \quad F_{rt} = -Qr^{-2}e^\alpha(\chi-1).
\end{equation*}
Since the metric is diagonal, raising an index $\zeta$ of a tensor component amounts to multiplying it by $g^{\zeta\zeta}$ (no sum in $\zeta$ is implied): $\tensor{F}{^\zeta_\eta} = \sum_\mu g^{\zeta\mu}F_{\mu\eta} = g^{\zeta\zeta} F_{\zeta\eta} + \sum_{\mu\neq\zeta} g^{\zeta\mu}F_{\mu\eta} = g^{\zeta\zeta} F_{\zeta\eta}$. Therefore, the only nonzero components of $(\tensor{F}{^\zeta_\eta})$ are
\begin{equation*}
    \tensor{F}{^t_r} = -Qr^{-2}\beta e^{-\alpha}(\chi-1), \quad \tensor{F}{^r_t} = -Qr^{-2}\beta^{-1} e^{\alpha}(\chi-1).
\end{equation*}
Raising the second index, we also find the only nonzero components of $(F^{\zeta\eta})$:
\begin{equation*}
    F^{tr} = -Qr^{-2}e^{-\alpha}(\chi-1), \quad F^{rt} = Qr^{-2}e^{-\alpha}(\chi-1).
\end{equation*}
Next we compute
\begin{equation*}
    (\delta F)^\mu = -\frac{1}{\sqrt{\av{g}}}\partial_\lambda\big( \sqrt{\av{g}} F^{\lambda\mu} \big) = -r^{-2}(\sin\theta)^{-1} e^{-\alpha} \partial_\lambda\big( r^2e^\alpha\sin\theta F^{\lambda\mu} \big), \quad \mu=t,r,\theta,\phi.
\end{equation*}
The expression being differentiated here only depends on the variables $r$ and $\theta$, but for $\lambda = \theta$ the term $F^{\lambda\mu}$ disappears. Hence only $\lambda=r$ must be considered: $(\delta F)^\mu = -r^{-2}e^{-\alpha} \big( r^2 e^\alpha F^{r\mu} \big)'$, which shows that $(\delta F)^\mu$ is only nonzero when $\mu = t$, in which case it is given by
\begin{equation*}
    (\delta F)^t = -r^{-2}e^{-\alpha} \big( r^2 e^\alpha F^{rt} \big)' = -Qr^{-2}e^{-\alpha} \chi'.
\end{equation*}
Lowering its index, we find the 1-form $\delta F = (\delta F)_t \ \rmd (ct)$ for
\begin{equation*}
    (\delta F)_t = Qr^{-2}\beta^{-1} e^\alpha \chi'.
\end{equation*}
We will also need $\rmd\delta F$, which is given by $\big( (\delta F)_t \big)' \ \rmd r\wedge\rmd (ct) = (\rmd\delta F)_{tr} \ \rmd (ct)\wedge\rmd r$ for
\begin{equation*}
    (\rmd\delta F)_{tr} = -\big( (\delta F)_t \big)' = -Q\big( r^{-2}\beta^{-1}e^\alpha\chi' \big)'.
\end{equation*}
We can finally write out $M$ from~\eqref{appM}:
\begin{align}
    M &= \Star(F + \varkappa^{-2} \rmd\delta F) \nonumber \\
    &= \Star \big( Qr^{-2}e^\alpha(\chi-1) - Q\varkappa^{-2}(r^{-2}\beta^{-1}e^\alpha\chi')' \big) \ \rmd(ct)\wedge\rmd r \nonumber \\
    &= Q\big( r^{-2}e^\alpha(\chi-1) - Q\varkappa^{-2}(r^{-2}\beta^{-1}e^\alpha\chi')' \big) \big( -r^2 e^{-\alpha} \sin\theta \ \rmd\theta\wedge\rmd\phi \big) \nonumber \\
    &= Q\sin\theta\bigg( -(\chi-1) + \varkappa^{-2} r^2 e^{-\alpha} (r^{-2}\beta^{-1}e^\alpha\chi')' \bigg) \ \rmd\theta\wedge\rmd\phi, \label{app_M_coords}
\end{align}
where we used equation~\eqref{hodge_dtdr} for $\Star\big( 
\rmd(ct)\wedge\rmd r \big)$. Therefore, in the equation
\begin{equation*}
    0 = \rmd M = Q \sin\theta \frac{\rmd}{\rmd r} \bigg( -(\chi-1) + \varkappa^{-2} r^2 e^{-\alpha} (r^{-2}\beta^{-1}e^\alpha\chi')' \bigg) \ \rmd r\wedge\rmd\theta\wedge\rmd\phi,
\end{equation*}
the term in parentheses must be equal to some constant $q$, giving us $M = Qq\sin\theta \ \rmd\theta\wedge\rmd\phi$ and
\begin{equation} \label{app_chi_eq}
    \frac{(\chi-1+q)e^{\alpha}}{r^2} = \varkappa^{-2}\frac{\rmd}{\rmd r}\left( \frac{\chi' e^\alpha}{r^2\beta} \right).
\end{equation}
Finally, it turns out that $q=1$ must be chosen in order for the charge at the singularity to be $Q$. This choice proves the $\chi$ equation as claimed in~\eqref{embltp}. To see why $q=1$, we simply need to figure out how $Q$ is related to $M$:

\begin{proposition} \label{prop_charge}
    The charge at the singularity is given by
    \begin{equation*}
        Q = \frac{1}{4\pi} \int_{\partial B_R(0)} M,
    \end{equation*}
    where $R>0$ is arbitrary and $\partial B_R(0)$ denotes the boundary of the 3-dimensional sphere
    \begin{equation*}
        B_R(0) := \big\{ct_0\big\} \times \big\{ (r,\theta,\phi): \, r\in [0,R), \ \theta\in [0,\pi], \ \phi\in [0,2\pi) \big\}
    \end{equation*}
    for an arbitrary instant $t_0\in\bbR$. This implies that the value of $q$ in~\eqref{app_chi_eq} is $q=1$.
\end{proposition}

\textbf{Proof}. In GR, the charge-density scalar field $\rho$ associated with a coordinate system like ours is defined in terms of the four-current \textit{vector} $J = (J^\mu)$ by the relation
\begin{equation} \label{GR_rho}
    \sqrt{\av{g}} J^t = c\rho.
\end{equation}
Note that this holds for an inertial coordinate system in SR, where we had the 1-form $J = -c\rho \ \rmd(ct)$ (equation~\eqref{J_components} with $\bm{j} = 0$) and thus, raising the index, $J^t = c\rho$. The reason why~\eqref{GR_rho} is the correct generalization of this definition for GR has to do with the property of charge conservation~\eqref{continuity} that must be satisfied by $\rho$ and $\bm{j}$: Equation~\eqref{divergence_cov}, namely $\delta J = 0$, can only be written as a continuity equation (that is, an equation stating that the sum of \textit{partial derivatives} of scalar quantities is zero) if one uses~\eqref{delta_coordinates}:
\begin{equation*}
    0 = \delta J = -\frac{1}{\sqrt{\av{g}}} \partial_\lambda \big( \sqrt{\av{g}} J^\lambda \big).
\end{equation*}
For this reason, the scalars $\sqrt{\av{g}}J^\lambda$ are identified with charge and current densities in GR. In particular, using the metric~\eqref{metric}, the 1-form $J$ is found by lowering the index from $J^t$ in~\eqref{GR_rho}:
\begin{equation} \label{J_new}
    J = -\frac{e^{2\alpha}}{\beta} \frac{c\rho}{\sqrt{\av{g}}} \ \rmd (ct).
\end{equation}
Of course, we have said before that $J=0$, and this remains true on $\mathcal{M}$ according to this formula, due to $\rho = Q\delta_{\bm{0}}$ being a $\delta$ distribution localized at the singularity. However, we now need to take into account the value of $J$ at the singularity because we are about to perform a volume integral of $\Star J$ over a ball centered at $r=0$, so we must use the more precise equation~\eqref{J_new}. With it, the computation of the Hodge dual of $J$, using equation~\eqref{hodge_dt} for $\Star\rmd(ct)$, reveals
\begin{equation} \label{star_J}
    \Star J = c\rho \ \rmd r\wedge\rmd\theta\wedge\rmd\phi
\end{equation}
and thus $\int_{B_R(0)} \Star J = c \, Q$. Now~\eqref{star_J} and the $M$ equation~\eqref{appM}, namely $\rmd M = (4\pi/c)\Star J$, give us the claimed relation between $Q$ and $M$:
\begin{equation*}
    \frac{4\pi}{c}c \, Q = \frac{4\pi}{c} \int_{B_R(0)} \Star J = \int_{B_R(0)} \rmd M = \int_{\partial B_R(0)} M,
\end{equation*}
where Stokes' theorem was used. To further simplify this, we use the expression $M = Qq \, \sin\theta \ \rmd\theta\wedge\rmd\phi$ (as mentioned just before the statement of this proposition), valid globally on the region where $r>0$ and particularly also along $\partial B_R(0)$, to arrive at
\begin{equation*}
    4\pi Q = Qq \int_{\partial B_R(0)} \sin\theta \ \rmd\theta\wedge\rmd\phi = Qq \int_0^{2\pi} \hspace{-0.7em}\int_0^\pi \sin\theta \ \rmd\theta \, \rmd\phi = 4\pi Qq.
\end{equation*}
Hence $q=1$. \qed

Next we expand the relevant equations of the EFE. Using the expressions for the components of $F$, $\delta F$ and $\rmd\delta F$ found above, the invariants of $F$ featuring in formula~\eqref{Tmunu_GR} for $T_{\mu\nu}$ are determined to be
\begin{align*}
    &F_{\zeta\eta}F^{\zeta\eta} = 2F_{tr} F^{tr} = -2Q^2r^{-4}(\chi-1)^2, \\
    &(\delta F)_\zeta (\delta F)^\zeta = (\delta F)_t (\delta F)^t = -Q^2r^{-4}\beta^{-1}(\chi')^2, \\
    &F_{\zeta\eta}(\rmd\delta F)^{\zeta\eta} = 2F^{tr}(\rmd\delta F)_{tr} = 2Q^2r^{-2}e^{-\alpha}(\chi-1)(r^{-2}\beta^{-1}e^\alpha\chi')'.
\end{align*}
Then we quickly calculate
\begin{align*}
    &\phantom{--} \tensor{F}{^\zeta_t}F_{\zeta t} = Q^2r^{-4}\beta^{-1}e^{2\alpha}(\chi-1)^2, & &\phantom{--} \tensor{F}{^\zeta_r}F_{\zeta r} = -Q^2r^{-4}\beta(\chi-1)^2, \\
    &-\frac{g_{tt}}{4} F_{\zeta\eta}F^{\zeta\eta} = -\frac{Q^2}{2}r^{-4}\beta^{-1}e^{2\alpha}(\chi-1)^2, & &-\frac{g_{rr}}{4} F_{\zeta\eta}F^{\zeta\eta} = \frac{Q^2}{2}r^{-4}\beta(\chi-1)^2, \\
    &\phantom{--} \tensor{F}{^\zeta_t}(\rmd\delta F)_{\zeta t} = -Q^2r^{-2}\beta^{-1}e^\alpha(\chi-1)\big( r^{-2}\beta^{-1}e^\alpha\chi' \big)', & &\phantom{--} \tensor{F}{^\zeta_r}(\rmd\delta F)_{\zeta r} = Q^2r^{-2}\beta e^{-\alpha}(\chi-1)\big( r^{-2}\beta^{-1}e^\alpha\chi' \big)', \\
    &- \frac{g_{tt}}{2}F_{\zeta\eta}(\rmd\delta F)^{\zeta\eta} = Q^2r^{-2}\beta^{-1}e^\alpha(\chi-1)\big( r^{-2}\beta^{-1}e^\alpha\chi' \big)', & 
    &- \frac{g_{rr}}{2}F_{\zeta\eta}(\rmd\delta F)^{\zeta\eta} = -Q^2r^{-2}\beta e^{-\alpha}(\chi-1)\big( r^{-2}\beta^{-1}e^\alpha\chi' \big)', \\
    &-(\delta F)_t(\delta F)_t = -Q^2r^{-4}\beta^{-2}e^{2\alpha}(\chi')^2, & &-(\delta F)_r(\delta F)_r = 0, \\
    &\phantom{--} \frac{g_{tt}}{2} (\delta F)_\zeta (\delta F)^\zeta = \frac{Q^2}{2}r^{-4}\beta^{-2}e^{2\alpha}(\chi')^2, & &\phantom{--} \frac{g_{rr}}{2} (\delta F)_\zeta (\delta F)^\zeta = -\frac{Q^2}{2}r^{-4}(\chi')^2,
\end{align*}
with which we put together the tensor components
\begin{equation*}
    T_{\mu\nu} =: T^{\mathrm{(M)}}_{\mu\nu} + \varkappa^{-2}T^{\mathrm{(BLTP)}}_{\mu\nu}, \quad \mu=\nu=t \text{ and } \mu=\nu=r,
\end{equation*}
piece by piece from definition~\eqref{Tmunu_GR} as follows:
\begin{align*}
    4\pi T^{\mathrm{(M)}}_{tt} &= \tensor{F}{^\zeta_t}F_{\zeta t} -\frac{g_{tt}}{4} F_{\zeta\eta}F^{\zeta\eta} = \frac{Q^2}{2}r^{-4}\beta^{-1}e^{2\alpha}(\chi-1)^2, \\
    4\pi T^{\mathrm{(M)}}_{rr} &= \tensor{F}{^\zeta_r}F_{\zeta r} -\frac{g_{rr}}{4} F_{\zeta\eta}F^{\zeta\eta} = -\frac{Q^2}{2}r^{-4}\beta(\chi-1)^2, \\
    4\pi T^{\mathrm{(BLTP)}}_{tt} &= 2\tensor{F}{^\zeta_t}(\rmd\delta F)_{\zeta t} - \frac{g_{tt}}{2}F_{\zeta\eta}(\rmd\delta F)^{\zeta\eta} - (\delta F)_t(\delta F)_t + \frac{g_{tt}}{2} (\delta F)_\zeta (\delta F)^\zeta \\
    &= -Q^2r^{-2}\beta^{-1}e^\alpha(\chi-1)\big( r^{-2}\beta^{-1}e^\alpha\chi' \big)' - \frac{Q^2}{2} r^{-4}\beta^{-2} e^{2\alpha}(\chi')^2, \\
    4\pi T^{\mathrm{(BLTP)}}_{rr} &= 2\tensor{F}{^\zeta_r}(\rmd\delta F)_{\zeta r} - \frac{g_{rr}}{2}F_{\zeta\eta}(\rmd\delta F)^{\zeta\eta} - (\delta F)_r(\delta F)_r + \frac{g_{rr}}{2} (\delta F)_\zeta (\delta F)^\zeta \\
    &= Q^2r^{-2}\beta e^{-\alpha}(\chi-1)\big( r^{-2}\beta^{-1}e^\alpha\chi' \big)' - \frac{Q^2}{2}r^{-4}(\chi')^2.
\end{align*}
This gives
\begin{align}
    T_{tt} &= \frac{Q^2}{8\pi}\left( \frac{e^{2\alpha}(\chi-1)^2}{r^{4}\beta} - 2\varkappa^{-2}\frac{e^\alpha(\chi-1)}{r^2\beta}\frac{\rmd}{\rmd r}\left(\frac{\chi' e^\alpha}{r^2\beta}\right) - \varkappa^{-2} \frac{e^{2\alpha}(\chi')^2}{r^4\beta^2} \right), \label{Tmunu_initially_tt} \\
    T_{rr} &= \frac{Q^2}{8\pi}\left( -\frac{\beta(\chi-1)^2 }{r^{4}} + 2\varkappa^{-2}\frac{\beta e^{-\alpha}(\chi-1)}{r^2}\frac{\rmd}{\rmd r}\left(\frac{\chi' e^\alpha}{r^2\beta}\right) - \varkappa^{-2}\frac{(\chi')^2}{r^4} \right). \label{Tmunu_initially_rr}
\end{align}
Now using~\eqref{Gmunu_initially} and also using the $\chi$ equation~\eqref{app_chi_eq} to simplify the $\rmd/(\rmd r)$ terms on the right sides of~\eqref{Tmunu_initially_tt} and~\eqref{Tmunu_initially_rr}, the EFE~\eqref{appEFE} for indices $(t,t)$ and $(r,r)$ become
\begin{align}
    \frac{e^{2\alpha}}{r\beta^2}\left( \frac{\beta'}{\beta} + \frac{\beta}{r} - \frac{1}{r} \right) &= \frac{GQ^2}{c^4}\left( \frac{e^{2\alpha}(1-\chi^2)}{r^{4}\beta} - \varkappa^{-2} \frac{e^{2\alpha}(\chi')^2}{r^4\beta^2} \right), \label{app1stEFE} \\
    \frac{1}{r}\left( -\frac{\beta'}{\beta} - \frac{\beta}{r} + \frac{1}{r} + 2\alpha' \right) &= \frac{GQ^2}{c^4}\left( -\frac{\beta(1-\chi^2) }{r^{4}} - \varkappa^{-2}\frac{(\chi')^2}{r^4} \right). \label{app2ndEFE}
\end{align}
From here we see that the combination $(r/2)\big( \beta^2e^{-2\alpha}$\eqref{app1stEFE}$+$\eqref{app2ndEFE}$\big)$ produces the $\alpha$ equation as claimed in the E-M-BLTP system~\eqref{embltp}, while~\eqref{app1stEFE} alone yields the $\beta$ equation after trivial simplifications.

It remains to check our previous claim that the other Einstein equations do not contribute any new information to the problem. We check that the $(t,t)$-- and the $(r,r)$-EFE combined with the identity $\nabla^\nu T_{\mu\nu} = 0$ (which holds independently from the EFE; see remarks right after equation~\eqref{Hilbert_Tmunu}) imply all the other $(\mu,\nu)$-EFE:

\begin{proposition} Equations~\eqref{appEFE} for $\mu=\nu=t$ and $\mu=\nu=r$ together imply the same equations for all other choices of the indices $\mu,\nu$.
\end{proposition}

\begin{proof}
    Given the diagonality of the metric and the simple form of the tensor $F$, one quickly checks that $(T_{\mu\nu})$ is diagonal just like $(G_{\mu\nu})$. Thus, only the equations corresponding to $\mu=\nu=\theta$ and $\mu=\nu=\phi$ need to be checked to be superfluous. Moreover, the relation $T_{\phi\phi} = T_{\theta\theta} \sin^2\theta$ is easy to derive from the definition of $(T_{\mu\nu})$, and, given that $(G_{\mu\nu})$ also satisfies it, we see that the $(\phi,\phi)$ equation in~\eqref{appEFE} is superfluous, as it is a multiple of the $(\theta,\theta)$ equation. It is more convenient now to work with the \textit{contravariant} tensors $(G^{\zeta\eta})$ and $(T^{\zeta\eta})$. Given that the metric~\eqref{metric_philambda} is diagonal, these tensors are also diagonal and verify 
    \begin{equation} \label{app_sintheta}
        G^{\phi\phi} = \frac{G^{\theta\theta}}{\sin^2\theta}, \quad T^{\phi\phi} = \frac{T^{\theta\theta}}{\sin^2\theta}.
    \end{equation}
    The divergence freeness of $(T_{\mu\nu})$ is now written in the form $\nabla_\nu T^{\mu\nu} = 0$, which is expanded in components as
    \begin{equation} \label{T_divfree}
        \partial_\nu T^{\mu\nu} + \Gamma^\mu_{\nu\xi} T^{\xi\nu} + \Gamma^\nu_{\nu\xi} T^{\mu\xi} = 0, \quad \mu = t,r,\theta,\phi.
    \end{equation}
    The Einstein tensor is also identically divergence-free, that is, independently from the EFE:
    \begin{equation} \label{G_divfree}
        \partial_\nu G^{\mu\nu} + \Gamma^\mu_{\nu\xi} G^{\xi\nu} + \Gamma^\nu_{\nu\xi} G^{\mu\xi} = 0, \quad \mu = t,r,\theta,\phi.
    \end{equation}
    Multiplying~\eqref{T_divfree} by $8\pi G/c^4$, then subtracting its $\mu=r$ version from that of~\eqref{G_divfree}, we get
    \begin{equation} \label{app_div_GT}
        \frac{\partial}{\partial x^\nu}\left( G^{r\nu} - \frac{8\pi G}{c^4}T^{r\nu} \right) + \Gamma^r_{\nu\xi}\left( G^{\xi\nu} - \frac{8\pi G}{c^4}T^{\xi\nu} \right) + \Gamma^\nu_{\nu\xi}\left( G^{r\xi} - \frac{8\pi G}{c^4}T^{r\xi} \right) = 0.
    \end{equation}
    In the first term between parentheses here, only $\nu = r$ can contribute anything, because $(G^{\zeta\eta})$ and $(T^{\zeta\eta})$ are diagonal, but the $\nu = r$ case also vanishes due to~\eqref{appEFE} having been assumed for $(\mu,\nu)=(r,r)$. For the same reasons (also using the $(\mu,\nu)=(t,t)$ equation), the third term of~\eqref{app_div_GT} vanishes, while in the second one only the values $\xi = \nu = \theta$ and $\xi = \nu = \phi$ can make nonzero contributions. Thus~\eqref{app_div_GT} reduces to
    \begin{multline*}
        0 = \Gamma^r_{\theta\theta}\left( G^{\theta\theta} - \frac{8\pi G}{c^4}T^{\theta\theta} \right) + \Gamma^r_{\phi\phi}\left( G^{\phi\phi} - \frac{8\pi G}{c^4}T^{\phi\phi} \right) = \left( \Gamma^r_{\theta\theta} + \frac{\Gamma^r_{\phi\phi}}{\sin^2\theta} \right)\left( G^{\theta\theta} - \frac{8\pi G}{c^4}T^{\theta\theta} \right) \\
       = -2r\beta^{-1} \left( G^{\theta\theta} - \frac{8\pi G}{c^4}T^{\theta\theta} \right),
    \end{multline*}
    where~\eqref{app_sintheta} was used and the Christoffel symbols $\Gamma^r_{\theta\theta} = -r\beta^{-1}$ and $\Gamma^r_{\phi\phi} = -r\beta^{-1}\sin^2\theta$ were calculated from their definition~\eqref{Ricci_Christoffel}. Since $r\beta^{-1}$ does not vanish, the $(\theta,\theta)$-EFE is established.
\end{proof}

\setcounter{subsection}{3}
\subsection{The field-energy formula} \label{subsec_app_energy}

Here we show a derivation of the integrand in~\eqref{energy_GR} for the electric-field energy $\mathcal{E}$ written as an expression of the unknowns $\alpha,\beta,\chi$. For the 1-form $P$ defined in~\eqref{noether_P} in terms of the vector field $X = \partial_{t}$, we calculate
\begin{equation*}
    P_\mu = T_{\mu\nu}X^\nu = T_{\mu\nu}\delta^\nu_t = T_{\mu t} = T_{tt} \delta_\mu^t,
\end{equation*}
where we used the fact that $(T_{\mu\nu})$ is diagonal. This gives
\begin{equation*}
    P = P_t \ \rmd(ct) \quad\text{for } P_t = T_{tt}.
\end{equation*}
The Hodge dual $\Star P$ is calculated with the help of~\eqref{hodge_dt}:
\begin{equation} \label{star_P_aux}
    \Star P = T_{tt} \big( 
    g^{tt}\sqrt{\av{g}} \ \rmd r\wedge\rmd\theta\wedge\rmd\phi \big) = -T_{tt} r^2\beta e^{-\alpha} \sin\theta \ \rmd r\wedge\rmd\theta\wedge\rmd\phi.
\end{equation}
Plugging this into~\eqref{energy_GR} and using the expression for $T_{tt}$ that can be read off~\eqref{app1stEFE}, we find
\begin{align*}
    \mathcal{E} &= -\int_{\{t=t_0\}} \Star P \\
    &= \frac{Q^2}{8\pi} \int_0^{2\pi} \int_0^\pi \int_0^\infty r^2\beta e^{-\alpha}\sin\theta \left(\frac{e^{2\alpha}(1 - \chi^2)}{r^4\beta} - \varkappa^{-2}\frac{e^{2\alpha}(\chi')^2}{r^4\beta^2} \right) \ \rmd r \, \rmd\theta \, \rmd\phi \\
    &= \frac{Q^2}{2} \int_0^\infty \frac{e^{\alpha}}{r^2} \left( 1 - \chi^2 - \varkappa^{-2}\frac{(\chi')^2}{\beta} \right) \ \rmd r,
\end{align*}
which is as claimed in~\eqref{energy_GR}. To justify why a minus sign appears in front of the integral in the definition~\eqref{energy_GR} of $\mathcal{E}$, we have to check that this result generalizes that of SR, namely~\eqref{energy_SR}. This is immediate: Plugging~\eqref{star_P_aux} directly into the energy integral without explicitly writing out $T_{tt}$ gives
\begin{align*}
    \mathcal{E} &= \int_{\{t=t_0\}} T_{tt} r^2\beta e^{-\alpha} \sin\theta \ \rmd r\wedge \rmd\theta\wedge\rmd\phi \\
    &= \int_{\{t=t_0\}} (g_{tt})^2 T^{tt} r^2\beta e^{-\alpha} \sin\theta \ \rmd r\wedge \rmd\theta\wedge\rmd\phi \\
    &= \int_0^{2\pi}\int_0^\pi\int_0^\infty \frac{e^{3\alpha}}{\beta} T^{tt} r^2\sin\theta \ \rmd r\,\rmd\theta\,\rmd\phi \\
    &= \int_{\bbR^3} \frac{e^{3\alpha}}{\beta} T^{tt} \ \rmd^3\bm{x},
\end{align*}
which reduces to~\eqref{energy_SR} (since $\frake = T^{tt}$) when the Minkowski metric is used, that is, when $\alpha \equiv 0$ and $\beta \equiv 1$.


} 

\end{document}